\documentclass[a4paper,11pt]{article}
\usepackage[utf8]{inputenc}
\usepackage{geometry}
\usepackage{fullpage}

\usepackage{graphicx}
\usepackage[export]{adjustbox}
\usepackage{multirow}
\usepackage{rotating}
\usepackage{pdflscape}
\usepackage{afterpage}
\usepackage{booktabs}
\usepackage{caption}
\usepackage{comment}

\usepackage[colorlinks=true, allcolors=blue]{hyperref}
\usepackage[shortlabels]{enumitem}
\usepackage{amsthm,amsmath,bm,bbm}
\usepackage{amssymb,mathtools}

\usepackage{authblk}

\usepackage{listings}

\allowdisplaybreaks[1]

\usepackage[multiple]{footmisc}

\usepackage{wrapfig}
\usepackage{enumitem}
\usepackage{cancel}
\usepackage{nccmath}
\usepackage{xcolor}
\usepackage{soul}

\usepackage{minitoc}

\usepackage{pdflscape}

\makeatletter
\newcommand{\pushright}[1]{\ifmeasuring@#1\else\omit\hfill$\displaystyle#1$\fi\ignorespaces}
\newcommand{\pushleft}[1]{\ifmeasuring@#1\else\omit$\displaystyle#1$\hfill\fi\ignorespaces}

\newcommand*\bigcdot{\mathpalette\bigcdot@{.7}}
\newcommand*\bigcdot@[2]{\mathbin{\vcenter{\hbox{\scalebox{#2}{$\m@th#1\bullet$}}}}}
\makeatother

\newtheorem{theorem}{Proposition}
\newtheorem{corollary}{Corollary}
\newtheorem{lemma}{Lemma}

\theoremstyle{remark}
\newtheorem{remark}{Remark}

\newtheorem{theoremalt}{{\normalfont\bfseries Proposition}}[theorem]
\newenvironment{theoremb}[1]{
  \renewcommand\thetheoremalt{{#1}}
  \theoremalt
}{\endtheoremalt}

\renewenvironment{proof}[1][\proofname]{{\bfseries #1.}}{\qed}

\def\independent{\perp\!\!\!\perp}
\def\var{\mathrm{var}}

\def\E{\mathrm{E}}
\def\P{\mathrm{P}}
\def\I{\mathrm{I}}

\def\H{\mathcal{H}}
\def\T{\mathcal{T}}

\def\pto{\overset{p}{\longrightarrow}}
\def\dto{\overset{d}{\longrightarrow}}

\allowdisplaybreaks

\usepackage[superscript]{cite}

\usepackage{tikz}
\usetikzlibrary{shapes.geometric, arrows, shapes, positioning}

\title{Sensitivity analysis for principal ignorability violation in estimating complier and noncomplier average causal effects}

\author[1]{Trang Quynh Nguyen}
\author[1]{Elizabeth A. Stuart}
\author[2]{Daniel O. Scharfstein}
\author[1]{Elizabeth L. Ogburn}
\affil[1]{Johns Hopkins Bloomberg School of Public Health, MD, USA}
\affil[2]{University of Utah School of Medicine, UT, USA}

\begin{document}

\maketitle


\begin{abstract}
    \normalsize\noindent An important strategy for identifying \textit{principal causal effects} (popular estimands in settings with noncompliance) is to invoke the principal ignorability (PI) assumption.
    As PI is untestable, it is important to gauge how sensitive effect estimates are to its violation. We focus on this task for the common one-sided noncompliance setting where 
    there are two principal strata, compliers and noncompliers. Under PI, compliers and noncompliers share the same outcome-mean-given-covariates function under the control condition. For sensitivity analysis, we allow this function to differ between compliers and noncompliers in several ways, indexed by an odds ratio, a generalized odds ratio, a mean ratio, or a standardized mean difference sensitivity parameter. We tailor sensitivity analysis techniques (with any sensitivity parameter choice) to several types of PI-based main analysis methods, including outcome regression, influence function (IF) based and weighting methods. We discuss range selection for the sensitivity parameter.
    We illustrate the sensitivity analyses with several outcome types from the JOBS II study.
    This application estimates nuisance functions parametrically -- for simplicity and accessibility. In addition, we establish rate conditions on nonparametric nuisance estimation for IF-based estimators to be asymptotically normal -- with a view to inform nonparametric inference.
    
    ~
    
\end{abstract}

\doparttoc 
\faketableofcontents 

\part{} 
\parttoc 

\section{Introduction}

The study of causal effects of a treatment is often complicated by noncompliance. The principal stratification framework \cite{Frangakis2002} defines types (principal strata) of study participants based on their \textit{potential} compliance to treatment conditions. In the one-sided noncompliance setting where individuals in the control condition do not have access to the active treatment, there are two principal strata: \textit{compliers}, who would take the treatment if offered, and \textit{noncompliers}, who would not. In the two-sided noncompliance setting where all individuals (assigned to either treatment or control) can access the treatment, there are four principal strata, often known as \textit{compliers}, \textit{always-takers}, \textit{never-takers}, and \textit{defiers}. 
\textit{Principal causal effects} are effects of \textit{treatment assignment} within each stratum, $\E[Y_1-Y_0\mid C]$, where $Y_1$ and $Y_0$ are potential outcomes \cite{Rubin1974} under assignment of active treatment and of control, respectively, and $C$ denotes principal stratum. Of common interest is the complier average causal effect (CACE), but other principal causal effects may also be of interest \cite{Rubin2006,Griffin2008}.

This paper focuses on one-sided noncompliance, which is common in studies where the treatment is designed and implemented by the study and is not otherwise available, e.g., job search training for unemployed workers \cite{Vinokur1995}, volunteering program for the elderly \cite{gruenewald2016BaltimoreExperienceCorps}, or weight management for people with mental illness \cite{daumit2013BehavioralWeightlossIntervention}. We will briefly comment on the two-sided non-compliance case in the Discussion section.

The challenge in identifying principal causal effects is that principal stratum membership $C$ is only partially observed; with one-sided noncompliance $C$ is not observed in the control condition. Effect identification thus requires untestable assumptions. One such assumption is \textit{exclusion restriction} \cite{Angrist1995} (ER), which posits that treatment assignment does not affect the outcome other than through its effect on treatment received. This means there is no effect on noncompliers, and effects on compliers explain the full effect of treatment assignment. ER is not suitable if treatment receipt is not strictly binary, i.e., noncompliers are exposed to some active ingredients in the treatment arm \cite{marshall2016CoarseningBiasHow,andresen2021InstrumentbasedEstimationBinarised}. This case may arise when an intervention includes several components, and only a major one is used to define compliance. It may also arise due to dichotomization, e.g., only people who attend more than a certain number of treatment sessions are classified as compliers \cite{gruenewald2016BaltimoreExperienceCorps}.
ER may also not hold if there are compensating behaviors or psychological effects due to being assigned to one condition as opposed to the other \cite{Feller2017}.

Another identification strategy does not restrict the noncomplier effect to zero, but instead invokes the \textit{principal ignorability} (PI) assumption \cite{Jo2009,Feller2017,Ding2017}. This assumption  posits that, conditional on a set of pre-treatment-assignment covariates $X$, the potential outcome under control $Y_0$ is independent (or mean-independent) of principal stratum $C$, i.e., compliers and noncompliers share the same conditional $Y_0$ distribution (or mean function). PI may be appealing for studies with rich baseline covariate data. As randomized trials and cohort studies tend to collect a lot of covariate data, one might hope that the covariates account for a substantial part of the dependence between $Y_0$ and $C$. On the other hand, most studies are not designed with noncompliance in mind, and thus not much attention is paid to measuring covariates that predict compliance type to render $C$ and $Y_0$ independent, which means PI may be violated.

\subsection{Our contribution}

In this paper we focus on the PI assumption. Specifically, we develop methods to evaluate the robustness of the estimated principal causal effects to violation of PI, in the one-sided noncompliance setting. We introduce several sensitivity parameterizations representing how (within levels of $X$) the mean of $Y_0$ differs between compliers and noncompliers. These are indexed by an odds ratio, generalized odds ratio, mean ratio, or standardized mean difference, suitable for use with different outcome types. In addition, we tailor sensitivity analysis techniques for pairing with a range of estimation methods that may be used for the PI-based main analysis, including weighting, outcome regression and influence function based estimation.

We illustrate the proposed sensivity analysis methods using the JOBS II Intervention Study \cite{Vinokur1995}, where unemployed workers were randomized to receive either a week-long training program to promote mental health and provide job search skills (\textit{treatment}) or a booklet with job search tips (\textit{control}). Just over half of those randomized to treatment actually attended the training, resulting in a setting with compliers and noncompliers. JOBS II has been used by authors investigating different aspects of principal stratification, e.g., identification and estimation under PI \cite{Jo2009,Feller2017}, alternative identification assumptions \cite{Jiang2021auxiliary}, bias due to failed assumptions \cite{Stuart2015}, and noncompliance combined with outcome missingness \cite{Jo2011}.
For our purpose, JOBS II is an interesting example for two reasons: (i) the study paid attention to the issue of noncompliance and collected baseline data on workers' motivation to participate in a hypothetical training program on job search skills, making this a prime case for invoking PI; and (ii) the study collected outcomes of several types (binary, continuous and bounded) to which the methods we propose are relevant.

\subsection{Related work}\label{sec:relatedwork}

To our knowledge, two methods have been proposed to assess sensitivity of effect estimates to PI violation. The method used in Ding and Lu (2017) \cite{Ding2017} is the closest to, and inspired, our work. In the one-sided noncompliance context, this method allows the mean of $Y_0$ given $X$ to differ between compliers and noncompliers by a ratio that serves as the sensitivity parameter, and estimates effects under each value of the sensitivity parameter by modifying a PI-based weighting estimator. The application was with a binary outcome, flu-related hospitalization. A drawback is that with a binary outcome this mean ratio parameter may yield predictions greater than~1. This motivated our expansion of the range of sensitivity parameterizations to accommodate different outcome types. Also, we consider sensitivity analysis techniques pairing with different types of PI-based estimators, not just the weighting estimator. The second sensitivity analysis method is that of Wang et al. (2023)
\cite{wang2023SensitivityAnalysesPrincipal} for survival outcomes, which imputes unobserved $C$ and $Y_0$ under a parametric model containing a hazard ratio sensitivity parameter. This work differs from our approach in that it relies on this parametric model for identification, whereas we make explicit the assumption required for identification and then use modeling only for estimation.
We also avoid refitting models for every value of the sensitivity parameter.

There are methods to assess sensitivity of principal causal effect estimates to violation of other assumptions: treatment assignment ignorability \cite{Schwartz2012,Mercatanti2017} and ER \cite{baiocchi2014InstrumentalVariableMethods}. These are not our focus.

To discuss sensitivity analysis, we will need to start with a description of PI-based estimation. While PI-based methods have been discussed in the literature, it has been in settings that are somewhat different, e.g., randomized treatment assignment \cite{Stuart2015,Feller2017,Ding2017} (which we do not require), a qualitatively different assumption \cite{Jiang2021auxiliary}, or two-sided rather than one-sided noncompliance \cite{jiang2022MultiplyRobustEstimationa}. The PI-based estimators we list in this paper share certain features (e.g., principal score weighting) with these earlier works, but are based on results for the current setting.

\smallskip

The paper proceeds as follows. Section \ref{sec:setting} presents the setting, the estimands, and identification under PI. Section \ref{sec:pi-estimators} introduces three types of PI-based estimators to be handled with different sensitivity analysis techniques. 
Sections \ref{sec:ratio-params} and \ref{sec:diff-param} present sensitivity analysis using ratio-type and difference-type sensitivity parameters, respectively, and address each of the three estimator types. Section~\ref{sec:technical-rest} covers topics relevant to the sensitivity analyses. Section \ref{sec:illustration} analyzes JOBS II data. Section \ref{sec:discussion} closes with a discussion. Proofs are provided in the Appendix. Code is provided in the R-package PIsens available at \url{https://github.com/trangnguyen74/PIsens}.

\section{Setting, estimands, and PI-based identification}\label{sec:setting}

\subsection{Setting, estimands, and standard assumptions}

Let $Z$ denote treatment assignment (1 for treatment, 0 for control), $Y$ denote the observed outcome, $Y_z$ the potential outcome had treatment $z$ been \textit{assigned} ($z=0,1$), and $X$ denote baseline covariates.
Let $S$ be a binary variable indicating whether the person actually receives the treatment ($S=1$) or not ($S=0$).
(More generally, $S$ can be any post-treatment variable of interest \cite{gruenewald2016BaltimoreExperienceCorps,mcconnell2008TruncationbyDeathProblemWhat,Griffin2008}.)
The principal stratification framework \cite{Frangakis2002} defines subpopulations (aka \textit{principal strata}, denoted by $C$) based on $S_1$ and $S_0$, the potential values of $S$ under assignment to treatment and to control. In the one-sided compliance setting, $S_0=0$, so only $S_1$ matters. Hence $C$ coincides with $S_1$ and there are two principal strata: \textit{compliers} $(C=1)$ who would and \textit{noncompliers} $(C=0)$ who would not take the treatment, if offered the treatment.
The ``full'' data for an individual are $(X,Z,C,Y_1,Y_0$); the observed data are $O:=(X,Z,S,Y)$.  Assume that we observe $n$ i.i.d. copies of $O$.

Here we are interested in the complier and noncomplier average causal effects (CACE and NACE). As the PI identification strategy is symmetric with respect to these two effects (and so are the sensitivity assumptions we consider), we focus on the generic estimand
\begin{align*}
    \Delta_c:=\E[Y_1-Y_0\mid C=c]=\E[Y_1\mid C=c]-\E[Y_0\mid C=c],
\end{align*}
where $c=1$ gives the CACE and $c=0$ gives the NACE.

Throughout we assume the usual causal inference assumptions:
\begin{center}
\begin{tabular}{ll}
    A0 (consistency): & $Y=ZY_1+(1-Z) Y_0$,~~~$S=ZC$,
    \\[.5em]
    A1 (treatment assignment ignorability): & $Z\independent (C,Y_1,Y_0)\mid X$,
    \\[.5em]
    A2 (treatment assignment positivity): & $0<\P(Z=1\mid X)<1$.
\end{tabular}
\end{center}

\noindent
Under A0, we write $O=(X,Z,ZC,Y)$ to simplify presentation.

As several expressions appear repeatedly in the paper, we will use the shorthand notation
\begin{align*}
    \tau_{zc}&:=\E[Y_z\mid C=c],
    \\
    \mu_{zc}(X)&:=\E[Y_z\mid X,C=c],
    \\
    \mu_0(X)&:=\E[Y\mid X,Z=0],
    \\
    e(X,Z)&:=\P(Z\mid X),
    \\
    \pi_c(X)&:=\P(C=c\mid X),
\end{align*}
for $z=0,1$, $c=0,1$. Here $\Delta_c=\tau_{1c}-\tau_{0c}$. Note the difference between $\mu_{zc}(X)$ which is the conditional mean of a \textit{potential} outcome within a principal stratum and $\mu_0(X)$ which concerns the \textit{observed} outcome in the control condition and does not condition on principal stratum. $e(X,1)$ is the \textit{propensity score}. $\pi_c(X)$ is the probability of being in stratum $c$ given covariate values, which we also refer to as the \textit{principal score}, following the literature \cite{Jo2009,Stuart2015,Feller2017,Ding2017}.

\smallskip

Proofs of all results in this section are provided in Appendix \ref{appendix:prelim}.

\subsection{The identification challenge and the PI assumption}\label{sec:main}

Identification of $\Delta_c=\tau_{1c}-\tau_{0c}$ amounts to identification of $\tau_{1c}$ and $\tau_{0c}$.
The challenge is that while A0-A2 identify $\tau_{1c}$, they are not sufficient to identify $\tau_{0c}$.
To see this, we start with the identity below.

\begin{lemma}\label{lm:starting-point}
    \begin{align}
    \overbrace{\E[Y_z\mid C=c]}^{\textstyle=:\tau_{zc}}&=\frac{\E\big\{\overbrace{\P(C=c\mid X)}^{\textstyle=:\pi_c(X)}\,\overbrace{\E[Y_z\mid X,C=c]}^{\textstyle=:\mu_{zc}(X)}\big\}}{\E[\P(C=c\mid X)]}.\label{estimand:tau.zc}
\end{align}
\end{lemma}

\noindent(To simplify presentation, it is left implicit that $\mu_{zc}(X)$ is only defined where $\pi_c(X)>0$.)

\smallskip

Lemma~\ref{lm:starting-point} says that $\tau_{zc}$ is equal to the weighted average of the stratum-specific potential outcome mean $\mu_{zc}(X)$ where the weight is proportional to the principal score $\pi_c(X)$.
This means $\tau_{zc}$ can be identified via identification of $\pi_c(X)$ and $\mu_{zc}(X)$, which we address next.

\begin{theorem}[Results without PI]\label{thm:id-tau.1c}
Under assumptions A0-A2,
\begin{align}
    \pi_c(X)
    &=\P(C=c\mid X,Z=1),\label{id:pi.c}
    \\
    \mu_{zc}(X)
    &=\E[Y\mid X,Z=z,C=c],\label{id:mu.1c}
    \\
    \tau_{1c}
    &=\frac{\E[\pi_c(X)\mu_{1c}(X)]}{\E[\pi_c(X)]}
    =\frac{\E\left[\frac{Z}{e(X,Z)}\I(C=c)Y\right]}{\E\left[\frac{Z}{e(X,Z)}\I(C=c)\right]},\label{id:tau.1c}
    \\
    \pi_1(X)&\mu_{01}(X)+\pi_0(X)\mu_{00}(X)=\overbrace{\E[Y\mid X,Z=0]}^{\textstyle=:\mu_0(X)}.\label{eq:mixture-mean}
    \end{align}
\end{theorem}

Proposition \ref{thm:id-tau.1c} shows that A0-A2 identify $\pi_c(X)$ and $\mu_{1c}(X)$, but not $\mu_{0c}(X)$. (The RHS of (\ref{id:mu.1c}) conditions on $C$, which is not observed for $Z=0$.) Hence $\tau_{1c}$ is identified,
but $\tau_{0c}$ is not, so $\Delta_c$ is not. 

The problem here is nonidentifiability of the stratum-specific conditional $Y_0$ mean functions $\mu_{0c}(X)$. These two functions, $\mu_{01}(X)$ for compliers (where $c=1$) and $\mu_{00}(X)$ for noncompliers (where $c=0$), are tied together as two unknowns in one equation, (\ref{eq:mixture-mean}), which we will call the \textit{mixture equation}.  
To identify them, some additional assumption is needed.

One such assumption is PI, which we state here as a conditional mean independence:
\begin{center}
\begin{tabular}{ll}
    A3 (principal ignorability): & $\overbrace{\E[Y_0\mid X,C=1]}^{\textstyle=:\mu_{01}(X)}=\overbrace{\E[Y_0\mid X,C=0]}^{\textstyle=:\mu_{00}(X)}$.
\end{tabular}
\end{center}
PI is also sometimes stated as $C\independent Y_0\mid X$ (which implies~A3). This version is more intuitive: it is satisfied if $X$ captures all common causes of $C$ and $Y_0$ \cite{Ding2017}.
Like other authors, we assume that A3 and A1 involve the same set of covariates; this can be relaxed.

A3 combined with (\ref{eq:mixture-mean}) 
solves the identification problem.

\begin{theorem}[PI based identification]\label{thm:pi-id:tau.0c}
Under assumptions A0-A3,
\begin{align}
    \mu_{0c}(X)
    &=\mu_0(X),\label{pi-id:mu.0c}
    \\
    \tau_{0c}
    &=\frac{\E[\pi_c(X)\mu_0(X)]}{\E[\pi_c(X)]}
    =\frac{\E\left[\frac{Z}{e(X,Z)}\I(C=c)\mu_0(X)\right]}{\E\left[\frac{Z}{e(X,Z)}\I(C=c)\right]}
    =\frac{\E\left[\frac{1-Z}{e(X,Z)}\pi_c(X)Y\right]}{\E\left[\frac{1-Z}{e(X,Z)}\pi_c(X)\right]}.\label{pi-id:tau.0c}
\end{align}
\end{theorem}

We will refer to the observed data functionals in Proposition \ref{thm:pi-id:tau.0c} that identify $\mu_{0c}(X)$ and $\tau_{0c}$ as $\mu_{0c}^\text{PI}(X)$ and $\tau_{0c}^\text{PI}$,
and the corresponding result for $\Delta_c$ (i.e., $\tau_{1c}-\tau_{0c}^\text{PI}$) as $\Delta_c^\text{PI}$. 

\smallskip

\begin{remark}[Sufficient PI version]
A3 involves $Y_0$ but not $Y_1$. Feller et al. (2017)\cite{Feller2017} call this assumption \textit{weak} PI to differentiate it from a different assumption (\textit{strong} PI) that involves both potential outcomes, $C\independent Y_z\mid X$ for $z=0,1$. While these labels suggest a difference in degree, these assumptions are qualitatively different. Strong PI implies that conditional on $X$, the average causal effect is constant across principal strata, which is generally not desired \cite{Feller2017}.
As A3 is sufficient (and strong PI is unnecessary), we simply refer to A3 as PI.
\end{remark}

\smallskip

PI is untestable. The sensitivity analyses in Sections \ref{sec:ratio-params} and \ref{sec:diff-param} will each replace PI with an alternative assumption (\textit{sensitivity assumption}) indexed by a \textit{sensitivity parameter} representing deviation from PI. Such an assumption obtains alternative identification results for $\mu_{0c}(X)$ and $\Delta_c$. The sensitivity analysis then shows, for a plausible range of the sensitivity parameter, how effect estimates depart from those obtained in a PI-based analysis.

\section{Three types of PI-based estimators from the lens of sensitivity analysis}\label{sec:pi-estimators}

It is desirable to develop sensitivity analysis methods that are simple modifications of PI-based methods. With this in mind, in this section we group estimators of $\Delta_c^\text{PI}$ into three types (each with a few example estimators), which we anticipate can be adapted for sensitivity analysis using different techniques (in subsequent sections). This grouping may be useful generally, say, where it is desirable to use a different sensitivity assumption not covered in this paper.

With three estimator types, the presentation from here through Section \ref{sec:diff-param} is slightly complex. Readers who are mainly looking to add a sensitivity analysis to an already conducted or planned PI-based analysis only need to focus on the type of their estimator and can ignore the others.

Proofs of results in this section are provided in Appendix \ref{appendix:pi-estimators}.

\subsection{Type A ($\approx$ outcome regression estimators)}\label{sec:typeA}

As PI-based analysis relies on the identification result $\mu_{0c}^\text{PI}(X)=\mu_0(X)$, an obvious sensitivity analysis technique (applicable to any PI-based method that involves estimating $\mu_0(X)$) is to replace $\mu_0(X)$ with the alternative formula for $\mu_{0c}(X)$ identified under the sensitivity assumption.
We aim to use this technique with type A (roughly \textit{outcome regression}) estimators. 

To be precise, type A estimators involve estimating $\mu_0(X)$ 
in order to first estimate the principal causal effect conditional on covariates (which under PI is $\mu_{1c}(X)-\mu_0(X)$) or a proxy for it, and then aggregate these conditional effects to estimate the average principal causal effect $\Delta_c^\text{PI}$. Examples include 
the principal-score-weighted outcome-regression estimator (aka the plug-in estimator) (\ref{pi-est:id2}) and the propensity-score-weighted outcome-regression estimator (\ref{pi-est:id3}):
\begin{align}
    \hat\Delta_{c,\pi\mu}^\text{PI}
    &:=\frac{\sum_{i=1}^n\hat\pi_c(X_i)[\hat\mu_{1c}(X_i)-{\color{red}\hat\mu_0(X_i)}]}{\sum_{i=1}^n\hat\pi_c(X_i)}\label{pi-est:id2},
    \\
    \hat\Delta_{c,e\mu}^\text{PI}
    &:=\frac{\sum_{i=1}^n\frac{Z_i\I(C_i=c)}{\hat e(X_i,Z_i)}[Y_i-{\color{red}\hat\mu_0(X_i)}]}{\sum_{i=1}^n\frac{Z_i\I(C_i=c)}{\hat e(X_i,Z_i)}},\label{pi-est:id3}
\end{align}
where the hat notation indicates an estimated function. These are justified by the $\tau_{1c}$ formulae in (\ref{id:tau.1c}) and the first two $\tau_{0c}^\text{PI}$ formulae in (\ref{pi-id:tau.0c}).
Also included in type A is a multiply robust outcome regression estimator, $\hat\Delta_{c,\textsc{ms}}^\text{PI}$, which we will present after explaining type B estimators.

For each estimator here we put in \textcolor{red}{red} the component to be replaced in sensitivity analysis.

\subsection{Type B ($\approx$ influence function based estimators)}\label{sec:typeB}

Type B estimators are a subset of estimators constructed based on the nonparametric influence function (IF) of $\Delta_c^\text{PI}$ (hence the rough label \textit{IF-based estimators}, although not all IF-based estimators belong in type B). To define this type precisely, let
\begin{align*}
    \nu_{zc}
    &:=\E[\pi_c(X)\mu_{zc}(X)],
    \\
    \pi_c
    &:=\E[\pi_c(X)],
    \\
    \nu_{0c}^\text{PI}
    &:=\E[\pi_c(X)\mu_0(X)].
\end{align*}
In this notation, $\Delta_c=(\nu_{1c}-\nu_{0c})/\pi_c$ and $\Delta_c^\text{PI}=(\nu_{1c}-\nu_{0c}^\text{PI})/\pi_c$.
A type~B estimator of $\Delta_c^\text{PI}$ is one that can be expressed as a combination of IF-based estimators of $\nu_{1c}$, $\nu_{0c}^\text{PI}$ and $\pi_c$. The sensitivity analysis technique will be to replace the  $\nu_{0c}^\text{PI}$ component with an IF-based estimator of $\nu_{0c}$ under the sensitivity assumption.
To obtain these estimators, we derive the relevant IFs.

\begin{theorem}[IFs for PI-based analysis]\label{thm:ifs-pi}
The IFs of $\pi_c$, $\nu_{1c}$, $\nu_{0c}^\textup{PI}$, and $\Delta_c^\textup{PI}$ are
\begin{align}
    \varphi_{\pi_c}(O)
    &=\frac{Z}{e(X,Z)}[\I(C=c)-\pi_c(X)]+\pi_c(X)-\pi_c,\label{if:pi.c}
    \\
    \varphi_{\nu_{1c}}(O)
    &=\frac{Z}{e(X,Z)}\I(C=c)[Y-\mu_{1c}(X)]+\frac{Z}{e(X,Z)}\mu_{1c}(X)[\I(C=c)-\pi_c(X)]+\pi_c(X)\mu_{1c}(X)-\nu_{1c},\label{if:nu.1c}
    \\
    \varphi_{\nu_{0c}^\textup{PI}}(O)
    &=\frac{1-Z}{e(X,Z)}\pi_c(X)[Y-\mu_0(X)]+\frac{Z}{e(X,Z)}\mu_0(X)[\I(C=c)-\pi_c(X)]+\pi_c(X)\mu_0(X)-\nu_{0c}^\textup{PI},\label{pi-if:nu.0c}
    \\
    \varphi_{\Delta_c^\textup{PI}}(O)
    &=\frac{1}{\pi_c}\left\{[\varphi_{\nu_{1c}}(O)+\nu_{1c}]-[\varphi_{\nu_{0c}^\textup{PI}}(O)+\nu_{0c}^\textup{PI}]-\Delta_c^\textup{PI}[\varphi_{\pi_c}(O)+\pi_c]\right\}.\label{pi-if:Delta.c}
\end{align}
\end{theorem}

The estimator that uses the IF of $\Delta_c^\text{PI}$ (with estimated nuisances) as the estimating function is a type B estimator. This is because due to (\ref{pi-if:Delta.c}), this estimator has the form
\begin{align}
    \hat\Delta_{c,\textsc{if}}^\text{PI}=\frac{\hat\nu_{1c,\textsc{if}}-{\color{red}\hat\nu_{0c,\textsc{if}}^\text{PI}}}{\hat\pi_{c,\textsc{if}}},\label{est:pi-if}
\end{align}
where (with $\P_n$ representing sample average)
\vspace{-.5em}
\begin{align*}
    \hat\nu_{1c,\textsc{if}}
    &:=\P_n\left\{\frac{Z}{\hat e(X,Z)}\I(C=c)[Y-\hat\mu_{1c}(X)]+\frac{Z}{\hat e(X,Z)}\hat\mu_{1c}(X)[\I(C=c)-\hat\pi_c(X)]+\hat\pi_c(X)\hat\mu_{1c}(X)\right\},
    \\
    {\color{red}\hat\nu_{0c,\textsc{if}}^\text{PI}}
    &:=\P_n\left\{\frac{1-Z}{\hat e(X,Z)}\hat\pi_c(X)[Y-\hat\mu_0(X)]+\frac{Z}{\hat e(X,Z)}\hat\mu_0(X)[\I(C=c)-\hat\pi_c(X)]+\hat\pi_c(X)\hat\mu_0(X)\right\},
    \\
    \hat\pi_{c,\textsc{if}}
    &:=\P_n\left\{\frac{Z}{\hat e(X,Z)}[\I(C=c)-\hat\pi_c(X)]+\hat\pi_c(X)\right\}
\end{align*}
are IF-based estimators of $\nu_{1c}$, $\nu_{0c}^\text{PI}$, $\pi_c$.

Another type B estimator is the H\'ajek-type \cite{hajek1971comment} estimator,
\begin{align}
    \hat\Delta_{c,\textsc{ifh}}^\text{PI}=\frac{\hat\nu_{1c,\textsc{ifh}}-{\color{red}\hat\nu_{0c,\textsc{ifh}}^\text{PI}}}{\hat\pi_{c,\textsc{ifh}}},\label{est:pi-ifh}
\end{align}
where $\hat\nu_{1c,\textsc{ifh}}$, $\color{red}{\hat\nu_{0c,\textsc{ifh}}^\text{PI}}$, $\hat\pi_{c,\textsc{ifh}}$ are a modified version of $\hat\nu_{1c,\textsc{if}}$, $\hat\nu_{0c,\textsc{if}}^\text{PI}$, $\hat\pi_{c,\textsc{if}}$, 
replacing $\frac{Z}{\hat e(X,Z)}$ with $\frac{Z}{\hat e(X,Z)}\big/\P_n[\frac{Z}{\hat e(X,Z)}]$ and $\frac{1-Z}{\hat e(X,Z)}$ with $\frac{1-Z}{\hat e(X,Z)}\big/\P_n[\frac{1-Z}{\hat e(X,Z)}]$.
(We call this modification \textit{H\'ajek-ization}.)

$\hat\Delta_{c,\textsc{if}}^\text{PI}$ and $\hat\Delta_{c,\textsc{ifh}}^\text{PI}$ are multiply robust (see Proposition \ref{thm:pi-multiplyrobust} below). $\hat\Delta_{c,\textsc{ifh}}^\text{PI}$ is range-preserving.

\paragraph{Circling back to type A.}

We now present the multiply robust \textit{outcome regression} estimator $\hat\Delta_{c,\textsc{ms}}^\text{PI}$ mentioned earlier. This is a multi-step estimator (the \textsc{ms} subscript is for ``multi-step'') that is based on expressing the IF of $\Delta_c^\text{PI}$ as a sum of three terms:
\vspace{-.5em}
\begin{align}
    \varphi_{\Delta_c^\text{PI}}(O)
    &=\frac{1}{\pi_c}\Big\{\overbrace{\frac{Z}{e(X,Z)}\I(C=c)[Y-\mu_{1c}(X)]}^{(*)}
    ~-~\overbrace{\frac{1-Z}{e(X,Z)}\pi_c(X)[Y-\mu_0(X)]}^{(**)}~+\nonumber
    \\
    &~~~~~~~~~~\underbrace{\Big[\frac{Z}{e(X,Z)}[\I(C=c)-\pi_c(X)]+\pi_c(X)\Big][\mu_{1c}(X)-\mu_0(X)-\Delta_c^\text{PI}]}_{(***)}\Big\},\label{adhoc}
\end{align}
and building steps that zero out the sample means of the terms. 
The resulting estimator is
\begin{align}
    \hat\Delta_{c,\textsc{ms}}^\text{PI}
    :=\frac{\sum_{i=1}^n\hat w(O_i)[\tilde\mu_{1c}(X_i)-{\color{red}\tilde\mu_0(X_i)}]}{\sum_{i=1}^n\hat w(O_i)}.\label{MS}
\end{align}
Here $\hat w(O):=\frac{Z}{\hat e(X,Z)}[\I(C=c)-\hat\pi_c(X)]+\hat\pi_c(X)$.
$\tilde\mu_{1c}(X)$ and $\tilde\mu_0(X)$ are specific estimators of $\mu_{1c}(X)$ and $\mu_0(X)$: 
$\tilde\mu_{1c}(X)$ is fit to (non)compliers in the treatment arm weighted by $1/\hat e(X,1)$, $\tilde\mu_0(X)$ is fit to control units weighted by $\hat\pi_c(X)/\hat e(X,0)$, and both are \textit{mean-recovering} models (i.e., on the sample to which the model is fit, the mean of model predictions equals outcome mean).
These models zero out the sample means of $\scriptsize(*)$ and $\scriptsize(**)$, and the weighted averaging in (\ref{MS}) zeros out the sample mean of $\scriptsize(***)$. ($\hat w(O)$ can also be H\'ajek-ized, for another version.)

\begin{remark}\label{rm:targeted-nuisance-estimation}
    The tilde notation here refers to this specific method of estimating $\mu$ functions for this estimator. The weighting targets the model to the relevant covariate space where it is used for prediction, and the mean-recovering feature ensures that predictions are on average unbiased (if the weights are correct). This targeted estimation technique can also be used (but is not required) for estimating $\mu_{1c}(X)$ and $\mu_0(X)$ for other estimators, and for estimating $\pi_c(X)$.
\end{remark}

\smallskip

$\hat\Delta_{c,\textsc{ms}}^\text{PI}$ shares the same multiply robust property of $\hat\Delta_{c,\textsc{if}}^\text{PI}$ and $\hat\Delta_{c,\textsc{ifh}}^\text{PI}$ (see Proposition~\ref{thm:pi-multiplyrobust}).

\begin{theorem}[multiply robust PI-based estimators]\label{thm:pi-multiplyrobust}
$\hat\Delta_{c,\textsc{if}}^\textup{PI}$, $\hat\Delta_{c,\textsc{ifh}}^\textup{PI}$ and $\hat\Delta_{c,\textsc{ms}}^\textup{PI}$ are consistent if one of the following three conditions hold:
\begin{enumerate}[(i),topsep=5pt,itemsep=-.2em]
    \item the propensity score $e(X,Z)$ and principal score $\pi_c(X)$ models are correctly specified; or
    \item the principal score model $\pi_c(X)$ and both outcome models $\mu_{1c}(X),\mu_0(X)$ are correctly specified; or
    \item the propensity score model $e(X,Z)$ and the outcome under control $\mu_0(X)$ model are correctly specified.
\end{enumerate}
\end{theorem}

For simplicity, we presume that estimation uses parametric models. While an active research topic on IF-based inference is data-adaptive nuisance estimation, we leave it to future work.

\subsection{Type C ($\approx$ other/weighting estimators)}

Type C estimators do not involve estimating $\mu_0(X)$ as a step in the estimation procedure. This type includes the pure weighting estimator
\begin{align}
    \hat\Delta_{c,e\pi}^\text{PI}
    &:=
    \frac{\sum_{i=1}^n\mfrac{Z_i\I(C_i=c)}{\hat e(X_i,Z_i)}Y_i}{\sum_{i=1}^n\mfrac{Z_i\I(C_i=c)}{\hat e(X_i,Z_i)}}
    -
    \frac{\sum_{i=1}^n\mfrac{(1-Z_i)\hat\pi_c(X_i)}{\hat e(X_i,Z_i)}Y_i}{\sum_{i=1}^n\mfrac{(1-Z_i)\hat\pi_c(X_i)}{\hat e(X_i,Z_i)}},\label{pi-est:id1}
\end{align}
justified by the second $\tau_{1c}$ formula in (\ref{id:tau.1c}) and the third $\tau_{0c}^\text{PI}$ formula in (\ref{pi-id:tau.0c}).
Also included in type C is the estimator that employs this same weighting scheme and uses the weighted sample to fit a model regressing outcome on treatment and covariates, say, to improve precision in estimating $\Delta_c^\text{PI}$ (in the spirit of~ \cite{wang2019AnalysisCovarianceRandomized,steingrimsson2017ImprovingPrecisionAdjusting}). 
For this type, we do not have a specific sensitivity analysis technique in mind, and will need to see whether the identification result under the sensitivity assumption allows a simple modification.

\smallskip

To sum up, we have defined three types of PI-based estimators: type A, whose defining feature is involving $\mu_0(X)$ estimation; type B, whose defining feature is having as a component an IF-based estimator of $\nu_{0c}^\text{PI}$; and type C, other estimators.
We now consider sensitivity analysis.

\section{Sensitivity analysis based on three ratio-type sensitivity paramters}\label{sec:ratio-params}

Recall that the challenge before invoking PI was that the stratum-specific conditional $Y_0$ means $\mu_{01}(X)$ and $\mu_{00}(X)$ are not identified, as they are two unknowns in the mixture equation
\begin{align}
    {\color{purple}\mu_{01}(X)}\pi_1(X)+{\color{purple}\mu_{00}(X)}\pi_0(X)=\mu_0(X).\tag{\ref{eq:mixture-mean}}
\end{align}
PI identifies 
$\mu_{01}(X)$ and $\mu_{00}(X)$ by equating them to each other. A sensitivity analysis replaces PI with a \textit{sensitivity assumption} that allows $\mu_{01}(X)$ and $\mu_{00}(X)$ to differ from each other. The assumption is indexed by a \textit{sensitivity parameter} indicating how and to what degree they differ.
To accommodate different outcome types (binary, bounded, unbounded) and different conceptualizations of how $\mu_{01}(X)$ and $\mu_{00}(X)$ may differ, we consider different parameterizations.
The following assumptions use an odds ratio (OR), a generalized odds ratio (GOR) and a mean ratio (MR) sensitivity parameter.
In all of them, $\rho=1$ recovers the PI case.

\medskip

\noindent
\begin{tabular}{ll}
    A4-OR (sensitivity odds ratio): 
    & $\displaystyle\frac{\mu_{01}(X)/[1-\mu_{01}(X)]}{\mu_{00}(X)/[1-\mu_{00}(X)]}=\rho$,
    \\[1.2em]
    A4-GOR (sensitivity generalized odds ratio):
    & $\displaystyle\frac{[\mu_{01}(X)-l]/[h-\mu_{01}(X)]}{[\mu_{00}(X)-l]/[h-\mu_{00}(X)]}=\rho$
    \\[.8em]
    & where $l,h$ are the lower and upper $Y_0$ bounds,
    \\[.2em]
    A4-MR (sensitivity mean ratio):
    & $\displaystyle\frac{\mu_{01}(X)}{\mu_{00}(X)}=\rho$,
\end{tabular}

\smallskip

\noindent
for some positive range of $\rho$ that is considered plausible.

\medskip

As mentioned in Section~\ref{sec:relatedwork}, a challenge with A4-MR is that it may predict out of the outcome range. For an example, consider an outcome on a 0 to 7 scale. Suppose that for some covariate value $x$, $\mu_0(x)=5$ and $\pi_1(x)=0.3$. Then a sensitivity MR value of 1.69 would imply $\mu_{01}(x)>7$.
A4-MR is thus more suitable if the outcome is single-signed and unbounded. Since most outcomes are practically bounded, if using A4-MR, the parameter range should be carefully selected to avoid predicting extreme $\mu_{0c}(X)$ values; we will discuss this in Section~\ref{sec:calibration}.

For binary outcomes, we propose A4-OR, the assumption that within levels of $X$ (i)~the odds of the outcome for compliers is $\rho$ times that for noncompliers, or equivalently (because ORs are symmetric), (ii)~the odds of being a complier for those with the outcome is $\rho$ times that for those without the outcome. A4-OR predicts $\mu_{0c}(X)$ within $[0,1]$.

More generally, for outcomes bounded on both ends, we propose A4-GOR, a generalization of A4-OR. (A4-OR is a special case with $l=0$ and $h=1$.) Figure~\ref{fig:gor-visualization} shows the connection between $\mu_{00}(X)$ and $\mu_{01}(X)$ for several GOR values.
If the outcome range varies with $X$, the bounds can be made $X$-value-specific, i.e., $l(X)$ and $h(X)$.
A4-GOR always predicts $\mu_{0c}(X)$ within the specified bounds. 
For a non-binary outcome, however, A4-GOR may still contradict with the observed outcome distribution in ways that are not obvious, e.g., predicting $\mu_{0c}(X)$ values far from where the outcome mass is concentrated.

\begin{figure}[t]
    \centering
    \caption{Connection between $\mu_{00}(X)$ and $\mu_{01}(X)$ under A4-GOR for different GOR values}
    \includegraphics[width=\textwidth]{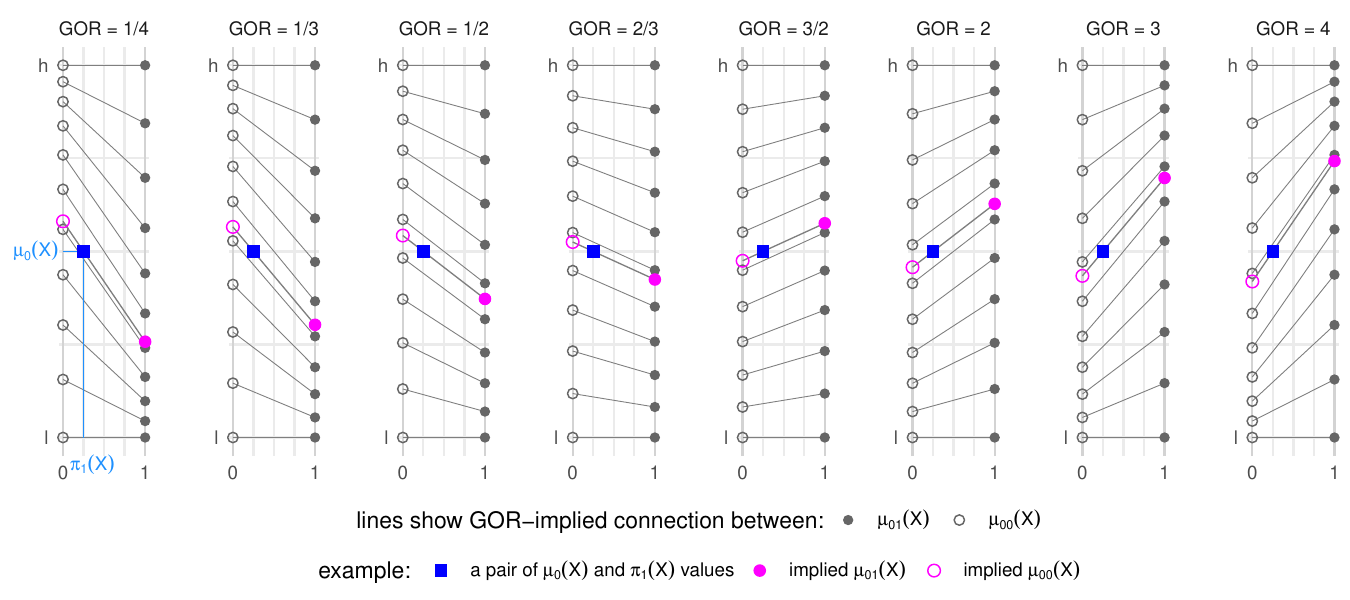}
    \label{fig:gor-visualization}
    \vspace{-2em}
\end{figure}

\begin{remark}[Exponential tilting connection]
A4-OR can be equivalently expressed as
\begin{align}
    \P(Y_0\mid X,C=1)=\P(Y_0\mid X,C=0)\frac{\exp((\ln\rho) Y_0)}{\E[\exp((\ln\rho) Y_0)\mid X,C=0]},\label{assumption:exponential-tilt}
\end{align}
which looks like exponential tilting assumptions used in the context of non-ignorable missingness and unobserved confounding \cite{Franks2020,Scharfstein2021,Robins2000}. The difference is that in these other problems, the assumption connects an unobserved distribution (e.g., that of missing data) to an observed distribution (that of non-missing data), whereas here the assumption relates two otherwise unidentified distributions whose mixture (and mixing ratio) is identified. 
Here the tilting-like assumption (\ref{assumption:exponential-tilt}) achieves identification with a binary outcome but not generally. If $Y_0$ is continuous, for example, (\ref{assumption:exponential-tilt}) (combined with the mixing weights $\pi_c(X)$) is not sufficient to identify the component distributions $\P(Y_0\mid X,C=c)$ (or their means) based on the mixture distribution $\P(Y_0\mid X)$.
\end{remark}

Proofs of all results in this section are provided in Appendix \ref{appendix:ratio-params}.

\subsection{Identification}

Combining any of the A4- assumptions with (\ref{eq:mixture-mean}), we can identify $\mu_{0c}(X)$, which then identifies $\tau_{0c}$. We present results for A4-GOR (which includes A4-OR as a special case) and A4-MR.

To maintain symmetry, let $\rho_1=\rho$, $\rho_0=1/\rho$.

\smallskip

\begin{theorem}[GOR- and MR-based identification]\label{thm:ratio.params-id}
Under assumptions A0-A2 combined with A4-GOR,
\begin{align}
    \mu_{0c}(X)=
    \begin{cases}
    \displaystyle\frac{\alpha_c(X)-\beta_c(X)}{2(\rho_c-1)\pi_c(X)}(h-l)+l & \text{if }\rho_c\neq1
    \\
    \mu_0(X) & \text{if }\rho_c=1
    \end{cases}
    ~~~=:\mu_{0c}^\textup{GOR}(X),\label{gor-id:mu.0c}
\end{align}
and under assumptions A0-A2 combined with A4-MR,
\begin{align}
    \mu_{0c}(X)=\gamma_c(X)\mu_0(X)=:\mu_{0c}^\textup{MR}(X),\label{mr-id:mu.0c}
\end{align}
where
\vspace{-.5em}
\begin{align*}
    \alpha_c(X)
    &:=\left[\pi_c(X)+{\textstyle\frac{\mu_0(X)-l}{h-l}}\right](\rho_c-1)+1,
    \\
    \beta_c(X)
    &:=\sqrt{[\alpha_c(X)]^2-4\pi_c(X){\textstyle\frac{\mu_0(X)-l}{h-l}}\rho_c(\rho_c-1)},
    \\
    \gamma_c(X)
    &:=\frac{\rho_c}{(\rho_c-1)\pi_c(X)+1}.
\end{align*}
\end{theorem}

\medskip

Identification of $\nu_{0c}$, $\tau_{0c}$, $\Delta_c$ follows from $\mu_{0c}(X)$ identification.
We will label the results of these parameters under A4-GOR and A4-MR with superscripts \textsuperscript{GOR} and \textsuperscript{MR}, respectively.

\subsection{Estimation}

Based on the above identification results, we now modify the PI-based estimators.
We let each resulting estimator inherit the label of the originating estimator, except replacing the superscript $^\text{PI}$ with one indicating the sensitivity assumption. 

Figure~\ref{fig:flowchart} provides a summary of the key techniques presented here and in the next section.

\subsubsection{Type A estimators}

These estimators are adapted by replacing the estimate of $\color{red}\mu_0(X)$ with 
estimates of $\color{blue}\mu_{0c}^\text{GOR}(X)$ or $\color{blue}\mu_{0c}^\text{MR}(X)$.
For example, this turns the principal score weighted outcome regression estimator $\hat\Delta_{c,\pi\mu}^\text{PI}$ (\ref{pi-est:id2}) (aka the plug-in estimator) into 
\begin{align}
    \hat\Delta_{c,\pi\mu}^\text{GOR}
    &:=\frac{\sum_{i=1}^n\hat\pi_c(X_i)[\hat\mu_{1c}(X_i)-{\color{blue}\hat\mu_{0c}^\text{GOR}(X_i)}]}{\sum_{i=1}^n\hat\pi_c(X_i)},
    ~~~
    \hat\Delta_{c,\pi\mu}^\text{MR}:=\frac{\sum_{i=1}^n\hat\pi_c(X_i)[\hat\mu_{1c}(X_i)-{\color{blue}\hat\mu_{0c}^\text{MR}(X_i)}]}{\sum_{i=1}^n\hat\pi_c(X_i)},\label{est:gor,mr-plugin}
\end{align}
where $\color{blue}{\hat\mu_{0c}^\text{GOR}(X_i)}$ and $\color{blue}{\hat\mu_{0c}^\text{MR}(X_i)}$ are $\mu_{0c}^\text{GOR}(X_i)$ (\ref{gor-id:mu.0c}) and $\mu_{0c}^\text{MR}(X_i)$ (\ref{mr-id:mu.0c}) evaluated at $\hat\mu_0(X_i)$  and $\hat\pi_c(X_i)$.
The other outcome-regression estimators $\hat\Delta_{c,e\mu}^\text{PI}$ (\ref{pi-est:id3}) and $\hat\Delta_{c,\textsc{ms}}^\text{PI}$ (\ref{MS}) are adapted similarly.

\subsubsection{Type B estimators}

Adaptation is based on the IFs of $\nu_{0c}^\textup{GOR}:=\E[\pi_c(X)\mu_{0c}^\text{GOR}(X)]$ and $\nu_{0c}^\textup{MR}:=\E[\pi_c(X)\mu_{0c}^\text{MR}(X)]$, which are provided in Proposition \ref{thm:ratio.params-ifs}.

\begin{theorem}[GOR- and MR-based IFs]\label{thm:ratio.params-ifs}
The IFs for $\nu_{0c}^\textup{GOR}$ and $\nu_{0c}^\textup{MR}$ are
\begin{align}
    \varphi_{\nu_{0c}^\textup{GOR}}(O)
    &=\frac{1-Z}{e(X,Z)}\epsilon_{\mu,c}^\textup{GOR}(X)[Y-\mu_0(X)]+\frac{Z}{e(X,Z)}\epsilon_{\pi,c}^\textup{GOR}(X)[\I(C=c)-\pi_c(X)]+\nonumber
    \\
    &~~~~~~~~~~~~~~~~~~~~~~~~~~~~~~~~~~~~~~~~~~~~~~~~~~~~~~~~~~~~+\pi_c(X)\mu_{0c}^\textup{GOR}(X)-\nu_{0c}^\textup{GOR},
    \\
    \varphi_{\nu_{0c}^\textup{MR}}(O)
    &=\frac{1-Z}{e(X,Z)}\epsilon_{\mu,c}^\textup{MR}(X)[Y-\mu_0(X)]+\frac{Z}{e(X,Z)}\epsilon_{\pi,c}^\textup{MR}(X)[\I(C=c)-\pi_c(X)]+\nonumber
    \\
    &~~~~~~~~~~~~~~~~~~~~~~~~~~~~~~~~~~~~~~~~~~~~~~~~~~~~~~~~~~~~+\pi_c(X)\mu_{0c}^\textup{MR}(X)-\nu_{0c}^\textup{MR},
\end{align}
where
\vspace{-1em}
\begin{alignat*}{2}
    \epsilon_{\mu,c}^\textup{GOR}(X)
    &:=
    \begin{cases}
    \frac{1}{2}-\frac{\alpha_c(X)}{2\beta_c(X)}+\frac{\rho_c\pi_c(X)}{\beta_c(X)} & \text{if }\rho_c\neq1
    \\
    \pi_c(X) & \text{if }\rho_c=1
    \end{cases},
    ~~~~
    &&\epsilon_{\mu,c}^\textup{MR}(X)
    :=\gamma_c(X)\pi_c(X),
    \\
    \epsilon_{\pi,c}^\textup{GOR}(X)
    &:=
    \begin{cases}
    \left[\frac{1}{2}\!-\!\frac{\alpha_c(X)}{2\beta_c(X)}\!+\!\frac{\rho_c\frac{\mu_0(X)-l}{h-l}}{\beta_c(X)}\right](h\!-\!l)+l & \text{if }\rho_c\neq1
    \\
    \mu_0(X) & \text{if }\rho_c=1
    \end{cases},
    ~~~
    &&\epsilon_{\pi,c}^\textup{MR}(X)
    :=\gamma_1(X)\gamma_0(X)\mu_0(X).
\end{alignat*}
\end{theorem}

\medskip

Based on Proposition \ref{thm:ratio.params-ifs}, under A4-GOR and A4-MR, we obtain estimators $\hat\Delta_{c,\textsc{if}}^\text{GOR}$ and $\hat\Delta_{c,\textsc{if}}^\text{MR}$ by replacing the $\color{red}{\hat\nu_{0c,\textsc{if}}^\text{PI}}$ component of $\hat\Delta_{c,\textsc{if}}^\text{PI}:=\frac{\hat\nu_{1c,\textsc{if}}-{\color{red}\hat\nu_{0c,\textsc{if}}^\text{PI}}}{\hat\delta_{c,\textsc{if}}}$ (\ref{est:pi-if})
with $\color{blue}{\hat\nu_{0c,\textsc{if}}^\text{GOR}}$ and $\color{blue}{\hat\nu_{0c,\textsc{if}}^\text{MR}}$, respectively, where
\vspace{-.5em}
\begin{align}
    {\color{blue}\hat\nu_{0c,\textsc{if}}^\text{GOR}}
    &:=\P_n\Big\{\frac{1-Z}{\hat e(X,Z)}{\hat\epsilon}_{\mu,c}^\text{GOR}(X)[Y-\hat\mu_0(X)]+
    \frac{Z}{\hat e(X,Z)}{\hat\epsilon}_{\pi,c}^\text{GOR}(X)[\I(C=c)-\hat\pi_c(X)]+\nonumber
    \\
    &~~~~~~~~~~~~~~~~~~~~~~~~~~~~~~~~~~~~~~~~~~~~~~~~~~~~~~~~~~~~~~~~~~~~~~~~~~~~~~~~+\hat\pi_c(X)\hat\mu_{0c}^\text{GOR}(X)\Big\},\label{gor-est:nu.0c}
    \\
    {\color{blue}\hat\nu_{0c,\textsc{if}}^\text{MR}}
    &:=\P_n\Big\{\frac{1-Z}{\hat e(X,Z)}{\hat\epsilon}_{\mu,c}^\text{MR}(X)[Y-\hat\mu_0(X)]+
    \frac{Z}{\hat e(X,Z)}{\hat\epsilon}_{\pi,c}^\text{MR}(X)[\I(C=c)-\hat\pi_c(X)]+\nonumber
    \\
    &~~~~~~~~~~~~~~~~~~~~~~~~~~~~~~~~~~~~~~~~~~~~~~~~~~~~~~~~~~~~~~~~~~~~~~~~~~~~~~~~+\hat\pi_c(X)\hat\mu_{0c}^\text{MR}(X)\Big\},
\end{align}
and the $\epsilon$ functions are estimated by evaluating them at $\hat\pi_c(X)$ and $\hat\mu_0(X)$.

$\hat\Delta_{c,\textsc{ifh}}^\text{PI}$ (\ref{est:pi-ifh}) is modified similarly to obtain sensitivity estimators $\hat\Delta_{c,\textsc{ifh}}^\text{GOR}$ and $\hat\Delta_{c,\textsc{ifh}}^\text{MR}$ by replacing $\color{red}{\hat\nu_{0c,\textsc{ifh}}^\text{PI}}$ with $\color{blue}{\hat\nu_{0c,\textsc{ifh}}^\text{GOR}}$ and $\color{blue}{\hat\nu_{0c,\textsc{ifh}}^\text{MR}}$, the H\'ajek-ized version of $\hat\nu_{0c,\textsc{if}}^\text{GOR}$ and $\hat\nu_{0c,\textsc{if}}^\text{MR}$.

\paragraph{\textit{Partial loss of robustness.}}

Proposition \ref{thm:pi-multiplyrobust} stated that several PI-based estimators are multiply robust, including type B estimators $\hat\Delta_{c,\textsc{if}}^\text{PI}$ (\ref{est:pi-if}) and $\hat\Delta_{c,\textsc{ifh}}^\text{PI}$ (\ref{est:pi-ifh}), and the multi-step type A estimator $\hat\Delta_{c,\textsc{ms}}^\text{PI}$ (\ref{MS}). The adaptation of these estimators for sensitivity analysis results in \textit{partial loss of robustness} (see Proposition \ref{thm:robustness-loss}).
The resulting GOR-based estimators 
($\hat\Delta_{c,\textsc{if}}^\text{GOR}$, $\hat\Delta_{c,\textsc{ifh}}^\text{GOR}$, $\hat\Delta_{c,\textsc{ms}}^\text{GOR}$)
depend on correct specification of models for $\pi_c(X)$ and $\mu_0(X)$ (i.e., they are inconsistent if either model is misspecified).
The MR-based counterparts
($\hat\Delta_{c,\textsc{if}}^\text{MR}$, $\hat\Delta_{c,\textsc{ifh}}^\text{MR}$, $\hat\Delta_{c,\textsc{ms}}^\text{MR}$)
depend on correct specification of the model for $\pi_c(X)$.

\begin{theorem}[Partial loss of robustness]\label{thm:robustness-loss}
    $\hat\Delta_{c,\textsc{if}}^\textup{GOR}$, $\hat\Delta_{c,\textsc{ifh}}^\textup{GOR}$ and $\hat\Delta_{c,\textsc{ms}}^\textup{GOR}$ are consistent for $\Delta_c^\textup{GOR}$ if 
    \begin{itemize}
        \item both the model for $\pi_c(X)$ and the model for $\mu_0(X)$ are correctly specified, AND
        \item either the model for $e(X,Z)$ or the model for $\mu_{1c}(X)$ is correctly specified.
    \end{itemize}
    \noindent
    $\hat\Delta_{c,\textsc{if}}^\textup{MR}$, $\hat\Delta_{c,\textsc{ifh}}^\textup{MR}$ and $\hat\Delta_{c,\textsc{ms}}^\textup{MR}$ are consistent for $\Delta_c^\textup{MR}$ if
    \begin{itemize}
        \item the model for $\pi_c(X)$ is correctly specified, AND 
        \item either the model for $e(X,Z)$ or both outcome models $\mu_{1c}(X),\mu_0(X)$ are correctly specified.
    \end{itemize}
\end{theorem}

\begin{remark}[Approximate robustness]\label{rm:approximate-robustness}
Among 
these sensitivity estimators, the type B estimators ($\hat\Delta_{c,\textsc{if}}^\text{GOR}$, $\hat\Delta_{c,\textsc{ifh}}^\text{GOR}$, $\hat\Delta_{c,\textsc{if}}^\text{MR}$, $\hat\Delta_{c,\textsc{ifh}}^\text{MR}$) are in a sense more robust than the multi-step type A estimators ($\hat\Delta_{c,\textsc{ms}}^\text{GOR}$, $\hat\Delta_{c,\textsc{ms}}^\text{MR}$): they have an \textit{approximate robustness} property with respect to the model component(s) whose correct specification they require for consistency.
Specifically, (i)~while all six estimators depend on a correct model for $\pi_c(X)$, the type B estimators provide a first-order correction of the bias (that would be incurred if simply using the plug-in estimator (\ref{est:gor,mr-plugin})) due to the deviation of the probability limit $\pi_c^\dagger(X)$ of $\hat\pi_c(X)$ from the true function $\pi_c(X)$.
Also, (ii)~while all three GOR-based estimators additionally depend on a correct model for $\mu_0(X)$, the type B estimators provide a first-order correction of the bias due to the deviation of the probability limit $\mu_0^\dagger(X)$ of $\hat\mu_0(X)$ from the true function $\mu_0(X)$.
(This first-order bias correction feature is also shared by the originating PI-based estimators $\hat\Delta_{c,\textsc{ms}}^\text{PI}$, $\hat\Delta_{c,\textsc{if}}^\text{PI}$, $\hat\Delta_{c,\textsc{ifh}}^\text{PI}$, and results in the robustness of those estimators.) 

We give a quick explanation of (ii) to make this concrete. (For full details concerning Remark \ref{rm:approximate-robustness}, see the Appendix.) If $\hat e(X,Z)$ and $\hat\pi_c(X)$ are correctly specified but $\hat\mu_0(X)$ is not, the probability limit of both $\hat\nu_{0c,\textsc{if}}^\text{GOR}$ and $\hat\nu_{0c,\textsc{ifh}}^\text{GOR}$ is the sum of two terms
\vspace{-.5em}
\begin{align}
    \E\Big\{\pi_c(X)\mu_{0c}^\text{GOR}[\mu_0^\dagger(X),\pi_c(X)]\Big\}+\E\Big\{\epsilon_{\mu,c}^\text{GOR}[\mu_0^\dagger(X),\pi_c(X)][\mu_0(X)-\mu_0^\dagger(X)]\Big\}
\end{align}
(which result from the last and first terms in (\ref{gor-est:nu.0c})). These are the first two terms in the Taylor expansion of the true parameter $\nu_{0c}^\text{GOR}$ treated as a function of $\mu_0()$ at the point $\mu_0^\dagger()$.
The first term coincides with the probability limit of the plug-in estimator, which is biased due to $\mu_0^\dagger(X)\neq\mu_0(X)$. The second term provides a first-order correction of this bias.
For this approximate robustness property to be beneficial, however, $\mu_0^\dagger(X)$ needs to be close to $\mu_0(X)$.

\end{remark}

\begin{figure}[!h]
\caption{Flowchart summarizing key sensitivity analysis techniques that are applicable given PI-based estimator type and sensitivity parameterization}\label{fig:flowchart}
\resizebox{\textwidth}{!}{%
\begin{tikzpicture}[node distance=2cm]

\node[draw,
    rounded rectangle,
    minimum width=2.5cm,
    minimum height=1cm] (main) {MAIN ANALYSIS};

\node[below of=main, yshift=.5cm,
    draw,
    rectangle,
    minimum height=1cm
] (estimator) {Conduct/plan main analysis using a PI-based estimator};

\node[below of=estimator,
    draw,
    rounded rectangle,
    minimum width=2.5cm,
    minimum height=1cm] (sens) {SENSITIVITY ANALYSIS};

\node[below of=sens, yshift=.4cm,
    draw,
    rectangle,
    minimum height=1cm
] (param) {Choose sensitivity parameterization};

\node[below of=param, yshift=-.1cm,
    diamond, 
    draw, 
    aspect=3, 
    text centered
] (Q1) {Which sens param?};

\node[below left of=Q1, xshift=-5cm, yshift=-1cm,
    diamond, 
    draw, 
    aspect=3,  
    text centered
] (Q2a) {Which estimator type?};

\node[below right of=Q1, xshift=5cm, yshift=-1cm,
    diamond, 
    draw, 
    aspect=3, 
    text centered
] (Q2b) {Which estimator type?};

\node[below of=Q2a, yshift=-0.8cm, xshift=-3.cm,
    rectangle,
    draw,
    text width=4cm
] (Aa) {Replace $\hat\mu_0(X)$ in the estimator with $\hat\mu_{0c}^\text{sens}(X)$\\\hfill section~4.2.1};

\node[below of=Aa, yshift=-0.3cm,
    rectangle,
    draw,
    text width=4cm
] (Ba) {Replace the $\hat\nu_{0c}^\text{PI}$ component of the estimator with $\hat\nu_{0c}^\text{sens}$\\\hfill section~4.2.2};

\node[below of=Q2a, yshift=-0.7cm, xshift=3cm,
    diamond,
    draw,
    aspect=3,
    text centered,
    inner sep=0
] (q3) {MR sens param?};

\node[below of=Ba, yshift=-0.7cm,
    rectangle,
    draw,
    text width=4cm
] (Ca1) {Scale $Y$ in control units by a factor of $\gamma_c(X)$ when estimating effect for stratum $c$\\\hfill section~4.2.3};

\node[below of=q3, yshift=-1.5cm,
    rectangle,
    draw,
    text width=2cm,
    text centered
] (Ca2) {Consider using a type~A (or~type~B) estimator instead};

\node[below of=Q2b, yshift=-2cm, xshift=-3.2cm,
    rectangle,
    draw,
    text width=4.5cm
] (Ab) {Estimate the bias $\xi_c$ using a simple estimator\\(Equivalent to replacing $\hat\mu_0(X)$ with $\hat\mu_{0c}^\text{SMDe}(X)$)\\\hfill section~5.2.1};

\node[below of=Ab, yshift=-1.5cm,
    rectangle,
    draw,
    text width=4.5cm
] (Bb) {Estimate the bias $\xi_c$ using an IF-based estimator\\(Equivalent to replacing $\hat\nu_{0c}^\text{PI}$ with $\hat\nu_{0c}^\text{SMDe}$)\\\hfill section~5.2.2};

\node[below of=Q2b, yshift=-2cm, xshift=3cm,
    rectangle,
    draw,
    text width=2.2cm,
    text centered
] (Cb) {Consider using a type~A (\mbox{or IF-based}) estimator instead};

\draw[->] (main) -- (estimator);
\draw[->] (estimator) -- (sens);
\draw[->] (sens) -- (param);
\draw[->] (param) -- (Q1);

\draw [->] (Q1) -| (Q2a) node[pos=0.25,fill=white,inner sep=0]{OR/GOR/MR} coordinate[pos=0.9] (*);
\draw [->] (Q1) -| (Q2b) node[pos=0.17,fill=white,inner sep=0]{SMD} coordinate[pos=0.9] (**);

\draw[->] (Q2a) -| (Aa) node[pos=0.75,fill=white,inner sep=0]{type A};
\draw[->] (Q2a) |- (Ba) node[pos=0.1,fill=white,inner sep=0]{type B};
\draw[->] (Q2a) -| (q3) node[pos=0.75,fill=white,inner sep=0]{type C};
\draw[->] (q3.west)|- (Ca1) node[pos=0.1,fill=white,inner sep=0]{yes};
\draw[->] (q3)-- (Ca2) node[pos=0.35,fill=white,inner sep=0]{no};
\draw[->] (Ca2) --  ++(2.2,0) |- (*);

\draw[->] (Q2b) -| (Ab) node[pos=0.75, fill=white, inner sep=0, text width=1cm, text centered]{simple type~A};
\draw[->] (Q2b) |- (Bb) node[pos=0.07, fill=white, inner sep=0]{IF-based};
\draw[->] (Q2b) -| (Cb) node[pos=0.75, fill=white, inner sep=0]{other};
\draw[->] (Cb) --  ++(1.5,0) |- (**);

\end{tikzpicture}%
}
\end{figure}
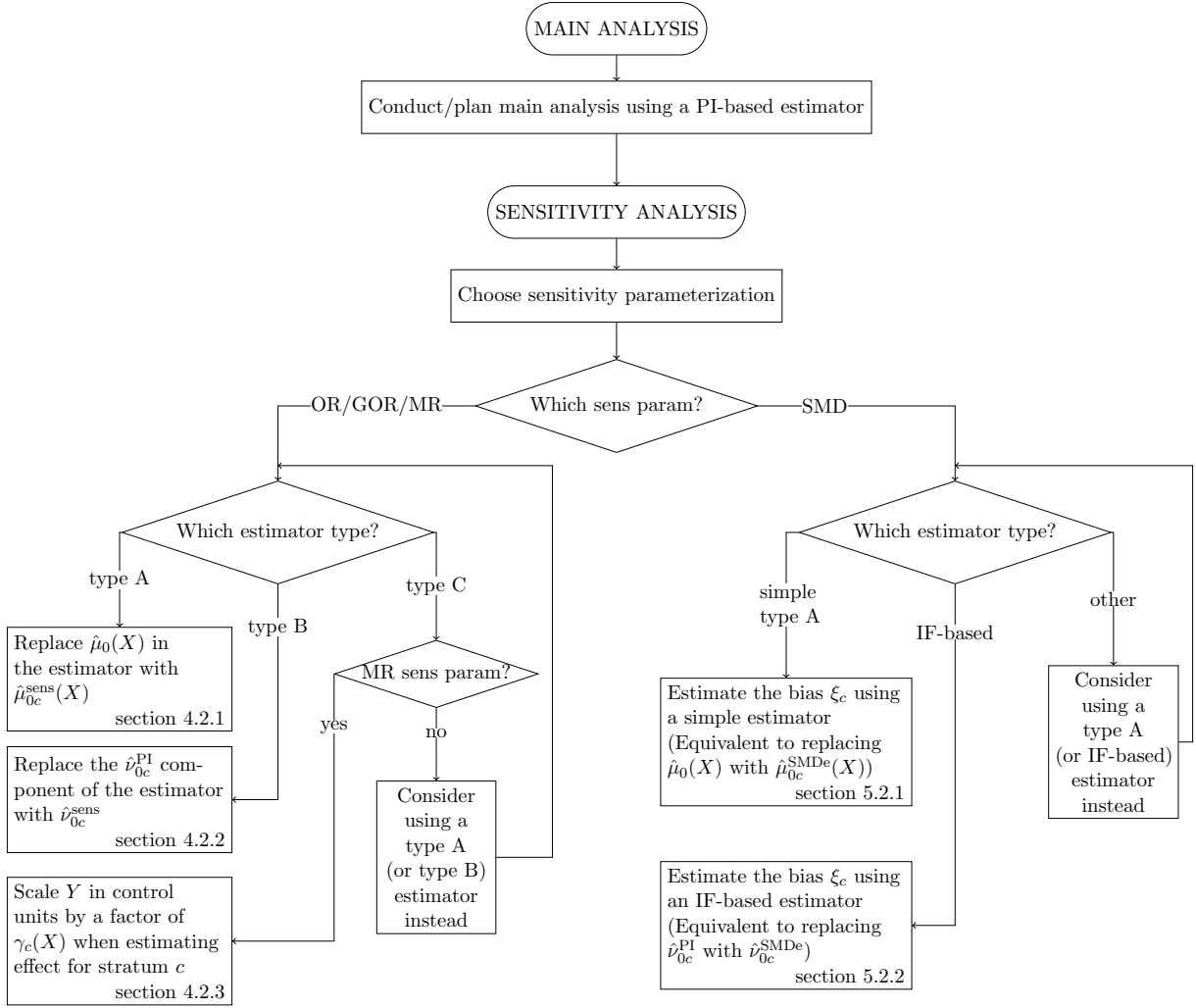

\subsubsection{Type C estimators}

We consider A4-MR and A4-GOR separately.
Under A4-MR, the convenient form of $\mu_{0c}^\text{MR}(X)$ (\ref{mr-id:mu.0c}) allows a simple adaptation of type C estimators: to estimate $\Delta_c^\text{MR}$, scale the outcome in control units by a factor of $\hat\gamma_c(X)$ then use the PI-based analysis method.
For the pure weighting estimator specifically, this adaptation results in the estimator
\begin{align*}
    \hat\Delta_{c,e\pi}^\text{MR}:=\frac{\sum_{i=1}^n\frac{Z_i\I(C_i=c)}{\hat e(X_i,Z_i)}Y_i}{\sum_{i=1}^n\frac{Z_i\I(C_i=c)}{\hat e(X_i,Z_i)}}-\frac{\sum_{i=1}^n\frac{(1-Z_i)\hat\pi_c(X_i)}{\hat e(X_i,Z_i)}{\color{blue}\hat\gamma_c(X_i)Y_i}}{\sum_{i=1}^n\frac{(1-Z_i)\hat\pi_c(X_i)}{\hat e(X_i,Z_i)}}.
\end{align*}
This outcome scaling technique is justified by the result below, a corollary of Proposition \ref{thm:ratio.params-id}.
\begin{corollary}[MR-based outcome scaling]
\begin{align}
    \tau_{0c}^\textup{MR}=\frac{\E\left[\frac{1-Z}{e(X,Z)}\pi_c(X)\gamma_c(X)Y\right]}{\E\left[\frac{1-Z}{e(X,Z)}\pi_c(X)\right]}.\label{mr-Y.scaling}
\end{align}
\end{corollary}

\begin{remark}
    When specializing to the randomized treatment setting, (\ref{mr-Y.scaling}) simplifies, and one expression of the specialized version of (\ref{mr-Y.scaling}) is $\E[\frac{\pi_c(X)\gamma_c(X)Y}{\pi_c}\mid Z=0]$, which appeared in Ding and Lu (2017, proposition 3)\cite{Ding2017}. Based on this expression, this paper characterizes the MR-based sensitivity analysis as an \textit{under/overweighting of the principal score} by a factor of $\gamma_c(X)$. Interestingly, this characterization breaks the interpretation of $\tau_{0c}$ as a weighted average (our starting point in Lemma \ref{lm:starting-point}, which we have maintained throughout). Our new insight here is that the appearance of $\gamma_c(X)$ in (\ref{mr-Y.scaling}) is due to the fact that under A4-MR the outcome mean $\mu_{0c}(X)$ is identified by $\gamma_c(X)\mu_0(X)$. It is thus natural to use the \textit{scaling the outcome} by a factor of $\gamma_c(X)$ characterization. Also, by leaving the principal score weights alone, this outcome scaling technique applies to type C estimators generally, not just the pure weighting estimator.
\end{remark}

Under A4-GOR, there is no result similar to (\ref{mr-Y.scaling}) that separates $Y$ from functions of $X$, therefore no simple modification is available for type C estimators. The pure weighting estimator $\hat\Delta_{c,e\pi}^\text{PI}$ (\ref{pi-est:id3}) (but not type C generally) can be adapted by replacing $Y$ in the second term with an estimate of $\mu_{0c}^\text{GOR}(X)$ (which requires estimating $\mu_0(X)$). For this estimator to reduce to $\hat\Delta_{c,e\pi}^\text{PI}$ when $\rho=1$, $\mu_0(X)$ has to be estimated by a $\tilde\mu_0(X)$ model (defined in (\ref{MS})). However, with $\mu_0(X)$ estimated, there are other options for estimating $\Delta_c^\text{GOR}$ that one might prefer to such modification, e.g., replacing the whole second term of $\hat\Delta_{c,e\pi}^\text{PI}$ with $\frac{\sum_{i=1}^n\frac{Z_i\I(C_i=c)}{\hat e(X_i,Z_i)}\hat\mu_{0c}^\text{GOR}}{\sum_{i=1}^n\frac{Z_i\I(C_i=c)}{\hat e(X_i,Z_i)}}$. This obtains the type A estimator $\hat\Delta_{c,e\mu}^\text{GOR}$, which inconveniently does not reduce to $\hat\Delta_{c,e\pi}^\text{PI}$ when $\rho=1$. Hence this is one place where we break the convention of respecting the primacy of the main analysis and recommend that, if a GOR-based sensitivity analysis is to be conducted, a type A (or type B) estimator be used for the main analysis.

\section{Sensitivity analysis based on a difference-type sensitivity parameter}\label{sec:diff-param}

A4-OR, A4-GOR and A4-MR all assume that the means of $Y_0$ differ between compliers and noncompliers in some multiplicative manner. If one believes the difference is additive, it is more appropriate to use a sensitivity parameter that involves $\mu_{01}(X)-\mu_{00}(X)$. We propose using a standardized mean difference (SMD).
For convenient notation, let
\begin{align*}
    \sigma_{0c}^2(X)
    &:=\var(Y_0\mid X,C=c),
    \\
    \sigma_0^2(X)
    &:=\var(Y\mid X,Z=0).
\end{align*}
A simple SMD-based assumption is

\medskip

\begin{tabular}{ll}
    A4-SMD: & 
    $\displaystyle\frac{\mu_{01}(X)-\mu_{00}(X)}{\sqrt{[\sigma_{01}^2(X)+\sigma_{00}^2(X)]/2}}=\eta$,~~~for a plausible range of $\eta$.
\end{tabular}

\medskip\smallskip

\noindent
The denominator here is an ``average'' standard deviation: the quadratic mean of $\sigma_{01}(X)$ and $\sigma_{00}(X)$ (the within-stratum conditional standard deviations of $Y_0$).
This standard deviation scale helps in selecting a range for $\eta$ and users can tap into intuition about SMDs from other contexts (e.g., measuring effect size\cite{cohen1988StatisticalPowerAnalysis} or covariate imbalance\cite{stuart2010MatchingMethodsCausal}).
$\eta=0$ recovers the PI case; $\eta=\pm 1$ indicates a substantial complier-noncomplier difference in the outcome under control.

Inconveniently, A4-SMD combined with A0-A2 only partially identifies $\mu_{0c}(X)$. For a simple sensitivity analysis, we consider the stronger assumption below, which supplements A4-SMD with an equal variance assumption:

\medskip

\begin{tabular}{ll}
    A4-SMDe (sensitivity SMD, equal variance): 
    & A4-SMD~~and~~$\sigma_{01}^2(X)=\sigma_{00}^2(X)$.
\end{tabular}

\subsection{Identification}

For symmetry, let $\eta_1=\eta$ and $\eta_0=-\eta$.

\begin{theorem}[SMDe-based identification]\label{thm:smd-id}

Under A0-A2 combined with A4-SMDe,
\begin{align}
    \mu_{0c}(X)
    &=\mu_0(X)+\eta_c\underbrace{\frac{\pi_{1-c}(X)\sigma_0(X)}{\sqrt{1+\eta^2\pi_1(X)\pi_0(X)}}}_{\textstyle=:\lambda_c(X)}
    =:\mu_{0c}^\textup{SMDe}(X),
    \\
    \tau_{0c}
    &=\tau_{0c}^\textup{PI}+\eta_c\underbrace{\frac{\E[\pi_c(X)\lambda_c(X)]}{\E[\pi_c(X)]}}_{\textstyle=:\xi_c}
    =:\tau_{0c}^\textup{SMDe},
    \\
    \Delta_c
    &=\Delta_c^\textup{PI}-\eta_c\,\xi_c
    =:\Delta_c^\textup{SMDe}.
    \label{SMDe:Delta.c}
\end{align}
\end{theorem}

If equal variance is not assumed, $\Delta_c$ is not point identified, but bounds can be obtained. The bounds can be narrowed if one additionally assumes that $\sigma_{01}^2(X)$ and $\sigma_{00}^2(X)$ differ from each other by less than a certain factor (see Proposition~\ref{thm:smd-bounds} in the Appendix).

\subsection{Estimation}

This sensitivity analysis requires estimating $\sigma_0^2(X)$. For simplicity, we use a quasi-likelihood approach assuming the outcome's conditional variance is proportional to a function of its mean. An alternative is to directly model $[Y-\hat\mu_0(X)]^2$ based on $X$ in control units.

With the simple result (\ref{SMDe:Delta.c}), each estimator of $\Delta_c^\text{SMDe}$ we obtain is an estimator of $\Delta_c^\text{PI}$ minus $\eta_c$ times an estimator of $\xi_c$. This is the case regardless of the type of the PI-based estimator.

\subsubsection{Simple type A estimators}

\noindent Adaptation of $\hat\Delta_{c,\pi\mu}^\text{PI}$ (\ref{pi-est:id2}) and $\hat\Delta_{c,e\mu}^\text{PI}$ (\ref{pi-est:id3}) by replacing $\color{red}{\mu_0(X)}$ with $\color{blue}{\mu_{0c}^\text{SMDe}(X)}$ yields the following estimators:
\vspace{-.5em}
\begin{align}
    \hat\Delta_{c,\pi\mu}^\text{SMDe}
    &:
    =\hat\Delta_{c,\pi\mu}^\text{PI}
    -\eta_c
    \underbrace{\frac{\sum_{i=1}^n\frac{\hat\pi_1(X_i)\hat\pi_0(X_i)\sqrt{\hat\sigma_0^2(X_i)}}{\sqrt{1+\eta^2\hat\pi_1(X_i)\hat\pi_0(X_i)}}}{\sum_{i=1}^n\hat\pi_c(X_i)}}_{\textstyle=:\hat\xi_{c,\pi\sigma}},\label{eq:xi.pi.sigma}
    \\
    \hat\Delta_{c,e\mu}^\text{SMDe}
    &:
    =\hat\Delta_{c,e\mu}^\text{PI}
    -\eta_c
    \underbrace{\frac{\sum_{i=1}^n\frac{Z_i\I(C_i=c)}{\hat e(X_i,Z_i)}\frac{\hat\pi_{1-c}(X)\sqrt{\hat\sigma_0^2(X)}}{\sqrt{1+\eta^2\hat\pi_1(X)\hat\pi_0(X)}}}{\sum_{i=1}^n\frac{Z_i\I(C_i=c)}{\hat e(X_i,Z_i)}}}_{\textstyle=:\hat\xi_{c,e\pi\sigma}}.\label{eq:xi.e.pi.sigma}
\end{align}

Rather than applying the same adaptation to the multi-step estimator $\hat\Delta_{c,\textsc{ms}}^\text{PI}$ (\ref{MS}), thanks to the special form of $\Delta_c^\text{SMDe}$, we can adapt $\hat\Delta_{c,\textsc{ms}}^\text{PI}$ the way we adapt other IF-based estimators.

\subsubsection{IF-based estimators (including type B and multi-robust type A)}
We adapt these estimators using IF-based estimators of $\xi_c$.
Let $\vartheta(X):=\frac{\pi_1(X)\pi_0(X)\sigma_0(X)}{\sqrt{1+\eta^2\pi_1(X)\pi_0(X)}}$ and $\vartheta:=\E[\vartheta(X)]$. Then $\xi_c=\vartheta/\pi_c$.

\begin{theorem}[SMDe-based IF]\label{thm:smde-if}

The IFs of $\vartheta$ and $\xi_c$ are
\begin{align}
    \varphi_{\vartheta}(O)
    &=
    \frac{1-Z}{e(X,Z)}\boldsymbol\dot\vartheta_{\sigma^2}(X)\{[Y-\mu_0(X)]^2-\sigma_0^2(X)\}+
    \frac{Z}{e(X,Z)}\boldsymbol\dot\vartheta_\pi(X)[C-\pi_1(X)]+\vartheta(X)-\vartheta,\label{if:smde-vartheta}
    \\
    \varphi_{\xi_c}(O)
    &=\frac{1}{\pi_c}\big\{[\varphi_\vartheta(O)+\vartheta]-\xi_c[\varphi_{\pi_c}(O)+\pi_c]\big\},
\end{align}
where
\vspace{-1em}
\begin{align*}
    \boldsymbol{\dot}\vartheta_{\sigma^2}(X)
    &:=\frac{\pi_1(X)\pi_0(X)}{2\sigma_0(X)\sqrt{1+\eta^2\pi_1(X)\pi_0(X)}},
    \\
    \boldsymbol{\dot}\vartheta_\pi(X)
    &:=\frac{[\pi_0(X)-\pi_1(X)][2+\eta^2\pi_1(X)\pi_0(X)]\sigma_0(X)}{2[1+\eta^2\pi_1(X)\pi_0(X)]^{3/2}}.
\end{align*}
\end{theorem}

\bigskip

Based on Proposition \ref{thm:smde-if}, we have the estimator
\begin{align}
    \hat\xi_{c,\textsc{if}}:=\hat\vartheta_{\textsc{if}}/\hat\pi_{c,\textsc{if}},
\end{align}
where
\vspace{-1em}
\begin{align*}
    \hat\vartheta_{\textsc{if}}
    :=
    \P_n\left[\frac{1-Z}{\hat e(X,Z)}\hat{\boldsymbol{\dot}\vartheta}_{\sigma^2}(X)\{[Y-\hat\mu_0(X)]^2-\hat\sigma_0^2(X)\}+
    \frac{Z}{\hat e(X,Z)}\hat{\boldsymbol\dot\vartheta}_\pi(X)[C-\hat\pi_1(X)]+\hat\vartheta(X)\right]
\end{align*}
is the IF-based estimator of $\vartheta$ (where $\vartheta(X)$, $\boldsymbol\dot\vartheta_{\sigma^2}(X)$ and $\boldsymbol\dot\vartheta_\pi(X)$ are estimated by plugging in $\hat\pi_1(X)$, $\hat\pi_0(X)$ and $\sqrt{\hat\sigma_0^2(X)}$), and $\hat\pi_{c,\textsc{if}}$ is the IF-based estimator of $\pi_c$ (defined under (\ref{est:pi-if})). In addition, we have the estimator based on H\'ajek-ized versions of $\hat\vartheta_{\textsc{if}}$ and $\hat\pi_{c,\textsc{if}}$,
\begin{align}
    \hat\xi_{c,\textsc{ifh}}:=\hat\vartheta_{\textsc{ifh}}/\hat\pi_{c,\textsc{ifh}}.
\end{align}
Then the adapted IF-based estimators are
\vspace{-.5em}
\begin{align}
    \hat\Delta_{c,\textsc{if}}^\text{SMDe}
    &:=\hat\Delta_{c,\textsc{if}}^\text{PI}-\eta_c\,\hat\xi_{c,\textsc{if}},
    \\
    \hat\Delta_{c,\textsc{ifh}}^\text{SMDe}
    &:=\hat\Delta_{c,\textsc{ifh}}^\text{PI}-\eta_c\,\hat\xi_{c,\textsc{ifh}},
    \\
    \hat\Delta_{c,\textsc{ms}}^\text{SMDe}
    &:=\hat\Delta_{c,\textsc{ms}}^\text{PI}-\eta_c\,\hat\xi_{c,\textsc{if}}.
\end{align}

\begin{remark}\label{rm:approximate-robustness2}
$\hat\xi_{c,\textsc{if}}$ and $\hat\xi_{c,\textsc{ifh}}$ depend on consistent estimation of $\pi_c(X)$ and $\sigma_0^2(X)$ (they are inconsistent if either component is inconsistent), but they have the \textit{approximately robust} property where (i) if $\hat\pi_1(X)$, $\hat\mu_0(X)$ and $\hat e(X,Z)$ are consistent but $\hat\sigma_0^2(X)$ is not, the estimator provides a first-order correction of the bias of the plug-in estimator due to the deviation of the probability limit ${\sigma_0^2}^\dagger(X)$ of $\hat\sigma_0^2(X)$ from the true $\sigma_0^2(X)$; and (ii) if $\hat\mu_0(X)$, $\hat\sigma_0^2(X)$ and $\hat e(X,Z)$ are consistent but $\hat\pi_c(X)$ is not, the estimator provides a first-order correction of the bias due to the deviation of the probability limit $\pi_c^\dagger(X)$ of $\hat\pi_c(X)$ from the true $\pi_c(X)$. (See details in the Appendix.)
\end{remark}

\subsubsection{Other estimators}
While any PI-based estimator can be paired with any $\xi_c$ estimator, to keep things simple it is reasonable to pair non-IF-based estimators with either $\hat\xi_{c,\pi\sigma}$ (\ref{eq:xi.pi.sigma}) or $\hat\xi_{c,e\pi\sigma}$ (\ref{eq:xi.e.pi.sigma}), which are not IF-based. As outcome modeling is needed to estimate $\xi_c$ for the sensitivity analysis, however, we recommend switching to a type A or IF-based estimator for the PI-based main analysis.

\section{Other topics}\label{sec:technical-rest}

\subsection{Using data in considering the range of the MR and SMD parameters}\label{sec:calibration}

We now return to the issue that certain sensitivity parameters may predict extreme $\mu_{0c}(X)$ values. A example concerns the outcome \textit{earnings} in our illustrative study. Since earnings span a large range, it may be intuitive to think about the earnings as differing in a multiplicative rather than additive manner, so a researcher may choose to use A4-MR for a sensitivity analysis. But earnings are not unbounded, and there is a maximum earning in the dataset, so we would be right to worry that certain sensitivity MR values may predict some $\mu_{0c}(X)$ values that are too high. 
A4-SMDe also has the same issue (to a lesser degree), where predicted $\mu_{0c}(X)$ values may be too high or too low. A4-GOR and A4-OR, on the other hand, predict within bounds.

We can use the data to gauge what values of the MR or SMD sensitivity parameter may be extreme, if we are willing to also specify bounds for the stratum-specific conditional $Y_0$ means, $\mu_{0c}(X)$. With A4-MR (and a non-negative outcome), we fix an upper bound (B) for $\mu_{0c}(X)$. With A4-SMDe, we fix a pair of upper ($\text{B}_h$) and lower ($\text{B}_l$) bounds. These can be informed by the observed outcome distribution, but are not necessarily bounds on the outcome itself. They are required to satisfy $\text{B}\geq\hat\mu_0(X)$ or $\text{B}_l\leq\hat\mu_0(X)\leq\text{B}_h$ for all $X$ values in the data.

For each $X$ value, we can obtain an interval for the MR/SMD sensitivity parameter that does not predict $\mu_{0c}(X)$ outside of these assumed bounds. (This interval is derived in Appendix~\ref{appendix:calibration}, see Propositions~\ref{thm:ok-mr} and \ref{thm:ok-smd}.) We estimate such intervals for all covariate values and examine the distributions of their upper and lower ends to judge which ranges of the sensitivity parameter should not be allowed -- see application in the illustrative example in Section~\ref{sec:illustration}.

Note that while this helps guard against mathematically implausible values, it does not replace careful consideration based on substantive knowledge,  which is important for deciding which range is practically plausible and relevant to the specific application.

\subsection{Confidence interval estimation}\label{sec:confidence-interval}

The application in this paper estimates nuisance functions (e.g., propensity score, principal score and outcome mean) parametrically, for simplicity. 
All the estimators in sections \ref{sec:pi-estimators}, \ref{sec:ratio-params} and \ref{sec:diff-param} are M-estimators. With parametric nuisance estimation, they are asymptotically normal and analytic standard errors can be derived using M-estimation calculus \cite{Stefanski2002}, and the bootstrap is also valid. In our illustration below, we bootstrap and construct BCa confidence intervals \cite{Efron1987}.

\subsection{Rate conditions for nonparametric estimation}\label{sec:rate-conditions}

With a view to inform nonparametric inference (not the focus of this paper), we derive rate conditions on nonparametric nuisance estimation for IF-based estimators (using sample splitting or cross fitting) to be $\sqrt{n}$-consistent and asymptotically normal. See Propositions~\ref{thm:rates-pi} and \ref{thm:rates-sens} in Appendix~\ref{appendix:rate-conditions} for these results under PI and under the sensitivity assumptions, respectively.
To our knowledge, our results are the first on rate conditions for sensitivity analyses for PI violation.
They show that while PI-based analysis only requires typical rate conditions on several error products of nuisance functions (e.g., $||\hat e_1(X)-e_1(X)||_2||\hat\pi_c(X)-\pi_c(X)||=o_p(n^{-1/2})$), the sensitivity analyses require rate conditions on single nuisance functions (due to the presence of square errors in the remainder bias term). Specifically, we require $||\hat\pi_c(X)-\pi_c(X)||_2=o_p(n^{-1/4})$ with all the sensitivity analyses, and additionally $||\hat\mu_0(X)-\mu_0(X)||_2=o_p(n^{-1/4})$ with the GOR- and SMDe-based sensitivity analyses, and $||\hat\sigma_0^2(X)-\sigma_0^2(X)||_2=o_p(n^{-1/4})$ with the SMDe-based sensitivity analyis. These results immediately connect to the earlier results on the robustness under PI, and (partial) loss of robustness under sensitivity assumptions, of IF-based estimation.

\subsection{Finite-sample bias}\label{sec:finite-sample-bias}

There is not an ideal choice for the placement of this topic. It is easier to read after reading the illustrative analysis in the next section. But we put it here for it is a small \textit{other} topic.

Many consistent estimators are biased in finite samples. Methods to reduce such bias
\cite{efron1994IntroductionBootstrap,Chang2015} are not often used, perhaps because the bias tends to be small, and the correction is complicated. The data example, however, reveals an interesting pattern of bias specific to sensitivity analysis that is worth noting. It is seen with the different outcomes and different estimators. An instance of this pattern is shown in Figure \ref{fig:finitesamplebias}; all instances are shown in Appendix~\ref{appendix:finite-sample-bias}.

\begin{figure}[t!]
    \centering
    \caption{Point estimate and iterated bootstrap mean estimates. Plots are shown for the outcome \textit{work for pay}.}
    \includegraphics[width=.8\textwidth]{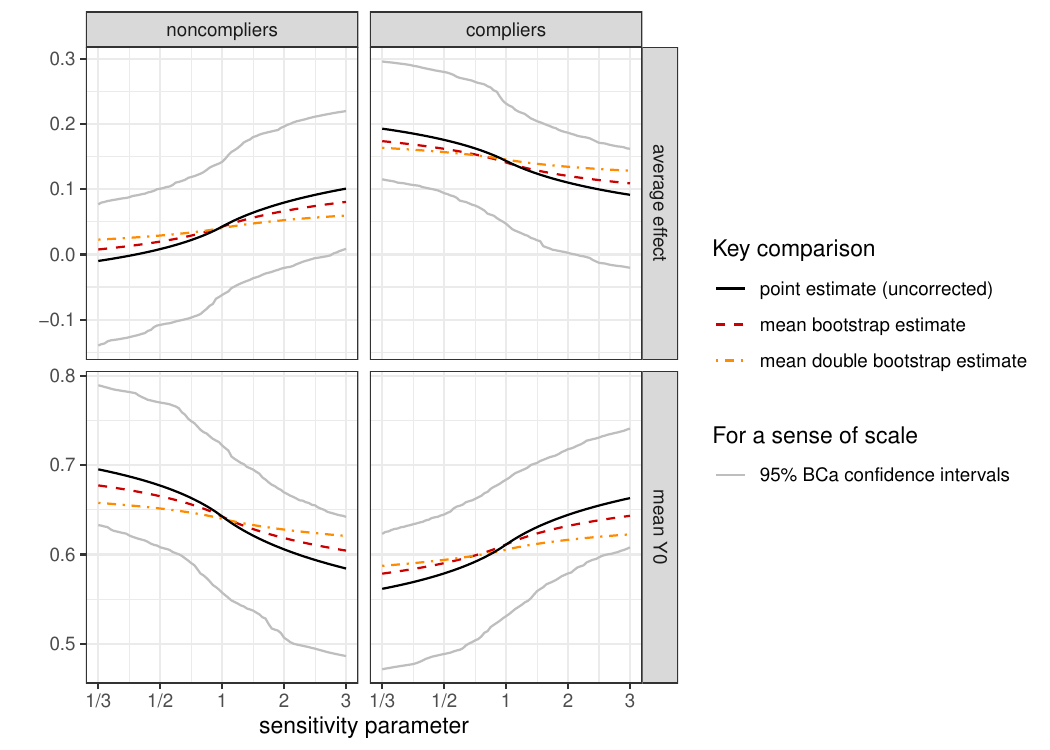}
    \label{fig:finitesamplebias}
\end{figure}

In Figure~\ref{fig:finitesamplebias}, the solid black curve is the point estimate (which we refer to generically as $\hat\theta$), the dashed red curve is the mean of bootstrap estimates ($\bar{\hat\theta}^*$), and the dashed orange curve is the mean of estimates from the double bootstrap (bootstrap of bootstrap samples)~($\bar{\hat\theta}^{**}$).
The shared pattern in all sensitivity analyses is that the slope of the $\bar{\hat\theta}^*$ curve is less steep than that of the $\hat\theta$ curve, and the slope of the $\bar{\hat\theta}^{**}$ is even less steep. 
(Note that the steepness of the curve indicates the degree to which sensitivity analysis estimates depart from the main analysis estimate.)
For the two outcomes \textit{work} and \textit{depressive symptoms}, where the differences between $\hat\theta$, $\bar{\hat\theta}^*$ and $\bar{\hat\theta}^{**}$ are minimal in the main analysis, this means that in the sensitivity analysis $\bar{\hat\theta}^*$ tends to be less extreme than $\hat\theta$, and $\bar{\hat\theta}^{**}$ tends to be even less extreme; and this gets more pronounced the farther the sensitivity parameter is from its null value.

Finite-sample bias deserves dedicated investigation, which is outside the scope of this paper.
This specific pattern, however, begs the question why. Our intuition is that it may be due to the fact that $\tau_{0c}$ is a weighted average of $\mu_{0c}(X)$ where the weights are $\pi_c(X)$, and under sensitivity assumptions the quantity being averaged $\mu_{0c}(X)$ depends on the weight $\pi_c(X)$. Specifically, 
with a fixed $\mu_0(X)$, $\mu_{0c}(X)$ is (i)~monotone decreasing in $\pi_c(X)$ for $\rho_c>1$ or $\eta_c>0$, and (ii)~monotone increasing in $\pi_c(X)$ for $\rho_c<1$ or $\eta_c<0$ (see Proposition~\ref{thm:mu.0c-monotone-in-pi.c} in Appendix~\ref{appendix:finite-sample-bias}). This results in a coupling of (a) any deviation (of 
the finite sample from the population) in the weight with (b) a deviation in the quantity being averaged -- in the opposite direction for case (i) and the same direction for case (ii). 
The resulting finite-sample bias is
an attenuation of the difference between the sensitivity analysis and main analysis estimates.

For the data example we use a bootstrap-based bias correction after conducting a focused simulation study (see Appendix~\ref{appendix:finite-sample-bias}). This bias correction is also implemented in our R-package.

\section{JOBS II illustration}\label{sec:illustration}

De-identified JOBS II data were accessed from the Inter-University Consortium for Political and Social Research data archive (\url{www.icpsr.umich.edu}). Our analysis focuses on the set of participants who were identified at initial screening as being at high risk for developing depression \cite{Vinokur1995}. For illustrative purposes, we further subset to participants with complete data (n=465) and treat the resulting dataset as if it were an observational study. (Due to this restriction of the sample, analysis results should be seen as merely illustrative and not taken as substantive findings.) We consider three outcomes: working for pay (binary), monthly earnings (non-negative), and depressive symptoms (a score ranging from 1 to 5) at six months post-treatment.
The study has a rich set of baseline covariates including demographics, household characteristics, employment history, motivation, and depressive symptoms. Given these covariates, we assume treatment assignment ignorability. We also assume PI in the main analysis.

Table~\ref{tab:Xdistributions} summarizes the covariate distribution (i)~in the full analysis sample; (ii)~stratified by compliance type in the treatment group (to give a sense of $X$-$C$ associations); and (iii)~stratified by the binary work-for-pay outcome in the control group (to give a sense of $X$-$Y_0$ associations).
Compared to noncompliers, compliers were more likely to be male, White, older and have a college degree. They were more likely to have ever married and have fewer cohabiting children, and less likely to have low household income. They were more likely to have had a professional job as their last steady job, to have been unemployed for a shorter time, and to report slightly higher job-seeking and program-participation motivation. In the control condition, participants who were younger, White, higher educated, unemployed for a shorter period, a manager at their last steady job, or reported higher motivation were more likely to be employed at six months.

We aim to illustrate the use of the sensitivity assumptions introduced above with the different outcomes, and show how sensitivity analysis effect estimates depart from PI-based estimates. For this purpose, any type A or type B estimator suffices. We suppose that a researcher has chosen to use the H\'ajek-type IF-based estimator (\ref{est:pi-ifh}) for the PI-based analysis. We will briefly describe an implementation of this estimator, and then will focus on the sensitivity analyses.

We report bias-corrected point estimates (see Section~\ref{sec:finite-sample-bias}) and BCa confidence intervals.

\begin{table}
    \centering
    \centering
    \caption{Baseline covariates in (1) full analysis sample; (2) propensity-score-weighted treatment group, stratified by compliance type; and (3) propensity-score-weighted control group, stratified by outcome \textit{work for pay}}\label{tab:Xdistributions}
    \vspace{-1em}
    \resizebox{\textwidth}{!}{%
    \begin{tabular}{l r@{\hspace{1\tabcolsep}}l c r@{\hspace{1\tabcolsep}}l r@{\hspace{1\tabcolsep}}l c r@{\hspace{1\tabcolsep}}l r@{\hspace{1\tabcolsep}}l}
        & \multicolumn{2}{c}{Full} && \multicolumn{4}{c}{Treatment group} && \multicolumn{4}{c}{Control group}
        \\
        & \multicolumn{2}{c}{analysis} && \multicolumn{4}{c}{propensity-score-weighted} && \multicolumn{4}{c}{propensity-score-weighted}
        \\\cline{5-8}\cline{10-13}
        & \multicolumn{2}{c}{sample} && \multicolumn{2}{c}{compliers} & \multicolumn{2}{c}{noncompliers} && \multicolumn{2}{c}{work} & \multicolumn{2}{c}{not work}
        \\
        & \multicolumn{2}{c}{(n=465)} && \multicolumn{2}{c}{(n=172)} & \multicolumn{2}{c}{(n=139)} && \multicolumn{2}{c}{(n=96)} & \multicolumn{2}{c}{(n=58)}
        \\
        & & && \multicolumn{2}{c}{(n.wt=256.6)} & \multicolumn{2}{c}{(n.wt=208.0)} && \multicolumn{2}{c}{(n.wt=303.3)} & \multicolumn{2}{c}{(n.wt=152.1)}
        \\\hline
        & \textit{mean} & \textit{(SD)} && \textit{mean} & \textit{(SD)} & \textit{mean} & \textit{(SD)} && \textit{mean} & \textit{(SD)} & \textit{mean} & \textit{(SD)}
        \\
        & \textit{or \%} & \textit{(count)} && \textit{or \%} & \textit{(count)} & \textit{or \%} & \textit{(count)} && \textit{or \%} & \textit{(count)} & \textit{or \%} & \textit{(count)}
        \\[.5em]
        Age & 36.5 & (9.9) && 39.0 & (9.7) & 33.5 & (9.8) && 35.2 & (9.4) & 38.6 & (11.3)
        \\[.2em]
        Sex (female) & 57.6\% & (268) && 53.4\% & (137) & 62.9\% & (130.9) && 59.4\% & (180.3) & 56.7\% & (86.3)
        \\[.2em]
        Race (white) & 81.7\% & (380) && 85.1\% & (218.5) & 78.6\% & (163.5) && 87.1\% & (264.3) & 73.7\% & (112.1)
        \\[.6em]
        Education
        \\
        ~~~less than high school & 10.5\% & (49) && 7.3\% & (18.8) & 16.5\% & (34.3) && 8.5\% & (25.9) & 16.3\% & (24.8)
        \\
        ~~~high school & 29.7\% & (138) && 26.3\% & (67.5) & 31.7\% & (66.0) && 25.6\% & (77.5) & 35.6\% & (54.2)
        \\
        ~~~some college & 38.9\% & (181) &&  37.2\% & (95.5) & 41.7\% & (86.7) && 42.6\% & (129.0) & 34.7\% & (52.8)
        \\
        ~~~Bachelor's degree & 13.1\% & (61) && 19.4\% & (49.9) & 5.7\% & (11.8) && 14.6\% & (44.3) & 9.5\% & (11.4)
        \\
        ~~~graduate studies & 7.7\% & (36) && 9.8\% & (25.0) & 4.4\% & (9.3) && 8.7\% & (26.4) & 3.9\% & (5.9)
        \\[.6em]
        Marital status
        \\
        ~~~never married & 34.4\% & (160) && 31.8\% & (81.6) & 38.0\% & (79.0) && 35.2\% & (106.8) & 35.5\% & (54.0)
        \\
        ~~~married & 38.7\% & (180) && 37.3\%  & (95.7) & 38.5\% & (80.2) && 35.4\% & (107.3) & 39.8\% & (60.6)
        \\
        ~~~divorced/separated/widowed & 26.9\% & (125) && 30.9\% & (79.2) & 23.5\% & (48.8) && 29.4\% & (89.2) & 24.6\% & (37.5)
        \\[.2em]
        Kids in household & 0.93 & (1.13) && 0.85 & (1.12) & 0.95 & (1.17) && 0.80 & (1.05) & 0.98 & (1.03)
        \\[.2em]
        Household income
        \\
        ~~~under 15K & 22.8\% & (106) && 19.3\% & (49.4) & 26.7\% & (55.6) && 19.0\% & (57.5) & 31.3\% & (47.6)
        \\
        ~~~15K to under 25K & 24.9\% & (116) && 22.1\% & (56.6) & 29.6\% & (61.6) && 34.3\% & (104.1) & 15.1\% & (23.0)
        \\
        ~~~25K to under 40K & 25.8\% & (120) && 28.6\% & (73.3) & 23.3\% & (48.6) && 25.6\% & (77.7) & 24.0\% & (36.5)
        \\
        ~~~40K to under 50K & 10.8\% & (50) && 12.6\% & (32.4) & 7.7\% & (16.0) && 6.7\% & (20.2) & 15.3\% & (23.2)
        \\
        ~~~50K or more & 15.7\% & (73) && 17.5\% & (44.9) & 12.6\% & (26.2) && 14.4\% & (43.8) & 14.3\% & (21.8)
        \\[.2em]
        Economic hardship & 3.62 & (0.92) && 3.52 & (0.92) & 3.78 & (0.92) && 3.73 & (0.91) & 3.52 & (1.00)
        \\[.6em]
        Occupation (last steady job)
        \\
        ~~~professional & 18.5\% & (86) && 26.7\% & (68.5) & 9.1\% & (18.9) && 17.1\% & (51.8) & 18.8\% & (28.6)
        \\
        ~~~managerial & 17.2\% & (80) && 14.7\% & (37.6) & 19.5\% & (40.6) && 18.4\% & (55.9) & 10.8\% & (16.4)
        \\
        ~~~clerical & 23.4\% & (109) && 23.9\% & (61.3) & 23.2\% & (48.2) && 22.9\% & (69.6) & 26.3\% & (40.1)
        \\
        ~~~sales & 6.5\% & (30) && 5.2\% & (13.3) & 7.7\% & (16.0) && 7.9\% & (23.8) & 3.0\% & (4.6)
        \\
        ~~~crafts/foremen & 12.9\% & (60) && 13.6\% & (34.8) & 12.2\% & (25.4) && 10.9\% & (33.0) & 18.8\% & (28.6)
        \\
        ~~~operative & 9.5\% & (44) && 5.2\% & (13.3) & 14.6\% & (30.4) && 9.5\% & (28.9) & 7.3\% & (11.1)
        \\
        ~~~labor/service & 12.0\% & (56) && 10.8\% & (27.7) & 13.7\% & (28.5) && 13.3\% & (40.3) & 15.0\% & (22.7)
        \\[.2em]
        Weeks unemployed & 9.3 & (11.0) && 8.1 & (10.3) & 10.4 & (11.1) && 8.0 & (9.4) & 10.5 & (12.5)
        \\[.6em]
        Motivation to participate & 5.34 & (0.80) && 5.50 & (0.79) & 5.19 & (0.78) && 5.41 & (0.73) & 5.37 & (0.83)
        \\[.2em]
        Job-seeking motivation & 82 & (17) && 84 & (15) & 81 & (19) && 85 & (16) & 76 & (17)
        \\[.2em]
        Job-seeking self-efficacy & 3.59 & (0.83) && 3.48 & (0.84) & 3.70 & (0.82) && 3.66 & (0.76) & 3.44 & (0.84)
        \\[.2em]
        Assertiveness & 2.99 & (0.82) && 2.90 & (0.82) & 3.07 & (0.82) && 2.97 & (0.81) & 2.94 & (0.79)
        \\[.6em]
        Depressive symptoms & 2.34 & (0.68) && 2.34 & (0.69) & 2.36 & (0.68) && 2.42 & (0.70) & 2.25 & (0.60)
        \\\hline
    \end{tabular}%
    }
    \caption*{\scriptsize n.wt = weighted subsample size. Ranges of continuous/interval variables: age 17 to 77; kids in households 0 to 5 (one observation $>$5 truncated to 5), economic hardship 1 to 5; weeks unemployed 1 to 52 (12 observations $>$52 truncated to 52); motivation to participate 1 to 6.5; job-seeking motivation 0 to 100; job-seeking self-efficacy 1 to 5; assertiveness 1 to 5; depressive symptoms 1 to 5.}
    \vspace{.5em}
    \caption{PI-based analysis results: point estimates (and 95\% BCa confidence intervals)}\label{tab:main}
    \vspace{-.5em}
    \resizebox{\textwidth}{!}{%
    \begin{tabular}{lccccccccccc}
        & \multicolumn{3}{c}{compliers} && \multicolumn{3}{c}{noncompliers}
        \\\cline{2-4}\cline{6-8}
        \\[-.9em]
        outcome 
        & mean $Y_1$ ($\tau_{11}$) & mean $Y_0$ ($\tau_{01}^\text{PI}$) & CACE ($\Delta_1^\text{PI}$) 
        && mean $Y_1$ ($\tau_{10}$) & mean $Y_0$ ($\tau_{00}^\text{PI}$) & NACE ($\Delta_0^\text{PI}$)
        \\\hline
        work
        & 75.4\% & 61.1\% & 14.3 percentage points && 68.5\% & 64.2\% & 4.3 percentage points
        \\
        & (69.4, 81.4) & (53.1, 68.4) & (4.7, 23.1)
        && (60.7, 75.2) & (55.7, 72.2) & ($-6.2$, 14.2)
        \\\hline
        earnings
        & \$1,279 & \$1,014 & \$266 && \$928 & \$835 & \$92 
        \\
        & (1,107, 1,452) & (802, 1,221) & (18, 530) && (776, 1,115) & (666, 972) & (-90, 318)
        \\\hline
        depressive
        & 1.90 & 2.07 & -0.18 && 2.05 & 2.02 & 0.02
        \\
        symptoms
        & (1.80, 1.99) & (1.96, 2.20) & (-0.32, -0.04) && (1.94, 2.16) & (1.88, 2.14) & (-0.12, 0.18)
        \\\hline
    \end{tabular}%
    }
    \caption*{\scriptsize Variable \textit{work} is binary. Actual \textit{earnings} range is \$0-5,667. Scale range of \textit{depressive symptoms} is 1 to 5.\hfill~}
    \vspace{1em}
\end{table}

\subsection{PI-based main analysis}\label{sec:jobs-main}

The estimator $\hat\Delta_{c,\textsc{ifh}}^\text{PI}$ requires estimating several nuisance functions. We make relatively simple choices, keeping in mind what applied researchers may use in practice.
We use logistic regression to fit the propensity score ($e(Z,X)$) and principal score ($\pi_c(X)$) models. These models include all baseline covariates, plus squares and square roots of continuous covariates; the inclusion of these additional terms is meant to improve covariate balance to be obtained from principal score and inverse propensity score weighting. We check balance as suggested in~ \cite{Ding2017}~\!\!: Figure~\ref{fig:love.plots} (in Appendix~\ref{appendix:example}) shows that covariate balance is improved (i)~between the treated and control groups after propensity score weighting, and (ii)~between treated (non)compliers and controls after principal score weighting combined with propensity score weighting. 

Next, we estimate the conditional outcome mean functions for treated compliers ($\mu_{11}(X)$), treated noncompliers ($\mu_{10}(X)$) and controls ($\mu_0(X)$). With the binary outcome \textit{work for pay}, we use logistic regression. 
For the outcome \textit{earnings}, the means are estimated conditional on working using gamma regression with log link.
and then multiplied with the probability of working predicted by the \textit{work for pay} model.
(Small detail: since we use a noncanonical link with the gamma model, the predictions are slightly mean-biased; we calibrate them by a multiplicative constant to remove this bias.)
For the \textit{depressive symptoms} outcome, we use a simple transformation to the $[0,1]$ interval (by subtracting $l=1$ and dividing by $h-l=4$), fit a quasi-logistic model to the transformed outcome to estimate the conditional means, and then transform the means back to the original scale.
These models include all baseline covariates.

We use targeted nuisance estimation (see Remark~\ref{rm:targeted-nuisance-estimation}). The $\pi_c(X)$, $\mu_{11}(X)$ and $\mu_{10}(X)$ models are fit to data (treated group, treated compliers and treated noncompliers, respectively) weighted by $1/\hat e(Z,X)$. The $\mu_0(X)$ model is fit twice, to the control group weighted by $\hat\pi_1(X)/\hat e(Z,X)$ and weighted by $\hat\pi_0(X)/\hat e(Z,X)$, for CACE and NACE estimation, respectively.

Results (see Table \ref{tab:main}) suggest that assignment to the intervention resulted in increased employment and earnings and decreased depressive symptoms for compliers. For noncompliers, effect estimates are close to null.

\subsection{Sensitivity analysis}

We now demonstrate sensitivity analyses that are OR-based for \textit{work for pay}, MR-based for \textit{earnings}, and GOR- and SMDe-based for \textit{depressive symptoms}.

\paragraph{OR-based sensitivity analysis: work for pay}\label{sec:mean-based-work3}\hfill

\noindent
We noted above that some baseline factors such as socio-economic advantage and motivation are positively associated both with being a complier ($C$) and with the \textit{work for pay} outcome under control ($Y_0$). One might be concerned whether, within subpopulations homogeneous in the observed covariates, there are other advantage type factors that are \textit{unobserved} that relate to $C$ and $Y_0$ in a similar way; in that case the PI-based analysis might have overestimated the CACE and underestimated the NACE. On the other hand, one might be concerned that among people with the same $X$, some may not have needed to participate in the training because they had good prospects of finding a job; in that case the PI-based analysis might have been biased in the opposite direction. We thus consider a range of sensitivity OR values spanning both sides of 1. 
The results of this sensitivity analysis (Figure \ref{fig:meanbased-sens}, top left) suggest that, even if (within levels of $X$) compliers had double the odds (relative to noncompliers) of getting work \textit{without} the intervention, the intervention's effect on having work for compliers would still be positive.

\begin{figure}[t]
    \caption{Sensitivity analysis results: point estimates and 95\% point-wise CIs for CACE, NACE and stratum-specific potential outcome means, for the range of the sensitivity parameter.}
    \includegraphics[width=.93\textwidth, right]{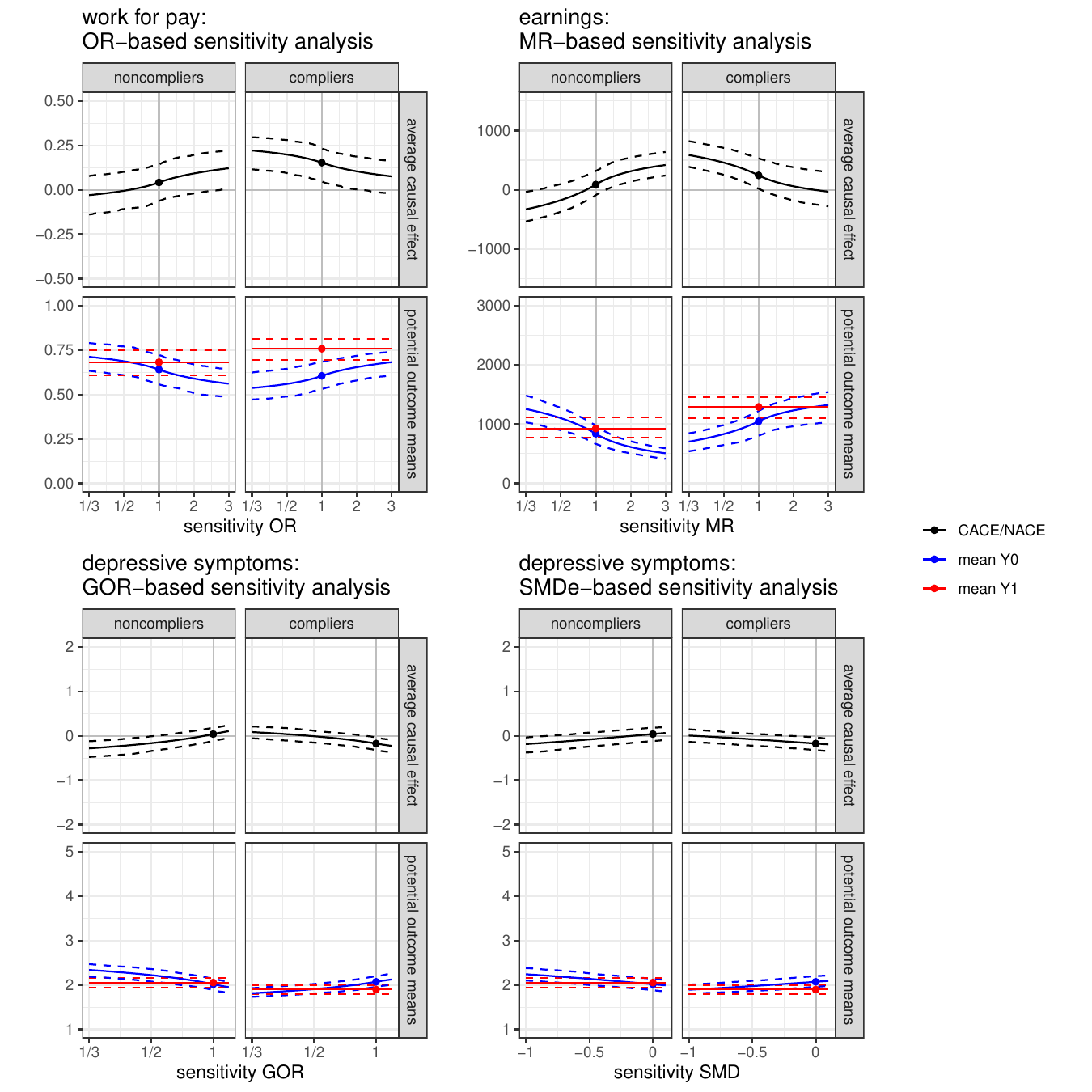}
    \vspace{-15pt}
    \label{fig:meanbased-sens}
\end{figure}

\paragraph{GOR-based sensitivity analysis: depressive symptoms}\hfill

\noindent
A concern may be that even among people with the same baseline covariate values (including baseline depressive symptoms score), compliers may be those who were more robust in some way (e.g., better at getting out of bed in the morning), and therefore may have better outcome (i.e., lower depressive symptoms at six months) under control than noncompliers. Therefore we consider sensitivity GOR values smaller than 1 (Figure \ref{fig:meanbased-sens}, bottom left).
The CACE estimate is quite sensitive to PI violation. It is negative (indicating a reduction in depressive symptoms) under PI, but as the sensitivity GOR deviates only slightly from 1, it quickly approaches zero. 

\paragraph{MR-based sensitivity analysis: earnings}\hfill

\noindent
With this outcome, we use the MR sensitivity parameter. To illustrate the method as it would typically be used, we treat \textit{earnings} as a stand-alone outcome, using $\hat\mu_0(X)$ as the only input, putting aside its connection with the \textit{work for pay} outcome.

We start with a tentative MR range from 1/3 to 3, which is covered in Figure~\ref{fig:meanbased-sens} (top right).
As mentioned earlier, it is challenging to choose what range to consider for the sensitivity parameter. Most important to this decision is substantive knowledge, including opinions of experts and study staff (who might know participants better than what is captured in the covariates). Such knowledge should be used, whenever it is available, to help \textit{rule in} which range of the sensitivity parameter is practical and relevant. 

As discussed in Section~\ref{sec:calibration}, the data can help \textit{rule out} some implausible ranges. Here we simply use the maximum reported earnings under control (\$5,667) as the upper bound B for $\mu_{0c}(X)$. After computing covariate-specific ``legal'' intervals for the sensitivity parameter, we use their end points to make the plot on the left in Figure~\ref{fig:calibration}, which shows the proportion of the sample with either $\hat\mu_{01}(X)$ or $\hat\mu_{00}(X)$ exceeding B under each MR value. We do not restrict the MR range based on this plot, as it suggests limited bound contradiction. (Alternatives include (i) restricting the MR range, or (ii) modifying the assumption to let the MR be $\rho$ for $X$ values where $\rho$ is in the legal interval, and otherwise be the legal value closest to $\rho$.)

\begin{figure}[h!]
    \caption{Bounds violation diagnostic: proportion contradicting bounds as a function of the sensitivity parameter}
    \label{fig:calibration}
    \centering
    \includegraphics[width=.47\textwidth]{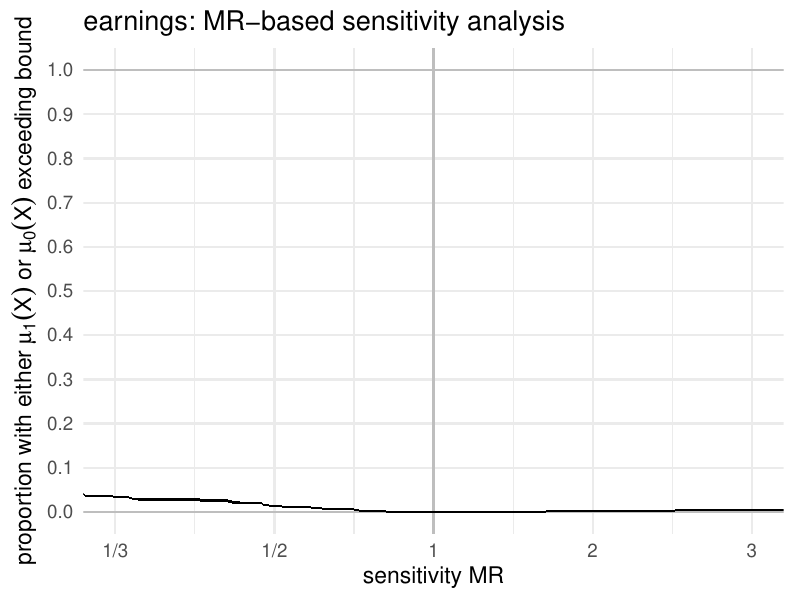}~~~
    \includegraphics[width=.47\textwidth]{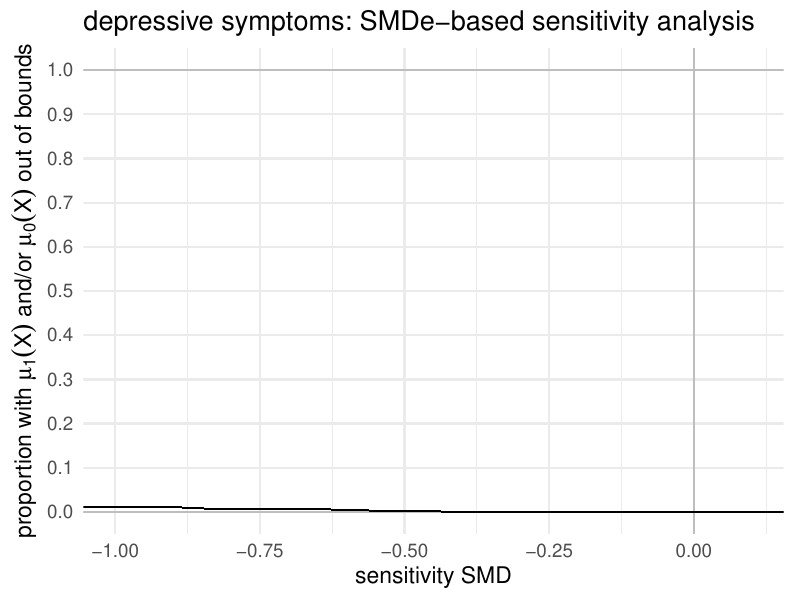}
\end{figure}

Another way to rely on the data is to examine what the MR values imply about the distributions of $\mu_{01}(X)$ values among compliers and of $\mu_{00}(X)$ values among noncompliers. Figure~\ref{fig:MRimplication} (in Appendix~\ref{appendix:example}) plots these implied distributions for several MR values on $[1/3,3]$. To judge whether such distributions are plausible, again, one should rely substantive knowledge if possible. Also, very large $\mu_{01}(X)$ or $\mu_{00}(X)$ values (especially those substantially larger than the maximum reported earnings) are suspect. Based on this, one might consider excluding MR values at the low end (1/3) and at the high end ($\geq2$).

Another possibility is to supplement the A4-MR with other assumptions based on substantive knowledge. Suppose, for example, that substantive experts think it is unlikely that being assigned to the intervention is harmful to noncompliers (a relaxation of the ER assumption). Based on the results plot in Figure \ref{fig:meanbased-sens}, this would narrow attention to the MR range above 1/2.

\paragraph{SMDe-based sensitivity analysis: depressive symptoms}
\hfill

\noindent
Suppose that for the depressive symptoms outcome, an investigator prefers a sensitivity analysis based on A4-SMDe, being more comfortable communicating about mean differences. Here also, we consider a sensitivity SMD range to the left of the null value, where within $X$ levels, complier and noncomplier outcome means under control may differ by up to one standard deviation.
Results (Figure \ref{fig:meanbased-sens}, bottom right) look similar to those from the GOR-based sensitivity analysis, although using a different sensitivity parameter.

We note two details. First, this sensitivity analysis requires estimating the conditional variance $\sigma_0^2(X)$. Using the quasi-likelihood approach, we assume that $\sigma_0^2(X)$ is proportional to $[\mu_0(X)-l][h-\mu_0(X)]$. This is equivalent to assuming that the outcome, after being shifted and rescaled to the [0,1] interval, follows a quasibinomial model conditional on covariates. Recall that in the PI-based analysis, we transformed this outcome to the [0,1] interval and fit a model with logit link. We now manually extract the dispersion parameter $\hat\phi$ from this model and use it to compute the variance estimate $\hat\sigma_0^2(X)=\hat\phi[\hat\mu_0(X)-l][h-\hat\mu_0(X)]$. 
Second, the plot on the right of Figure~\ref{fig:calibration} shows that for the SMD range considered there is minimal contradiction with the $\mu_{0c}(X)$ bounds, which here are simply set to the minimum and maximum depressive symptom scores. This is expected, as we consider a modest SMD range.


We do not conduct an SMDe-based sensitivity analysis for the outcome earnings, because the equal variance part of A4-SMDe is likely grossly incorrect for that outcome.


\section{Discussion}\label{sec:discussion}

This paper substantially expands options for sensitivity analysis for PI violation in the estimation of complier and noncomplier average causal effects in two ways. First, we consider several sensitivity models with different sensitivity parameters (OR, GOR, MR, SMD) suitable for different outcome types and reflecting different ways compliers and noncompliers may differ with respect to outcome under control. Second, rather than proposing one estimator under the sensitivity model, we tailor sensitivity analysis techniques to different types of estimators (outcome regression, IF-based and weighting) that may be used for the PI-based main analysis. 

There are several future directions for this line of sensitivity analysis. 
One is to incorporate data-adaptive nuisance estimation. As noted, the robustness available for PI-based analysis via IF-based estimation is partially lost for sensitivity analysis, making it more important that we estimate nuisance functions well. We provide rate conditions, but otherwise leave this to future work. 
Also important is how to handle missing data. Missing-at-random cases can be handled by standard techniques, but given the difference in compliance type observability between treatment arms, one may wish to allow certain not-at-random missingness, e.g., outcome missingness that depends on compliance type \cite{frangakis1999AddressingComplicationsIntentiontotreat}.
Another extension is to adapt the methods 
to accommodate two-sided noncompliance and non-binary $S$, 
which are also common settings.

For the two-sided noncompliance case, extension is conceptually straightforward: wherever a PI assumption is used to disentangle a mixture it can be replaced with a sensitivity assumption. With binary $Z$ and $S$ there are four mixtures, so if PI assumptions are invoked to disentangle all four, then replacing those assumptions requires four sensitivity parameters. If one assumes away one principal stratum (say, defiers) to identify stratum prevalences and covariate distributions, then two mixtures remain, which means a PI-based analysis requires two PI assumptions and the sensitivity analysis involves two sensitivity parameters -- see~ \cite{jiang2022MultiplyRobustEstimationa} for a sensitivity analysis using two MR parameters. While the idea is simple, works needs to be done to consider different (types of) PI-based estimators and pair them with sensitivity analysis techniques.

The methods in this paper belong to a mean-centric approach to sensitivity analysis. Each assumes a connection between two conditional outcome mean functions of complier and noncompliers. For a binary outcome, the sensitivity analysis based on A4-OR fully respects the observed outcome distribution. For continuous outcomes, however, the sensitivity analyses based on A4-GOR, A4-MR and A4-SMDe alone may conflict with the observed outcome distribution. The MR-based model may predict out of range because it treats the outcome as unbounded. The other two methods use some additional information: the GOR-based model takes in user-specified outcome bounds and respects those bounds; the SMDe-based model is informed about conditional outcome variability and with that information offers a scale-free sensitivity parameter. To mitigate the out-of-range prediction problem that affects the MR-based and to a lesser degree of the SMD-based method, we propose a simple technique that requires an additional assumption of bounds on stratum-specific conditional outcome means. There remains, however, the risk of more subtle conflict (e.g., predicting mean outcome in the tail of the distribution).
A different approach is to avoid conflicting with the observed data distribution \cite{Scharfstein2021,Franks2020,Robins2000} all together by anchoring on the conditional distribution of the outcome  under control rather than just its mean plus bounds/variance. 
Such sensitivity analysis (described briefly in the preprint ~\cite{nguyen2023sensitivity}) will be presented in a separate manuscript.

One last comment: This paper provides technical solutions for doing sensitivity analysis, but does not address how to choose a relevant range for the sensitivity parameter and how to elicit and use expert opinion for this purpose. This is a topic that should receive more attention.

\section*{Acknowledgements}

This work is partially supported by grants R03MH128634, R01MH115487 and U24OD023382 from the National Institutes of Health, and N00014-21-1-2820 from the Office of Naval Research. The quality of the work and the clarity of its presentation have been improved thanks to feedback from anonymous reviewers. TQN thanks Drs.~Ilya Shpitser, Bonnie Smith and Razieh Nabi for helpful discussions about influence functions, and Drs.~Constantine Frangakis and Scott Zeger for thought-provoking comments at an early presentation of this work. The authors appreciate the participants, staff and investigators of the JOBS II study, and the ICPSR data archive.

\bibliographystyle{ama}
\bibliography{refs.bib}

\begin{thebibliography}{10}

\bibitem{Frangakis2002}
Frangakis Constantine~E., Rubin Donald~B.. {Principal stratification in causal
  inference}  {\it Biometrics. } 2002;58:21--29.

\bibitem{Rubin1974}
Rubin Donald~B. {Estimating causal effects of treatments in randomized and
  nonrandomized studies.}  {\it Journal of Educational Psychology. }
  1974;66:688--701.

\bibitem{Rubin2006}
Rubin Donald~B.. {Causal inference through potential outcomes and principal
  stratification: Application to studies with ``censoring'' due to death}  {\it
  Statistical Science. } 2006;21:299--309.

\bibitem{Griffin2008}
Griffin Beth~Ann, McCaffrey Daniel~F., Morral Andrew~R.. {An application of
  principal stratification to control for institutionalization at follow-up in
  studies of substance abuse treatment programs}  {\it The Annals of Applied
  Statistics. } 2008;2:1034--1055.

\bibitem{Vinokur1995}
Vinokur Amiram~D., Price Richard~H., Schul Yaacov. {Impact of the JOBS
  intervention on unemployed workers varying in risk for depression}  {\it
  American Journal of Community Psychology. } 1995;23:39--74.

\bibitem{gruenewald2016BaltimoreExperienceCorps}
Gruenewald Tara~L., Tanner Elizabeth~K., Fried Linda~P., et al. The {Baltimore}
  {Experience} {Corps} {Trial}: {Enhancing} {generativity} via
  {intergenerational} {activity} {engagement} in {later} {life}  {\it The
  Journals of Gerontology Series B: Psychological Sciences and Social Sciences.
  } 2016;71:661--670.

\bibitem{daumit2013BehavioralWeightlossIntervention}
Daumit Gail~L., Dickerson Faith~B., Wang Nae-Yuh, et al. A behavioral
  weight-loss intervention in persons with serious mental illness  {\it New
  England Journal of Medicine. } 2013;368:1594--1602.

\bibitem{Angrist1995}
Angrist Joshua~D., Imbens Guido~W.. {Two-stage least squares estimation of
  average causal effects in models with variable treatment intensity}  {\it
  Journal of the American Statistical Association. } 1995;90:431--442.

\bibitem{marshall2016CoarseningBiasHow}
Marshall John. Coarsening {bias}: {How} {coarse} {treatment} {measurement}
  {upwardly} {biases} {instrumental} {variable} {estimates}  {\it Political
  Analysis. } 2016;24:157--171.

\bibitem{andresen2021InstrumentbasedEstimationBinarised}
Andresen Martin~E, Huber Martin. Instrument-based estimation with binarised
  treatments: issues and tests for the exclusion restriction  {\it The
  Econometrics Journal. } 2021;24:536--558.

\bibitem{Feller2017}
Feller Avi, Mealli Fabrizia, Miratrix Luke. {Principal score methods:
  Assumptions, extensions, and practical considerations}  {\it Journal of
  Educational and Behavioral Statistics. } 2017;42:726--758.

\bibitem{Jo2009}
Jo~Booil, Stuart Elizabeth~A.. {On the use of propensity scores in principal
  causal effect estimation}  {\it Statistics in Medicine. } 2009;28:2857--2875.

\bibitem{Ding2017}
Ding Peng, Lu~Jiannan. {Principal stratification analysis using principal
  scores}  {\it Journal of the Royal Statistical Society. Series B: Statistical
  Methodology. } 2017;79:757--777.

\bibitem{Jiang2021auxiliary}
Jiang Zhichao, Ding Peng. {Identification of causal effects within principal
  strata using auxiliary variables}  {\it Statistical Science. } 2021;36:1--49.

\bibitem{Stuart2015}
Stuart Elizabeth~A., Jo~Booil. {Assessing the sensitivity of methods for
  estimating principal causal effects}  {\it Statistical Methods in Medical
  Research. } 2015;24:657--674.

\bibitem{Jo2011}
Jo~Booil, Vinokur Amiram~D.. {Sensitivity analysis and bounding of causal
  effects with alternative identifying assumptions}  {\it Journal of
  Educational and Behavioral Statistics. } 2011;36:415--440.

\bibitem{wang2023SensitivityAnalysesPrincipal}
Wang Craig, Zhang Yufen, Mealli Fabrizia, Bornkamp Björn. Sensitivity analyses
  for the principal ignorability assumption using multiple imputation  {\it
  Pharmaceutical Statistics. } 2023;22:64--78.

\bibitem{Schwartz2012}
Schwartz Scott, Li~Fan, Reiter Jerome~P.. {Sensitivity analysis for unmeasured
  confounding in principal stratification settings with binary variables}  {\it
  Statistics in Medicine. } 2012;31:949--962.

\bibitem{Mercatanti2017}
Mercatanti Andrea, Li~Fan. {Do debit cards decrease cash demand?: Causal
  inference and sensitivity analysis using principal stratification}  {\it
  Journal of the Royal Statistical Society. Series C: Applied Statistics. }
  2017;66:759--776.

\bibitem{baiocchi2014InstrumentalVariableMethods}
Baiocchi Michael, Cheng Jing, Small Dylan~S.. Instrumental variable methods for
  causal inference  {\it Statistics in Medicine. } 2014;33:2297--2340.

\bibitem{jiang2022MultiplyRobustEstimationa}
Jiang Zhichao, Yang Shu, Ding Peng. Multiply robust estimation of causal
  effects under principal ignorability  {\it Journal of the Royal Statistical
  Society. Series B: Statistical Methodology. } 2022;84:1423--1445.

\bibitem{mcconnell2008TruncationbyDeathProblemWhat}
McConnell Sheena, Stuart Elizabeth~A., Devaney Barbara. The
  {truncation}-by-{death} {problem}: {What} {to} {do} in an {experimental}
  {evaluation} {when} the {outcome} {is} {not} {always} {defined}  {\it
  Evaluation Review. } 2008;32:157--186.

\bibitem{hajek1971comment}
Hájek Jaroslav. Comment on “{An} essay on the logical foundations of survey
  sampling, part one” by {Basu, D}  in {\it The foundations of survey
  sampling}:236Toronto: Holt, Rinehart, and Winston 1971.

\bibitem{wang2019AnalysisCovarianceRandomized}
Wang Bingkai, Ogburn Elizabeth~L., Rosenblum Michael. Analysis of covariance in
  randomized trials: {More} precision and valid confidence intervals, without
  model assumptions  {\it Biometrics. } 2019;75:1391--1400.

\bibitem{steingrimsson2017ImprovingPrecisionAdjusting}
Steingrimsson Jon~Arni, Hanley Daniel~F., Rosenblum Michael. Improving
  precision by adjusting for prognostic baseline variables in randomized trials
  with binary outcomes, without regression model assumptions  {\it Contemporary
  Clinical Trials. } 2017;54:18--24.

\bibitem{Franks2020}
Franks Alexander~M., D'Amour Alexander, Feller Avi. {Flexible sensitivity
  analysis for observational studies without observable implications}  {\it
  Journal of the American Statistical Association. } 2020;115:1730--1746.

\bibitem{Scharfstein2021}
Scharfstein Daniel~O., Nabi Razieh, Kennedy Edward~H., Huang Ming-Yueh, Bonvini
  Matteo, Smid Marcela. {Semiparametric sensitivity analysis: Unmeasured
  confounding in observational studies}   2021.
\newblock arxiv: 2104.08300.

\bibitem{Robins2000}
Robins James~M., Rotnitzky Andrea, Scharfstein Daniel~O.. {Sensitivity analysis
  for selection bias and unmeasured confounding in missing data and causal
  inference models}  in {\it Statistical Models in Epidemiology: The
  Environment and Clinical Trials}:1--94New York, NY: Springer New York 2000.

\bibitem{cohen1988StatisticalPowerAnalysis}
Cohen Jacob. {\it Statistical {Power} {Analysis} for the {Behavioral}
  {Sciences}}.
\newblock New York: Routledge2nd~ed. 1988.

\bibitem{stuart2010MatchingMethodsCausal}
Stuart Elizabeth~A.. Matching {methods} for {causal} {inference}: {A} {review}
  and a {look} {forward}  {\it Statistical Science. } 2010;25.

\bibitem{Stefanski2002}
Stefanski Leonard~A., Boos Dennis~D.. {The calculus of M-estimation}  {\it The
  American Statistician. } 2002;56:29--38.

\bibitem{Efron1987}
Efron Bradley. {Better bootstrap confidence intervals}  {\it Journal of the
  American Statistical Association. } 1987;82:171--185.

\bibitem{efron1994IntroductionBootstrap}
Efron Bradley, Tibshirani R.~J.. {\it An {Introduction} to the {Bootstrap}}.
\newblock CRC Press 1994.
\newblock Google-Books-ID: gLlpIUxRntoC.

\bibitem{Chang2015}
Chang Jinyuan, Hall Peter. {Double-bootstrap methods that use a single
  double-bootstrap simulation}  {\it Biometrika. } 2015;102:203--214.

\bibitem{frangakis1999AddressingComplicationsIntentiontotreat}
Frangakis C., Rubin Donald~B. Addressing complications of intention-to-treat
  analysis in the combined presence of all-or-none treatment-noncompliance and
  subsequent missing outcomes  {\it Biometrika. } 1999;86:365--379.

\bibitem{nguyen2023sensitivity}
Nguyen Trang~Quynh, Stuart Elizabeth~A., Scharfstein Daniel~O., Ogburn
  Elizabeth~L.. Sensitivity analysis for principal ignorability violation in
  estimating complier and noncomplier average causal effects   2023.
\newblock arXiv:2303.05052v1 (preprint version 1).

\bibitem{kennedy2023SemiparametricDoublyRobust}
Kennedy Edward~H.. Semiparametric doubly robust targeted double machine
  learning: a review   2023.
\newblock arXiv:2203.06469 [stat].

\bibitem{nowok2016SynthpopBespokeCreation}
Nowok Beata, Raab Gillian~M., Dibben Chris. synthpop: {Bespoke} {Creation} of
  {Synthetic} {Data} in {R}  {\it Journal of Statistical Software. }
  2016;74:1--26.

\end{thebibliography}


\newpage
\pagenumbering{arabic}
\renewcommand*{\thepage}{A\arabic{page}}

\appendix
\noindent{\Large\textbf{APPENDIX}}
\small
\part{} 
\vspace{-4em}
\parttoc


\section{Some lemmas}

There are three lemmas that simplify many proofs in this paper. They are stated here. Other lemmas are introduced specifically where they are needed.

\addcontentsline{toc}{subsection}{Lemma \ref{lm:independence}}
\begin{lemma}\label{lm:independence}
If $A\independent(B,C)$ then $A\independent B\mid C$.
\end{lemma}

\addcontentsline{toc}{subsection}{Lemma \ref{lm:parttowhole}}
\begin{lemma}\label{lm:parttowhole}
\begin{align*}
    \E[A\mid B,D=d]=\E\left[\frac{\I(D=d)}{\P(D=d\mid B)}A\mid B\right].
\end{align*}
\end{lemma}

\addcontentsline{toc}{subsection}{Lemma \ref{lm:if-ratioOfMeans}}
\begin{lemma}\label{lm:if-ratioOfMeans}
Let $\alpha$ and $\beta$ be two parameters with IFs $\varphi_\alpha(O)$ and $\varphi_\beta(O)$ under a probability model $\mathcal{P}$. Then the IF of $\gamma:=\alpha/\beta$ is
\begin{align}
    \varphi_\gamma(O)=\frac{1}{\beta}[\varphi_\alpha(O)-\gamma\,\varphi_\beta(O)].
\end{align}
Moreover, if $\varphi_\alpha(O)$ and $\varphi_\beta(O)$ are of the form
\begin{align}
    \begin{matrix}
    \varphi_\alpha(O)=\phi_\alpha(O)-\alpha,
    \\
    \varphi_\beta(O)=\phi_\beta(O)-\beta,
    \end{matrix}\label{if:simple.form}
\end{align}
where the functions $\phi_\alpha(O)$ and $\phi_\beta(O)$ do not involve $\alpha,\beta$, then
\begin{align}
    \varphi_\gamma(O)=\frac{1}{\beta}[\phi_\alpha(O)-\gamma\,\phi_\beta(O)].\label{if:simple.result}
\end{align}
\end{lemma}

\bigskip

\addcontentsline{toc}{subsection}{Proof of Lemma \ref{lm:independence}}
\noindent
\begin{proof}[Proof of Lemma \ref{lm:independence}]
First, note that if $A\independent(B,C)$ then $A\independent C$. This is because
\begin{align*}
    \P(A\mid C)
    &=\E[\P(A\mid B,C)\mid C] && \text{(law of total probability)}
    \\
    &=\E[\P(A)\mid C] && (A\independent(B,C))
    \\
    &=\P(A).
\end{align*}
That $A\independent B\mid C$ follows,
\begin{align*}
    \P(B\mid C)
    &=\frac{\P(B,C)}{\P(C)} && \text{(Bayes' rule)}
    \\
    &=\frac{\P(B,C\mid A)}{\P(C\mid A)} && (A\independent(B,C)\text{ and }A\independent C)
    \\
    &=\P(B\mid A,C). && (\text{Bayes' rule})
\end{align*}
\end{proof}

\addcontentsline{toc}{subsection}{Proof of Lemma \ref{lm:parttowhole}}
\noindent
\begin{proof}[Proof of Lemma \ref{lm:parttowhole}]
\begin{align*}
    \text{RHS}
    &=\E\left[\frac{\I(D=d)}{\P(D=d\mid B)}A\mid B\right]
    \\
    &=\E\left\{\E\left[\frac{\I(D=d)}{\P(D=d\mid B)}A\mid B,D\right]\mid B\right\}
    \\
    &=\E\left\{\frac{\I(D=d)}{\P(D=d\mid B)}\E[A\mid B,D=d]\mid B\right\}
    \\
    &=\E\left\{\frac{\I(D=d)}{\P(D=d\mid B)}\mid B\right\}\E[A\mid B,D=d]
    \\
    &=\text{LHS}.
\end{align*}
\end{proof}

\addcontentsline{toc}{subsection}{Proof of Lemma \ref{lm:if-ratioOfMeans}}
\noindent
\begin{proof}[Proof of Lemma \ref{lm:if-ratioOfMeans}]
Let the density of data $O$ in a parametric submodel of $\mathcal{P}$ be indexed by parameter $\theta$. Then $\varphi_\gamma(O)$ is the function that satisfies
\begin{align*}
    \frac{\partial\gamma(\theta)}{\partial\theta}\Big|_{\theta=\theta^0}=\E[S(O;\theta_0)\varphi_\gamma(O)]~~~\text{and}~~\E[\varphi_\gamma(O)]=0,
\end{align*}
where $\theta^0$ is the true value of the parameter, $S(O;\theta_0)$ is the score function of the model evaluated at the true parameter, and the expectations are taken under the true distribution. We start with the derivative on LHS,
\begin{align*}
    \frac{\partial\gamma(\theta)}{\partial\theta}
    &=\frac{\partial\gamma}{\partial\alpha}\frac{\partial\alpha}{\partial\theta}+\frac{\partial\gamma}{\partial\beta}\frac{\partial\beta}{\partial\theta}
    \\
    &=\frac{1}{\beta}\frac{\partial\alpha}{\partial\theta}-\frac{\alpha}{\beta^2}\frac{\partial\beta}{\partial\theta}
    \\
    &=\frac{1}{\beta}\left[\frac{\partial\alpha}{\partial\theta}-\gamma\frac{\partial\beta}{\partial\theta}\right].
\end{align*}
Evaluating both sides at $\theta=\theta_0$,
\begin{align*}
    \E[S(O;\theta_0)\varphi_\gamma(O)]
    &=\frac{1}{\beta}\Big\{\E[S(O;\theta_0)\varphi_\alpha(O)]-\gamma\,\E[S(O;\theta_0)\varphi_\beta(O)]\Big\}
    \\
    &=\E\Big\{S(O;\theta_0)\frac{1}{\beta}[\varphi_\alpha(O)-\gamma\varphi_\beta(O)]\Big\}
\end{align*}
As $\E\left\{\mfrac{1}{\beta}[\varphi_\alpha(O)-\gamma\varphi_\beta(O)]\right\}=0$,
\begin{align*}
    \varphi_\gamma(O)
    &=\frac{1}{\beta}[\varphi_\alpha(O)-\gamma\varphi_\beta(O)].
\end{align*}
For the case with the specific form of $\varphi_\alpha(O)$ and $\varphi_\beta(O)$, we plug them in and obtain
\begin{align*}
    \varphi_\gamma(O)
    &=\frac{1}{\beta}[\phi_\alpha(O)-\alpha-\gamma\phi_\beta(O)+\gamma\beta]
    \\
    &=\frac{1}{\beta}[\phi_\alpha(O)-\gamma\phi_\beta(O)].
\end{align*}
\end{proof}

\section{Proofs for Section \ref{sec:setting} -- PI-based identification}\label{appendix:prelim}

\addcontentsline{toc}{subsection}{Proof of Lemma \ref{lm:starting-point}}
\noindent
\begin{proof}[Proof of Lemma \ref{lm:starting-point}]
\begin{align*}
    \overbrace{\E[Y_z\mid C=c]}^{\textstyle=:\tau_{zc}}
    &=\E\{\overbrace{\E[Y_z\mid X,C=c]}^{\textstyle=:\mu_{zc}(X)}\mid C=c\} && \text{(iterated expectation)}
    \\
    &=\E\left[\mu_{zc}(X)\frac{\P(X\mid C=c)}{\P(X)}\right] && \text{(re-expression after writing in integral form)}
    \\
    &=\E\Big[\mu_{zc}(X)\frac{\,\overbrace{\P(C=c\mid X)}^{\textstyle=:\pi_c(X)}\,}{\P(C=c)}\Big] && \text{(Bayes' rule)}
    \\
    &=\frac{\E[\pi_c(X)\mu_{zc}(X)]}{\E[\pi_c(X)]}.\tag{\ref{estimand:tau.zc}}
\end{align*}
\end{proof}

\bigskip

\addcontentsline{toc}{subsection}{Proof of Proposition \ref{thm:id-tau.1c}}
\noindent
\begin{proof}[Proof of Proposition \ref{thm:id-tau.1c}]

\begin{align*}
    \overbrace{\P(C=c\mid X)}^{\textstyle=:\pi_c(X)}
    &=\P(C=c\mid X,Z=1), && (\text{under A1})\tag{\ref{id:pi.c}}
    \\
    \overbrace{\E[Y_1\mid X,C=c]}^{\textstyle=:\mu_{zc}(X)}
    &=\E[Y_1\mid X,Z=1,C=c] && (Y_1\independent Z\mid X,C~\text{due to A1 by Lemma \ref{lm:independence}})
    \\
    &=\E[Y\mid X,Z=1,C=c]. && (\text{under A0 and A2})\tag{\ref{id:mu.1c}}
\end{align*}

The first formula for $\tau_{1c}$ in (\ref{id:tau.1c}) is basically (\ref{estimand:tau.zc}) from Lemma \ref{lm:starting-point}, but now interpreted in terms of (\ref{id:pi.c}) and (\ref{id:mu.1c}).
In addition,
\begin{align}
    \E\left[\frac{Z}{e(X,Z)}\I(C=c)Y\right]
    &=\E\left[\left(\left\{\E\left[\frac{Z}{e(X,Z)}\I(C=c)Y\mid X,Z,C\right]\mid X,Z\right\}\mid X\right)\right]=\E[\pi_c(X)\mu_{1c}(X)],\nonumber
    \\
    \E\left[\frac{Z}{e(X,Z)}\I(C=c)\right]
    &=\E\left(\E\left\{\E\left[\frac{Z}{e(X,Z)}\I(C=c)\mid X,Z\right]\mid X\right\}\right)
    =\E[\pi_c(X)],\label{delta.c-form2}
\end{align}
which together give the second formula for $\tau_{1c}$ in (\ref{id:tau.1c}).

Lastly,
\begin{align*}
    \E[Y_0\mid X]
    &=\E[Y_0\mid X,Z=0] && \text{(under A1)}
    \\
    &=\overbrace{\E[Y\mid X,Z=0]}^{\textstyle=:\mu_0(X)}, && \text{(under A0 and A2)}
\end{align*}
and at the same time
\begin{align*}
    \E[Y_0\mid X]
    &=\E\{\E[Y_0\mid X,C]\mid X\} && \text{(iterated expectation)}
    \\
    &=\sum_{c=0}^1\overbrace{\E[Y_0\mid X,C=c]}^{\textstyle=:\mu_{0c}(X)}\overbrace{\P(C=c\mid X)}^{\textstyle=:\pi_c(X)}.
\end{align*}
It follows that
\begin{align*}
    \mu_{00}(X)\pi_0(X)+\mu_{01}(X)\pi_1(X)=\mu_0(X).\tag{\ref{eq:mixture-mean}}
\end{align*}
\end{proof}

\bigskip\bigskip

\addcontentsline{toc}{subsection}{Proof of Proposition \ref{thm:pi-id:tau.0c}}
\noindent
\begin{proof}[Proof of Proposition \ref{thm:pi-id:tau.0c}]
That 
\begin{align*}
    \mu_{0c}(X)=\mu_0(X)\tag{\ref{pi-id:mu.0c}}
\end{align*}
follows directly from (\ref{eq:mixture-mean}) combined with A3, and justifies the first $\tau_{0c}^\text{PI}$ formula in (\ref{pi-id:tau.0c}).

In addition, the combination of
\begin{align*}
    \E\left[\frac{1-Z}{e(X,Z)}\pi_c(X)Y\right]
    &=\E\left(\E\left\{\E\left[\frac{1-Z}{e(X,Z)}\pi_c(X)Y\mid X,Z\right]\mid X\right\}\right)=\E[\pi_c(X)\mu_0(X)],
    \\
    \E\left[\frac{1-Z}{e(X,Z)}\pi_c(X)\right]
    &=\E\left\{\E\left[\frac{1-Z}{e(X,Z)}\pi_c(X)\mid X\right]\right\}
    =\E[\pi_c(X)],
\end{align*}
gives the second $\tau_{0c}^\text{PI}$ formula in (\ref{pi-id:tau.0c}).

Also, (\ref{delta.c-form2}) combined with 
\begin{align*}
    \E\left[\frac{Z}{e(X,Z)}\I(C=c)\mu_0(X)\right]
    &=\E\left(\E\left\{\E\left[\frac{Z}{e(X,Z)}\I(C=c)\mu_0(X)\mid X,Z\right]\mid X\right\}\right)=\E[\pi_c(X)\mu_0(X)]
\end{align*}
gives the third $\tau_{0c}^\text{PI}$ formula in (\ref{pi-id:tau.0c}).
\end{proof}

\section{Proof for Section \ref{sec:pi-estimators} -- IF-based PI-based estimation}\label{appendix:pi-estimators}

\addcontentsline{toc}{subsection}{Proof of Proposition \ref{thm:ifs-pi}}
\noindent
\begin{proof}[Proof of Proposition \ref{thm:ifs-pi}]
Consider a paramatric submodel (of the nonparametric model) of O, 
\begin{align}
    f(O,\theta)=~&f_1(X,\theta_1)\,
    f_2(Z\mid X,\theta_2)\,
    f_3(C\mid X,Z=1,\theta_3)\times\nonumber
    \\
    &[f_4(Y\mid X,Z=1,C=1,\theta_4)]^{ZC}\,
    [f_5(Y\mid X,Z=1,C=0,\theta_5)]^{Z(1-C)}\,
    [f_6(Y\mid X,Z=0,\theta_6)]^{1-Z}
\end{align}
Based on this factorization, the observed data Hilbert space (i.e., the space of mean-zero finite-variance 1-dimensional functions of observed data equipped with the covariance inner product) is the direct sum of six subspaces:
\begin{align*}
    \H=\T_1\oplus\T_2\oplus\T_3\oplus\T_4\oplus\T_5\oplus\T_6,
\end{align*}
where
\begin{align*}
    \T_1&=\{g_1(X):\E[g_1(X)]=0\},
    \\
    \T_2&=\{g_2(X,Z):\E[g_2(X,Z)\mid X]=0\},
    \\
    \T_3&=\{Zg_3(X,C):\E[g_3(X,C)\mid X,Z=1]=0\},
    \\
    \T_4&=\{ZCg_4(X,Y):\E[g_4(X,Y)\mid X,Z=1,C=1]=0\},
    \\
    \T_5&=\{Z(1-C)g_5(X,Y):\E[g_5(X,Y)\mid X,Z=1,C=0]=0\},
    \\
    \T_6&=\{(1-Z)g_6(X,Y):\E[g_6(X,Y)\mid X,Z=0]=0\}.
\end{align*}

\noindent Assume regularity conditions hold that allow interchanging integration and derivation.

\bigskip

\noindent
\textit{Part 0: Some useful derivatives}

\begin{align}
    f_1'(X):=\frac{\partial f_1(X,\theta_1)}{\partial\theta_1}\Big|_{\theta_1=\theta_1^0}
    &=S_1(X)f_1(X),\label{f_1'(X)}
    \\
    \pi_1'(X):=\frac{\partial\pi_1(X,\theta_3)}{\partial\theta_3}\Big|_{\theta_3=\theta_3^0}
    &=\frac{\partial}{\partial\theta_3}\int c\,f_3(c\mid X,Z=1,\theta_3)dc\Big|_{\theta_3=\theta_3^0}\nonumber
    \\
    &=\int c\frac{\partial f_3(c\mid X,Z=1,\theta_3)}{\partial\theta_3}\Big|_{\theta_3=\theta_3^0}dc\nonumber
    \\
    &=\int c\,S_3(X,Z=1,C=c)f_3(c\mid X,Z=1)dc\nonumber
    \\
    &=\E[S_3(X,Z,C)C\mid X,Z=1]\nonumber
    \\
    &=\E\{S_3(X,Z,C)[C-\pi_1(X)]\mid X,Z=1\}+\underbrace{\E\{S_3(X,Z,C)\pi_1(X)\mid X,Z=1\}}_{\textstyle=0\text{ because }\E[S_3\mid X,Z=1]=0}\nonumber
    \\
    &=\E\left[S_3(X,Z,C)\frac{Z}{e(X,Z)}[C-\pi_1(X)]\mid X\right], && (\text{Lemma \ref{lm:parttowhole}})\label{pi_1'(X)}
    \\
    \mu_{11}'(X):=\frac{\partial\mu_{11}(X)}{\partial\theta_4}\Big|_{\theta_4=\theta_4^0}
    &=\frac{\partial}{\partial\theta_4}\int y\,f_4(y\mid X,Z=1,C=1,\theta_4)dy\Big|_{\theta_4=\theta_4^0}\nonumber
    \\
    &=\int y\frac{\partial}{\partial\theta_4}f_4(y\mid X,Z=1,C=1,\theta_4)\Big|_{\theta_4=\theta_4^0}dy\nonumber
    \\
    &=\int y\,S_4(X,Z=1,C=1,y)f_4(y\mid X,Z=1,C=1)dy\nonumber
    \\
    &=\E[S_4(X,Z,C,Y)Y\mid X,Z=1,C=1]\nonumber
    \\
    &=\E\{S_4(X,Z,C,Y)[Y-\mu_{11}(X)]\mid X,Z=1,C=1\}+\nonumber
    \\
    &~~~~~~~~~~~~~~~~~~~~~~~~~~~~~~\underbrace{\E\{S_4(X,Z,C,Y)\mu_{11}(X)\mid X,Z=1,C=1\}}_{\textstyle=0\text{ because }\E[S_4\mid X,Z=1,C=1]=0}\nonumber
    \\
    &=\E\Big[S_4(X,Z,C,Y)\frac{ZC}{\underbrace{\P(Z=1,C=1\mid X)}_{\textstyle e(X,Z)\pi_1(X)}}[Y-\mu_{11}(X)]\mid X\Big], && (\text{Lemma \ref{lm:parttowhole}})\label{mu_11'(X)}
    \\
    \mu_0'(X):=\frac{\partial\mu_0(X,\theta_6)}{\partial\theta_6}\Big|_{\theta_6=\theta_6^0}
    &=\frac{\partial}{\partial\theta_6}\int y\,f_6(y\mid X,Z=0,\theta_6)dy\Big|_{\theta_6=\theta_6^0}\nonumber
    \\
    &=\int y\frac{\partial}{\partial\theta_6}f_6(y\mid X,Z=0,\theta_6)\Big|_{\theta_6=\theta_6^0}dy\nonumber
    \\
    &=\int y\,S_6(X,Z=0,y)f_6(y\mid X,Z=0)dy\nonumber
    \\
    &=\E[S_6(X,Z,Y)Y\mid X,Z=0]\label{mu_0'(X)-intermediate-result}
    \\
    &=\E\{S_6(X,Z,Y)[Y-\mu_0(X)]\mid X,Z=0\}+\underbrace{\E\{S_6(X,Z,Y)\mu_0(X)\mid X,Z=0\}}_{\textstyle=0\text{ because }\E[S_6\mid X,Z=0]=0}\nonumber
    \\
    &=\E\left[S_6(X,Z,Y)\frac{1-Z}{e(X,Z)}[Y-\mu_0(X)]\mid X\right]. && (\text{Lemma \ref{lm:parttowhole}})\label{mu_0'(X)}
\end{align}

\bigskip

\noindent
\textit{Part 1: The IF for $\pi_c$}
\smallskip

\noindent
We just need to consider $\pi_1$ and then infer for $\pi_c$ for $c=0,1$.
We write $\pi_1$ as a function of parameter $\theta$ of the parametric submodel,
\begin{align*}
    \pi_1(\theta)=\E_{\theta_1}[\pi_1(X,\theta_3)]=\int\pi_1(x,\theta_3)f_1(x,\theta_1)dx.
\end{align*}
This function involves only $\theta_1,\theta_3$, so the IF for $\pi_c$, denoted $\varphi_{\pi_1}(O)$, is the sum of two terms $\varphi_{\pi_1,1}\in\T_1$ and $\varphi_{\pi_1,3}\in\T_3$, such that
\begin{align*}
    \frac{\partial\pi_1(\theta)}{\partial\theta_1}\Big|_{\theta=\theta^0}
    &=\E[S_1(X,\theta_1^0)\varphi_{\pi_1,1}(X)],
    \\
    \frac{\partial\pi_1(\theta)}{\partial\theta_3}\Big|_{\theta=\theta^0}
    &=\E[S_3(X,Z,C,\theta_3^0)\varphi_{\pi_1,3}(X,Z,C)],
\end{align*}
where $\theta^0$ is the true value of $\theta$; $S_1,S_3$ are score functions for $\theta_1,\theta_3$; and the expectations are taken w.r.t. the truth. 
To simplify notation, we suppress the parameter when a function is evaluated at the true value of the parameter, e.g., $f_1(X)=f_1(X,\theta_1^0)$, $S_1(X)=S_1(X,\theta_1^0)$, etc.

\begin{align*}
    \frac{\partial\pi_1(\theta)}{\partial\theta_1}\Big|_{\theta=\theta^0}
    &=\int\pi_1(x)f_1'(x)dx
    \\
    &=\int\pi_1(x)S_1(x)f_1(x)=\E[S_1(X)\pi_1(X)] && \text{(by (\ref{f_1'(X)}))}
    \\
    &=\E\{S_1(X)\underbrace{[\pi_1(X)-\pi_1]}_{\textstyle=:\varphi_{\delta,1}(X)\in\T_1}\},
    \\
    \frac{\partial\pi_1(\theta)}{\partial\theta_3}\Big|_{\theta=\theta^0}
    &=\E[\pi_1'(X)]
    \\
    &=\E\Big\{S_3(X,Z,C)\underbrace{\frac{Z}{e(X,Z)}[C-\pi_1(X)]}_{\textstyle=:\varphi_{\pi_1,3}(X,Z,C)\in\T_3}\Big\}. && \text{(by (\ref{pi_1'(X)}))}
\end{align*}
Hence
\begin{align*}
    \varphi_{\pi_1}(O)=\varphi_{\pi_1,3}(X,Z,C)+\varphi_{\delta,1}(X)=\frac{Z}{e(X,Z)}[C-\pi_1(X)]+\pi_1(X)-\pi_1,
\end{align*}
and more generally, the IF of $\pi_c$ is:
\begin{align*}
    \varphi_{\pi_c}(O)=\frac{Z}{e(X,Z)}[\I(C=c)-\pi_c(X)]+\pi_c(X)-\pi_c.\tag{\ref{if:pi.c}}
\end{align*}

\medskip

\noindent
\textit{Part 2: The IF for $\nu_{1c}$}
\smallskip

\noindent
Consider $\nu_{11}$.
\begin{align*}
    \nu_{11}(\theta)=\E_{\theta_1}[\pi_1(X,\theta_3)\mu_{11}(X,\theta_4)]=\int\pi_1(x,\theta_3)\mu_{11}(x,\theta_4)f_1(x,\theta_1)dx.
\end{align*}
This function involves $\theta_1,\theta_3,\theta_4$ so the IF of $\nu_{11}$, denoted $\varphi_{\nu_{11}}(O)$, is the sum of three terms $\varphi_{\nu_{11},1}\in\T_1$, $\varphi_{\nu_{11},3}\in\T_3$ and $\varphi_{\nu_{11},4}\in\T_4$, such that
\begin{align*}
    \frac{\partial\nu_{11}(\theta)}{\partial\theta_1}\Big|_{\theta=\theta^0}
    &=\E[S_1(X)\varphi_{\nu_{11},1}(X)],
    \\
    \frac{\partial\nu_{11}(\theta)}{\partial\theta_3}\Big|_{\theta=\theta^0}
    &=\E[S_3(X,Z,C)\varphi_{\nu_{11},3}(X,Z,C)],
    \\
    \frac{\partial\nu_{11}(\theta)}{\partial\theta_4}\Big|_{\theta=\theta^0}
    &=\E[S_4(X,Z,C,Y)\varphi_{\nu_{11},4}(X,Z,C,Y)].
\end{align*}
\begin{align*}
    \frac{\partial\nu_{11}(\theta)}{\partial\theta_1}\Big|_{\theta=\theta^0}
    &=\int\pi_1(x)\mu_{11}(x)f_1'(x)dx
    \\
    &=\E[S_1(X)\pi_1(X)\mu_{11}(X)] && \text{(by (\ref{f_1'(X)}))}
    \\
    &=\E\Big\{S_1(X)\underbrace{[\pi_1(X)\mu_{11}(X)-\nu_{11}]}_{\textstyle=:\varphi_{\nu_{11},1}(X)\in\T_1}\Big\},
    \\
    \frac{\partial\nu_{11}(\theta)}{\partial\theta_3}\Big|_{\theta=\theta^0}
    &=\E[\pi'_1(X)\mu_{11}(X)]
    \\
    &=\E\Big\{S_3(X,Z,C)\underbrace{\frac{Z}{e(X,Z)}\mu_{11}(X)[C-\pi_1(X)]}_{\textstyle=:\varphi_{\nu_{11},3}(X,Z,C)\in\T_3}\Big\}, && \text{(by (\ref{pi_1'(X)}))}
    \\
    \frac{\partial\nu_{11}(\theta)}{\partial\theta_4}\Big|_{\theta=\theta^0}
    &=\E[\pi_1(X)\mu_{11}'(X)]
    \\
    &=\E\Big\{S_4(X,Z,C,Y)\underbrace{\frac{Z}{e(X,Z)}C[Y-\mu_{11}(X)]}_{\textstyle=:\varphi_{\nu_{11},4}(X,Z,C,Y)\in\T_4}\Big\}. && \text{(by (\ref{mu_11'(X)}))}
\end{align*}
Hence
\begin{align*}
    \varphi_{\tau_{11}}(O)=\frac{Z}{e(X,Z)}C[Y-\mu_{11}(X)]+\frac{Z}{e(X,Z)}\mu_{11}[C-\pi_1(X)]+\pi_1(X)\mu_{11}(X)-\nu_{11},
\end{align*}
and more generally, the IF of $\nu_{1c}$ is:
\begin{align*}
    \varphi_{\tau_{1c}}(O)
    &=\frac{Z}{e(X,Z)}\I(C=c)[Y-\mu_{1c}(X)]+\frac{Z}{e(X,Z)}\mu_{1c}[\I(C=c)-\pi_c(X)]+\pi_c(X)\mu_{1c}(X)-\nu_{1c}.\tag{\ref{if:nu.1c}}
\end{align*}

\medskip

\noindent
\textit{Part 3: The IF for $\nu_{0c}^\textup{PI}$.}
\smallskip

\noindent
Consider $\nu_{01}^\text{PI}$.
\begin{align*}
    \nu_{01}^\text{PI}(\theta)=\int\pi_1(x,\theta_3)\mu_0(x,\theta_6)f_1(x,\theta_1)dx.
\end{align*}
This function involves $\theta_1$, $\theta_3$ and $\theta_6$ so the IF of $\nu_{01}^\text{PI}$, denoted $\varphi_{\nu_{01}^\text{PI}}(O)$, is the sum of three terms $\varphi_{\nu_{01}^\text{PI},1}\in\T_1$, $\varphi_{\nu_{01}^\text{PI},3}\in\T_3$ and $\varphi_{\nu_{01}^\text{PI},6}\in\T_6$, such that
\begin{align*}
    \frac{\partial\nu_{01}^\text{PI}(\theta)}{\partial\theta_1}\Big|_{\theta=\theta^0}
    &=\E[S_1(X)\varphi_{\nu_{01}^\text{PI},1}(X)],
    \\
    \frac{\partial\nu_{01}^\text{PI}(\theta)}{\partial\theta_3}\Big|_{\theta=\theta^0}
    &=\E[S_3(X,Z,C)\varphi_{\nu_{01}^\text{PI},3}(X,Z,C)],
    \\
    \frac{\partial\nu_{01}^\text{PI}(\theta)}{\partial\theta_6}\Big|_{\theta=\theta^0}
    &=\E[S_6(X,Z,Y)\varphi_{\nu_{01}^\text{PI},6}(X,Z,Y)].
\end{align*}
\begin{align*}
    \frac{\partial\nu_{01}^\text{PI}(\theta)}{\partial\theta_1}\Big|_{\theta=\theta^0}
    &=\int\pi_1(x)\mu_0(x)f_1'(x)dx
    \\
    &=\E[S_1(X)\pi_1(X)\mu_0(X)] && \text{(by (\ref{f_1'(X)}))}
    \\
    &=\E\Big\{S_1(X)\underbrace{[\pi_1(X)\mu_0(X)-\nu_{01}^\text{PI}]}_{\textstyle=:\varphi_{\nu_{01}^\text{PI},1}(X)\in\T_1}\Big\},
    \\
    \frac{\partial\nu_{01}^\text{PI}(\theta)}{\partial\theta_3}\Big|_{\theta=\theta^0}
    &=\E[\pi_1'(X)\mu_0(X)]
    \\
    &=\E\Big\{S_3(X,Z,C)\underbrace{\frac{Z}{e(X,Z)}\mu_0(X)[C-\pi_1(X)]}_{\textstyle=:\varphi_{\nu_{01}^\text{PI},3}(X,Z,C)\in\T_3}\Big\}, && \text{(by (\ref{pi_1'(X)}))}
    \\
    \frac{\partial\nu_{01}^\text{PI}(\theta)}{\partial\theta_6}\Big|_{\theta=\theta^0}
    &=\E[\pi_1(X)\mu_0'(X)]
    \\
    &=\E\Big\{S_6(X,Z,Y)\underbrace{\frac{1-Z}{e(X,Z)}\pi_1(X)[Y-\mu_0(X)]}_{\textstyle=:\varphi_{\nu_{01}^\text{PI},6}(X,Z,Y)\in\T_6}\Big\}. && \text{(by (\ref{mu_0'(X)}))}
\end{align*}
Hence
\begin{align*}
    \varphi_{\nu_{01}^\text{PI}}(O)=\frac{1-Z}{e(X,Z)}\pi_1(X)[Y-\mu_0(X)]+\frac{Z}{e(X,Z)}\mu_0(X)[C-\pi_1(X)]+\pi_1(X)\mu_0(X)-\nu_{01}^\text{PI},
\end{align*}
and more generally, the IF for $\nu_{01}^\text{PI}$ is:
\begin{align*}
    \varphi_{\nu_{0c}^\text{PI}}(O)=\frac{1-Z}{e(X,Z)}\pi_c(X)[Y-\mu_0(X)]+\frac{Z}{e(X,Z)}\mu_0(X)[\I(C=c)-\pi_c(X)]+\pi_c(X)\mu_0(X)-\nu_{0c}^\text{PI}.\tag{\ref{pi-if:nu.0c}}
\end{align*}

\noindent
\textit{Part 4: The IF for $\Delta_c^\textup{PI}$}

\smallskip

\noindent
This IF is obtained by applying Lemma \ref{lm:if-ratioOfMeans} to obtain the IFs for $\tau_{1c}=\nu_{1c}/\pi_c$ and $\tau_{0c}^\text{PI}/\pi_c$ and then combining the two IFs. Specifically, by Lemma \ref{lm:if-ratioOfMeans}, the IFs of $\tau_{1c}$ and $\tau_{0c}^\text{PI}$ are
\begin{align*}
    \varphi_{\tau_{1c}}(O)
    &=\frac{1}{\pi_c}\Big\{[\varphi_{\nu_{1c}}(O)+\nu_{1c}]-\tau_{1c}[\varphi_{\pi_c}(O)+\pi_c]\Big\},
    \\
    \varphi_{\tau_{0c}^\text{PI}}(O)
    &=\frac{1}{\pi_c}\Big\{[\varphi_{\nu_{0c}^\text{PI}}(O)+\nu_{0c}^\text{PI}]-\tau_{0c}^\text{PI}[\varphi_{\pi_c}(O)+\pi_c]\Big\}.
\end{align*}
The difference of these two IFs is
\begin{align*}
    \varphi_{\Delta_c^\text{PI}}(O)=\frac{1}{\pi_c}\Big\{[\varphi_{\nu_{1c}}(O)+\nu_{1c}]-[\varphi_{\nu_{0c}^\text{PI}}(O)+\nu_{0c}^\text{PI}]-\Delta_c^\text{PI}[\varphi_{\pi_c}(O)+\pi_c]\Big\}.\tag{\ref{pi-if:Delta.c}}
\end{align*}
\end{proof}

\bigskip

\addcontentsline{toc}{subsection}{Proof of (\ref{adhoc})}
\noindent
\begin{proof}[Proof of (\ref{adhoc})]
To obtain the expression of the IF for $\Delta_c^\text{PI}$ in (\ref{adhoc}), we rewrite terms in (\ref{pi-if:Delta.c}).
    \begin{align*}
        \varphi_{\pi_c}(O)+\pi_c
        &=\frac{Z}{e(X,Z)}[\I(C=c)-\pi_c(X)]+\pi_c(X),
        \\
        \varphi_{\nu_{1c}}(O)+\nu_{1c}
        &=\frac{Z}{e(X,Z)}\I(C=c)[Y-\mu_{1c}(X)]+\frac{Z}{e(X,Z)}\mu_{1c}(X)[\I(C=c)-\pi_c(X)]+\pi_c(X)\mu_{1c}(X)
        \\
        &=\frac{Z}{e(X,Z)}\I(C=c)[Y-\mu_{1c}(X)]+\left[\frac{Z}{e(X,Z)}[\I(C=c)-\pi_c(X)]+\pi_c(X)\right]\mu_{1c}(X),
        \\
        \varphi_{\nu_{0c}^\text{PI}}(O)+\nu_{0c}^\text{PI}
        &=\frac{1-Z}{e(X,Z)}\pi_c(X)[Y-\mu_0(X)]+\frac{Z}{e(X,Z)}\mu_0(X)[\I(C=c)-\pi_c(X)]+\pi_c(X)\mu_0(X)
        \\
        &=\frac{1-Z}{e(X,Z)}\pi_c(X)[Y-\mu_0(X)]+\left[\frac{Z}{e(X,Z)}[\I(C=c)-\pi_c(X)]+\pi_c(X)\right]\mu_0(X).
    \end{align*}
    Plugging these terms back in (\ref{pi-if:Delta.c}) and combining terms that share $\left[\frac{Z}{e(X,Z)}[\I(C=c)-\pi_c(X)]+\pi_c(X)\right]$, we obtain
    \begin{align*}
    \varphi_{\Delta_c^\text{PI}}(O)
    &=\frac{1}{\pi_c}\left\{\frac{Z}{e(X,Z)}\I(C=c)[Y-\mu_{1c}(X)]-\frac{1-Z}{e(X,Z)}\pi_c(X)[Y-\mu_0(X)]\right.+
    \\
    &~~~~~~~~~~~~~~~~~~~~\left.\left[\frac{Z}{e(X,Z)}[\I(C=c)-\pi_c(X)]+\pi_c(X)\right][\mu_{1c}(X)-\mu_0(X)-\Delta_c^\text{PI}]\right\}.\tag{\ref{adhoc}}
    \end{align*}
\end{proof}

\bigskip

\addcontentsline{toc}{subsection}{Proof of Proposition \ref{thm:pi-multiplyrobust}}
\noindent
\begin{proof}[Proof of Proposition \ref{thm:pi-multiplyrobust}]
We will start with $\hat\Delta_{c,\textsc{if}}^\mathrm{PI}$, then move to $\hat\Delta_{c,\textsc{ifh}}^\mathrm{PI}$, and consider $\hat\Delta_{c,\textsc{ms}}^\mathrm{PI}$ last.

\medskip

\noindent
\textit{Part 1: Multiple robustness of }$\hat\Delta_{c,\textsc{if}}^\mathrm{PI}$
\smallskip

\noindent
Assume regularity conditions hold that ensure convergence of the nuisance functions to certain limit functions, $\hat e(X,Z)\pto e^\dagger(X,Z)$, $\hat\pi_c(X)\pto\pi_c^\dagger(X)$,  $\hat\mu_{1c}(X)\pto\mu_{1c}^\dagger(X)$, $\hat\mu_0(X)\pto\mu_0^\dagger(X)$. Then by the continuous mapping theorem,
\begin{align*}
    \hat\pi_{c,\textsc{if}}&\pto\E\left\{\frac{Z}{e^\dagger(X,Z)}[\I(C=c)-\pi_c^\dagger(X)]+\pi_c^\dagger(X)\right\}=:\pi_c^\dagger,
    \\
    \hat\nu_{1c,\textsc{if}}&\pto\E\left\{\frac{Z}{e^\dagger(X,Z)}\I(C=c)[Y-\mu_{1c}^\dagger(X)]+\frac{Z}{e^\dagger(X,Z)}\mu_{1c}^\dagger(X)[\I(C=c)-\pi_c^\dagger(X)]+\pi_c^\dagger(X)\mu_{1c}^\dagger(X)\right\}=:\nu_{1c}^\dagger,
    \\
    \hat\nu_{0c,\textsc{if}}^\text{PI}&\pto\E\left\{\frac{1-Z}{e^\dagger(X,Z)}\pi_c^\dagger(X)[Y-\mu_0^\dagger(X)]+\frac{Z}{e^\dagger(X,Z)}\mu_0^\dagger(X)[\I(C=c)-\pi_c^\dagger(X)]+\pi_c^\dagger(X)\mu_0^\dagger(X)\right\}=:{\nu_{0c}^\text{PI}}^\dagger.
\end{align*}

First, consider $\pi_c^\dagger$. If the propensity score model is correctly specified, $e^\dagger(X,Z)=e(X,Z)$, so
\begin{align*}
    \pi_c^\dagger
    &=\E\left\{\frac{Z}{e(X,Z)}\I(C=c)+\left[1-\frac{Z}{e(X,Z)}\right]\pi_c^\dagger(X)\right\} && \text{(rearranging terms)}
    \\
    &=\E\left[\frac{Z}{e(X,Z)}\I(C=c)\right]+\E\left\{\E\left[1-\frac{Z}{e(X,Z)}\mid X\right]\pi_c^\dagger(X)\right\}
    \\
    &=\pi_c+0=\pi_c.
\end{align*}
If the principal score model is correctly specified, $\pi_c^\dagger(X)=\pi_c(X)$, so
\begin{align*}
    \pi_c^\dagger
    &=\E\left\{\frac{Z}{e^\dagger(X,Z)}[\I(C=c)-\pi_c(X)]+\pi_c(X)\right\}
    \\
    &=\E\left\{\frac{Z}{e^\dagger(X,Z)}\E[\I(C=c)-\pi_c(X)\mid X,Z]\right\}+\E[\pi_c(X)]
    \\
    &=0+\pi_c=\pi_c.
\end{align*}

Next, consider $\nu_{1c}^\dagger$. If the propensity score model is correctly specified, $e^\dagger(X,Z)=e(X,Z)$, so
\begin{align*}
    \nu_{1c}^\dagger
    &=\E\left\{\frac{Z}{e(X,Z)}\I(C=c)Y+\left[1-\frac{Z}{e(X,Z)}\right]\pi_c^\dagger(X)\mu_{1c}^\dagger(X)\right\} && \text{(rearranging terms)}
    \\
    &=\E\left[\frac{Z}{e(X,1)}\I(C=c)Y\right]+\E\left\{\E\left[1-\frac{Z}{e(X,1)}\mid X\right]\pi_c^\dagger(X)\mu_{1c}^\dagger(X)\right\}
    \\
    &=\nu_{1c}+0=\nu_{1c}.
\end{align*}
In the principal score model and outcome model are correct, $\pi_c^\dagger(X)=\pi_c(X)$ and $\mu_{1c}^\dagger(X)=\mu_{1c}(X)$, so
\begin{align*}
    \nu_{1c}^\dagger
    &=\E\left\{\frac{Z}{e^\dagger(X,Z)}\I(C=c)[Y-\mu_{1c}(X)]+\frac{Z}{e^\dagger(X,Z)}\mu_{1c}(X)[\I(C=c)-\pi_c(X)]+\pi_c(X)\mu_{1c}(X)\right\}
    \\
    &=\E\left\{\frac{Z}{e^\dagger(X,Z)}\I(C=c)\E[Y-\mu_{1c}(X)\mid X,Z,C]\right\}+
    \E\left\{\frac{Z}{e^\dagger(X,Z)}\mu_{1c}(X)\big\{\E[\I(C=c)-\pi_c(X)\mid X,Z]\big\}\right\}+
    \\
    &~~~~~\E[\pi_c(X)\mu_{1c}(X)]
    \\
    &=0+0+\nu_{1c}=\nu_{1c}.
\end{align*}

Lastly, consider ${\nu_{0c}^\text{PI}}^\dagger$. If both the propensity score and principal score models are correct, $e^\dagger(X,Z)=e(X,Z)$, $\pi_c^\dagger(X)=\pi_c(X)$, so
\begin{align*}
    {\nu_{0c}^\text{PI}}^\dagger
    &=\E\left\{\frac{1-Z}{e(X,Z)}\pi_c(X)Y+\pi_c(X)\mu_0^\dagger(X)\left[1-\frac{1-Z}{e(X,Z)}\right]+\frac{Z}{e(X,Z)}\mu_0^\dagger(X)[\I(C=c)-\pi_c(X)]\right\}
    \\
    &\pushright{\text{(rearranging terms)}}
    \\
    &=\E\left[\frac{1-Z}{e(X,Z)}\pi_c(X)Y\right]+\E\left\{\pi_c(X)\mu_0^\dagger(X)\left[1-\frac{1-Z}{e(X,Z)}\mid X\right]\right\}+
    \\
    &~~~~~\E\left\{\frac{Z}{e(X,Z)}\mu_0^\dagger(X)\E[\I(C=c)-\pi_c(X)\mid X,Z]\right\}
    \\
    &=\nu_{0c}^\text{PI}+0+0=\nu_{0c}^\text{PI}.
\end{align*}
If both the propensity score model and the outcome model are correct, then $e^\dagger(X,Z)=e(X,Z)$ and $\mu_0^\dagger(X)=\mu_0(X)$, so
\begin{align*}
    {\nu_{0c}^\text{PI}}^\dagger
    &=\E\left\{\frac{1-Z}{e(X,Z)}\pi_c^\dagger(X)[Y-\mu_0(X)]+\frac{Z}{e(X,Z)}\I(C=c)\mu_0(X)+\pi_c^\dagger(X)\mu_0(X)\left[1-\frac{Z}{e(X,Z)}\right]\right\}
    \\
    &\pushright{\text{(rearranging terms)}}
    \\
    &=\E\left\{\frac{1-Z}{e(X,Z)}\pi_c^\dagger(X)\E[Y-\mu_0(X)\mid X,Z]\right\}+\E\left[\frac{Z}{e(X,Z)}\I(C=c)\mu_0(X)\right]+\E\left\{\pi_c^\dagger(X)\mu_0(X)\E\left[1-\frac{Z}{e(X,Z)}\mid X\right]\right\}
    \\
    &=0+\nu_{0c}^\text{PI}+0=\nu_{0c}^\text{PI}.
\end{align*}
If the principal score model and the outcome model is correct, $\pi_c^\dagger(X)=\pi_c(X)$, $\mu_0^\dagger(X)=\mu_0(X)$, so
\begin{align*}
    {\nu_{0c}^\text{PI}}^\dagger
    &=\E\left\{\frac{1-Z}{e^\dagger(X,Z)}\pi_c(X)[Y-\mu_0(X)]+\frac{Z}{e^\dagger(X,Z)}\mu_0(X)[\I(C=c)-\pi_c(X)]+\pi_c(X)\mu_0(X)\right\}
    \\
    &=\E\left\{\frac{1-Z}{e^\dagger(X,Z)}\pi_c(X)\E[Y-\mu_0(X)\mid X,Z]\right\}+\E\left\{\frac{Z}{e^\dagger(X,Z)}\mu_0(X)\E[\I(C=c)-\pi_c(X)\mid X,Z]\right\}+\E[\pi_c(X)\mu_0(X)]
    \\
    &=0+0+\nu_{0c}^\text{PI}.
\end{align*}
Collecting the above results, we have
\begin{itemize}
    \item $\pi_c^\dagger=\pi_c$ if either $\hat e(X,Z)$ or $\hat\pi_c(X)$ is correctly specified;
    \item $\nu_{1c}^\dagger=\nu_{1c}$ if either $\hat e(X,Z)$ is correctly specified or both $\hat\pi_c(X)$ and $\hat\mu_{1c}(X)$ are correctly specified;
    \item ${\nu_{0c}^\text{PI}}^\dagger=\nu_{0c}^\text{PI}$ if $\hat e(X,Z)$ and $\hat\pi_c(X)$ are correctly specified, or if $\hat e(X,Z)$ and $\hat\mu_0(X)$ are correctly specified, or if $\hat\pi_c(X)$ and $\hat\mu_0(X)$ are correctly specified.
\end{itemize}

It follows that $\hat\Delta_{c,\textsc{if}}^\text{PI}$ is consistent if one of the following is true
\begin{enumerate}[(a)]
    \item $\hat e(X,Z)$ and $\hat\pi_c(X)$ are correctly specified;
    \item $\hat e(X,Z)$ and $\hat\mu_0(X)$ are correctly specified;
    \item $\hat\pi_c(X)$ and both $\mu_{1c}(X)$, $\mu_0(X)$ are correctly specified.
\end{enumerate}

\bigskip

\noindent
\textit{Part 2: Multiple robustness of }$\hat\Delta_{c,\textsc{ifh}}^\mathrm{PI}$
\smallskip

\noindent
Here we assume the same regularity conditions hold as in the proof for the multiple robustness of $\hat\Delta_{c,\textsc{if}}^\text{PI}$, and use the same notation used in that proof for the probability limits of the nuisance functions.

$\hat\Delta_{c,\textsc{ifh}}^\text{PI}$ is a modification of $\hat\Delta_{c,\textsc{if}}^\text{PI}$ where 
$\displaystyle\frac{Z}{\hat e(X,Z)}$ is replaced with $\displaystyle\frac{Z}{\hat e(X,Z)}/\P_n\left[\frac{Z}{\hat e(X,Z)}\right]$ and 
$\displaystyle\frac{1-Z}{\hat e(X,Z)}$ is replaced with $\displaystyle\frac{1-Z}{\hat e(X,Z)}/\P_n\left[\frac{1-Z}{\hat e(X,Z)}\right]$. We have
\begin{align*}
    \frac{Z}{\hat e(X,Z)}/\P_n\left[\frac{Z}{\hat e(X,Z)}\right]
    &\pto\frac{Z}{e^\dagger(X,1)}/\E\left[\frac{e(X,1)}{e^\dagger(X,1)}\right],
    \\
    \frac{1-Z}{\hat e(X,Z)}/\P_n\left[\frac{1-Z}{\hat e(X,Z)}\right]
    &\pto\frac{1-Z}{e^\dagger(X,0)}/\E\left[\frac{e(X,0)}{e^\dagger(X,0)}\right].
\end{align*}
Let
\begin{align*}
    e^{\dagger\dagger}(Z,X):=Z\,e^\dagger(X,1)\E\left[\frac{e(X,1)}{e^\dagger(X,1)}\right]+(1-Z)e^\dagger(X,0)\E\left[\frac{e(X,0)}{e^\dagger(X,0)}\right].
\end{align*}
Then we can write
\begin{align*}
    \hat\pi_{c,\textsc{if}}&\pto\E\left\{\frac{Z}{e^{\dagger\dagger}(X,Z)}[\I(C=c)-\pi_c^\dagger(X)]+\pi_c^\dagger(X)\right\}=:\pi_c^{\dagger\dagger},
    \\
    \hat\nu_{1c,\textsc{if}}&\pto\E\left\{\frac{Z}{e^{\dagger\dagger}(X,Z)}\I(C=c)[Y-\mu_{1c}^\dagger(X)]+\frac{Z}{e^{\dagger\dagger}(X,Z)}\mu_{1c}^\dagger(X)[\I(C=c)-\pi_c^\dagger(X)]+\pi_c^\dagger(X)\mu_{1c}^\dagger(X)\right\}=:\nu_{1c}^{\dagger\dagger},
    \\
    \hat\nu_{0c,\textsc{if}}^\text{PI}&\pto\E\left\{\frac{1-Z}{e^{\dagger\dagger}(X,Z)}\pi_c^\dagger(X)[Y-\mu_0^\dagger(X)]+\frac{Z}{e^{\dagger\dagger}(X,Z)}\mu_0^\dagger(X)[\I(C=c)-\pi_c^\dagger(X)]+\pi_c^\dagger(X)\mu_0^\dagger(X)\right\}=:{\nu_{0c}^\text{PI}}^{\dagger\dagger},
\end{align*}
which reminds of the probability limits at the start of the proof for $\hat\Delta_{c,\textsc{if}}^\text{PI}$. Note also that when $\hat e(X,Z)$ is correctly specified, $e^{\dagger\dagger}(X,Z)=e(X,Z)$. From this point the arguments are identical to the arguments in the proof for $\hat\Delta_{c,\textsc{if}}^\text{PI}$ above. Hence, $\hat\Delta_{c,\textsc{ifh}}^\text{PI}$ shares the same robustness property with $\hat\Delta_{c,\textsc{if}}^\text{PI}$.

\bigskip

\noindent
\textit{Part 3: Multiple robustness of }$\hat\Delta_{c,\textsc{ms}}^\mathrm{PI}$
\smallskip

\noindent
The estimators $\tilde\mu_{1c}(X)$ and $\tilde\mu_0(X)$ (of the nuisance functions $\mu_{1c}(X)$ and $\mu_0(X)$) and the estimator $\hat\Delta_{c,\mathrm{MS}}^\mathrm{PI}$ (which is based on these nuisance estimators) solve the set of equations
\begin{align*}
    &\P_n\left\{\frac{Z\,\I(C=c)}{\hat e(X,Z)}[Y-\mu_{1c}(X)]\right\}=0,
    \\
    &\P_n\left\{\frac{1-Z}{\hat e(X,Z)}\hat\pi_c(X)[Y-\mu_0(X)]\right\}=0,
    \\
    &\P_n\left\{\left[\frac{Z}{\hat e(X,Z)}[\I(C=c)-\hat\pi_c(X)]+\hat\pi_c(X)\right][\mu_{1c}(X)-\mu_0(X)-\Delta_c^\text{PI}]\right\}=0,
\end{align*}
where $\hat e(X,Z)$ and $\hat\pi_c(X)$ are estimators of $e(X,Z)$ and $\pi_c(X)$ that are plugged in.

Assume regularity conditions hold such that $\hat e(X,Z)\pto e^\dagger(X,Z)$, $\hat\pi_c(X)\pto\pi_c^\dagger(X)$, $\tilde\mu_{1c}(X)\pto\mu_{1c}^*(X)$, $\tilde\mu_0(X)\pto\mu_0^*(X)$. Then we have
\begin{align}
    &\E\left\{\frac{Z}{e^\dagger(X,Z)}\I(C=c)[Y-\mu_{1c}^*(X)]\right\}=0,\label{eq:robust2}
    \\
    &\E\left\{\frac{1-Z}{e^\dagger(X,Z)}\pi_c^\dagger(X)[Y-\mu_0^*(X)]\right\}=0.\label{eq:robust3}
\end{align}
Applying the continuous mapping theorem and factoring the numerator of the result, we have 
\begin{align*}
    \hat\Delta_{c,\mathrm{MS}}^\mathrm{PI}
    &\pto
    \frac{\overbrace{\E\left\{\left[\frac{Z}{e^\dagger(X,Z)}[\I(C=c)-\pi_c^\dagger(X)]+\pi_c^\dagger(X)\right]\mu_{1c}^*(X)\right\}}^{\textstyle=:\nu_{1c}^*}-\overbrace{\E\left\{\left[\frac{Z}{e^\dagger(X,Z)}[\I(C=c)-\pi_c^\dagger(X)]+\pi_c^\dagger(X)\right]\mu_0^*(X)\right\}}^{\textstyle=:{\nu_{0c}^\text{PI}}^*}}
    {\underbrace{\E\left[\frac{Z}{e^\dagger(X,Z)}[\I(C=c)-\pi_c^\dagger(X)]+\pi_c^\dagger(X)\right]}_{\textstyle=:\pi_c^\dagger}}
    \\
    &=:{\Delta_c^\text{PI}}^*.
\end{align*}
In the proof of multiple robustness of $\hat\Delta_{c,\textsc{if}}^\text{PI}$, we have shown that $\pi_c^\dagger=\pi_c$ if either $\hat e(X,Z)$ or $\hat\pi_c(X)$ is correctly specified.

Consider $\nu_{1c}^*$. If the propensity score model is correctly specified, $e^\dagger(X,Y)=e(X,Y)$, so
\begin{align*}
    \nu_{1c}^*
    &=\E\left\{\left[\dfrac{Z}{e(X,Z)}[\I(C=c)-\pi_c^\dagger(X)]+\pi_c^\dagger(X)\right]\mu_{1c}^*(X)\right\}
    \\
    &=\E\left[\frac{Z}{e(X,Z)}\I(C=c)\mu_{1c}^*(X)\right]+\E\left\{\pi_c^\dagger(X)\mu_{1c}^*(X)\E\left[1-\frac{Z}{e(X,Z)}\mid X\right]\right\}
    \\
    &\overset{(\ref{eq:robust2})}{=}\E\left[\frac{Z}{e(X,Z)}\I(C=c)Y\right]+0=\nu_{1c}.
\end{align*}
If the principal score model and the stratum-specific outcome under treatment model are correctly specified, $\pi_c^\dagger(X)=\pi_c(X)$ and $\mu_{1c}^*(X)=\mu_{1c}(X)$, so
\begin{align*}
    \nu_{1c}^*
    &=\E\left\{\left[\frac{Z}{e^\dagger(X,Z)}[\I(C=c)-\pi_c(X)]+\pi_c(X)\right]\mu_{1c}(X)\right\}
    \\
    &=\E\left\{\frac{Z}{e^\dagger(X,Z)}\mu_{1c}(X)\E[\I(C=c)-\pi_c(X)\mid X,Z]\right\}+\E[\pi_c(X)\mu_{1c}(X)]
    \\
    &=0+\nu_{1c}=\nu_{1c}.
\end{align*}

Lastly, consider ${\nu_{0c}^\text{PI}}^*$. If the propensity score model and the outcome under control model are correctly specified, $e^\dagger(X,Z)=e(X,Z)$ and $\mu_0^*(X)=\mu_0(X)$, so
\begin{align*}
    {\nu_{0c}^\text{PI}}^*
    &=\E\left\{\left[\dfrac{Z}{e(X,Z)}[\I(C=c)-\pi_c^\dagger(X)]+\pi_c^\dagger(X)\right]\mu_0(X)\right\}
    \\
    &=\E\left[\frac{Z}{e(X,Z)}\I(C=c)\mu_0(X)\right]+\E\left\{\pi_c^\dagger(X)\mu_0(X)\E\left[1-\frac{Z}{e(X,Z)}\mid X\right]\right\}
    \\
    &=\nu_{0c}^\text{PI}+0=\nu_{0c}^\text{PI}.
\end{align*}
If the propensity score model and principal score model are correctly specified,
\begin{align*}
    {\nu_{0c}^\text{PI}}^*
    &=\E\left\{\left[\frac{Z}{e(X,Z)}[\I(C=c)-\pi_c(X)]+\pi_c(X)\right]\mu_0^*(X)\right\}
    \\
    &=\E\left\{\frac{Z}{e(X,Z)}\E[\I(C=c)-\pi_c(X)\mid X,Z]\right\}+\E[\pi_c(X)\mu_0^*(X)]
    \\
    &=0+\E\left[\frac{1-Z}{e(X,Z)}\pi_c(X)\mu_0^*(X)\right]
    \\
    &\overset{(\ref{eq:robust3})}{=}\E\left[\frac{1-Z}{e(X,Z)}\pi_c(X)Y\right]
    =\nu_{0c}^\text{PI}.
\end{align*}
If the principal score model and the outcome under control model are correctly specified,
\begin{align*}
    {\nu_{0c}^\text{PI}}^*
    &=\E\left\{\left[\frac{Z}{e^\dagger(X,Z)}[\I(C=c)-\pi_c(X)]+\pi_c(X)\right]\mu_0(X)\right\}
    \\
    &=\E\left\{\frac{Z}{e^\dagger(X,Z)}\E[\I(C=c)-\pi_c(X)\mid X,Z]\right\}+\E[\pi_c(X)\mu_0(X)]
    \\
    &=0+\nu_{0c}^\text{PI}.
\end{align*}
It follows from the above results that $\hat\Delta_{c,\textsc{ms}}^\text{PI}$ shares the same robustness property with $\hat\Delta_{c,\textsc{if}}^\text{PI}$ and $\hat\Delta_{c,\textsc{ifh}}^\text{PI}$
\end{proof}

\section{Proofs for Section \ref{sec:ratio-params} -- Sens analysis with ratio-type sens params}\label{appendix:ratio-params}

\addcontentsline{toc}{subsection}{Proof of Proposition \ref{thm:ratio.params-id}}
\noindent
\begin{proof}[Proof of Proposition \ref{thm:ratio.params-id}]
\smallskip

\noindent
\textit{Part 1: assuming A0, A1, A2, A4-GOR}
\smallskip

\noindent
Under these assumptions, we have two equations with two unknowns
\begin{align*}
    \begin{cases}
        \pi_1(X)\mu_{01}(X)+\pi_0(X)\mu_{00}(X)=\mu_0(X)
        \\
        \displaystyle\frac{[\mu_{01}(X)-l]/[h-\mu_{01}(X)]}{[\mu_{00}(X)-l]/[h-\mu_{00}(X)]}=\rho
    \end{cases}.
\end{align*}
Let
\begin{align*}
    a&:=\pi_1(X),
    \\
    u&:=[\mu_{01}(X)-l]/(h-l),
    \\
    v&:=[\mu_{00}(X)-l]/(h-l),
    \\
    m&:=[\mu_0(X)-l]/(h-l).
\end{align*}
Then our two equations become
\begin{align*}
    \begin{cases}
    au+(1-a)v=m
    \\
    \displaystyle\frac{u/(1-u)}{v/(1-v)}=\rho
    \end{cases},
\end{align*}
where $a\in(0,1)$, $m\in[0,1]$ and $\rho>0$,
subject to the condition $u,v\in[0,1]$.

The first equation gives $\displaystyle v=\frac{m-au}{1-a}$. Plugging this into the second equation, we obtain (after some algebra)
\begin{align*}
    a(\rho-1)u^2-[(a+m)(\rho-1)+1]u+m\rho=0.
\end{align*}
For $\rho=1$, this reduces to $u=m$, which corresponds to the PI case. For $\rho\neq 1$, this is a quadratic equation, with two roots
\begin{align*}
    u_1=\frac{[(a+m)(\rho-1)+1]+\sqrt{d}}{2a(\rho-1)},~~~u_2=\frac{[(a+m)(\rho-1)+1]-\sqrt{d}}{2a(\rho-1)},
\end{align*}
where
\vspace{-1em}
\begin{align*}
    d:=[(a+m)(\rho-1)+1]^2-4am\rho(\rho-1).
\end{align*}
These two roots for $u$ respectively imply two values for $v$:
\begin{align*}
    v_1=\frac{[(m-a)(\rho-1)-1]-\sqrt{d}}{2(1-a)(\rho-1)},~~~v_2=\frac{[(m-a)(\rho-1)-1]+\sqrt{d}}{2(1-a)(\rho-1)},
\end{align*}
and we note (after some algebra) another helpful expression of $d$
\begin{align*}
    d=[(m-a)(\rho-1)-1]^2+4(1-a)m(\rho-1).
\end{align*}

Now we check these candidate values for $u$ and $v$ against the condition $u,v\in(0,1)$.
If $\rho>1$, it can be shown that
$d\geq[(m-a)(\rho-1)-1]^2$, which implies $v_2\geq0$ but $v_1<0$, ruling out the candidate $v_1$.
If $\rho<1$ then $d\geq[(a+m)(\rho-1)+1]^2$, which implies $u_2\geq0$ but $u_1<0$, ruling out the candidate $u_1$. In both cases, the choice left is
\begin{align*}
    u=u_2=\frac{(a+m)(\rho-1)+1-\sqrt{d}}{2a(\rho-1)},~~~v=v_2=\frac{(m-a)(\rho-1)-1+\sqrt{d}}{2(1-a)(\rho-1)}.
\end{align*}
That $u_2$ and $v_2$ are also $\leq1$ is clear from the first equation, which says that their weighted average is $\leq1$. The solution to the set of two equations is thus $(u_2,v_2)$.

The $u_2$ formula is spelled out as
\begin{align*}
    \frac{\mu_{01}(X)-l}{h-l}=\frac{\overbrace{[\pi_1(X)+{\textstyle\frac{\mu_0(X)-l}{h-l}}](\rho_1-1)+1}^{\textstyle=:\alpha_1(X)}
    -\overbrace{\sqrt{[\alpha_1(X)]^2-4\pi_1(X)\textstyle{\frac{\mu_0(X)-l}{h-l}}\rho_1(\rho_1-1)}}^{\textstyle=:\beta_1(X)}}{2\pi_1(X)(\rho_1-1)},
\end{align*}
which implies
\begin{align*}
    \mu_{01}(X)=\frac{\alpha_1(X)-\beta_1(X)}{2\pi_1(X)(\rho-1)}(h-l)+l.
\end{align*}
We could use the $v_2$ formula to express $\mu_{00}(X)$, replacing $\rho$ with $1/\rho_0$ and $a$ with $1-\pi_0(X)$, and simplify expressions. Or we can simply use symmetry to conclude that, under A0-A2 combined with A4-GOR, for $c=0,1$,
\begin{align*}
    \mu_{0c}(X)
    &=\frac{\alpha_c(X)
    -\beta_c(X)}{2\pi_c(X)(\rho_c-1)}(h-l)+l
    =:\mu_{0c}^\text{GOR}(X).\tag{\ref{gor-id:mu.0c}}
\end{align*}

\bigskip

\noindent
\textit{Part 2: assuming A0, A1, A2, A4-MR}
\smallskip

\noindent
Under these assumptions, we have two equations with two unknowns
\begin{align*}
    \begin{cases}
        \pi_1(X)\mu_{01}(X)+\pi_0(X)\mu_{00}(X)=\mu_0(X)
        \\
        \mu_{01}(X)/\mu_{00}(X)=\rho
    \end{cases}.
\end{align*}
Let
\begin{align*}
    a&:=\pi_1(X),
    \\
    m&:=\mu_0(X),
    \\
    u&:=\mu_{01}(X),
    \\
    v&:=\mu_{00}(X).
\end{align*}
Then our two equations become
\begin{align*}
    \begin{cases}
    au+(1-a)v=m
    \\
    u/v=\rho
    \end{cases}.
\end{align*}
The second equation implies that $u=\rho v$. Plugging this in the first equation and solving for $v$, we have
\begin{align*}
    v=\frac{m}{\rho a+ (1-a)}.
\end{align*}
Plugging this back in the second equation, we obtain
\begin{align*}
    u=\frac{\rho m}{\rho a+ (1-a)}.
\end{align*}
Therefore,
\begin{align*}
    \mu_{01}(X)=u
    &=\frac{\rho_1m}{\rho_1a+(1-a)}=\frac{\rho_1 m}{(\rho_1-1)a+1}=\frac{\rho_1\mu_0(X)}{(\rho_1-1)\pi_1(X)+1},
    \\
    \mu_{00}(X)=v
    &=\frac{m}{(1/\rho_0)a+(1-a)}=\frac{\rho_0 m}{a+\rho_0(1-a)}=\frac{\rho_0m}{(\rho_0-1)(1-a)+1}=\frac{\rho_0\mu_0(X)}{(\rho_0-1)\pi_0(X)+1}.
\end{align*}
Hence, for $c=0,1$,
\begin{align*}
    \mu_{0c}(X)=\frac{\rho_c\mu_0(X)}{(\rho_c-1)\pi_c(X)+1}=:\mu_{0c}^\text{MR}(X).\tag{\ref{mr-id:mu.0c}}
\end{align*}

\end{proof}

\bigskip\bigskip

\addcontentsline{toc}{subsection}{Proof of Proposition \ref{thm:ratio.params-ifs}}
\noindent
\begin{proof}[Proof of Proposition \ref{thm:ratio.params-ifs}]
The case with $\rho=1$ (i.e., PI) is already covered in Proposition \ref{thm:ifs-pi}, so we only need to consider the case with $\rho\neq1$.

\medskip

\noindent
\textit{Part 1: The IF of $\nu_{0c}^\textup{GOR}$}
\smallskip

\noindent
Consider $\nu_{01}^\text{GOR}$.
\vspace{-1em}
\begin{align*}
    \nu_{01}^\text{GOR}(\theta)=\E_{\theta_1}[\pi_1(X,\theta_3)\mu_{01}^\text{GOR}(X,\theta_3,\theta_6)]=\int\pi_1(x,\theta_3)\mu_{01}^\text{GOR}(x,\theta_3,\theta_6)f_1(x,\theta_1)dx.
\end{align*}
This function involves $\theta_1,\theta_3,\theta_4$ so the IF of $\nu_{01}^\text{GOR}$, denoted $\varphi_{\nu_{01}^\text{GOR}}(O)$, is the sum of three terms $\varphi_{\nu_{01}^\text{GOR},1}\in\T_1$, $\varphi_{\nu_{01}^\text{GOR},3}\in\T_3$ and $\varphi_{\nu_{01}^\text{GOR},6}\in\T_6$, such that
\begin{align*}
    \frac{\partial\nu_{01}^\text{GOR}(\theta)}{\partial\theta_1}\Big|_{\theta=\theta^0}
    &=\E[S_1(X)\varphi_{\nu_{01}^\text{GOR},1}(X)],
    \\
    \frac{\partial\nu_{01}^\text{GOR}(\theta)}{\partial\theta_3}\Big|_{\theta=\theta^0}
    &=\E[S_3(X,Z,C)\varphi_{\nu_{01}^\text{GOR},3}(X,Z,C)],
    \\
    \frac{\partial\nu_{01}^\text{GOR}(\theta)}{\partial\theta_6}\Big|_{\theta=\theta^0}
    &=\E[S_6(X,Z,C,Y)\varphi_{\nu_{01}^\text{GOR},6}(X,Z,C,Y)].
\end{align*}
\begin{align}
    \frac{\partial\nu_{01}^\text{GOR}(\theta)}{\partial\theta_1}\Big|_{\theta=\theta^0}
    &=\int\pi_1(x)\mu_{01}^\text{GOR}(x)f_1'(x)dx\nonumber
    \\
    &=\int\pi_1(x)\mu_{01}^\text{GOR}(x)S_1(x)f_1(x)dx=\E[S_1(X)\pi_1(X)\mu_{01}^\text{GOR}(X)] && \text{(by (\ref{f_1'(X)}))}\nonumber
    \\
    &=\E\big\{S_1(X)\underbrace{[\pi_1(X)\mu_{01}^\text{GOR}(X)-\nu_{01}^\text{GOR}]}_{\textstyle\varphi_{\nu_{01}^\text{GOR},1}(X)\in\T_1}\big\}.\nonumber
    \\
    \frac{\partial\nu_{01}^\text{GOR}(\theta)}{\partial\theta_3}\Big|_{\theta=\theta^0}
    &=\E\left[\pi_1'(X)\mu_{01}^\text{GOR}(X)+\pi_1(X)\frac{\partial\mu_{01}^\text{GOR}(X)}{\partial\pi_1(X)}\pi_1'(X)\right]\nonumber
    \\
    &=\E\Big\{\underbrace{\left[\mu_{01}^\text{GOR}(X)+\pi_1(X)\frac{\partial\mu_{01}^\text{GOR}(X)}{\partial\pi_1(X)}\right]}_{(*)}\pi_1'(X)\Big\},\label{nu01^GOR.deriv3}
    \\
    \frac{\partial\nu_{01}^\text{GOR}(\theta)}{\partial\theta_6}\Big|_{\theta=\theta^0}
    &=\E\Big[\underbrace{\pi_1(X)\frac{\partial\mu_{01}^\text{GOR}(X)}{\partial\mu_0(X)}}_{(**)}\mu_0'(X)\Big].\label{nu01^GOR.deriv6}
\end{align}
We now derive $(*)$ and $(**)$. These are functions of $X$, so to reduce notational burden, we suppress the $(X)$ notation until the last line.
\begin{align}
    (*)
    &=\mu_{01}^\text{GOR}+\pi_1\frac{h-l}{2(\rho_1-1)}\left\{\frac{\partial\alpha_1/\partial\pi_1-\partial\beta_1/\partial\pi_1}{\pi_1}-\frac{\alpha_1-\beta_1}{\pi_1^2}\right\}\nonumber
    \\
    &=\mu_{01}^\text{GOR}+\frac{h-l}{2(\rho_1-1)}\left\{\partial\alpha_1/\partial\pi_1-\partial\beta_1/\partial\pi_1-\frac{\alpha_1-\beta_1}{\pi_1}\right\}\nonumber
    \\
    &=\mu_{01}^\text{GOR}+\frac{h-l}{2(\rho_1-1)}\left\{(\rho_c-1)-\frac{1}{2\beta_1}\left[2\alpha_1(\rho_1-1)-4\frac{\mu_0-l}{h-l}\rho_1(\rho_1-1)\right]-\frac{\alpha_1-\beta_1}{\pi_1}\right\}\nonumber
    \\
    &=\mu_{01}^\text{GOR}+\left[\frac{1}{2}-\frac{\alpha_1}{2\beta_1}+\frac{\rho_1\frac{\mu_0-l}{h-l}}{\beta_1}\right](h-l)
    -\frac{\alpha_1-\beta_1}{2(\rho_1-1)\pi_1}(h-l)\nonumber
    \\
    &=\left[\frac{1}{2}-\frac{\alpha_1(X)}{2\beta_1(X)}+\frac{\rho_1\frac{\mu_0(X)-l}{h-l}}{\beta_1(X)}\right](h-l)+l\nonumber
    \\
    &=:\epsilon_{\pi,1}^\text{GOR}(X),\label{epsilon_mu,1}
    \\
    (**)
    &=\frac{h-l}{2(\rho_1-1)}\left\{\partial\alpha_1/\partial\mu_0-\partial\beta_1/\partial\mu_0\right\}\nonumber
    \\
    &=\frac{h-l}{2(\rho_1-1)}\left\{\frac{\rho_1-1}{h-l}-\frac{1}{2\beta_1}\left[2\alpha_1\frac{\rho_1-1}{h-l}-4\pi_1\frac{\rho_1(\rho_1-1)}{h-l}\right]\right\}\nonumber
    \\
    &=\frac{1}{2}-\frac{\alpha_1(X)}{2\beta_1(X)}+\frac{\rho_1\pi_1(X)}{\beta_1(X)}\nonumber
    \\
    &=\epsilon_{\mu,1}^\text{GOR}(X).\label{epsilon_pi,1}
\end{align}
Plugging (\ref{epsilon_mu,1}) into (\ref{nu01^GOR.deriv3}), and plugging (\ref{epsilon_pi,1}) into (\ref{nu01^GOR.deriv6}), we obtain
\begin{align*}
    \frac{\partial\nu_{01}^\text{GOR}(\theta)}{\partial\theta_3}\Big|_{\theta=\theta^0}
    &=\E[\epsilon_{\pi,1}^\text{GOR}(X)\pi_1'(X)]
    \\
    &=\E\left\{\epsilon_{\pi,1}^\text{GOR}(X)\E\left[S_3(X,Z,C)\frac{Z}{e(X,Z)}[C-\pi_1(X)]\mid X\right]\right\} && \text{(by (\ref{pi_1'(X)}))}
    \\
    &=\E\Big\{S_3(X,Z,C)\underbrace{\frac{Z}{e(X,Z)}\epsilon_{\pi,1}^\text{GOR}(X)[C-\pi_1(X)]}_{\textstyle=:\varphi_{\nu_{01}^\text{GOR},3}(X,Z,C)\in\T_3}\Big\},
    \\
    \frac{\partial\nu_{01}^\text{GOR}(\theta)}{\partial\theta_6}\Big|_{\theta=\theta^0}
    &=\E[\epsilon_{\mu,1}^\text{GOR}(X)\mu_0'(X)]
    \\
    &=\E\left\{\epsilon_{\mu,1}^\text{GOR}(X)\E\left[S_6(X,Z,Y)\frac{1-Z}{e(X,Z)}[Y-\mu_0(X)]\mid X\right]\right\} && \text{(by (\ref{mu_0'(X)}))}
    \\
    &=\E\Big\{S_6(X,Z,Y)\underbrace{\frac{1-Z}{e(X,Z)}\epsilon_{\mu,1}^\text{GOR}(X)[Y-\mu_0(X)]}_{\textstyle=:\varphi_{\nu_{01}^\text{GOR},6}(X,Z,Y)\in\T_6}\Big\}.
\end{align*}
Therefore,
\begin{align*}
    \varphi_{\nu_{01}^\text{GOR}}(O)=\frac{1-Z}{e(X,Z)}\epsilon_{\mu,1}^\text{GOR}(X)[Y-\mu_0(X)]+\frac{Z}{e(X,Z)}\epsilon_{\pi,1}^\text{GOR}(X)[C-\pi_1(X)]+\pi_1(X)\mu_{01}^\text{GOR}(X)-\nu_{01}^\text{GOR}.
\end{align*}
Hence, for $c=0,1$, the IF of $\nu_{0c}^\text{GOR}$ is
\begin{align*}
    \varphi_{\nu_{0c}^\text{GOR}}(O)=\frac{1-Z}{e(X,Z)}\epsilon_{\mu,c}^\text{GOR}(X)[Y-\mu_0(X)]+\frac{Z}{e(X,Z)}\epsilon_{\pi,c}^\text{GOR}(X)[\I(C=c)-\pi_c(X)]+\pi_c(X)\mu_{0c}^\text{GOR}(X)-\nu_{0c}^\text{GOR}.
\end{align*}

\bigskip

\noindent
\textit{Part 2: The IF of $\nu_{0c}^\textup{MR}$}
\smallskip

\noindent
Consider $\nu_{01}^\text{MR}$.
\vspace{-1em}
\begin{align*}
    \nu_{01}^\text{MR}(\theta)=\E_{\theta_1}[\pi_1(X,\theta_3)\overbrace{\gamma_1(X,\theta_3)}^{\textstyle\frac{\rho}{(\rho-1)\pi_1(x,\theta_3)+1}}\mu_0(X,\theta_3,\theta_6)]
    =\int\pi_1(x,\theta_3)\gamma_1(x,\theta_3)\mu_0(x,\theta_3,\theta_6)f_1(x,\theta_1)dx.
\end{align*}
This function involves $\theta_1,\theta_3,\theta_4$ so the IF of $\nu_{01}^\text{MR}$, denoted $\varphi_{\nu_{01}^\text{MR}}(O)$, is the sum of three terms $\varphi_{\nu_{01}^\text{MR},1}\in\T_1$, $\varphi_{\nu_{01}^\text{MR},3}\in\T_3$ and $\varphi_{\nu_{01}^\text{MR},6}\in\T_6$, such that
\begin{align*}
    \frac{\partial\nu_{01}^\text{MR}(\theta)}{\partial\theta_1}\Big|_{\theta=\theta^0}
    &=\E[S_1(X)\varphi_{\nu_{01}^\text{MR},1}(X)],
    \\
    \frac{\partial\nu_{01}^\text{MR}(\theta)}{\partial\theta_3}\Big|_{\theta=\theta^0}
    &=\E[S_3(X,Z,C)\varphi_{\nu_{01}^\text{MR},3}(X,Z,C)],
    \\
    \frac{\partial\nu_{01}^\text{MR}(\theta)}{\partial\theta_6}\Big|_{\theta=\theta^0}
    &=\E[S_6(X,Z,C,Y)\varphi_{\nu_{01}^\text{MR},6}(X,Z,C,Y)],
\end{align*}
\begin{align*}
    \frac{\partial\nu_{01}^\text{MR}(\theta)}{\partial\theta_1}\Big|_{\theta=\theta^0}
    &=\int\pi_1(x)\gamma_1(x)\mu_0(x)f_1'(x)dx\nonumber
    \\
    &=\int\pi_1(x)\gamma_1(x)\mu_0(x)S_1(x)f_1(x)dx=\E[S_1(X)\pi_1(X)\gamma_1(X)\mu_0(X)] && \text{(by (\ref{f_1'(X)}))}\nonumber
    \\
    &=\E\big\{S_1(X)\underbrace{[\pi_1(X)\gamma_1(X)\mu_0(X)-\nu_{01}^\text{MR}]}_{\textstyle\varphi_{\nu_{01}^\text{MR},1}(X)\in\T_1}\big\}.\nonumber
    \\
    \frac{\partial\nu_{01}^\text{MR}(\theta)}{\partial\theta_3}\Big|_{\theta=\theta^0}
    &=\E\left[\pi_1'(X)\gamma_1(X)\mu_0(X)-\pi_1(X)\frac{\rho(\rho-1)}{[(\rho-1)\pi_1(X)+1]^2}\mu_0(X)\pi_1'(X)\right]\nonumber
    \\
    &=\E\left\{\gamma_1(X)\left[1-\frac{(\rho-1)\pi_1(X)}{(\rho-1)\pi_1(X)+1}\right]\mu_0(X)\pi_1'(X)\right\}
    \\
    &=\E[\underbrace{\gamma_1(X)\gamma_0(X)\mu_0(X)}_{\textstyle=:\epsilon_{\pi,1}^\text{MR}(X)}\pi_1'(X)]
    \\
    &=\E\left\{\epsilon_{\pi,1}^\text{MR}(X)\E\left[S_3(X,Z,C)\frac{Z}{e(X,Z)}[C-\pi_1(X)]\mid X\right]\right\} && \text{(by (\ref{pi_1'(X)}))}
    \\
    &=\E\Big\{S_3(X,Z,C)\underbrace{\frac{Z}{e(X,Z)}\epsilon_{\pi,1}^\text{MR}(X)[C-\pi_1(X)]}_{\textstyle=:\varphi_{\nu_{01},3}^\text{MR}(X,Z,C)\in\T_3}\Big\}.
    \\
    \frac{\partial\nu_{01}^\text{MR}(\theta)}{\partial\theta_6}\Big|_{\theta=\theta^0}
    &=\E[\underbrace{\gamma_1(X)\pi_1(X)}_{\textstyle=:\epsilon_{\mu,1}^\text{MR}(X)}\mu_0'(X)]
    \\
    &=\E\left\{\epsilon_{\mu,1}^\text{MR}(X)\E\left[S_6(X,Z,Y)\frac{1-Z}{e(X,Z)}[Y-\mu_0(X)]\right]\right\} && \text{(by (\ref{mu_0'(X)}))}
    \\
    &=\E\Big\{S_6(X,Z,Y)\underbrace{\frac{1-Z}{e(X,Z)}\epsilon_{\mu,1}^\text{MR}(X)[Y-\mu_0(X)]}_{\textstyle=:\varphi_{\nu_{01}^\text{MR},6}(X,Z,Y)\in\T_6}\Big\}.
\end{align*}
Therefore,
\begin{align*}
    \varphi_{\nu_{01}^\text{MR}}(O)=\frac{1-Z}{e(X,Z)}\epsilon_{\mu,1}^\text{MR}(X)[Y-\mu_0(X)]+\frac{Z}{e(X,Z)}\epsilon_{\pi,1}^\text{MR}(X)[C-\pi_1(X)]+\pi_1(X)\mu_{01}^\text{MR}(X)-\nu_{01}^\text{MR}.
\end{align*}
Hence, for $c=0,1$, the IF of $\nu_{0c}^\text{MR}$ is
\begin{align*}
    \varphi_{\nu_{0c}^\text{MR}}(O)=\frac{1-Z}{e(X,Z)}\epsilon_{\mu,c}^\text{MR}(X)[Y-\mu_0(X)]+\frac{Z}{e(X,Z)}\epsilon_{\pi,c}^\text{MR}(X)[\I(C=c)-\pi_c(X)]+\pi_c(X)\mu_{0c}^\text{MR}(X)-\nu_{0c}^\text{MR}.
\end{align*}
\end{proof}

\bigskip

\addcontentsline{toc}{subsection}{Proof of Proposition \ref{thm:robustness-loss}}
\noindent
\begin{proof}[Proof of Proposition \ref{thm:robustness-loss}]
\smallskip

\noindent
\textit{Part 1: The shared partial robustness of} $\hat\Delta_{c,\textsc{ms}}^\mathrm{GOR}$, $\hat\Delta_{c,\textsc{if}}^\mathrm{GOR}$ \textit{and} $\hat\Delta_{c,\textsc{ifh}}^\mathrm{GOR}$
\smallskip

\noindent
Recall from the proof of multiply robustness of $\hat\Delta_{c,\mathrm{MS}}^\mathrm{PI}$, $\hat\Delta_{c,\mathrm{IF}}^\mathrm{PI}$ and $\hat\Delta_{c,\mathrm{IFH}}^\mathrm{PI}$ that it is based on this result:
\begin{itemize}
    \item the component that estimates $\pi_c$ is consistent if either $\hat e(X,Z)$ or $\hat\pi_c(X)$ is correctly specified;
    \item the component that estimates $\nu_{1c}$ is consistent if either $\hat e(X,Z)$ is correctly specified or both $\hat\pi_c(X)$ and $\hat\mu_{1c}(X)$ are correctly specified;
    \item the component that estimates $\nu_{0c}^\text{PI}$ is consistent if $\hat e(X,Z)$ and $\hat\pi_c(X)$ are correctly specified, or if $\hat e(X,Z)$ and $\hat\mu_0(X)$ are correctly specified, or if $\hat\pi_c(X)$ and $\hat\mu_0(X)$ are correctly specified.
\end{itemize}
For $\hat\Delta_{c,\textsc{ms}}^\text{GOR}$, $\hat\Delta_{c,\textsc{if}}^\text{GOR}$ and $\hat\Delta_{c,\textsc{ifh}}^\text{GOR}$, the third bullet above is replaced with
\begin{itemize}
    \item the component that estimates $\nu_{0c}^\text{GOR}$ is consistent if $\hat\pi_c(X)$ and $\hat\mu_0(X)$ are correctly specified.
\end{itemize}
(The proof of this statement uses the same kind of reasoning used to prove $\hat\Delta_{c,\mathrm{MS}}^\mathrm{PI}$, $\hat\Delta_{c,\mathrm{IF}}^\mathrm{PI}$ and $\hat\Delta_{c,\mathrm{IFH}}^\mathrm{PI}$ are multiply robust, so is left out here.)
It follows that $\hat\Delta_{c,\textsc{ms}}^\text{GOR}$, $\hat\Delta_{c,\textsc{if}}^\text{GOR}$ and $\hat\Delta_{c,\textsc{ifh}}^\text{GOR}$ are consistent if the following conditions hold:
\begin{enumerate}[(1)]
    \item $\hat\pi_c(X)$ and $\hat\mu_0(X)$ are both correctly specified; AND
    \item either $\hat e(X,Z)$ or $\hat\mu_{1c}(X)$ is correctly specified.
\end{enumerate}

\bigskip
\noindent
\textit{Part 2: The shared partial robustness of} $\hat\Delta_{c,\textsc{ms}}^\mathrm{MR}$, $\hat\Delta_{c,\textsc{if}}^\mathrm{MR}$ \textit{and} $\hat\Delta_{c,\textsc{ifh}}^\mathrm{MR}$
\smallskip

\noindent
We follow the reasoning above. For $\hat\Delta_{c,\textsc{ms}}^\text{MR}$, $\hat\Delta_{c,\textsc{if}}^\text{MR}$ and $\hat\Delta_{c,\textsc{ifh}}^\text{MR}$, bullet 3 above is replaced with
\begin{itemize}
    \item the component that estimates $\nu_{0c}^\text{MR}$ is consistent if $\hat\pi_c(X)$ is correctly specified, and either $\hat e(X,Z)$ or $\hat\mu_0(X)$ is correctly specified.
\end{itemize}
(The proof of this statement uses the same kind of reasoning used to prove $\hat\Delta_{c,\mathrm{MS}}^\mathrm{PI}$, $\hat\Delta_{c,\mathrm{IF}}^\mathrm{PI}$ and $\hat\Delta_{c,\mathrm{IFH}}^\mathrm{PI}$ are multiply robust, so is left out here.)
As a result, $\hat\Delta_{c,\textsc{ms}}^\text{MR}$, $\hat\Delta_{c,\textsc{if}}^\text{MR}$ and $\hat\Delta_{c,\textsc{ifh}}^\text{MR}$ are consistent if the following conditions hold:
\begin{enumerate}[(1)]
    \item $\hat\pi_c(X)$ is correctly specified; AND
    \item either $\hat e(X,Z)$ is correctly specified or both $\hat\mu_{1c}(X),\hat\mu_0(X)$ are correctly specified.
\end{enumerate}
\end{proof}

\bigskip

\addcontentsline{toc}{subsection}{Elaboration of Remark \ref{rm:approximate-robustness}}
\noindent
\begin{proof}[Elaboration of Remark \ref{rm:approximate-robustness}]
\smallskip

\noindent
\textit{Part 1: The approximate robustness of} $\hat\Delta_{c,\textsc{if}}^\mathrm{GOR}$ \textit{and} $\hat\Delta_{c,\textsc{ifh}}^\mathrm{GOR}$
\smallskip

\noindent
This is about the component that estimates $\nu_{0c}^\text{GOR}$. 
In the proof of Proposition \ref{thm:robustness-loss} above, we mentioned that for $\hat\Delta_{c,\textsc{ms}}^\text{GOR}$, $\hat\Delta_{c,\textsc{if}}^\text{GOR}$ and $\hat\Delta_{c,\textsc{ifh}}^\text{GOR}$, this component of  is consistent if $\hat\pi_c(X)$ and $\hat\mu_0(X)$ are both correctly specified. However, for $\hat\Delta_{c,\textsc{if}}^\text{GOR}$ and $\hat\Delta_{c,\textsc{ifh}}^\text{GOR}$ (but not $\hat\Delta_{c,\textsc{ms}}^\text{GOR}$), this component has an approximate robustness w.r.t. to these two functions whose correct specification they require for consistency, which we now explain. 

It can be shown that the probability limits of $\hat\nu_{0c,\textsc{if}}^\text{GOR}$ and $\hat\nu_{0c,\textsc{ifh}}^\text{GOR}$ are both
\begin{align}
    \E\left\{\epsilon_{\pi,c}^\text{GOR}[\mu_0(X),\pi_c^\dagger(X)][\pi_c(X)-\pi_c^\dagger(X)]\right\}+\E\left\{\pi_c^\dagger(X)\mu_{0c}^\text{GOR}[\mu_0(X),\pi_c^\dagger(X)]\right\}\label{OR:if-appx.robust1}
\end{align}
when $\hat e(X,Z)$ and $\hat\mu_0(X)$ are correctly specified but $\hat\pi_c(X)$ is not, and are both
\begin{align}
    \E\left\{\epsilon_{\mu,c}^\text{GOR}[\mu_0^\dagger(X),\pi_c(X)][\mu_0(X)-\mu_0^\dagger(X)]\right\}+\E\left\{\pi_c(X)\mu_{0c}^\text{GOR}[\mu_0^\dagger(X),\pi_c(X)]\right\}\label{OR:if-appx.robust2}
\end{align}
when $\hat e(X,Z)$ and $\hat\pi_c(X)$ are correctly specified but $\hat\mu_0(X)$ is not. In both of these, the second term is the probability limit of the biased plug-in estimator. The first term in (\ref{OR:if-appx.robust1}) is the first term in the Taylor expansion of the true parameter $\nu_{0c}^\text{GOR}$ (treated as a function of $\pi_c()$) at the point $\pi_c^\dagger()$. The first term in (\ref{OR:if-appx.robust2}) is the first term in the Taylor expansion of $\nu_{0c}^\text{GOR}$ (treated as a function of $\mu_0()$) at the point $\mu_0^\dagger()$. 
Note though that for this approximate robustness property to be beneficial (reducing bias), $\pi_c^\dagger(X)$ and $\mu_0^\dagger(X)$ need to be close to $\pi_c(X)$ and $\mu_0(X)$, respectively.

\bigskip

\noindent
\textit{Part 2: The approximate robustness of} $\hat\Delta_{c,\textsc{if}}^\mathrm{MR}$ \textit{and} $\hat\Delta_{c,\textsc{ifh}}^\mathrm{MR}$
\smallskip

\noindent
This is about the component that estimates $\nu_{0c}^\text{MR}$. 
In the proof of Proposition \ref{thm:robustness-loss} above, we mentioned that for $\hat\Delta_{c,\textsc{ms}}^\text{MR}$, $\hat\Delta_{c,\textsc{if}}^\text{MR}$ and $\hat\Delta_{c,\textsc{ifh}}^\text{MR}$, this component of  is not consistent if $\hat\pi_c(X)$ is mis-specified. However, for $\hat\Delta_{c,\textsc{if}}^\text{MR}$ and $\hat\Delta_{c,\textsc{ifh}}^\text{MR}$ (but not $\hat\Delta_{c,\textsc{ms}}^\text{MR}$), this component has an approximate robustness w.r.t. to this function whose correct specification they require for consistency. 

It can be shown that the probability limits of $\hat\nu_{0c,\textsc{if}}^\text{MR}$ and $\hat\nu_{0c,\textsc{ifh}}^\text{MR}$ are both
\begin{align}
    \E\left\{\epsilon_{\pi,c}^\text{MR}[\mu_0(X),\pi_c^\dagger(X)][\pi_c(X)-\pi_c^\dagger(X)]\right\}+\E\left\{\pi_c^\dagger(X)\mu_{0c}^\text{MR}[\mu_0(X),\pi_c^\dagger(X)]\right\}
\end{align}
when $\hat e(X,Z)$ and $\hat\mu_0(X)$ are correctly specified but $\hat\pi_c(X)$ is not. Here the second term is the probability limit of the biased plug-in estimator. The first term is the first term in the Taylor expansion of the true parameter $\nu_{0c}^\text{MR}$ (treated as a function of $\pi_c()$) at the point $\pi_c^\dagger()$.  
For this approximate robustness property to be beneficial (reducing bias), $\pi_c^\dagger(X)$ need to be close to $\pi_c(X)$.
\end{proof}

\section{Proofs and additional results for Section \ref{sec:diff-param} -- Sens analysis with SMD}

\addcontentsline{toc}{subsection}{Proposition \ref{thm:smd-bounds}}
\begin{theoremb}{\ref*{thm:smd-id}b}\label{thm:smd-bounds}
    \textit{Under A0-A2 combined with A4-SMD, $\Delta_c$ lies between the two bounds}
    \begin{align}
        \Delta_c^\text{PI}-\eta_c\,\E\left[\frac{\pi_1(X)\pi_0(X)\sigma_0(X)}{\sqrt{1\pm|\pi_1(X)-\pi_0(X)|+\eta^2\pi_1(X)\pi_0(X)}}\right]/\pi_c.
    \end{align}
    \textit{Under A0-A2 combined with A4-SMD and the assumption that $\frac{1}{k}\leq\frac{\sigma_{01}^2(X)}{\sigma_{00}^2(X)}\leq k$ for a specified $k>1$, $\Delta_c$ lies between the two bounds}
    \begin{align}
        \Delta_c^\text{PI}-\eta_c\,\E\left[\frac{\pi_1(X)\pi_0(X)\sigma_0(X)}{\sqrt{1\pm\frac{k-1}{k+1}|\pi_1(X)-\pi_0(X)|+\eta^2\pi_1(X)\pi_0(X)}}\right]/\pi_c.
    \end{align}
\end{theoremb}

\bigskip

\noindent
To prove Propositions \ref{thm:smd-id} and \ref{thm:smd-bounds}, we will use Lemma \ref{lm:mixture-variance}.

\addcontentsline{toc}{subsection}{Lemma \ref{lm:mixture-variance}}

\begin{lemma}[Variance of mixture of two distributions]\label{lm:mixture-variance}
    Consider two distributions with means $\mu_1,\mu_2$ and variances $\sigma_1^2,\sigma_2^2$. The mixture of these two distributions by mixing ratio is $p:(1-p)$ has variance
    \begin{align}
        \sigma^2=[p\sigma_1^2+(1-p)\sigma_2^2]+p(1-p)(\mu_1-\mu_2)^2.\label{eq:mixture-variance}
    \end{align}
\end{lemma}

\addcontentsline{toc}{subsection}{Proof of Lemma \ref{lm:mixture-variance}}
\noindent
\begin{proof}[Proof of Lemma \ref{lm:mixture-variance}]
    Applying the law of total variance, we derive the variance of the mixture (total variance) based on the means and variances of the components (conditional means and variances).
    \begin{align*}
        \sigma^2
        &=
    \overbrace{[p\sigma_1^2+(1-p)\sigma_2^2]}^{\textstyle\text{expectation of conditional variance}}+~~~
    \overbrace{\big\{[p\mu_1^2+(1-p)\mu_2^2]-[p\mu_1+(1-p)\mu_2]^2\big\}}^{\textstyle\text{variance of conditional expectation}}
    \\
    &=[p\sigma_1^2+(1-p)\sigma_2^2]+p(1-p)\mu_1^2+(1-p)[1-(1-p)]\mu_2^2-2p(1-p)\mu_1\mu_2
    \\
    &=[p\sigma_1^2+(1-p)\sigma_2^2]+p(1-p)(\mu_1-\mu_2)^2.\tag{\ref{eq:mixture-variance}}
    \end{align*}
\end{proof}

\addcontentsline{toc}{subsection}{Proof of Propositions \ref{thm:smd-id} and \ref{thm:smd-bounds}}
\noindent
\begin{proof}[Proof of Propositions \ref{thm:smd-id} and \ref{thm:smd-bounds}]
The proof here covers both propositions. We will refer to the combination of A4-SMD and $\frac{1}{k}\leq\frac{\sigma_{01}^2(X)}{\sigma_{00}^2(X)}\leq k$ as A4-SMDr. Proposition \ref{thm:smd-bounds} is covered in what are referred to below as case 1 (assuming A4-SMD) and case 2 (assuming A4-SMDr), and Proposition \ref{thm:smd-id} is covered in case 3 (assuming A4-SMDe).

We will rely on the mixture mean and mixture variance equations
\begin{align}
    \mu_0(X)
    &=\pi_1(X)\mu_{01}(X)+\pi_0(X)\mu_{00}(X),\tag{\ref{eq:mixture-mean}},
    \\
    \sigma_0^2(X)
    &=[\pi_1(X)\sigma_{01}^2(X)+\pi_0(X)\sigma_{00}^2(X)]+\pi_1(X)\pi_0(X)[\mu_{01}(X)-\mu_{00}(X)]^2. && \text{(by Lemma \ref{lm:mixture-variance})}\label{mixture-variance}
\end{align}
As the terms in these equations are all functions of $X$, we suppress the $(X)$ notation to simplify presentation,
\begin{align}
    \mu_0
    &=\pi_1\mu_{01}+\pi_0\mu_{00}\tag{\ref{eq:mixture-mean}}
    \\
    \sigma_0^2
    &=(\pi_1\sigma_{01}^2+\pi_0\sigma_{00}^2)+\pi_1\pi_0(\mu_{01}-\mu_{00})^2.\tag{\ref{mixture-variance}}
\end{align}
We will first see what can be deduced from the combination of (\ref{mixture-variance}) with each of the three A4-SMD assumptions, before combining that result with (\ref{eq:mixture-mean}) to obtain the final result.

To simplify notation, let $\pi_\text{diff}:=|\pi_1(X)-\pi_0(X)|$.

\begin{enumerate}[leftmargin=1.3cm]
    \item[\textit{Case 1.}] 
    Assume A4-SMD. 
    
    Note that
    \begin{align*}
        (1-\pi_\text{diff})(\sigma_{01}^2/2+\sigma_{00}^2/2)
        \leq (\pi_1\sigma_{01}^2+\pi_0\sigma_{00}^2)
        \leq (1+\pi_\text{diff})(\sigma_{01}^2/2+\sigma_{00}^2/2).
    \end{align*}
    (To see this, suppose $\pi_1\geq\pi_0$. Then $1-\pi_\text{diff}=2\pi_0$ and $1+\pi_\text{diff}=2\pi_1$, so the LHS above is equal to $\pi_0(\sigma_{01}^2+\sigma_{00}^2)$, while the RHS is equal to $\pi_1(\sigma_{01}^2+\sigma_{00}^2)$.)

    Under A4-SMD, this becomes
    \begin{align*}
        (1-\pi_\text{diff})(\mu_{01}-\mu_{00})^2/\eta^2
        \leq (\pi_1\sigma_{01}^2+\pi_0\sigma_{00}^2)
        \leq (1+\pi_\text{diff})(\mu_{01}-\mu_{00})^2/\eta^2,
    \end{align*}
    which, combined with (\ref{mixture-variance}), implies (after some simple algebra)
    \begin{align*}
        \frac{\sigma_0^2\eta^2}{1+\pi_\text{diff}+\eta^2\pi_1\pi_0}\leq(\mu_{01}-\mu_{00})^2\leq\frac{\sigma_0^2\eta^2}{1-\pi_\text{diff}+\eta^2\pi_1\pi_0},
    \end{align*}
    so $\mu_{01}-\mu_{00}$ is bounded between
    \begin{align}
        \frac{\sigma_0\eta}{\sqrt{1\pm\pi_\text{diff}+\eta^2\pi_1\pi_0}}.\label{SMD:mu.diff}
    \end{align}
    
    \item[\textit{Case 2.}]
    Assume A4-SMDr.
    
    Note that, if $\pi_1\geq\pi_0$,
    \begin{align*}
        (1-\pi_\text{diff})(\sigma_{01}^2/2+\sigma_{00}^2/2)+\pi_\text{diff}\sigma_{01}^2=(\pi_1\sigma_{01}^2+\pi_0\sigma_{00}^2)=(1+\pi_\text{diff})(\sigma_{01}^2/2+\sigma_{00}^2/2)-\pi_\text{diff}\sigma_{00}^2,
    \end{align*}
    and if $\pi_1<\pi_0$,
    \begin{align*}
        (1-\pi_\text{diff})(\sigma_{01}^2/2+\sigma_{00}^2/2)+\pi_\text{diff}\sigma_{00}^2=(\pi_1\sigma_{01}^2+\pi_0\sigma_{00}^2)=(1+\pi_\text{diff})(\sigma_{01}^2/2+\sigma_{00}^2/2)-\pi_\text{diff}\sigma_{01}^2.
    \end{align*}
    The variance restriction assumption $\frac{1}{k}\leq\frac{\sigma_{01}^2}{\sigma_{00}^2}\leq k$ implies that both $\sigma_{00}^2$ and $\sigma_{01}^2$ are greater than or equal to $\frac{2}{1+k}(\sigma_{01}^2/2+\sigma_{00}^2/2)$. It follows that
    \begin{align*}
        \underbrace{\left(1-\pi_\text{diff}+\frac{2\pi_\text{diff}}{1+k}\right)}_{1-\frac{k-1}{k+1}\pi_\text{diff}}(\sigma_{01}^2/2+\sigma_{00}^2/2)
        \leq
        (\pi_1\sigma_{01}^2+\pi_0\sigma_{00}^2)
        \leq
        \underbrace{\left(1+\pi_\text{diff}-\frac{2\pi_\text{diff}}{1+k}\right)}_{1+\frac{k-1}{k+1}\pi_\text{diff}}(\sigma_{01}^2/2+\sigma_{00}^2/2).
    \end{align*}
    Under the SMD part of A4-SMDr, this becomes
    \begin{align*}
        \left(1-\frac{k-1}{k+1}\pi_\text{diff}\right)(\mu_{01}-\mu_{00})^2/\eta^2
        \leq
        (\pi_1\sigma_{01}^2+\pi_0\sigma_{00}^2)
        \leq
        \left(1+\frac{k-1}{k+1}\pi_\text{diff}\right)(\mu_{01}-\mu_{00})^2/\eta^2.
    \end{align*}
    which, combined with (\ref{mixture-variance}), implies (after some simple algebra)
    \begin{align*}
        \frac{\sigma_0^2\eta^2}{1+\frac{k-1}{k+1}\pi_\text{diff}+\eta^2\pi_1\pi_0}\leq(\mu_{01}-\mu_{00})^2\leq\frac{\sigma_0^2\eta^2}{1-\frac{k-1}{k+1}\pi_\text{diff}+\eta^2\pi_1\pi_0},
    \end{align*}
    so $\mu_{01}-\mu_{00}$ is bounded between
    \begin{align}
        \frac{\sigma_0\eta}{\sqrt{1\pm\frac{k-1}{k+1}\pi_\text{diff}+\eta^2\pi_1\pi_0}}.\label{SMDr:mu.diff}
    \end{align}
    
    \item[\textit{Case 3.}]
    Assume A4-SMDe.
    
    In this case $\pi_1\sigma_{01}^2+\pi_0\sigma_{00}^2=\sigma_{01}^2/2+\sigma_{00}^2/2=(\mu_{01}-\mu_{00})^2/\eta^2$. Combining this with (\ref{mixture-variance}) obtains
    \begin{align}
        \mu_{01}-\mu_{00}=\frac{\sigma_0\eta}{\sqrt{1+\eta^2\pi_1\pi_0}}.\label{SMDe:mu.diff}
    \end{align}
\end{enumerate}

Now we combine the above intermediate results with the mixture mean equation (\ref{eq:mixture-mean}).
For case 3, combining (\ref{SMDe:mu.diff}) with (\ref{eq:mixture-mean}), we have a set of two linear equations, with the solution
\begin{align*}
    \mu_{01}
    &=\mu_0+\frac{\eta\pi_0\sigma_0}{\sqrt{1+\eta^2\pi_1\pi_0}},
    \\
    \mu_{00}
    &=\mu_0-\frac{\eta\pi_1\sigma_0}{\sqrt{1+\eta^2\pi_1\pi_0}},
\end{align*}
which can be collectively expressed as
\begin{align*}
    \mu_{0c}=\mu_0+\eta_c\frac{\pi_{1-c}\sigma_0}{\sqrt{1+\eta^2\pi_1\pi_0}}.
\end{align*}
It follows that
\begin{align}
    \Delta_c=\Delta_c^\text{PI}-\eta_c\E\left[\frac{\pi_1(X)\pi_0(X)\sigma_0(X)}{\sqrt{1+\eta^2\pi_1(X)\pi_0(X)}}\right]/\pi_c=:\Delta_c^\text{SMDe}.\tag{\ref{SMDe:Delta.c}}
\end{align}
For cases 1 and 2 we combine each of the bounds of $\mu_{01}-\mu_{00}$ (in (\ref{SMD:mu.diff}) and (\ref{SMDr:mu.diff}), respectively) with (\ref{eq:mixture-mean}) and solve for the corresponding bounds of $\mu_{0c}(X)$. The bounds for $\Delta_c$ then follow.
\end{proof}

\bigskip\bigskip

\addcontentsline{toc}{subsection}{Proof of Proposition \ref{thm:smde-if}}
\noindent
\begin{proof}[Proof of Proposition \ref{thm:smde-if}]
We derive the IF of
\begin{align*}
    \vartheta:=\E\Big[\overbrace{\frac{\pi_1(X)\pi_0(X)\sigma_0(X)}{\sqrt{1+\eta^2\pi_1(X)\pi_0(X)}}}^{\textstyle=:\vartheta(X)}\Big].
\end{align*}
The function
\begin{align*}
    \vartheta(\theta)=\E_{\theta_1}\left[\vartheta(X,\theta_3,\theta_6)\right]
\end{align*}
involves $\theta_1,\theta_3,\theta_6$ so the IF of $\vartheta$, $\varphi_\vartheta(O)$, is the sum of three terms $\varphi_{\vartheta,1}\in\mathcal{T}_1$, $\varphi_{\vartheta,3}\in\mathcal{T}_3$, $\varphi_{\vartheta,6}\in\mathcal{T}_6$, such that
\begin{align*}
    \frac{\partial\vartheta(\theta)}{\partial\theta_1}\Big|_{\theta=\theta^0}
    &=\E[S_1(X)\varphi_{\vartheta,1}(X)],
    \\
    \frac{\partial\vartheta(\theta)}{\partial\theta_3}\Big|_{\theta=\theta^0}
    &=\E[S_3(X,Z,C)\varphi_{\vartheta,3}(X,Z,C)],
    \\
    \frac{\partial\vartheta(\theta)}{\partial\theta_6}\Big|_{\theta=\theta^0}
    &=\E[S_6(X,Z,Y)\varphi_{\vartheta,6}(X,Z,Y)].
\end{align*}
\begin{align}
    \frac{\partial\vartheta(\theta)}{\partial\theta_1}\Big|_{\theta=\theta^0}
    &=\int\vartheta(x)f_1'(x)dx\nonumber
    \\
    &=\int\vartheta(x)S_1(x)f_1(x)dx=\E[S_1(X)\vartheta(X)] && \text{(by (\ref{f_1'(X)}))}\nonumber
    \\
    &=\E\{S_1(X)\underbrace{[\vartheta(X)-\vartheta]}_{\textstyle=:\varphi_{\vartheta,1}(X)\in\T_1}\}.\label{if-vartheta-1}
    \\
    \frac{\partial\vartheta(\theta)}{\partial\theta_3}\Big|_{\theta=\theta^0}
    &=\E\left[\frac{\partial\vartheta(X,\theta)}{\partial\theta_3}\Big|_{\theta=\theta^0}\right]\nonumber
    \\
    &=\E\Big[\underbrace{\frac{\partial\vartheta(X)}{\partial\pi_1(X)}}_{\textstyle\underbrace{\left[\frac{(\pi_0-\pi_1)}{\sqrt{1+\eta^2\pi_1\pi_0}}
    -\frac{\eta^2(\pi_0-\pi_1)\pi_1\pi_0}{2(1+\eta^2\pi_1\pi_0)^{3/2}}\right]\sigma_0}_{\displaystyle \frac{(2+\eta^2\pi_1\pi_0)(\pi_0-\pi_1)\sigma_0}{2(1+\eta^2\pi_1\pi_0)^{3/2}}}}\pi_1'(X)\Big]\nonumber
    \\
    &=\E\Big\{\underbrace{\frac{[2+\eta^2\pi_1(X)\pi_0(X)][\pi_0(X)-\pi_1(X)]\sigma_0(X)}{2[1+\eta^2\pi_1(X)\pi_0(X)]^{3/2}}}_{\textstyle=:\boldsymbol\dot\vartheta_\pi(X)}
    \E\left[S_3(X,Z,C)\frac{Z}{e(X,Z)}[C-\pi_1(X)]\mid X\right]\Big\}. && \text{(by (\ref{pi_1'(X)}))}\nonumber
    \\
    &=\E\Big\{S_3(X,Z,C)\underbrace{\frac{Z}{e(X,Z)}\boldsymbol\dot\vartheta_\pi(X)[C-\pi_1(X)]}_{\textstyle=:\varphi_{\vartheta,3}(X,Z,C)\in\T_3}\Big\}.\label{if-vartheta-3}
    \\
    \frac{\partial\vartheta(\theta)}{\partial\theta_6}\Big|_{\theta=\theta^0}
    &=\E\Big[\underbrace{\frac{\pi_1(X)\pi_0(X)}{\sqrt{1+\eta^2\pi_1(X)\pi_0(X)}}\frac{1}{2\sigma_0(X)}}_{\textstyle=:\boldsymbol\dot\vartheta_{\sigma^2}(X)}
    \underbrace{\frac{\partial\sigma^2_0(X,\theta_6)}{\partial\theta_6}\Big|_{\theta_6=\theta_6^0}}_{\textstyle=:{\sigma_0^2}'(X)}\Big].\label{if-vartheta-6-midstep}
\end{align}
\begin{align*}
    {\sigma^2_0}'(X)
    &=\frac{\partial}{\partial\theta_6}\big\{\E_{\theta_6}[Y^2\mid X,Z=0]-\mu_0^2(X,\theta_6)\big\}\Big|_{\theta_6=\theta_6^0}
    \\
    &=\frac{\partial}{\partial\theta_6}\int y^2f_6(y\mid X,Z=0)dy\Big|_{\theta_6=\theta_6^0}-2\mu_0(X)\mu_0'(X)
    \\
    &=\E[S_6(X,Z,Y)Y^2\mid X,Z=0]-2\mu_0(X)\mu_0'(X)
    \\
    &=\E[S_6(X,Z,Y)Y^2\mid X,Z=0]-2\mu_0(X)\E[S_6(X,Z,Y)Y\mid X,Z=0] && \text{(by (\ref{mu_0'(X)-intermediate-result}))}
    \\
    &=\E\big\{S_6(X,Z,Y)[Y^2-2\mu_0(X)Y]\mid X,Z=0\big\}
    \\
    &=\E\big\{S_6(X,Z,Y)[Y-\mu_0(X)]^2\mid X,Z=0\big\}-\underbrace{\E[S_6(X,Z,Y)\mid X,Z=0]}_{0}\mu_0^2(X)
    \\
    &=\E\Big[S_6(X,Z,Y)\{[Y-\mu_0(X)]^2-\sigma^2_0(X)\}\mid X,Z=0\Big]+\underbrace{\E[S_6(X,Z,Y)\mid X,Z=0]}_{0}\sigma_0^2(X)
    \\
    &=\E\left[S_6(X,Z,Y)\frac{1-Z}{e(X,Z)}\{[Y-\mu_0(X)]^2-\sigma^2_0(X)\}\mid X\right].
\end{align*}
Plugging this into (\ref{if-vartheta-6-midstep}) obtains
\begin{align}
    \frac{\partial\vartheta(\theta)}{\partial\theta_6}\Big|_{\theta=\theta^0}
    &=\E\left\{\boldsymbol\dot\vartheta_{\sigma^2}(X)\E\left[S_6(X,Z,Y)\frac{1-Z}{e(X,Z)}\{[Y-\mu_0(X)]^2-\sigma^2_0(X)\}\mid X\right]\right\}\nonumber
    \\
    &=\E\Big[S_6(X,Z,Y)\underbrace{\frac{1-Z}{e(X,Z)}\boldsymbol\dot\vartheta_{\sigma^2}(X)\{[Y-\mu_0(X)]^2-\sigma^2_0(X)\}}_{\textstyle=:\varphi_{\vartheta,6}(X,Z,Y)\in\T_6}\Big].\label{if-vartheta-6}
\end{align}
Combining results from (\ref{if-vartheta-6}), (\ref{if-vartheta-3}) and (\ref{if-vartheta-1}), we have the IF of $\vartheta$
\begin{align*}
    \varphi_\vartheta(O)=\frac{1-Z}{e(X,Z)}\boldsymbol\dot\vartheta_{\sigma^2}(X)\{[Y-\mu_0(X)]^2-\sigma^2_0(X)\}+\frac{Z}{e(X,Z)}\boldsymbol\dot\vartheta_\pi(X)[C-\pi_1(X)]+\vartheta(X)-\vartheta.\tag{\ref{if:smde-vartheta}}
\end{align*}
The IF of $\xi_c$ is obtained by applying Lemma \ref{lm:if-ratioOfMeans} to $\xi_c$ as the ratio of $\vartheta$ to $\pi_c$.

\end{proof}

\bigskip

\addcontentsline{toc}{subsection}{Elaboration of Remark \ref{rm:approximate-robustness2}}
\noindent
\begin{proof}[Elaboration of Remark \ref{rm:approximate-robustness2}]
It can be shown that the probability limits of $\hat\xi_{c,\textsc{if}}$ and $\hat\xi_{c,\textsc{ifh}}$ are both
\begin{align}
    \E\{\boldsymbol\dot\vartheta_\pi[\pi_1^\dagger(X),\sigma_0(X)][\pi_1(X)-\pi_1^\dagger(X)]\}+\E\{\vartheta[\pi_1^\dagger(X),\sigma_0(X)]\}\label{vit1}
\end{align}
when all modelling components but $\hat\pi_1(X)$ are consistent, and are both
\begin{align}
    \E\{\boldsymbol\dot\vartheta_{\sigma^2}[\pi_1,\sigma_0^\dagger(X)][\sigma_0^2(X)-\sigma_0^\dagger(X)]\}+\E\{\vartheta[\pi_1^\dagger(X),\sigma_0(X)]\}\label{vit2}
\end{align}
when all modelling components but $\hat\sigma^2_0(X)$ are consistent. In both of these, the second term is the probability limit of the biased plug-in estimator. The first term in (\ref{vit1}) is the first term in the Taylor expansion of the true parameter $\vartheta$ (treated as a function of $\pi_1()$) at the point $\pi_1^\dagger()$. The first term in (\ref{vit2}) is the first term in the Taylor expansion of $\vartheta$ (treated as a function of $\sigma^2_0()$) at the point ${\sigma^2_0}^\dagger()$. Note though that for this approximate robustness property to be beneficial (reducing bias), $\pi_1^\dagger(X)$ and ${\sigma^2_0}^\dagger(X)$ need to be close to $\pi_1(X)$ and $\sigma^2_0(X)$, respectively.
\end{proof}

\section{Additional results for Section \ref{sec:calibration} -- Using data to consider the range of sens param MR/SMD}\label{appendix:calibration}

\addcontentsline{toc}{subsection}{Proposition \ref{thm:ok-mr}}
\noindent
\begin{theorem}[Okay MR interval]\label{thm:ok-mr}
    Assume A0-A2 and A4-MR with sensitivity parameter $\rho\in(0,\infty)$. Additionally assume $0\leq\mu_0(X)\leq\textup{B}$ for a positive constant \textup{B}. Let $x$ be a value in the support of $X$ such that $0<\pi_1(x)<1$ and $\mu_0(x)>0$. Let
    \begin{align*}
        \rho_1(x)&:=
        \begin{cases}
            \frac{1}{\pi_1(x)}\left[\frac{\mu_0(x)}{\textup{B}}-\pi_0(x)\right] & \text{if}~\textup{B}<\frac{\mu_0(x)}{\pi_0(x)}
            \\
            0 & \text{otherwise}
        \end{cases},
        \\
        \rho_2(x)&:=
        \begin{cases}
            \left\{\frac{1}{\pi_0(x)}\left[\frac{\mu_0(x)}{\textup{B}}-\pi_1(x)\right]\right\}^{-1} & \text{if}~\textup{B}<\frac{\mu_0(x)}{\pi_1(x)}
            \\
            \infty & \text{otherwise}
        \end{cases}.
    \end{align*}
    Then $\rho$ values in the interval
    \begin{align}
    [\rho_1(x),\rho_2(x)]\cap(0,\infty)\label{eq:ok-mr}
    \end{align}
    imply $\mu_{0c}(x)\leq\textup{B}$, for $c=0,1$.
\end{theorem}

\addcontentsline{toc}{subsection}{Proof of Proposition~\ref{thm:ok-mr}}
\noindent
\begin{proof}[Proof of Propositions \ref{thm:ok-mr}]
Under A0-A2 and A4-MR, by Proposition~\ref{thm:ratio.params-id}, 
$$\mu_{01}(x)=\frac{\rho\mu_0(x)}{\rho\pi_1(x)+\pi_0(x)},~~~\mu_{00}(x)=\frac{(1/\rho)\mu_0(x)}{(1/\rho)\pi_0(x)+\pi_1(x)}.$$

First, we derive conditions for $\mu_{01}(x)\leq\text{B}$. If $\text{B}\geq\mu_0(x)/\pi_1(x)$,
\begin{align*}
    \mu_{01}(x)
    =\frac{\rho\mu_0(x)/[\rho\pi_1(x)]}{[\rho\pi_1(x)+\pi_0(x)]/[\rho\pi_1(x)]}
    =\frac{\mu_0(x)/\pi_1(x)}{1+\frac{1}{\rho}\frac{\pi_0(x)}{\pi_1(x)}}
    \leq\frac{\text{B}}{1+\frac{1}{\rho}\frac{\pi_0(x)}{\pi_1(x)}}
    <\text{B}.
\end{align*}
If $\text{B}<\mu_0(x)/\pi_1(x)$, there is no guarantee that $\mu_{01}(x)$ does not exceed B, and we need some condition on $\rho$. Set $\mu_{01}(x)\leq\text{B}$ and solve for $\rho$:
\begin{align*}
    \mu_{01}(x)\leq\text{B}
    \Longleftrightarrow&~\frac{\mu_0(x)}{\pi_1(x)+(1/\rho)\pi_0(x)}\leq\text{B}
    \\
    \Longleftrightarrow&~\frac{\mu_0(x)}{\text{B}}\leq\pi_1(x)+(1/\rho)\pi_0(x) 
    && (\text{because}~\text{B}>0~\text{and}~\pi_1(x)+(1/\rho)\pi_0(x)>0)
    \\
    \Longleftrightarrow&~\frac{1}{\pi_0(x)}\left[\frac{\mu_0(x)}{\text{B}}-\pi_1(x)\right]\leq 1/\rho 
    && (\text{because}~\pi_0(x)>0)
    \\
    \Longleftrightarrow&~\rho\leq\left\{\frac{1}{\pi_0(x)}\left[\frac{\mu_0(x)}{\text{B}}-\pi_1(x)\right]\right\}^{-1}.
    && (\text{because}~\frac{1}{\pi_0(x)}\left[\frac{\mu_0(x)}{\text{B}}-\pi_1(x)\right]>0~\text{when}~\text{B}<\frac{\mu_0(x)}{\pi_1(x)})
\end{align*}

Next, consider $\mu_{00}(x)$. Similar reasoning obtains that (i) if $\text{B}\geq\mu_0(x)/\pi_0(x)$ then $\mu_{00}(x)<\text{B}$ regardless of $\rho$ value, but (ii) if $\text{B}<\mu_0(x)/\pi_0(x)$ then $\mu_{00}(x)\leq\text{B}$ if and only if $\displaystyle\rho\geq\frac{1}{\pi_1(x)}\left[\frac{\mu_0(x)}{\text{B}}-\pi_0(x)\right]$.

Combining these results for both $\mu_{01}(x)$ and $\mu_{00}(x)$, we get the interval in Proposition~\ref{thm:ok-mr}.
\end{proof}

\addcontentsline{toc}{subsection}{Proposition \ref{thm:ok-smd}}
\noindent
\begin{theorem}[Okay SMD interval]\label{thm:ok-smd}
    Assume A0-A2 and A4-SMDe with sensitivity parameter $\eta$. Additionally assume $\textup{B}_l\leq\mu_0(X)\leq\textup{B}_h$ for constants $\textup{B}_l<\textup{B}_h$. Let $x$ be a value in the support of $X$ such that $0<\pi_1(x)<1$ and $\sigma_0^2(x)>0$. Let
    \begin{align*}
    r_h(x)&:=[\textup{B}_h-\mu_0(x)]^2/\sigma_0^2(x),
    \\
    r_l(x)&:=[\textup{B}_l-\mu_0(x)]^2/\sigma_0^2(x),
    \\
    \eta_1(x)
    &:=\begin{cases}
        -\sqrt{\frac{r_l(x)}{\pi_0^2(x)-\pi_1(x)\pi_0(x)r_l(x)}} & \text{if}~~r_l(x)<\frac{\pi_0(x)}{\pi_1(x)}
        \\
        -\infty & \text{otherwise}
    \end{cases},
    \\
    \eta_2(x)
    &:=\begin{cases}
        -\sqrt{\frac{r_h(x)}{\pi_1^2(x)-\pi_1(x)\pi_0(x)r_h(x)}} & \text{if}~~r_h(x)<\frac{\pi_1(x)}{\pi_0(x)}
        \\
        -\infty & \text{otherwise}
    \end{cases},
    \\
    \eta_3(x)
    &:=\begin{cases}
        \sqrt{\frac{r_h(x)}{\pi_0^2(x)-\pi_1(x)\pi_0(x)r_h(x)}} & \text{if}~~r_h(x)<\frac{\pi_0(x)}{\pi_1(x)}
        \\
        \infty & \text{otherwise}
    \end{cases},
    \\
    \eta_4(x)
    &:=\begin{cases}
        \sqrt{\frac{r_l(x)}{\pi_1^2(x)-\pi_1(x)\pi_0(x)r_l(x)}} & \text{if}~~r_l(x)<\frac{\pi_1(x)}{\pi_0(x)}
        \\
        \infty & \text{otherwise}
    \end{cases}.
\end{align*}
    Then $\eta$ values in the interval
    \begin{align}
    \left[\max\{\eta_1(x),\eta_2(x)\},\min\{\eta_3(x),\eta_4(x)\}\right]\cap(-\infty,\infty)\label{eq:ok-smd}
\end{align}
imply $\textup{B}_l\leq\mu_{0c}(x)\leq\textup{B}_h$, for $c=0,1$.
\end{theorem}

\addcontentsline{toc}{subsection}{Proof of Proposition \ref{thm:ok-smd}}
\noindent
\begin{proof}[Proof of Propositions \ref{thm:ok-smd}]
Under A0-A2 and A4-SMDe, by Proposition~\ref{thm:smd-id},
$$\mu_{01}(x)=\mu_0(x)+\eta\frac{\pi_0(x)\sigma_0(x)}{\sqrt{1+\eta^2\pi_0(x)\pi_1(x)}},~~~\mu_{00}(x)=\mu_0(x)-\eta\frac{\pi_1(x)\sigma_0(x)}{\sqrt{1+\eta^2\pi_0(x)\pi_1(x)}}.$$

First, we seek conditions for $\mu_{01}(x)\geq\text{B}_l$, or equivalently,
\begin{align*}
    \eta\pi_0(x)\sigma_0(x)\geq[\text{B}_l-\mu_0(x)]\sqrt{1+\eta^2\pi_0(x)\pi_1(x)}.
\end{align*}
Because the RHS is non-positive, this inequality holds for $\eta\geq0$. For the $\eta<0$ case, we square both sides and flip the inequality, then collect terms with $\eta$ to obtain
\begin{align*}
    &\eta^2\left\{\pi_0^2(x)\sigma_0^2(x)-\pi_0(x)\pi_1(x)[\text{B}_l-\mu_0(x)]^2\right\}\leq[\text{B}_l-\mu_0(x)]^2
    \\
    \Longleftrightarrow
    &~\eta^2\Big\{\pi_0^2(x)-\pi_0(x)\pi_1(x)\underbrace{
    \frac{[\text{B}_l-\mu_0(x)]^2}{\sigma_0^2(x)}}_{\textstyle=:r_l(x)}\Big\}\leq\underbrace{
    \frac{[\text{B}_l-\mu_0(x)]^2}{\sigma_0^2(x)}}_{\textstyle=:r_l(x)}.
\end{align*}
This inequality holds if $r_l(x)\geq\frac{\pi_0(x)}{\pi_1(x)}$. Otherwise it holds if $-\sqrt{\frac{r_l(x)}{\pi_0^2(x)-\pi_0(x)\pi_1(x)r_l(x)}}\leq\eta<0$. Combining results, we have:
\begin{itemize}
    \item[1)] $\mu_{01}(x)\geq\text{B}_l$ in two cases: (i) $r_l(x)\geq\frac{\pi_0(x)}{\pi_1(x)}$; or (ii) $r_h(x)<\frac{\pi_0(x)}{\pi_1(x)}$ and $\eta\geq-\sqrt{\frac{r_l(x)}{\pi_0^2(x)-\pi_0(x)\pi_1(x)r_l(x)}}$. (This is the basis of the definition of $\eta_1(x)$.)
\end{itemize}

Next, we use similar reasoning to obtain the following:
\begin{itemize}
    \item[2)] $\mu_{00}(x)\leq\text{B}_h$ in two cases: (i) $r_h(x)\geq\frac{\pi_1(x)}{\pi_0(x)}$; or (ii) $r_h(x)<\frac{\pi_1(x)}{\pi_0(x)}$ and $\eta\geq-\sqrt{\frac{r_h(x)}{\pi_1^2(x)-\pi_0(x)\pi_1(x)r_h(x)}}$. (This is the basis of the definition of $\eta_2(x)$.)
    \item[3)] $\mu_{01}\leq\text{B}_h$ in two cases: (i) $r_h(x)\geq\frac{\pi_0(x)}{\pi_1(x)}$; or (ii) $r_h(x)<\frac{\pi_0(x)}{\pi_1(x)}$ and $\eta\leq\sqrt{\frac{r_h(x)}{\pi_0^2(x)-\pi_0(x)\pi_1(x)r_h(x)}}$. (This is the basis of the definition of $\eta_3(x)$.)
    \item[4)] $\mu_{00}(x)\geq\text{B}_l$ in two cases: (i) $r_l(x)\geq\frac{\pi_1(x)}{\pi_0(x)}$; or (ii) $r_l(x)<\frac{\pi_1(x)}{\pi_0(x)}$ and $\eta\leq\sqrt{\frac{r_l(x)}{\pi_1^2(x)-\pi_0(x)\pi_1(x)r_l(x)}}$. (This is the basis of the definition of $\eta_4(x)$.)
\end{itemize}
Combining the above four results, we obtain that the $\eta$ interval such that $\mu_{01}(x),\mu_{00}(x)\in[\text{B}_l,\text{B}_h]$ is the interval defined in Proposition~\ref{thm:ok-smd}.
\end{proof}

\section{Additional results for Section \ref{sec:rate-conditions} -- Nonparametric rate conditions}\label{appendix:rate-conditions}

\addcontentsline{toc}{subsection}{Proposition \ref{thm:rates-pi}}
\noindent
\begin{theorem}[PI-based rate conditions]\label{thm:rates-pi}
    Assume
    \begin{itemize}
        \item positivity: for some $\epsilon>0$ and all $x$ values in the support of $X$, $\P(\epsilon\leq e_1(x)<1-\epsilon)=1$,
        \item consistency: all the nuisance functions are mean squared error convergent, i.e., $||\tilde e_1(X)-e_1(X)||_2=o_p(1)$, $\tilde\pi_c(X)-\pi_c(X)||_2=o_p(1)$, $||\tilde\mu_{1c}(X)-\mu_{1c}(X)||_2=o_p(1)$, $||\tilde\mu_0(X)-\mu_0(X)||_2=o_p(1)$,
        \item bounded propensity score estimation: $||1/\tilde e_1(X)||_2=O_p(1)$, $||1/\tilde e_0(X)||_2=O_p(1)$.
    \end{itemize}
    The PI-based IF-based estimator $\hat\Delta_{c,\textup{IF}}^\textup{PI}$ is root-n consistent and asymptotically normal if the nuisance estimators satisfy the error rate conditions:
    \begin{itemize}
        \item $e\pi$-rate: $||\tilde e_1(X)-e_1(X)||_2||\tilde\pi_c(X)-\pi_c(X)||_2=o_p(n^{-1/2})$,
        \item $e\mu$-rates: $\begin{cases}
            ||\tilde e_1(X)-e_1(X)||_2||\tilde\mu_{1c}(X)-\mu_{1c}(X)||_2=o_p(n^{-1/2})
            \\
            ||\tilde e_1(X)-e_1(X)||_2||\tilde\mu_0(X)-\mu_0(X)||_2=o_p(n^{-1/2})
        \end{cases},$
        \item $\pi\mu_0$-rate: $||\tilde\pi_c(X)-\pi_c(X)||_2||\tilde\mu_0(X)-\mu_0(X)||_2=o_p(n^{-1/2})$.
    \end{itemize}
    Under these conditions, the asymptotic variance of $\hat\Delta_{c,\textup{IF}}^\textup{PI}$ is the variance of the IF of the principal causal effect $\Delta_c^\text{PI}$.
\end{theorem}

\addcontentsline{toc}{subsection}{Proposition \ref{thm:rates-sens}}
\noindent
\begin{theorem}[Rate conditions for sensitivity analyses]\label{thm:rates-sens}
    Assume
    \begin{itemize}
        \item positivity: for some $\epsilon>0$ and all $x$ values in the support of $X$, $\P(\epsilon\leq e_1(x)<1-\epsilon)=1$,
        \item consistency: all the nuisance functions are mean squared error convergent,
        \item bounded propensity score estimation: $||1/\tilde e_1(X)||_2=O_p(1)$, $||1/\tilde e_0(X)||_2=O_p(1)$.
    \end{itemize}
    The OR-based and GOR-based estimators $\hat\Delta_{c,\textup{IF}}^\textup{OR}$ and  $\hat\Delta_{c,\textup{IF}}^\textup{GOR}$ are root-n consistent and asymptotically normal if the nuisance estimators satisfy the error rate conditions:
    \begin{itemize}
        \item $e\pi$-rate: $||\tilde e_1(X)-e_1(X)||_2||\tilde\pi_c(X)-\pi_c(X)||_2=o_p(n^{-1/2})$,
        \item $e\mu$-rates: $\begin{cases}
            ||\tilde e_1(X)-e_1(X)||_2||\tilde\mu_{1c}(X)-\mu_{1c}(X)||_2=o_p(n^{-1/2})
            \\
            ||\tilde e_1(X)-e_1(X)||_2||\tilde\mu_0(X)-\mu_0(X)||_2=o_p(n^{-1/2})
        \end{cases},$
        \item $\pi$-rate: $||\tilde\pi_c(X)-\pi_c(X)||_2=o_p(n^{-1/4})$,
        \item $\mu_0$ rate: $||\tilde\mu_0(X)-\mu_0(X)||_2=o_p(n^{-1/4})$.
    \end{itemize}
    The MR-based IF-based estimator $\hat\Delta_{c,\textup{IF}}^\textup{MR}$ is root-n consistent and asymptotically normal if the nuisance estimators satisfy the error rate conditions:
    \begin{itemize}
        \item $e\pi$-rate: $||\tilde e_1(X)-e_1(X)||_2||\tilde\pi_c(X)-\pi_c(X)||_2=o_p(n^{-1/2})$,
        \item $e\mu$-rates: $\begin{cases}
            ||\tilde e_1(X)-e_1(X)||_2||\tilde\mu_{1c}(X)-\mu_{1c}(X)||_2=o_p(n^{-1/2})
            \\
            ||\tilde e_1(X)-e_1(X)||_2||\tilde\mu_0(X)-\mu_0(X)||_2=o_p(n^{-1/2})
        \end{cases},$
        \item $\pi\mu_0$-rate: $||\tilde\pi_c(X)-\pi_c(X)||_2||\tilde\mu_0(X)-\mu_0(X)||_2=o_p(n^{-1/2})$,
        \item $\pi$-rate: $||\tilde\pi_c(X)-\pi_c(X)||_2=o_p(n^{-1/4})$.
    \end{itemize}
    The SMDe-based IF-based estimator $\hat\Delta_{c,\textup{IF}}^\textup{SMDe}$ is root-n consistent and asymptotically normal if the nuisance estimators satisfy the error rate conditions:
    \begin{itemize}
        \item $e\pi$-rate: $||\tilde e_1(X)-e_1(X)||_2||\tilde\pi_c(X)-\pi_c(X)||_2=o_p(n^{-1/2})$,
        \item $e\mu$-rates: $\begin{cases}
            ||\tilde e_1(X)-e_1(X)||_2||\tilde\mu_{1c}(X)-\mu_{1c}(X)||_2=o_p(n^{-1/2})
            \\
            ||\tilde e_1(X)-e_1(X)||_2||\tilde\mu_0(X)-\mu_0(X)||_2=o_p(n^{-1/2})
        \end{cases}$,
        \item $e\sigma^2$-rate: $||\tilde e_1(X)-e_1(X)||_2||\tilde\sigma_0^2(X)-\sigma_0^2(X)||_2=o_p(n^{-1/2})$,
        \item $\pi$-rate: $||\tilde\pi_c(X)-\pi_c(X)||_2=o_p(n^{-1/4})$,
        \item $\mu_0$-rate: $||\tilde\mu_0(X)-\mu_0(X)||_2=o_p(n^{-1/4})$,
        \item $\sigma^2$-rate: $||\tilde\sigma_0^2(X)-\sigma_0^2(X)||_2=o_p(n^{-1/4})$.
    \end{itemize}
    Under these (respective) conditions, the asymptotic variance of the IF-based estimator is the variance of the IF of the principal causal effect.
\end{theorem}


\addcontentsline{toc}{subsection}{Proof of Propositions \ref{thm:rates-pi}-\ref{thm:rates-sens}: the shared parts 1-3 and 5}
\noindent
\begin{proof}[Proof of Propositions \ref{thm:rates-pi}-\ref{thm:rates-sens}: the shared parts 1-3 and 5]
\hfill

\smallskip
\noindent We want to establish conditions for
\begin{align*}
    \hat\Delta_c=\frac{\hat\nu_{1c}-\hat\nu_{0c}}{\hat\pi_c}=\frac{\P_n[\tilde\phi_{1c}(O)]-\P_n[\tilde\phi_{0c}(O)]}{\P_n[\tilde\phi_{\pi_c}(O)]}
\end{align*}
to be $\sqrt{n}$-consistent, where $\phi_{\nu_{1c}}(O)=\varphi_{\nu_{1c}}(O)+\nu_{1c}$, $\phi_{\nu_{0c}}(O)=\varphi_{\nu_{0c}}(O)+\nu_{0c}$ and $\phi_{\pi_c}(O)=\varphi_{\pi_c}(O)+\pi_c$ are the uncentered IFs of $\nu_{1c}$, $\nu_{0c}$ and $\pi_c$, respectively, and the $\tilde{}$ notation indicates that the functions are evaluated at estimated nuisance parameters. Throughout
\begin{align*}
    \phi_{\pi_c}(O)
    &:=\frac{Z}{e_1(X)}[\I(C=c)-\pi_c(X)]+\pi_c(X),
    \\
    \phi_{\nu_{1c}}(O)
    &:=\frac{Z}{e_1(X)}\I(C=c)[Y-\mu_{1c}(X)]+\frac{Z}{e_1(X)}\mu_{1c}(X)[\I(C=c)-\pi_c(X)]+\pi_c(X)\mu_{1c}(X),
\end{align*}
while
$\phi_{\nu_{0c}}(O)$ takes on different forms specific to the PI-based main analysis and the sensitivity analyses.

\medskip

\noindent
\underline{Part 1 (setting the stage)}: To start, we assume positivity, consistency and bounded propensity score estimation as stated in the two propositions.

Write
\begin{align*}
    \hat\Delta_c-\Delta_c
    &=\frac{\hat\nu_{1c}-\hat\nu_{0c}}{\hat\pi_c}-\Delta_c
    \\
    &=\frac{1}{\hat\pi_c}
    [(\hat\nu_{1c}-\hat\nu_{0c})-(\nu_{1c}-\nu_{0c})]+\Delta_c\frac{\pi_c}{\hat\pi_c}-\Delta_c
    \\
    &=\frac{1}{\hat\pi_c}[(\hat\nu_{1c}-\nu_{1c})-(\hat\nu_{0c}-\nu_{0c})-\Delta_c(\hat\pi_c-\pi_c)]
    \\
    &=\left[\frac{1}{\hat\pi_c}-\frac{1}{\pi_c}+\frac{1}{\pi_c}\right][(\hat\nu_{1c}-\nu_{1c})-(\hat\nu_{0c}-\nu_{0c})-\Delta_c(\hat\pi_c-\pi_c)]
    \\
    &=\frac{1}{\pi_c}[(\hat\nu_{1c}-\nu_{1c})-(\hat\nu_{0c}-\nu_{0c})-\Delta_c(\hat\pi_c-\pi_c)]
    -\frac{(\hat\pi_c-\pi_c)}{\hat\pi_c\pi_c}[(\hat\nu_{1c}-\nu_{1c})-(\hat\nu_{0c}-\nu_{0c})-\Delta_c(\hat\pi_c-\pi_c)]
    \\
    &=\frac{1}{\pi_c}[(\hat\nu_{1c}-\nu_{1c})-(\hat\nu_{0c}-\nu_{0c})-\Delta_c(\hat\pi_c-\pi_c)]+
    \\
    &~~~~
    -\frac{1}{\hat\pi_c\pi_c}(\hat\pi_c-\pi_c)(\hat\nu_{1c}-\nu_{1c})+\frac{1}{\hat\pi_c\pi_c}(\hat\pi_c-\pi_c)(\hat\nu_{0c}-\nu_{0c})
    +\frac{\Delta_c}{\hat\pi_c\pi_c}(\hat\pi_c-\pi_c)^2.
\end{align*}
This means that if $\hat\nu_{1c}$, $\hat\nu_{0c}$ and $\hat\pi_c$ are $\sqrt{n}$-consistent, then $\hat\Delta_c$ is $\sqrt{n}$-consistent, because then the first term above is $O_p(n^{-1/2})$ while the other terms are $O_p(n)$ therefore $o_p(n^{-1/2})$.

\medskip

\noindent
\underline{Part 2 ($\sqrt{n}$-consistency of $\hat\pi_c$)}:
We now derive conditions for $\hat\pi_c$ to be $\sqrt{n}$-consistent. As the results are well-known, this is more for completeness and as a review of the relevant theory.

We apply the theory in \cite{kennedy2023SemiparametricDoublyRobust}, decomposing the error in estimating $\pi_c$ into the sum of three terms:
\begin{align*}
    \hat\pi_c-\pi_c
    &=\P_n[\tilde\phi_{\pi_c}(O)]-\P[\phi_{\pi_c}(O)]
    \\
    &=
    \underbrace{(\P_n-\P)[\phi_{\pi_c}(O)]}_{\textstyle T_0}+
    \underbrace{(\P_n-\P)[\tilde\phi_{\pi_c}(O)-\phi_{\pi_c}(O)]}_{\textstyle T_1}+
    \underbrace{\P[\tilde\phi_{\pi_c}(O)-\phi_{\pi_c}(O)]}_{\textstyle T_2}.
\end{align*}

$T_0$ is the difference between the sample average and the population mean of a fixed function of data (where this function has finite variance under the assumption of bounded propensity score), so we can invoke the Central Limit Theorem and found this term to be $O_p(n^{-1/2})$.


Writing
\begin{align*}
    \tilde\phi_{\pi_c}(O)
    &=\overbrace{\frac{Z}{e_1(X)}\left\{1-\frac{1}{\tilde e_1(X)}[\tilde e_1(X)-e_1(X)]\right\}}^{Z/e_1(X)}\overbrace{\{[\I(C=c)-\pi_c(X)]-[\tilde\pi_c(X)-\pi_c(X)]\}}^{\I(C=c)-\tilde\pi_c(X)}+\overbrace{\{\pi_c(X)+[\tilde\pi_c(X)-\pi_c(X)]\}}^{\tilde\pi_c(X)},
\end{align*}
we obtain
\begin{align}
    \tilde\phi_{\pi_c}(O)-\phi_{\pi_c}(O)
    &=[\tilde e_1(X)-e_1(X)]\times\left\{-\frac{Z}{e_1(X)}\frac{1}{\tilde e_1(X)}[\I(C=c)-\pi_c(X)]\right\}+\nonumber
    \\
    &~~~~[\tilde\pi_c(X)-\pi_c(X)]\times\left\{1-\frac{Z}{e_1(X)}\right\}+\nonumber
    \\
    &~~~~[\tilde e_1(X)-e_1(X)][\tilde\pi_c(X)-\pi_c(X)]\times\left\{\frac{Z}{e_1(X)}\frac{1}{\tilde e_1(X)}\right\}.\label{eq:error-phi_pi.c}
\end{align}

By triangle inequality and the fact that $\frac{1}{\tilde e_1(X)}=O_p(1)$, (\ref{eq:error-phi_pi.c}) implies
\begin{align*}
    ||\tilde\phi_{\pi_c}(O)-\phi_{\pi_c}(O)||_2\leq \tilde D||\tilde e_1(X)-\hat e(X)||_2+\tilde D||\tilde\pi_c(X)-\pi_c(X)||_2,
\end{align*}
where $\tilde D=O_p(1)$. We will $\tilde D$ as generic notation for a quantity that is $O_p(1)$ that may take different values in different places.
The sample splitting lemma in \cite{kennedy2023SemiparametricDoublyRobust} provides that: If $\P_n$ is the empirical distribution from $(O_1,\dots,O_n)$ and $\hat f$ is estimated from $(O_{n+1},\dots,O_N)$ where these two samples are independent, then $(\P_n-\P)(\hat f-f)=O_p(||\hat f-f||_2/\sqrt{n})$.
This means if we use sample splitting (or cross fitting), then the empirical process term $T_1$ vanishes fast enough,
\begin{align*}
    T_1=O_p\left(||\tilde\phi_{\pi_c}(O)-\phi_{\pi_c}(O)||_2/\sqrt{n}\right)=O_p(o_p(1)/\sqrt{n})=o_p(n^{-1/2}).
\end{align*}

In addition, (\ref{eq:error-phi_pi.c}) implies
\begin{align*}
    T_2=\P\left\{\frac{1}{\tilde e_1(X)}[\tilde e_1(X)-e_1(X)][\tilde\pi_c(X)-\pi_c(X)]\right\},
\end{align*}
so
\begin{align*}
    |T_2|
    &\leq \tilde D\big|\P\big([\tilde e_1(X)-e_1(X)][\tilde\pi_c(X)-\pi_c(X)]\big)\big| 
    \\
    &\leq \tilde D||[\tilde e_1(X)-e_1(X)][\tilde\pi_c(X)-\pi_c(X)]||_2 & (\text{Jensen's inequality})
    \\
    &\leq \tilde D||\tilde e_1(X)-e_1(X)||_2||\tilde\pi_c(X)-\pi_c(X)||_2. & (\text{Cauchy-Schwarz inequality})
\end{align*}
Then $T_2=o_p(n^{-1/2})$ if the product of the two estimation errors on the RHS is $o_p(n^{-1/2})$. This is satisfied, for example, if $||\tilde e_1(X)-e_1(X)||_2=O_p(n^{-1/(2p)})$, $||\tilde\pi_c(X)-\pi_c(X)||_2=O_p(n^{1/(2q)})$ and $\frac{1}{p}+\frac{1}{q}>1$.

To sum up, we have added to the initial conditions the following
\begin{itemize}
    \item sample splitting (or cross fitting)
    \item $e\pi$-rate: $||\tilde e_1(X)-e_1(X)||_2||\tilde\pi_c(X)-\pi_c(X)||_2=o_p(n^{-1/2})$.
\end{itemize}

\medskip

As sampling splitting is the general technique to take care of the empirical process term, we will use sampling splitting across the board, and for the rest of the proof (dealing with $\hat\nu_{1c}$ and $\hat\nu_{0c}$) will only focus on the remainder term.

\medskip

\noindent
\underline{Part 3 ($\sqrt{n}$-consistency of $\hat\nu_{1c}$)}: We now turn to $\hat\nu_{1c}$.
\begin{align*}
    \hat\nu_{1c}-\nu_{1c}
    &=\P_n[\tilde\phi_{\nu_{1c}}(O)]-\P[\phi_{\nu_{1c}}(O)]
    \\
    &=\underbrace{(\P_n-P)[\phi_{\nu_{1c}}(O)]}_{T_0}+\underbrace{(\P_n-\P)[\tilde\phi_{\nu_{1c}}(O)-\phi_{\nu_{1c}}(O)]}_{T_1}+\underbrace{\P[\tilde\phi_{\nu_{1c}}(O)-\phi_{\nu_{1c}}(O)]}_{T_2}.
\end{align*}
From this point, to ease notation, we mostly drop the $(X)$ notation from $e_z(X)$, $\pi_c(X)$, $\mu_{1c}(X)$, $\mu_0(X)$ and $\tilde e_z(X)$, $\tilde\pi_c(X)$, $\tilde\mu_{1c}(X)$, $\tilde\mu_0(X)$. It's important to note that in the abbreviated notation here, $\pi_c$ stands for the conditional probability $\pi_c(X):=\P(C=c\mid X)$, not the marginal $\P(C=c)$.
\begin{align*}
    \tilde\phi_{\nu_{1c}}(O)-\phi_{\nu_{1c}}(O)
    &=\underbrace{\frac{Z}{\tilde e_1}\I(C=c)(Y-\tilde\mu_{1c})}_{(*)}+\underbrace{\frac{Z}{\tilde e_1}\tilde\mu_{1c}[\I(C=c)-\tilde\pi_c]}_{(**)}+\underbrace{\tilde\pi_c\tilde\mu_{1c}}_{(***)}+
    \\
    &~~~~~~-\left[\frac{Z}{e_1}\I(C=c)(Y-\mu_{1c})+\frac{Z}{e_1}\mu_{1c}[\I(C=c)-\pi_c]+\pi_c\mu_{1c}\right].
\end{align*}
With
\begin{align*}
    (*)
    &=\frac{Z}{e_1}\left[1-\frac{1}{\tilde e_1}(\tilde e_1-e_1)\right]\I(C=c)[(Y-\mu_{1c})-(\tilde\mu_{1c}-\mu_{1c})],
    \\
    (**)
    &=\frac{Z}{e_1}\left[1-\frac{1}{\tilde e_1}(\tilde e_1-e_1)\right][\mu_{1c}+(\tilde\mu_{1c}-\mu_{1c})]\{[\I(C=c)-\pi_c]-(\tilde\pi_c-\pi_c)\},
    \\
    (***)
    &=[\pi_c+(\tilde\pi_c-\pi_c)][\mu_{1c}+(\tilde\mu_{1c}-\mu_{1c})],
\end{align*}
we obtain
\begin{align}
    \tilde\phi_{\nu_{1c}}(O)-\phi_{\nu_{1c}}(O)
    &=(\tilde e_1-e_1)\times\left\{-\frac{Z}{e_1}\frac{1}{\tilde e_1}\I(C=c)(Y-\mu_{1c})-\frac{Z}{e_1}\frac{1}{\tilde e_1}\mu_{1c}[\I(C=c)-\pi_c]\right\}+\nonumber
    \\
    &~~~~(\tilde\pi_c-\pi_c)\times\left\{\left(1-\frac{Z}{e_1}\right)\mu_{1c}\right\}+\nonumber
    \\
    &~~~~(\tilde\mu_{1c}-\mu_{1c})\times\left\{\left(1-\frac{Z}{e_1}\right)\pi_c\right\}+\nonumber
    \\
    &~~~~(\tilde e_1-e_1)(\tilde\pi_c-\pi_c)\times
    \left\{\frac{Z}{e_1}\frac{1}{\tilde e_1}\mu_{1c}\right\}+\nonumber
    \\
    &~~~~(\tilde e_1-e_1)(\tilde\mu_{1c}-\mu_{1c})\times
    \left\{\frac{Z}{e_1}\frac{1}{\tilde e_1}\pi_c\right\}+\nonumber
    \\
    &~~~~(\tilde\pi_c-\pi_c)(\tilde\mu_{1c}-\mu_{1c})\times
    \left\{1-\frac{Z}{e_1}\right\}+\nonumber
    \\
    &~~~~(\tilde e_1-e_1)(\tilde\pi_c-\pi_c)(\tilde\mu_{1c}-\mu_{1c})\times
    \left\{\frac{Z}{e_1}\frac{1}{\tilde e_1}\right\},\label{eq:error-phi_1c}
\end{align}
and
\begin{align*}
    T_2=\P\left\{(\tilde e_1-e_1)(\tilde\pi_c-\pi_c)\frac{\mu_{1c}}{\tilde e_1}+(\tilde e_1-e_1)(\tilde\mu_{1c}-\mu_{1c})\frac{\pi_c}{\tilde e_1}+(\tilde e_1-e_1)(\tilde\pi_c-\pi_c)(\tilde\mu_{1c}-\mu_{1c})\frac{1}{\tilde e_1}\right\}.
\end{align*}
Using similar reasoning as above, we have
\begin{align*}
    |T_2|\leq\tilde D||\tilde e_1-e_1||_2||\tilde\pi_c-\pi_c||_2+\tilde D||\tilde e_1-e_1||_2||\tilde\mu_{1c}-\mu_{1c}||_2,
\end{align*}
so we need both of the terms on the RHS to be $o_p(n^{-1/2})$. 
Thus we have added a rate condition:
\begin{itemize}
    \item $e\mu_{1c}$-rate: $||\tilde e_1-e_1||_2||\tilde\mu_{1c}-\mu_{1c}||_2=o_p(n^{-1/2})$
\end{itemize}

\medskip

\noindent
\underline{Part 4 ($\sqrt{n}$-consistency of $\hat\nu_{0c}$)}: This part of the proof is unique to the form of $\nu_{0c}$, and thus will be presented separately for each of the four propositions shortly.

\medskip

\noindent
\underline{Part 5 (asymptotic distribution of $\hat\Delta_c$)}: Under the combination of all the conditions,
\begin{align*}
    \hat\Delta_c-\Delta_c
    &=(\P_n-\P)\left\{\frac{1}{\pi_c}\left[\phi_{\nu_{1c}}(O)-\phi_{\nu_{0c}}(O)-\Delta_c\phi_{\pi_c}(O)\right]\right\}+o_p(n^{-1/2})
    \\
    &=(\P_n-\P)[\varphi_{\Delta_c}(O)]+o_p(n^{-1/2}),
\end{align*}
and by the Central Limit Theorem and Slusky's lemma,
\begin{align*}
    \sqrt{n}(\hat\Delta_c-\Delta_c)\dto\mathrm{N}(0,\var(\varphi_{\Delta_c}(O)).
\end{align*}
The application of the CLT is based on the condition that $\var(\varphi_{\Delta_c}(O)<\infty$, which is satisfied under bounded propensity score estimation.

\end{proof}

\bigskip

\addcontentsline{toc}{subsection}{Proof of Proposition \ref{thm:rates-pi}: the unique part 4}
\noindent
\begin{proof}[Proof of Proposition \ref{thm:rates-pi}: the unique part 4]

\begin{align*}
    \hat\nu_{0c}^\text{PI}-\nu_{0c}^\text{PI}
    &=\P_n[\tilde\phi_{\nu_{0c}}^\text{PI}(O)]-\P[\phi_{\nu_{0c}}^\text{PI}(O)]
    \\
    &=
    \underbrace{(\P_n-\P)[\phi_{\nu_{0c}}^\text{PI}(O)]}_{\textstyle T_0}+
    \underbrace{(\P_n-\P)[\tilde\phi_{\nu_{0c}}^\text{PI}(O)-\phi_{\nu_{0c}}^\text{PI}(O)]}_{\textstyle T_1}+
    \underbrace{\P[\tilde\phi_{\nu_{0c}}^\text{PI}(O)-\phi_{\nu_{0c}}^\text{PI}(O)]}_{\textstyle T_2}
\end{align*}
\begin{align*}
    \tilde \phi_{\nu_{0c}^\text{PI}}(O)-\phi_{\nu_{0c}^\text{PI}}(O)
    &=
    \underbrace{\frac{1-Z}{\tilde e_0}\tilde\pi_c[Y-\tilde\mu_0]}
    _{(*)}+
    \underbrace{\frac{Z}{\tilde e_1}\tilde\mu_0[\I(C=c)-\tilde\pi_c]}
    _{(**)}+
    \underbrace{\tilde\pi_c\tilde\mu_0}
    _{(***)}
    +
    \\
    &~~~~~~-\left[\frac{1-Z}{e_0}\pi_c[Y-\mu_0]+\frac{Z}{e_1}\mu_0[\I(C=c)-\pi_c]+\pi_c\mu_0\right].
\end{align*}
With
\begin{align*}
    (*)
    &=\frac{1-Z}{e_0}\left[1+\frac{1}{\tilde e_0}(\tilde e_1-e_1)\right][\pi_c+(\tilde\pi_c-\pi_c)][(Y-\mu_0)-(\tilde\mu_0-\mu_0)],
    \\
    %
    (**)
    &=\frac{Z}{e_1}\left[1-\frac{1}{\tilde e_1}(\tilde e_1-e_1)\right][\mu_0+(\tilde\mu_0-\mu_0)]\{[\I(C=c)-\pi_c]-(\tilde\pi_c-\pi_c)\},
    \\
    %
    (***)
    &=[\pi_c+(\tilde\pi_c-\pi_c)][\mu_0+(\tilde\mu_0-\mu_0)]
    ,
\end{align*}
we obtain
\begin{align}
    \tilde \phi_{\nu_{0c}^\text{PI}}(O)-\phi_{\nu_{0c}^\text{PI}}(O)
    &=
    (\tilde e_1-e_1)\times
    \left\{\frac{1-Z}{e_0}\frac{1}{\tilde e_0}\pi_c(Y-\mu_0)-\frac{Z}{e_1}\frac{1}{\tilde e_1}\mu_0[\I(C=c)-\pi_c]\right\}
    +\nonumber
    \\
    &~~~~
    (\tilde\pi_c-\pi_c)\times
    \left\{\frac{1-Z}{e_0}(Y-\mu_0)+\left(1-\frac{Z}{e_1}\right)\mu_0\right\}
    +\nonumber
    \\
    &~~~~
    (\tilde\mu_0-\mu_0)\times
    \left\{\frac{Z}{e_1}[\I(C=c)-\pi_c]+\left(1-\frac{1-Z}{e_0}\right)\pi_c\right\}
    +\nonumber
    \\
    &~~~~
    (\tilde e_1-e_1)(\tilde\pi_c-\pi_c)\times
    \left\{\frac{1-Z}{e_0}\frac{1}{\tilde e_0}(Y-\mu_0)+\frac{Z}{e_1}\frac{1}{\tilde e_1}\mu_0\right\}
    +\nonumber
    \\
    &~~~~
    (\tilde e_1-e_1)(\tilde\mu_0-\mu_0)\times
    \left\{-\frac{1-Z}{e_0}\frac{1}{\tilde e_0}\pi_c-\frac{Z}{e_1}\frac{1}{\tilde e_1}[\I(C=c)-\pi_c]\right\}
    +\nonumber
    \\
    &~~~~
    (\tilde\pi_c-\pi_c)(\tilde\mu_0-\mu_0)\times
    \left\{1-\frac{1-Z}{e_0}-\frac{Z}{e_1}\right\}+\nonumber
    \\
    &~~~~
    (\tilde e_1-e_1))(\tilde\pi_c-\pi_c)(\tilde\mu_0-\mu_0)\times
    \left\{\frac{Z}{e_1}\frac{1}{\tilde e_1}-\frac{1-Z}{e_0}\frac{1}{\tilde e_0}\right\}
    .\label{eq:error-phi_PI}
\end{align}
%
\begin{align*}
    T_2
    &=\P\left[(\tilde e_1-e_1)(\tilde\pi_c-\pi_c)\frac{\mu_0}{\tilde e_1}-(\tilde e_1-e_1)(\tilde\mu_0-\mu_0)\frac{\pi_c}{\tilde e_0}-(\tilde\pi_c-\pi_c)(\tilde\mu_0-\mu_0)+\right.
    \\
    &~~~~~~~~~
    \left.+(\tilde e_1-e_1)(\tilde\pi_c-\pi_c)(\tilde\mu_0-\mu_0)\left(\frac{1}{\tilde e_1}-\frac{1}{\tilde e_0}\right)\right],
\end{align*}
so
\begin{align*}
    |T_2|
    &\leq \tilde D||\tilde e_1(X)-e_1(X)||_2||\tilde\pi_c(X)-\pi_c(X)||_2+
    \\
    &~~~~\tilde D||\tilde e_1(X)-e_1(X)||_2||\tilde\mu_0(X)-\mu_0(X)||_2+
    \\
    &~~~~\tilde D||\tilde\pi_c(X)-\pi_c(X)||_2||\tilde\mu_0(X)-\mu_0(X)||_2
\end{align*}
where $\tilde D=O_p(1)$. We need the terms on the RHS to be $o_p(n^{-1/2})$.



\medskip

In conclusion, we have established the following conditions for $\hat\Delta_c^\text{PI}$ to be $\sqrt{n}$-consistent for $\Delta_c^\text{PI}$: positivity, sample splitting, consistency, bounded propensity score estimation, plus the rate conditions
\begin{itemize}
    \item $e\pi$-rate: $||\tilde e_1(X)-e_1(X)||_2||\tilde\pi_c(X)-\pi_c(X)||_2=o_1(n^{-1/2})$
    \item $e\mu$-rates: $\begin{cases}
        ||\tilde e_1(X)-e_1(X)||_2||\tilde\mu_{1c}(X)-\mu_{1c}(X)||_2=o_1(n^{-1/2})
        \\
        ||\tilde e_1(X)-e_1(X)||_2||\tilde\mu_0(X)-\mu_0(X)||_2=o_1(n^{-1/2})
    \end{cases}$
    \item $\pi\mu_0$-rate: $||\tilde\pi_c(X)-\pi_c(X)||_2||\tilde\mu_0(X)-\mu_0(X)||_2=o_1(n^{-1/2})$
\end{itemize}
\end{proof}

\addcontentsline{toc}{subsection}{Proof of Proposition \ref{thm:rates-sens}: the unique parts 4}
\noindent
\begin{proof}[Proof of Proposition \ref{thm:rates-sens}: the unique parts 4]

\medskip
\noindent
\underline{Part 4 for $\hat\Delta_{c,\text{IF}}^\text{OR}$ and $\hat\Delta_{c,\text{IF}}^\text{GOR}$}

\smallskip
\noindent
The proof is basically the same for $\hat\Delta_{c,\text{IF}}^\text{OR}$ and $\hat\Delta_{c,\text{IF}}^\text{GOR}$. We present it for $\hat\Delta_{c,\text{IF}}^\text{OR}$, which is simpler.

\begin{align*}
    \tilde\phi_{\nu_{0c}}^\text{OR}-\phi_{\nu_{0c}}^\text{OR}
    &=
    \underbrace{\frac{1-Z}{\tilde e_0}\overbrace{\left(\frac{1}{2}-\frac{\tilde\alpha_c}{2\tilde\beta_c}+\frac{\rho_c\tilde\pi_c}{\tilde\beta_c}\right)}^{\tilde A}(Y-\tilde \mu_0)}_{(*)}
    +\underbrace{\frac{Z}{\tilde e_1}\overbrace{\left(\frac{1}{2}-\frac{\tilde\alpha_c}{2\tilde\beta_c}+\frac{\rho_c\tilde\mu_0}{\tilde\beta_c}\right)}^{\tilde B}[\I(C=c)-\tilde\pi_c]}_{(**)}
    +\underbrace{\overbrace{\frac{\tilde\alpha_c-\tilde\beta_c}{2(\rho_c-1)}}^{\tilde C}}_{(***)}+
    \\
    &~~~~\left[\frac{1-Z}{ e_0}\left(\frac{1}{2}-\frac{\alpha_c}{2\beta_c}+\frac{\rho_c\pi_c}{\beta_c}\right)(Y-\mu_0)
    +\frac{Z}{ e_1}\left(\frac{1}{2}-\frac{\alpha_c}{2\beta_c}+\frac{\rho_c\mu_0}{\beta_c}\right)[\I(C=c)-\pi_c]+\frac{\alpha_c-\beta_c}{2(\rho_c-1)}\right],
\end{align*}
where
\begin{align*}
    \tilde\phi_{\nu_{0c}}^\text{OR}-\phi_{\nu_{0c}}^\text{OR}
    &=
    \underbrace{\frac{1-Z}{\tilde e_0}\overbrace{\left(\frac{1}{2}-\frac{\tilde\alpha_c}{2\tilde\beta_c}+\frac{\rho_c\tilde\pi_c}{\tilde\beta_c}\right)}^{\tilde A}(Y-\tilde \mu_0)}_{(*)}
    +\underbrace{\frac{Z}{\tilde e_1}\overbrace{\left(\frac{1}{2}-\frac{\tilde\alpha_c}{2\tilde\beta_c}+\frac{\rho_c\tilde\mu_0}{\tilde\beta_c}\right)}^{\tilde B}[\I(C=c)-\tilde\pi_c]}_{(**)}
    +\underbrace{\overbrace{\frac{\tilde\alpha_c-\tilde\beta_c}{2(\rho_c-1)}}^{\tilde C}}_{(***)}+
    \\
    &~~~~\left[\frac{1-Z}{ e_0}\left(\frac{1}{2}-\frac{\alpha_c}{2\beta_c}+\frac{\rho_c\pi_c}{\beta_c}\right)(Y-\mu_0)
    +\frac{Z}{ e_1}\left(\frac{1}{2}-\frac{\alpha_c}{2\beta_c}+\frac{\rho_c\mu_0}{\beta_c}\right)[\I(C=c)-\pi_c]+\frac{\alpha_c-\beta_c}{2(\rho_c-1)}\right],
\end{align*}
and
\begin{align*}
    \alpha_c=(\rho_c-1)(\pi_c+\mu_0)+1,~~~\beta_c=\sqrt{\alpha_c^2-4\pi_c\mu_0\rho_c(\rho_c-1)}.
\end{align*}
To avoid getting lost in the details, we will use a trick. Note that we can write 
\begin{align*}
    \tilde A-A&=a_1(\tilde\pi_c-\pi_c)+a_2(\tilde\mu_0-\mu_0)+a_3(\tilde\pi_c-\pi_c)(\tilde\mu_0-\mu_0),
    \\
    \tilde B-B&=b_1(\tilde\pi_c-\pi_c)+b_2(\tilde\mu_0-\mu_0)+b_3(\tilde\pi_c-\pi_c)(\tilde\mu_0-\mu_0),
    \\
    \tilde C-C&=c_1(\tilde\pi_c-\pi_c)+c_2(\tilde\mu_0-\mu_0)+c_3(\tilde\pi_c-\pi_c)(\tilde\mu_0-\mu_0),
\end{align*}
where we put off deriving the functions $a_1,a_2,b_1,b_2,c_1,c_2$ until we need them. Then we have
\begin{align*}
    (*)
    &=
    \frac{1-Z}{e_0}\left[1+\frac{1}{\tilde e_0}(\tilde e_1-e_1)\right]
    [A+a_1(\tilde\pi_c-\pi_c)+a_2(\tilde\mu_0-\mu_0)+a_3(\tilde\pi_c-\pi_c)(\tilde\mu_0-\mu_0)]
    [(Y-\mu_0)-(\tilde\mu_0-\mu_0)],
    \\
    (**)
    &=
    \frac{Z}{e_1}\left[1-\frac{1}{\tilde e_1}(\tilde e_1-e_1)\right]
    [B+b_1(\tilde\pi_c-\pi_c)+b_2(\tilde\mu_0-\mu_0)+b_3(\tilde\pi_c-\pi_c)(\tilde\mu_0-\mu_0)]
    \{[\I(C=c)-\pi_c]-(\tilde\pi_c-\pi_c)\},
    \\
    (***)
    &=
    C+c_1(\tilde\pi_c-\pi_c)+c_2(\tilde\mu_0-\mu_0)+c_3(\tilde\pi_c-\pi_c)(\tilde\mu_0-\mu_0).
\end{align*}
Based on these, we can obtain $\tilde\phi_{\nu_{0c}}^\text{OR}-\phi_{\nu_{0c}}^\text{OR}$ to be the sum of a number of terms. As we are focusing on $T_2$ (which is the expectation of $\tilde\phi_{\nu_{0c}}^\text{OR}-\phi_{\nu_{0c}}^\text{OR}$), we can ignore the terms that obviously have expectation zero, i.e., terms that involve $(Y-\mu_0)$ or $[\I(C=c)-\pi_c]$. The other terms are
\begin{align*}
    &(\tilde\pi_c-\pi_c)
    \left\{
    -\frac{Z}{e_1}B+c_1\right\}
    =(\tilde\pi_c-\pi_c)
    \left\{
    \left(1-\frac{Z}{e_1}\right)B-B+c_1\right\},
    \\
    &(\tilde\mu_0-\mu_0)
    \left\{
    -\frac{1-Z}{e_0}A
    +c_2\right\}
    =(\tilde\mu_0-\mu_0)
    \left\{
    \left(1-\frac{1-Z}{e_0}\right)A
    -A+c_2\right\},
    \\
    &(\tilde e_1-e_1)(\tilde\pi_c-\pi_c)
    \left\{
    \frac{Z}{e_1}\frac{1}{\tilde e_1}B\right\},
    \\
    &(\tilde e_1-e_1)(\tilde\mu_0-\mu_0)
    \left\{
    -\frac{1-Z}{e_0}\frac{1}{\tilde e_0}A
    \right\},
    \\
    &(\tilde\pi_c-\pi_c)(\tilde\mu_0-\mu_0)
    \left\{-\frac{1-Z}{e_0}a_1
    -\frac{Z}{e_1}b_2
    +c_3\right\}
    =(\tilde\pi_c-\pi_c)(\tilde\mu_0-\mu_0)
    \left\{\left(1-\frac{1-Z}{e_0}\right)a_1
    +\left(1-\frac{Z}{e_1}\right)b_2
    +c_3-a_1-b_2\right\},
    \\
    &(\tilde e_1-e_1)(\tilde\pi_c-\pi_c)(\tilde\mu_0-\mu_0)
    \left\{-\frac{1-Z}{e_0}\frac{1}{\tilde e_0}a_1
    \right\},
    \\
    &(\tilde\pi_c-\pi_c)^2
    \left\{-\frac{Z}{e_1}b_1\right\},
    \\
    &(\tilde e_1-e_1)(\tilde\pi_c-\pi_c)^2
    \left\{\frac{Z}{e_1}\frac{1}{\tilde e_1}b_1\right\},
    \\
    &(\tilde e_1-e_1)(\tilde\mu_0-\mu_0)(\tilde\pi_c-\pi_c)^2
    \left\{\frac{Z}{e_1}\frac{1}{\tilde e_1}b_3\right\},
    \\
    &(\tilde\mu_0-\mu_0)^2
    \left\{-\frac{1-Z}{e_0}a_2\right\},
    \\
    &(\tilde e_1-e_1)(\tilde\mu_0-\mu_0)^2
    \left\{-\frac{1-Z}{e_0}\frac{1}{\tilde e_0}a_2\right\},
    \\
    &(\tilde e_1-e_1)(\tilde\pi_c-\pi_c)(\tilde\mu_0-\mu_0)^2
    \left\{-\frac{1-Z}{e_0}\frac{1}{\tilde e_0}a_3\right\}.
\end{align*}

\begin{align*}
    T_2
    &=
    \P\Big\{
    (\tilde\pi_c-\pi_c)\left(-B+c_1\right)+(\tilde\mu_0-\mu_0)(-A+c_2)
    \\
    &~~~~~~~~
    (\tilde e_1-e_1)(\tilde\pi_c-\pi_c)\frac{1}{\tilde e_1}B
    -(\tilde e_1-e_1)(\tilde\mu_0-\mu_0)\frac{1}{\tilde e_0}A
    +(\tilde\pi_c-\pi_c)(\tilde\mu_0-\mu_0)(c_3-a_1-b_2)+
    \\
    &~~~~~~~~
    -(\tilde\pi_c-\pi_c)^2 b_1-(\tilde\mu_0-\mu_0)^2 a_2+
    \\
    &~~~~~~~~
    -(\tilde e_1-e_1)(\tilde\pi_c-\pi_c)(\tilde\mu_0-\mu_0)\frac{1}{\tilde e_0}a_1
    +(\tilde e_1-e_1)(\tilde\pi_c-\pi_c)^2\frac{1}{\tilde e_1}b_1
    -(\tilde e_1-e_1)(\tilde\mu_0-\mu_0)^2\frac{1}{\tilde e_0}a_2+
    \\
    &~~~~~~~~
    (\tilde e_1-e_1)(\tilde\mu_0-\mu_0)(\tilde\pi_c-\pi_c)^2\frac{1}{\tilde e_1}b_3
    -(\tilde e_1-e_1)(\tilde\pi_c-\pi_c)(\tilde\mu_0-\mu_0)^2\frac{1}{\tilde e_0}a_3
    \Big\}
\end{align*}
To examine the first order error terms in $T_2$, we need to derive $c_1$ and $c_2$.
\begin{align*}
    \tilde C-C=\frac{1}{2(\rho_c-1)}[(\tilde\alpha_c-\alpha_c)-(\tilde\beta_c-\beta_c)].
\end{align*}
It can be shown that
\begin{align*}
    \tilde\alpha_c-\alpha_c
    &=(\rho_c-1)(\tilde\pi_c-\pi_c)+(\rho_c-1)(\tilde\mu_0-\mu_0),
    \\
    \tilde\beta_c-\beta_c
    &=\frac{\rho-1}{\beta_c+\tilde\beta_c}[(\tilde\alpha_c+\alpha_c-4\rho_c\mu_0)(\tilde\pi_c-\pi_c)+(\tilde\alpha_c+\alpha_c-4\rho\pi_c)(\tilde\mu_0-\mu_0)-4\rho_c(\tilde\mu_0-\mu_0)(\tilde\pi_c-\pi_c)],
\end{align*}
so
\begin{align*}
    \tilde C-C
    &=(\tilde\pi_c-\pi_c)\underbrace{\frac{1}{2}\left[1-\frac{\tilde\alpha_c+\alpha_c-4\rho_c\mu_0}{\beta_c+\tilde\beta_c}\right]}_{c_1}
    +(\tilde\mu_0-\mu_0)\underbrace{\frac{1}{2}\left[1-\frac{\tilde\alpha_c+\alpha_c-4\rho_c\pi_c}{\beta_c+\tilde\beta_c}\right]}_{c_2}
    +(\tilde\pi_c-\pi_c)(\tilde\mu_0-\mu_0)\underbrace{\frac{-2\rho_c}{\beta_c+\tilde\beta_c}}_{c_3}.
\end{align*}
Now consider
\begin{align*}
    -B+c_1
    &=-\left(\frac{1}{2}-\frac{\alpha_c}{2\beta_c}+\frac{\rho_c\mu_0}{\beta_c}\right)+\frac{1}{2}\left[1-\frac{\tilde\alpha_c+\alpha_c-4\rho_c\mu_0}{\beta_c+\tilde\beta_c}\right]
    \\
    &=\frac{1}{2}\left[\frac{\alpha_c-2\rho_c\mu_0}{\beta_c}-\frac{\tilde\alpha_c+\alpha_c-4\rho_c\mu_0}{\beta_c+\tilde\beta_c}\right]
    \\
    &=\frac{1}{2\beta_c(\beta_c+\tilde\beta_c)}[(\alpha_c\beta_c+\alpha_c\tilde\beta_c-2\rho_c\mu_0\beta_c-2\rho_c\mu_0\tilde\beta_c)-(\tilde\alpha_c\beta_c+\alpha_c\beta_c-4\rho_c\mu_0\beta_c)]
    \\
    &=\frac{1}{2\beta_c(\beta_c+\tilde\beta_c)}[(\alpha_c\tilde\beta_c-\tilde\alpha_c\beta_c)-2\rho_c\mu_0(\tilde\beta_c-\beta_c)]
    \\
    &=\frac{1}{2\beta_c(\beta_c+\tilde\beta_c)}[(\alpha_c-2\rho_c\mu_0)(\tilde\beta_c-\beta_c)-\beta_c(\tilde\alpha_c-\alpha_c)]
    \\
    &=\frac{\rho_c-1}{2\beta_c(\beta_c+\tilde\beta_c)}
    \Big\{\frac{(\alpha_c-2\rho_c\mu_0)}{\beta_c+\tilde\beta_c}[(\tilde\alpha_c+\alpha_c-4\rho_c\mu_0)(\tilde\pi_c-\pi_c)+(\tilde\alpha_c+\alpha_c-4\rho\pi_c)(\tilde\mu_0-\mu_0)-4\rho_c(\tilde\mu_0-\mu_0)(\tilde\pi_c-\pi_c)]
    \\
    &~~~~~~~~~~~~~~~~~~~~~~~~-\beta_c[(\tilde\pi_c-\pi_c)+(\tilde\mu_0-\mu_0)]\Big\}
    \\
    &=(\tilde\pi_c-\pi_c)\frac{\rho_c-1}{2\beta_c(\beta_c+\tilde\beta_c)}\left[\frac{(\alpha_c-2\rho_c\mu_0)(\tilde\alpha_c+\alpha_c-4\rho_c\mu_0)}{\beta_c+\tilde\beta_c}-\beta_c\right]+
    \\
    &~~~~(\tilde\mu_0-\mu_0)\frac{\rho_c-1}{2\beta_c(\beta_c+\tilde\beta_c)}\left[\frac{(\alpha_c-2\rho_c\mu_0)(\tilde\alpha_c+\alpha_c-4\rho_c\pi_c)}{\beta_c+\tilde\beta_c}-\beta_c\right]+
    \\
    &~~~~(\tilde\pi_c-\pi_c)(\tilde\mu_0-\mu_0)\left[-\frac{2\rho_c(\rho_c-1)}{\beta_c(\beta_c+\tilde\beta_c)}\right].
\end{align*}
This means the $(\tilde\pi_c-\pi_c)(-B+c_1)$ term is actually the sum of several higher-order terms. So is the $(\tilde\mu_0-\mu_0)(-A+c_2)$ term. The first-order error terms thus drop out. For $T_2$ to be $o_p(n^{-1/2})$, the second-order error terms need to be $o_p(n^{-1/2})$. Therefore, in addition to the rate conditions required under PI, we also require the $(\tilde\pi_c-\pi_c)^2$ and $(\tilde\mu_0-\mu_0)^2$ terms to vanish at $\sqrt{n}$ rate.

\medskip

In conclusion, under the following conditions $\hat\Delta_c^\text{OR}$ is $\sqrt{n}$-consistent for $\Delta_c^\text{OR}$: positivity, sample splitting, consistency, bounded propensity score estimation, plus the rate conditions
\begin{itemize}
    \item $e\pi$-rate: $||\tilde e_1(X)-e_1(X)||_2||\tilde\pi_c(X)-\pi_c(X)||_2=o_p(n^{-1/2})$
    \item $e\mu$-rates: $\begin{cases}
        ||\tilde e_1(X)-e_1(X)||_2||\tilde\mu_0(X)-\mu_0(X)||_2=o_p(n^{-1/2})
        \\
        ||\tilde e_1(X)-e_1(X)||_2||\tilde\mu_{1c}(X)-\mu_{1c}(X)||_2=o_p(n^{-1/2})
    \end{cases}$
    \item $\pi$-rate: $||\tilde\pi_c(X)-\pi_c(X)||_2=o_p(n^{-1/4})$
    \item $\mu_0$-rate: $||\tilde\mu_0(X)-\mu_0(X)||_2=o_p(n^{-1/4})$
\end{itemize}

\medskip
\noindent
\underline{Part 4 for $\hat\Delta_{c,\text{IF}}^\text{MR}$}

\begin{align*}
    \hat\nu_{0c}^\text{MR}-\nu_{0c}^\text{MR}
    &=\P_n[\tilde\phi_{\nu_{0c}}^\text{MR}(O)]-\P[\phi_{\nu_{0c}}^\text{MR}(O)]
    \\
    &=
    \underbrace{(\P_n-\P)[\phi_{\nu_{0c}}^\text{MR}(O)]}_{\textstyle T_0}+
    \underbrace{(\P_n-\P)[\tilde\phi_{\nu_{0c}}^\text{MR}(O)-\phi_{\nu_{0c}}^\text{MR}(O)]}_{\textstyle T_1}+
    \underbrace{\P[\tilde\phi_{\nu_{0c}}^\text{MR}(O)-\phi_{\nu_{0c}}^\text{MR}(O)]}_{\textstyle T_2},
\end{align*}
where
\begin{align*}
    \tilde\phi_{\nu_{0c}}^\text{MR}-\phi_{\nu_{0c}}^\text{MR}
    &=
    \underbrace{\frac{1-Z}{\tilde e_0}\overbrace{\tilde\gamma_c\tilde\pi_c}^{\tilde A}(Y-\tilde\mu_0)}
    _{(*)}
    +\underbrace{\frac{Z}{\tilde e_1}\overbrace{\tilde\gamma_1\tilde\gamma_0\tilde\mu_0}^{\tilde B}[\I(C=c)-\tilde\pi_c]}
    _{(**)}
    +\underbrace{\overbrace{\tilde\gamma_c\tilde\pi_c\tilde\mu_0}^{\tilde C}}_{(***)}+
    \\
    &~~~~~-\left[\frac{1-Z}{e_0}\gamma_c\pi_c(Y-\mu_0)+\frac{Z}{e_1}\gamma_1\gamma_0\mu_0[\I(C=c)-\pi_c]+\gamma_c\pi_c\mu_0\right],
\end{align*}
and
\begin{align*}
    \gamma_c
    =\frac{\rho_c}{(\rho_c-1)\pi_c+1},
    ~~~
    \gamma_{1-c}
    =\frac{1}{(\rho_c-1)\pi_c+1}.
\end{align*}
We use the same tricks as above.
\begin{align*}
    \tilde A-A&=a_1(\tilde\pi_c-\pi_c),
    \\
    \tilde B-B&=b_1(\tilde\pi_c-\pi_c)+b_2(\tilde\mu_0-\mu_0)+b_3(\tilde\pi_c-\pi_c)(\tilde\mu_0-\mu_0),
    \\
    \tilde C-C&=c_1(\tilde\pi_c-\pi_c)+c_2(\tilde\mu_0-\mu_0)+c_3(\tilde\pi_c-\pi_c)(\tilde\mu_0-\mu_0),
\end{align*}
and
\begin{align*}
    (*)
    &=
    \frac{1-Z}{e_0}\left[1+\frac{1}{\tilde e_0}(\tilde e_1-e_1)\right]
    [A+a_1(\tilde\pi_c-\pi_c)]
    [(Y-\mu_0)-(\tilde\mu_0-\mu_0)],
    \\
    (**)
    &=
    \frac{Z}{e_1}\left[1-\frac{1}{\tilde e_1}(\tilde e_1-e_1)\right]
    [B+b_1(\tilde\pi_c-\pi_c)+b_2(\tilde\mu_0-\mu_0)+b_3(\tilde\pi_c-\pi_c)(\tilde\mu_0-\mu_0)]
    \{[\I(C=c)-\pi_c]-(\tilde\pi_c-\pi_c)\},
    \\
    (***)
    &=
    C+c_1(\tilde\pi_c-\pi_c)+c_2(\tilde\mu_0-\mu_0)+c_3(\tilde\pi_c-\pi_c)(\tilde\mu_0-\mu_0).
\end{align*}
Consider $\tilde\phi_{\nu_{0c}}^\text{MR}-\phi_{\nu_{0c}}^\text{MR}$ as a sum of terms. Leaving out the obvious terms with expectation zero, the remaining terms are 
\begin{align*}
    &(\tilde\pi_c-\pi_c)
    \left\{-\frac{Z}{e_1}B+c_1\right\}
    =(\tilde\pi_c-\pi_c)
    \left\{\left(1-\frac{Z}{e_1}\right)B-B+c_1\right\},
    \\
    &(\tilde\mu_0-\mu_0)
    \left\{-\frac{1-Z}{e_0}A+c_2\right\}
    =(\tilde\mu_0-\mu_0)
    \left\{\left(1-\frac{1-Z}{e_0}\right)A-A+c_2\right\},
    \\
    &(\tilde e_1-e_1)(\tilde\pi_c-\pi_c)
    \left\{\frac{Z}{e_1}\frac{1}{\tilde e_1}B\right\},
    \\
    &(\tilde e_1-e_1)(\tilde\mu_0-\mu_0)
    \left\{-\frac{1-Z}{e_0}\frac{1}{\tilde e_0}A\right\},
    \\
    &(\tilde\pi_c-\pi_c)(\tilde\mu_0-\mu_0)
    \left\{-\frac{1-Z}{e_0}a_1-\frac{Z}{e_1}b_2+c_3\right\}
    \\
    &~~~~=(\tilde\pi_c-\pi_c)(\tilde\mu_0-\mu_0)
    \left\{\left(1-\frac{1-Z}{e_0}\right)a_1+\left(1-\frac{Z}{e_1}\right)b_2+c_3-a_1-b_2\right\},
    \\
    &(\tilde e_1-e_1)(\tilde\pi_c-\pi_c)(\tilde\mu_0-\mu_0)
    \left\{-\frac{1-Z}{e_0}\frac{1}{\tilde e_0}a_1\right\},
    \\
    &(\tilde\pi_c-\pi_c)^2
    \left\{-\frac{Z}{e_1}b_1\right\},
    \\
    &(\tilde e_1-e_1)(\tilde\pi_c-\pi_c)^2
    \left\{\frac{Z}{e_1}\frac{1}{\tilde e_1}b_1\right\},
    \\
    &(\tilde e_1-e_1)(\tilde\mu_0-\mu_0)(\tilde\pi_c-\pi_c)^2
    \left\{\frac{Z}{e_1}\frac{1}{\tilde e_1}b_3\right\},
\end{align*}
\begin{align*}
    T_2
    &=
    \P\Big\{
    (\tilde\pi_c-\pi_c)\left(-B+c_1\right)+(\tilde\mu_0-\mu_0)(-A+c_2)
    \\
    &~~~~~~~~
    (\tilde e_1-e_1)(\tilde\pi_c-\pi_c)\frac{1}{\tilde e_1}B
    -(\tilde e_1-e_1)(\tilde\mu_0-\mu_0)\frac{1}{\tilde e_0}A
    +(\tilde\pi_c-\pi_c)(\tilde\mu_0-\mu_0)(c_3-a_1-b_2)
    -(\tilde\pi_c-\pi_c)^2b_1+
    \\
    &~~~~~~~~
    -(\tilde e_1-e_1)(\tilde\pi_c-\pi_c)(\tilde\mu_0-\mu_0)\frac{1}{\tilde e_0}a_1
    +(\tilde e_1-e_1)(\tilde\pi_c-\pi_c)^2\frac{1}{\tilde e_1}b_1+
    \\
    &~~~~~~~~
    (\tilde e_1-e_1)(\tilde\mu_0-\mu_0)(\tilde\pi_c-\pi_c)^2\frac{1}{\tilde e_1}b_3
    \Big\}
\end{align*}
To examine the first order error terms in $T_2$, we derive $c_1$ and $c_2$.
\begin{align*}
    \tilde C-C
    &=\tilde\gamma_c\tilde\pi_c\tilde\mu_0-\gamma_c\pi_c\mu_0
    \\
    &=\mu_0(\tilde\gamma_c\tilde\pi_c-\gamma_c\pi_c)+\gamma_c\pi_c(\tilde\mu_0-\mu_0)+(\tilde\gamma_c\tilde\pi_c)(\tilde\mu_0-\mu_0)
    \\
    &=\underbrace{\mu_0\gamma_c\tilde\gamma_{1-c}}_{c_1}(\tilde\pi_c-\pi_c)+\underbrace{\gamma_c\pi_c}_{c_2}(\tilde\mu_0-\mu_0)+\underbrace{\gamma_c\tilde\gamma_{1-c}}_{c_3}(\tilde\pi_c-\pi_c)(\tilde\mu_0-\mu_0).
\end{align*}
This implies
\begin{align*}
    -A+c_2
    &=-\gamma_c\pi_c+\gamma_c\pi_c=0,
    \\
    -B+c_1
    &=-\gamma_1\gamma_0\mu_0+\gamma_c\tilde\gamma_{1-c}\mu_0
    \\
    &=\gamma_c\mu_0(\tilde\gamma_{1-c}-\gamma_{1-c})
    \\
    &=\gamma_c\mu_0\frac{(\rho_c-1)(\pi_c-\tilde\pi_c)}{[(\rho_c-1)\tilde\pi_c+1][\rho_c-1)\pi_c+1]}
    \\
    &=-\gamma_1\gamma_0\mu_0\tilde\gamma_c(\rho_c-1)(\tilde\pi_c-\pi_c).
\end{align*}
This means $(-A+c_2)(\tilde\mu_0-\mu_0)=0$ and $(-B+c_1)(\tilde\pi_c-\pi_c)=-\gamma_1\gamma_0\tilde\gamma_c(\rho-1)\mu_0(\tilde\pi_c-\pi_c)^2$ is a second-order term. Some tedious algebra shows that this term does not cancel out with the other $(\tilde\pi_c-\pi_c)^2$ term. Also, some tedious algebra shows that $c_3-a_1-b_2$ is non-zero. Therefore
\begin{align*}
    |T_2|\leq
    &\tilde D||\tilde e_1-e_1||_2||\tilde\pi_c-\pi_c||_2+
    \tilde D||\tilde e_1-e_1||_2||\tilde\mu_0-\mu_0||_2+
    \tilde D||\tilde\pi_c-\pi_c||_2||\tilde\mu_0-\mu_0||_2+
    \tilde D||\tilde\pi_c-\pi_c||_2^2,
\end{align*}
where $\tilde D=O_p(1)$. We thus require the terms on the RHS to be $o_p(n^{-1/2})$.

In conclusion, under the following conditions $\hat\Delta_c^\text{MR}$ is $\sqrt{n}$-consistent for $\Delta_c^\text{MR}$: positivity, sample splitting, consistency, bounded propensity score estimation, plus the rate conditions
\begin{itemize}
    \item $e\pi$-rate: $||\tilde e_1(X)-e_1(X)||_2||\tilde\pi_c(X)-\pi_c(X)||_2=o_p(n^{-1/2})$
    \item $e\mu$-rates: $\begin{cases}
        ||\tilde e_1(X)-e_1(X)||_2||\tilde\mu_0(X)-\mu_0(X)||_2=o_p(n^{-1/2})
        \\
        ||\tilde e_1(X)-e_1(X)||_2||\tilde\mu_{1c}(X)-\mu_{1c}(X)||_2=o_p(n^{-1/2})
    \end{cases}$
    \item $\pi\mu_0$-rate: $||\tilde\pi_c(X)-\pi_c(X)||_2||\tilde\mu_0(X)-\mu_0(X)||_2=o_p(n^{-1/2})$
    \item $\pi$-rate: $||\tilde\pi_c(X)-\pi_c(X)||_2=o_p(n^{-1/4})$
\end{itemize}

\medskip
\noindent
\underline{Part 4 for $\hat\Delta_{c,\text{IF}}^\text{SMDe}$}

\smallskip
\noindent
Recall that 
\begin{align*}
    \nu_{0c}^\text{SMDe}=\nu_{0c}^\text{PI}-\eta_c\underbrace{\E\Big[\overbrace{\frac{\pi_1(X)\pi_0(X)\sigma_0(X)}{\sqrt{1+\eta^2\pi_1(X)\pi_0(X)}}}^{\vartheta(X)}\Big]}_{\vartheta},
\end{align*}
so
\begin{align*}
    \hat\nu_{0c}^\text{SMDe}=\hat\nu_{0c}^\text{PI}-\eta_c\hat\vartheta,~~~\text{where}~\hat\vartheta=\P_n[\tilde\phi_\vartheta(O)],
\end{align*}
and $\phi_\vartheta(O)=\varphi_\vartheta(O)+\vartheta$ is the uncentered IF of $\vartheta$. We thus adopt all the conditions for $\hat\nu_{0c}^\text{PI}$ to be $\sqrt{n}$-consistent for $\nu_{0c}^\text{PI}$, and seek additional conditions (if any) for $\hat\vartheta$ to be $\sqrt{n}$-consistent for $\vartheta$.

\begin{align*}
    \hat\vartheta-\vartheta
    &=\P_n[\tilde\phi_\vartheta(O)]-\P[\phi_\vartheta(O)]
    \\
    &=(\P_n-P)[\phi_\vartheta(O)]+(\P_n-P)[\tilde\phi_\vartheta(O)-\phi_\vartheta(O)]+\P[\tilde\phi_\vartheta(O)-\phi_\vartheta(O)].
\end{align*}
Since $\vartheta$ is symmetric w.r.t. $c$ (being 1 or 0), to ease notation, let $\pi(X)=\pi_1(X)$, and again drop all $(X)$ notation from this point. Then
\begin{align*}
    &\tilde\phi_\vartheta(O)-\phi_\vartheta(O)=
    \\
    &\underbrace{\frac{1-Z}{\tilde e_0}\overbrace{\frac{\tilde\pi(1-\tilde\pi)}{2\sqrt{\tilde\sigma_0^2[1+\eta^2\tilde\pi(1-\tilde\pi)]}}}^{\tilde A}[(Y-\tilde\mu_0)^2-\tilde\sigma_0^2]}_{(*)}
    +\underbrace{\frac{Z}{\tilde e_1}\overbrace{\frac{(1-2\tilde\pi)[2+\eta^2\tilde\pi(1-\tilde\pi)]\sqrt{\tilde\sigma_0^2}}{2[1+\eta^2\tilde\pi(1-\tilde\pi)]^{3/2}}}^{\tilde B}(C-\tilde\pi)}_{(**)}
    +\underbrace{\overbrace{\frac{\tilde\pi(1-\tilde\pi)\sqrt{\tilde\sigma_0^2}}{\sqrt{1+\eta^2\tilde\pi(1-\tilde\pi)}}}^{\tilde G}}_{(***)}+
    \\
    &
    -\left\{\frac{1-Z}{ e_0}\frac{\pi(1-\pi)}{2\sqrt{\sigma_0^2[1+\eta^2\pi(1-\pi)]}}[(Y-\mu_0)^2-\sigma_0^2]+\frac{Z}{ e_1}\frac{(1-2\pi)[2+\eta^2\pi(1-\pi)]\sqrt{\sigma_0^2}}{2[1+\eta^2\pi(1-\pi)]^{3/2}}(C-\pi)+\frac{\pi(1-\pi)\sqrt{\sigma_0^2}}{\sqrt{1+\eta^2\pi(1-\pi)}}\right\}.
\end{align*}
Using the same tricks as before, we write
\begin{align*}
    \tilde A-A
    &=a_1(\tilde\pi-\pi)+a_2(\tilde\sigma_0^2-\sigma_0^2)+a_3(\tilde\pi-\pi)(\tilde\sigma_0^2-\sigma_0^2),
    \\
    \tilde B-B
    &=b_1(\tilde\pi-\pi)+b_2(\tilde\sigma_0^2-\sigma_0^2)+b_3(\tilde\pi-\pi)(\tilde\sigma_0^2-\sigma_0^2),
    \\
    \tilde G-G
    &=g_1(\tilde\pi-\pi)+g_2(\tilde\sigma_0^2-\sigma_0^2)+g_3(\tilde\pi-\pi)(\tilde\sigma_0^2-\sigma_0^2),
\end{align*}
so
\begin{align*}
    (*)
    &=
    \frac{1-Z}{e_0}\left[1+\frac{1}{\tilde e_0}(\tilde e_1-e_1)\right]
    [A+a_1(\tilde\pi-\pi)+a_2(\tilde\sigma_0^2-\sigma_0^2)+a_3(\tilde\pi-\pi)(\tilde\sigma_0^2-\sigma_0^2)]\times
    \\
    &~~~~
    \{[(Y-\mu_0)^2-\sigma_0^2]-2(\tilde\mu_0-\mu_0)(Y-\mu_0)+(\tilde\mu_0-\mu_0)^2-(\tilde\sigma_0^2-\sigma_0^2)]\},
    \\
    (**)
    &=
    \frac{Z}{e_1}\left[1-\frac{1}{\tilde e_1}(\tilde e_1-e_1)\right][B+b_1(\tilde\pi-\pi)+b_2(\tilde\sigma_0^2-\sigma_0^2)+b_3(\tilde\pi-\pi)(\tilde\sigma_0^2-\sigma_0^2)][(C-\pi)-(\tilde\pi-\pi)],
    \\
    (***)
    &=
    G+g_1(\tilde\pi-\pi)+g_2(\tilde\sigma_0^2-\sigma_0^2)+g_3(\tilde\pi-\pi)(\tilde\sigma_0^2-\sigma_0^2).
\end{align*}
Consider $\tilde\phi_\vartheta(O)-\phi_\vartheta(O)$ as a sum of terms. Again, we ignore the terms whose expectations are obviously zero, i.e., those that involve $[(Y-\mu_0)^2-\sigma_0^2]$, or $(Y-\mu_0)$, or $(C-\pi)$. We also put aside third-order terms that are dominated by second-order terms. The remaining terms are:
\begin{align*}
    &
    (\tilde\pi-\pi)\times
    \left\{-\frac{Z}{e_1}B+g_1\right\}
    =(\tilde\pi-\pi)\times
    \left\{\left(1-\frac{Z}{e_1}\right)B-B+g_1\right\},
    \\
    &
    (\tilde\sigma_0^2-\sigma_0^2)\times
    \left\{-\frac{1-Z}{e_0}A+g_2\right\}
    =(\tilde\sigma_0^2-\sigma_0^2)\times
    \left\{\left(1-\frac{1-Z}{e_0}\right)A-A+g_2\right\},
    \\
    &
    (\tilde\mu_0-\mu_0)^2\times
    \left\{\frac{1-Z}{e_0}A\right\},
    \\
    &
    (\tilde\pi-\pi)^2\times
    \left\{-\frac{Z}{e_1}b_1\right\},
    \\
    &
    (\tilde\sigma_0^2-\sigma_0^2)^2\times
    \left\{-\frac{1-Z}{e_0}a_2\right\},
    \\
    &
    (\tilde e_1-e_1)(\tilde\pi-\pi)\times
    \left\{\frac{Z}{e_1}\frac{1}{\tilde e_1}B\right\},
    \\
    &
    (\tilde e_1-e_1)(\tilde\sigma_0^2-\sigma_0^2)\times
    \left\{-\frac{1-Z}{e_0}\frac{1}{\tilde e_0}A\right\},
    \\
    &
    (\tilde\pi-\pi)(\tilde\sigma_0^2-\sigma_0^2)\times
    \left\{-\frac{1-Z}{e_0}a_1-\frac{Z}{e_1}b_2+g_3\right\}
    =(\tilde\pi-\pi)(\tilde\sigma_0^2-\sigma_0^2)\times
    \left\{\left(1-\frac{1-Z}{e_0}\right)a_1+\left(1-\frac{Z}{e_1}\right)b_2+g_3-a_1-b_2\right\}
    .
\end{align*}
Hence
\begin{align*}
    T_2
    &=\P\Big\{
    (\tilde\pi-\pi)(-B+g_1)+(\tilde\sigma_0^2-\sigma_0^2)(-A+g_2)+
    \\
    &~~~~~~~~~-(\tilde\pi-\pi)^2b_1-(\tilde\sigma_0^2-\sigma_0^2)^2a_2+(\tilde\mu_0-\mu_0)^2A+
    \\
    &~~~~~~~~~(\tilde e_1-e_1)(\tilde\pi-\pi)\frac{1}{\tilde e_1}B
    -(\tilde e_1-e_1)(\tilde\sigma_0^2-\sigma_0^2)\frac{1}{\tilde e_0}A
    +(\tilde\pi-\pi)(\tilde\sigma_0^2-\sigma_0^2)(g_3-a_1-b_2)
    +\text{higher-order terms}
    \Big\}.
\end{align*}
We want to examine the coefficients $(-B+g_1)$ and $(-A+g_2)$ of the two first-order terms in the expression of $T_2$ above. Let $u=\frac{\pi(1-\pi)}{\sqrt{1+\eta^2\pi(1-\pi)}}$, $v=\sigma_0=\sqrt{\sigma_0^2}$. Then
\begin{align*}
    \tilde G-G
    &=\tilde u\tilde v-uv=(\tilde u-u)v+(\tilde v-v)u+(\tilde u-u)(\tilde v-v).
\end{align*}
As $u$ involves $\pi$ but not $\sigma_0^2$ and $v$ involves $\sigma_0^2$ but not $\pi$, we can obtain $g_1$ and $g_2$ respectively from the first and second terms in the RHS above. 

First use the second term to derive $g_2$.
\begin{align*}
    (\tilde v-v)u
    &=(\sqrt{\tilde\sigma_0^2}-\sqrt{\sigma_0^2})u
    \\
    &=\frac{\tilde\sigma_0^2-\sigma_0^2}{\sqrt{\tilde\sigma_0^2}+\sqrt{\sigma_0^2}}u
    \\
    &=(\tilde\sigma_0^2-\sigma_0^2)\underbrace{\frac{\pi(1-\pi)}{(\tilde\sigma_0+\sigma_0)\sqrt{1+\eta^2\pi(1-\pi)}}}_{g_2}.
\end{align*}
\begin{align*}
    -A+g_2
    &=\frac{\pi(1-\pi)}{\sqrt{1+\eta^2\pi(1-\pi)}}\left(\frac{1}{\tilde\sigma_0+\sigma_0}-\frac{1}{2\sigma_0}\right)
    \\
    &=-\frac{\pi(1-\pi)}{\sqrt{1+\eta^2\pi(1-\pi)}}\frac{\tilde\sigma_0-\sigma_0}{2\sigma_0(\tilde\sigma_0+\sigma_0)}
    \\
    &=-\frac{\pi(1-\pi)}{\sqrt{1+\eta^2\pi(1-\pi)}}\frac{\tilde\sigma_0^2-\sigma_0^2}{2\sigma_0(\tilde\sigma_0+\sigma_0)^2},
\end{align*}
which means $(\tilde\sigma_0^2-\sigma_0^2)(-A+g_2)\propto(\tilde\sigma_0^2-\sigma_0^2)^2$ is a second-order term.

Now we turn to learn about $g_1$, using the first term in the expression of $\tilde G-G$ above.
Let $w=\pi(1-\pi)$. Then
\begin{align*}
    \tilde w-w
    &=(\tilde\pi-\pi)(1-2\pi)-(\tilde\pi-\pi)^2,
    \\
    \frac{1}{\sqrt{1+\eta^2\tilde w}}-\frac{1}{\sqrt{1+\eta^2w}}
    &=-\frac{\eta^2(\tilde w-w)}{\sqrt{1+\eta^2\tilde w}\sqrt{1+\eta^2w}(\sqrt{1+\eta^2\tilde w}+\sqrt{1+\eta^2w})},
\end{align*}
so
\begin{align*}
    \tilde u-u
    &=\frac{\tilde w}{\sqrt{1+\eta^2\tilde w}}-\frac{w}{\sqrt{1+\eta^2 w}}
    \\
    &=\frac{1}{\sqrt{1+\eta^2w}}(\tilde w-w)
    +w\left(\frac{1}{\sqrt{1+\eta^2\tilde w}}-\frac{1}{\sqrt{1+\eta^2 w}}\right)
    +(\tilde w-w)\left(\frac{1}{\sqrt{1+\eta^2\tilde w}}-\frac{1}{\sqrt{1+\eta^2 w}}\right)
    \\
    &=(\tilde w-w)\frac{1}{\sqrt{1+\eta^2w}}\left[1-\frac{\eta^2w}{\sqrt{1+\eta^2\tilde w}[\sqrt{1+\eta^2\tilde w}+\sqrt{1+\eta^2w}]}\right]-\frac{\eta^2(\tilde w-w)^2}{\sqrt{1+\eta^2\tilde w}\sqrt{1+\eta^2w}(\sqrt{1+\eta^2\tilde w}+\sqrt{1+\eta^2w})}
    \\
    &=(\tilde w-w)\frac{1}{\sqrt{1+\eta^2w}}\frac{1+\eta^2(\tilde w-w)+\sqrt{1+\eta^2\tilde w}\sqrt{1+\eta^2w}}{\sqrt{1+\eta^2\tilde w}[\sqrt{1+\eta^2\tilde w}+\sqrt{1+\eta^2w}]}-\frac{\eta^2(\tilde w-w)^2}{\sqrt{1+\eta^2\tilde w}\sqrt{1+\eta^2w}(\sqrt{1+\eta^2\tilde w}+\sqrt{1+\eta^2w})}
    \\
    &=(\tilde w-w)\frac{1+\sqrt{1+\eta^2\tilde w}\sqrt{1+\eta^2 w}}{\sqrt{1+\eta^2\tilde w}\sqrt{1+\eta^2w}(\sqrt{1+\eta^2\tilde w}+\sqrt{1+\eta^2w})}
    \\
    &=(\tilde\pi-\pi)\times
    \\
    &~~~~\left\{\frac{(1-2\pi)[1+\sqrt{1+\eta^2\tilde w}\sqrt{1+\eta^2 w}]}{\sqrt{1+\eta^2\tilde w}\sqrt{1+\eta^2w}(\sqrt{1+\eta^2\tilde w}+\sqrt{1+\eta^2w})}-(\tilde\pi-\pi)\frac{1+\sqrt{1+\eta^2\tilde w}\sqrt{1+\eta^2 w}}{\sqrt{1+\eta^2\tilde w}\sqrt{1+\eta^2w}(\sqrt{1+\eta^2\tilde w}+\sqrt{1+\eta^2w})}\right\}.
\end{align*}
This means $g_1$ is equal to the expression in the brackets above times $v=\sigma_0$. We can drop the second term which involves an additional $(\tilde\pi-\pi)$ factor and just consider the leading term in $g_1$, which we denote by $g_1^*$. That is,
\begin{align*}
    g_1^*=\frac{(1-2\pi)[1+\sqrt{1+\eta^2\tilde w}\sqrt{1+\eta^2 w}]}{\sqrt{1+\eta^2\tilde w}\sqrt{1+\eta^2w}(\sqrt{1+\eta^2\tilde w}+\sqrt{1+\eta^2w})}\sigma_0
    =\frac{(1-2\pi)\sigma_0}{\sqrt{1+\eta^2w}}\frac{1+\sqrt{1+\eta^2\tilde w}\sqrt{1+\eta^2 w}}{\sqrt{1+\eta^2\tilde w}(\sqrt{1+\eta^2\tilde w}+\sqrt{1+\eta^2w})}.
\end{align*}
We can now examine $(-B+g_1^*)$. Note that 
\begin{align*}
    B
    &=\frac{(1-2\pi)(2+\eta^2w)\sigma_0}{2(1+\eta^2w)^{3/2}}
    =\frac{(1-2\pi)\sigma_0}{\sqrt{1+\eta^2w}}
    \frac{1+\sqrt{1+\eta^2w}\sqrt{1+\eta^2 w}}{\sqrt{1+\eta^2w}(\sqrt{1+\eta^2w}+\sqrt{1+\eta^2w})},
\end{align*}
which looks very similar to $g_1^*$. Let $t=\sqrt{1+\eta^2w}$. Then
\begin{align*}
    -B+g_1^*
    &=\frac{(1-2\pi)\sigma_0}{\sqrt{1+\eta^2w}}\left[\frac{1+t\tilde t}{\tilde t(t+\tilde t)}-\frac{1+t^2}{2t^2}\right]
    \\
    &=\frac{(1-2\pi)\sigma_0}{\sqrt{1+\eta^2w}}\frac{2t^2+2t^3\tilde t-t\tilde t-\tilde t^2-t^3\tilde t-t^2\tilde t^2}{2t^2\tilde t(t+\tilde t)}
    \\
    &=\frac{(1-2\pi)\sigma_0}{\sqrt{1+\eta^2w}}\frac{-t(\tilde t-t)-(\tilde t^2-t^2)-t^2\tilde t(\tilde t-t)}{2t^2\tilde t(t+\tilde t)}
    \\
    &=-\frac{(1-2\pi)\sigma_0}{\sqrt{1+\eta^2w}}\frac{(\tilde t-t)(2t+\tilde t+t^2\tilde t)}{2t^2\tilde t(t+\tilde t)}
    \\
    &=-\frac{(1-2\pi)\sigma_0}{\sqrt{1+\eta^2w}}\frac{(\tilde t^2-t^2)(2t+\tilde t+t^2\tilde t)}{2t^2\tilde t(t+\tilde t)^2}
    \\
    &=-\frac{(1-2\pi)\sigma_0}{\sqrt{1+\eta^2w}}\frac{\eta^2(\tilde w-w)(2t+\tilde t+t^2\tilde t)}{2t^2\tilde t(t+\tilde t)^2}
    \\
    &\propto(\tilde\pi-\pi),
\end{align*}
which means $(\tilde\pi-\pi)(-B+g_1)\propto(\tilde\pi-\pi)^2$ is a second-order term.


As a result
\begin{align*}
    |T_2|\leq~
    &\tilde D||\tilde\pi-\pi||_2^2+\tilde D||\tilde\sigma_0^2-\sigma_0^2||_2^2+\tilde D||\tilde\mu_0-\mu_0||_2^2+
    \\
    &\tilde D||\tilde e_1-e_1||_2||\tilde\pi-\pi||_2+D||\tilde e_1-e_1||_2||\tilde\sigma_0^2-\sigma_0^2||_2+D||\tilde\pi-\pi||_2||\tilde\sigma_0^2-\sigma_0^2||_2,
\end{align*}
where $\tilde D=O_p(1)$,
which give us the rate conditions.

In conclusion, under the following conditions $\hat\Delta_c^\text{SMDe}$ is $\sqrt{n}$-consistent for $\Delta_c^\text{SMDe}$: positivity, sample splitting, consistency, bounded propensity score estimation, plus the rate conditions
\begin{itemize}
    \item $e\pi$-rate: $||\tilde e_1(X)-e_1(X)||_2||\tilde\pi_c(X)-\pi_c(X)||_2=o_p(n^{-1/2})$
    \item $e\mu$-rates: $\begin{cases}
        ||\tilde e_1(X)-e_1(X)||_2||\tilde\mu_0(X)-\mu_0(X)||_2=o_p(n^{-1/2})
        \\
        ||\tilde e_1(X)-e_1(X)||_2||\tilde\mu_{1c}(X)-\mu_{1c}(X)||_2=o_p(n^{-1/2})
    \end{cases}$
    \item $e\sigma^2$-rate: $||\tilde e_1(X)-e_1(X)||_2||\tilde\sigma_0^2(X)-\sigma_0^2(X)||_2=o_p(n^{-1/2})$
    \item $\pi$-rate: $||\tilde\pi_c(X)-\pi_c(X)||_2=o_p(n^{-1/4})$
    \item $\mu_0$-rate: $||\tilde\mu_0(X)-\mu_0(X)||_2=o_p(n^{-1/4})$
    \item $\sigma^2$-rate: $||\tilde\sigma_0^2(X)-\sigma_0^2(X)||_2=o_p(n^{-1/4})$
\end{itemize}
    
\end{proof}

\section{Expanded content of Section~\ref{sec:finite-sample-bias} -- Finite-sample bias}\label{appendix:finite-sample-bias}

\subsection{The pattern in all four sensitivity analyses}

Figure \ref{fig:finitesamplebias} in the paper only covers the OR-based sensitivity analysis on \textit{work for pay}. Figure \ref{fig:finitesamplebias-all} presents the pattern from all four sensitivity analyses. Figure \ref{fig:finitesamplebias-all-noCI} removes the confidence intervals to allow zooming into the patterns more closely. To show the pattern more clearly, we expand the ranges of the sensitivity GOR and SMD for depressive symptoms to make them symmetric.

Note one difference compared to Figure~\ref{fig:finitesamplebias}: to save length the current plots additionally include bias-corrected point estimates. These result from the bias correction methods that will be explained and investigated in Section~\ref{assec:bias-corrections}.

\begin{figure}[h]
    \caption{Expansion of Figure \ref{fig:finitesamplebias} for all four sensitivity analyses}
    \label{fig:finitesamplebias-all}
    \centering
    \includegraphics[width=.9\textwidth]{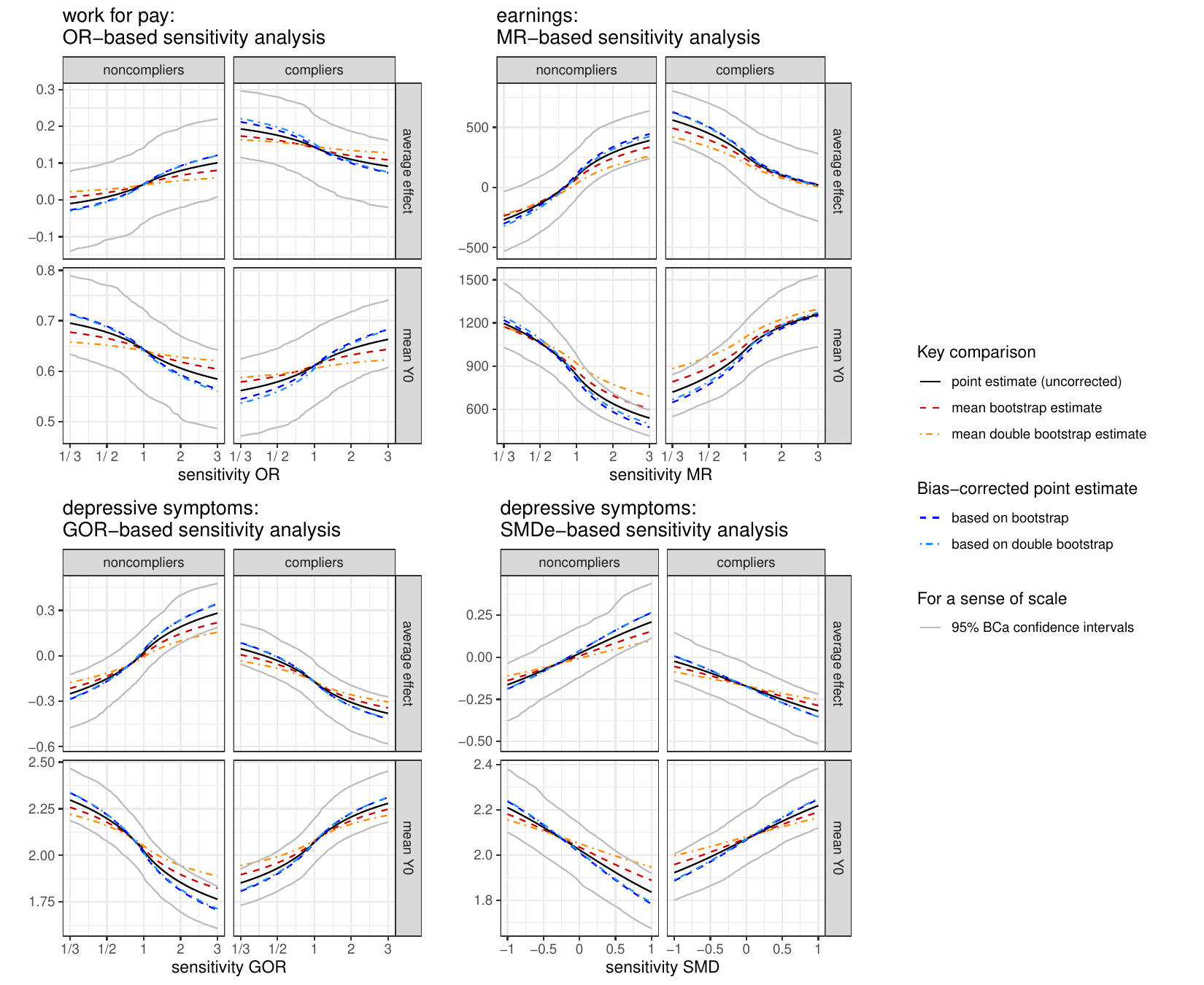}
\end{figure}

\begin{figure}[h]
    \caption{Expansion of Figure \ref{fig:finitesamplebias} for all four sensitivity analyses, CIs removed}
    \label{fig:finitesamplebias-all-noCI}
    \includegraphics[width=\textwidth]{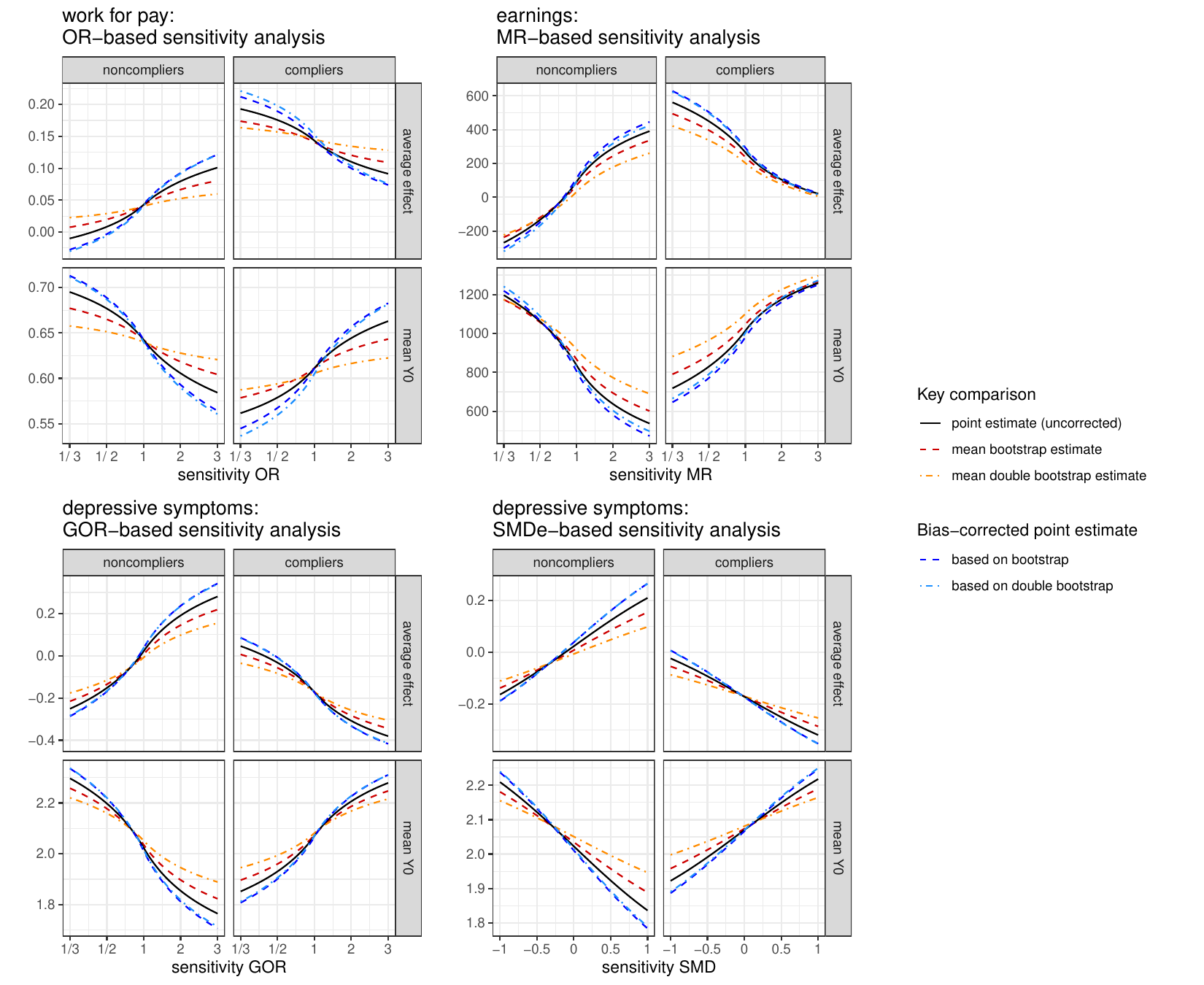}
\end{figure}

\newpage

\subsection{The pattern across different estimators}

JOBS II analysis in the paper uses one estimator. To show that the finite-sample bias pattern appears for all estimators, we show additional plots in Figures \ref{fig:13} to \ref{fig:16} below.

For the most part, the plots look very similar across estimators. For the unbounded estimator IF0, the mean double bootstrap estimate (BM2) for $\tau_{00}$ (and thus that for $\Delta_0$) is extreme. Interestingly, the mean bootstrap estimate (BM1) is not extreme. However, we can only see this problem after it happens. Therefore we recommend not using this kind of bias correction with an unbounded estimator.

\begin{landscape}
    \begin{figure}
        \caption{JOBS II result  -- \textit{work}, OR-based: Finite sample bias pattern with estimators $\hat\Delta_{c,\pi\mu}$ (pimu), $\hat\Delta_{c,e\mu}$ (emu), $\hat\Delta_{c,\textsc{ms}}$ (MS), $\hat\Delta_{c,\textsc{if}}$ (IF), $\hat\Delta_{c,\textsc{ifh}}$ (IFH). 0 and 1 indicated non-targeted and targeted nuisance estimation, respectively.}
        \label{fig:13}
        \centering
        \includegraphics[width=1.5\textwidth]{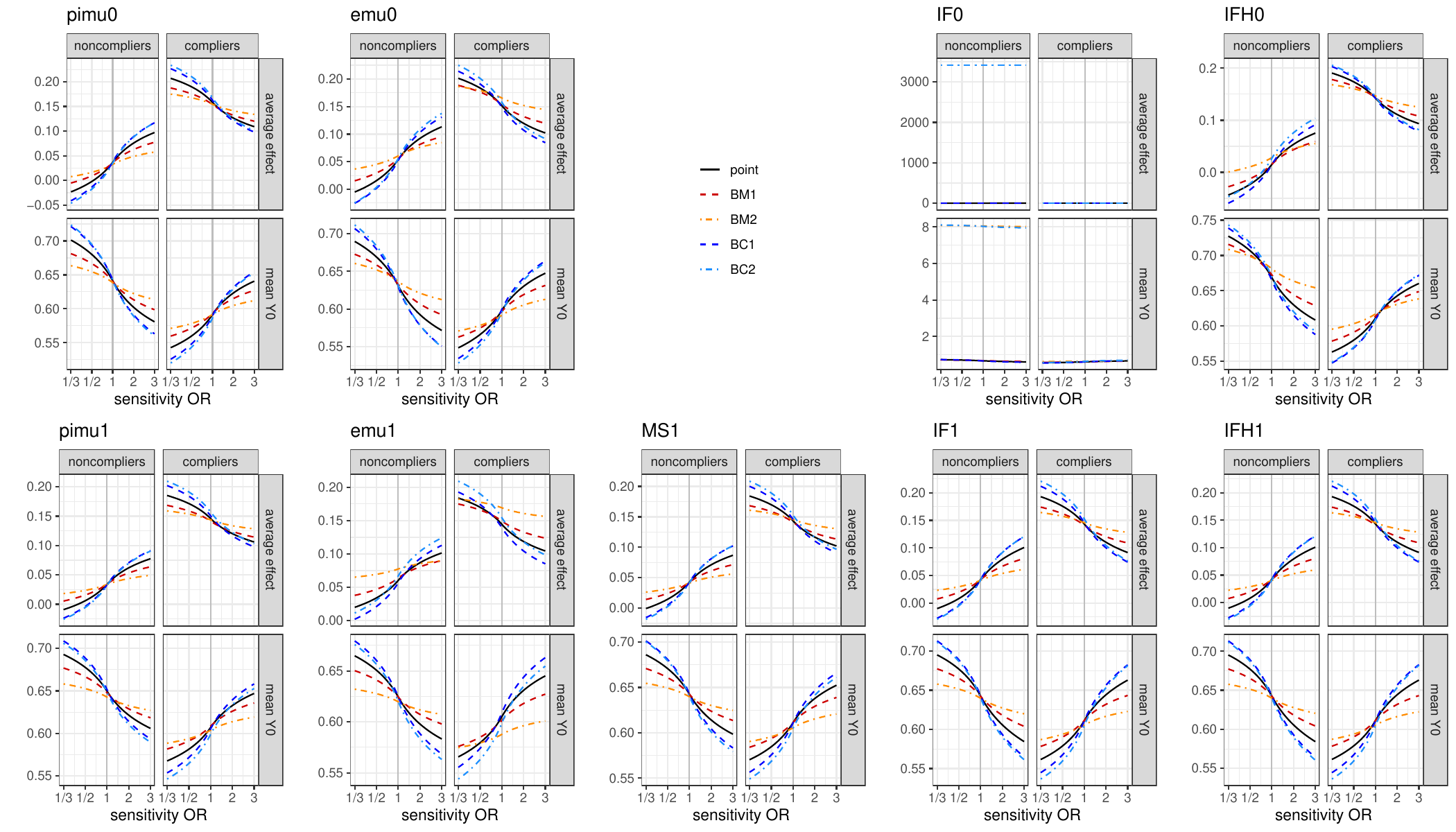}
    \end{figure}
\end{landscape}

\begin{landscape}
    \begin{figure}
        \caption{JOBS II result  -- \textit{depressive symptoms}, GOR-based: Finite sample bias pattern with estimators $\hat\Delta_{c,\pi\mu}$ (pimu), $\hat\Delta_{c,e\mu}$ (emu), $\hat\Delta_{c,\textsc{ms}}$ (MS), $\hat\Delta_{c,\textsc{if}}$ (IF), $\hat\Delta_{c,\textsc{ifh}}$ (IFH). 0 and 1 indicated non-targeted and targeted nuisance estimation, respectively.}
        \label{fig:14}
        \centering
        \includegraphics[width=1.5\textwidth]{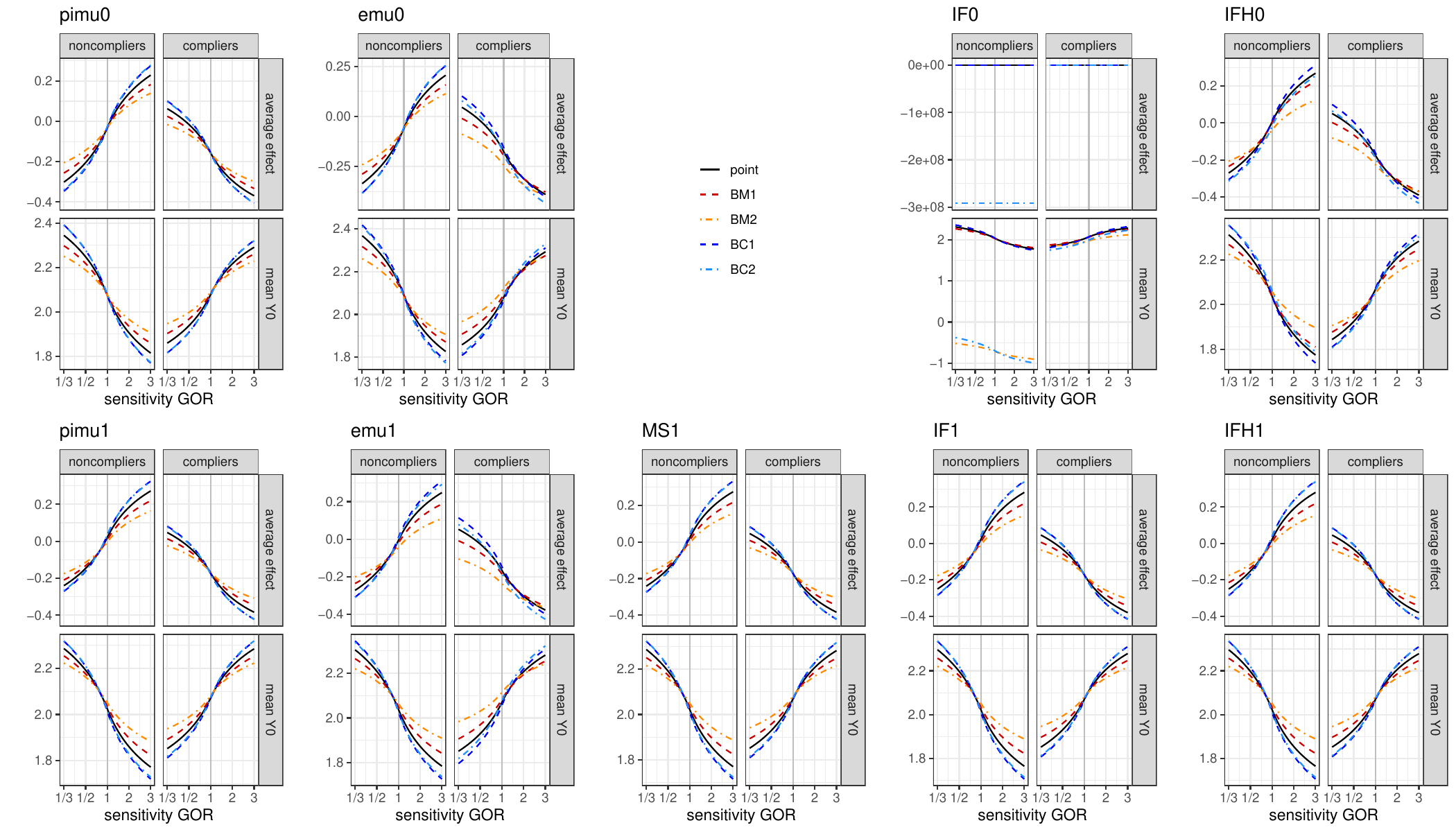}
    \end{figure}
\end{landscape}

\begin{landscape}
    \begin{figure}
        \caption{JOBS II result  -- \textit{depressive symptoms}, SMDe-based: Finite sample bias pattern with estimators $\hat\Delta_{c,\pi\mu}$ (pimu), $\hat\Delta_{c,e\mu}$ (emu), $\hat\Delta_{c,\textsc{ms}}$ (MS), $\hat\Delta_{c,\textsc{if}}$ (IF), $\hat\Delta_{c,\textsc{ifh}}$ (IFH). 0 and 1 indicated non-targeted and targeted nuisance estimation, respectively.}
        \label{fig:15}
        \centering
        \includegraphics[width=1.5\textwidth]{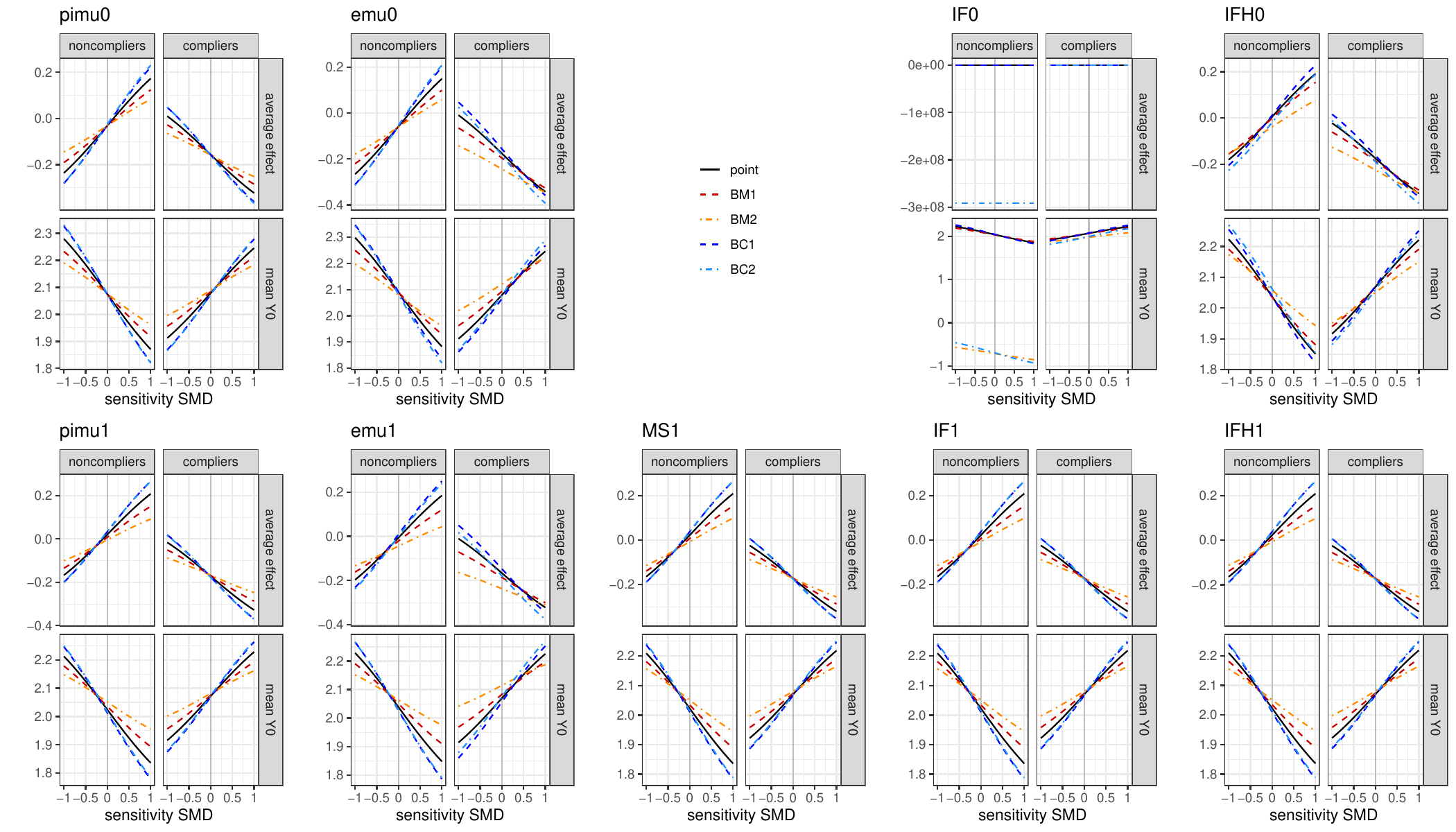}
    \end{figure}
\end{landscape}

\begin{landscape}
    \begin{figure}
        \caption{JOBS II result  -- \textit{earnings}, MR-based: Finite sample bias pattern with estimators $\hat\Delta_{c,\pi\mu}$ (pimu), $\hat\Delta_{c,e\mu}$ (emu), $\hat\Delta_{c,\textsc{ms}}$ (MS), $\hat\Delta_{c,\textsc{if}}$ (IF), $\hat\Delta_{c,\textsc{ifh}}$ (IFH). 0 and 1 indicated non-targeted and targeted nuisance estimation, respectively.}
        \label{fig:16}
        \centering
        \includegraphics[width=1.5\textwidth]{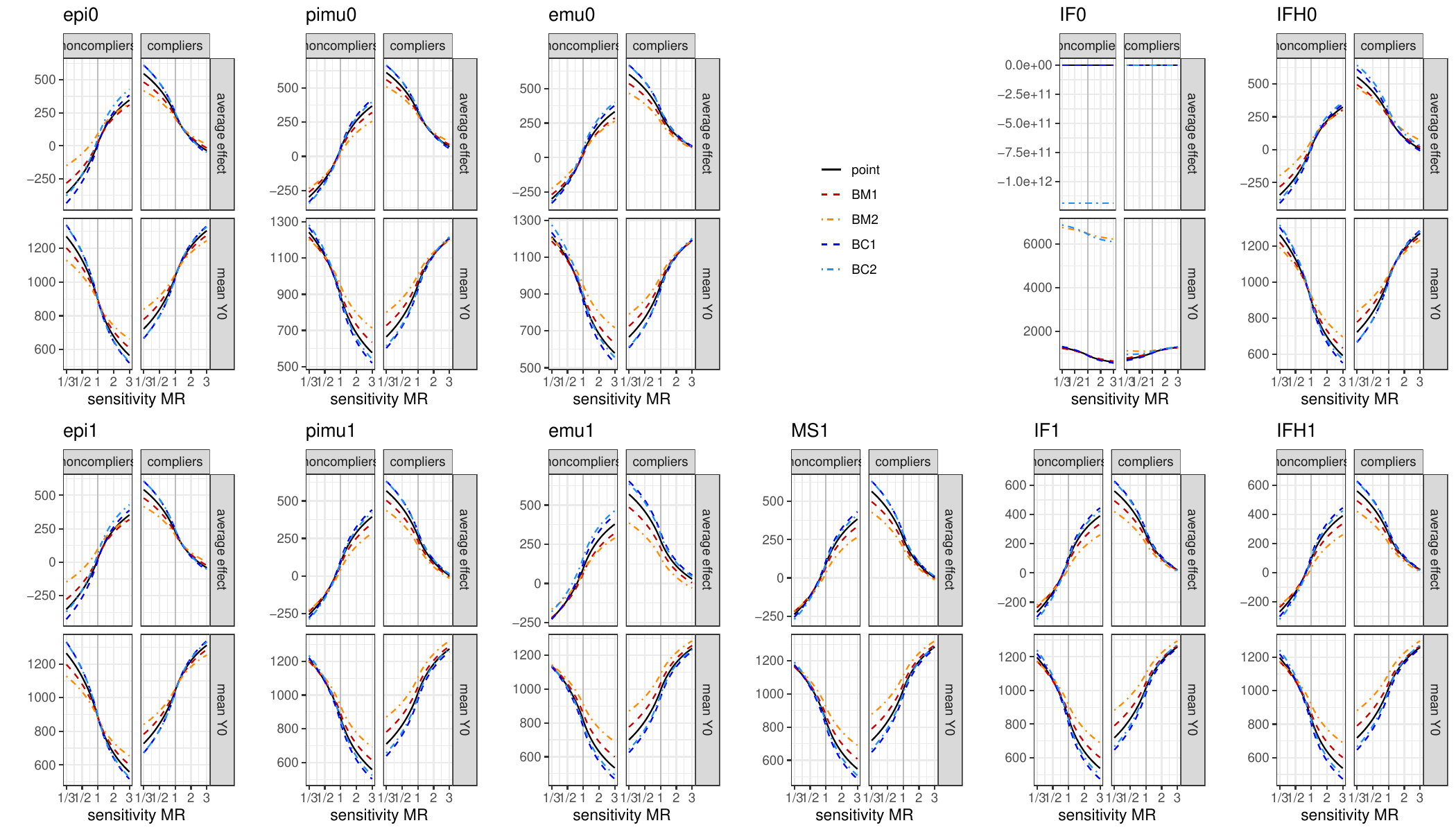}
    \end{figure}
\end{landscape}

\subsection{Proposition~\ref{thm:mu.0c-monotone-in-pi.c} -- $\mu_{0c}(X)$ monotone in $\pi_c(X)$ in sensitivity analysis}\label{appendix-sec:mu.0c-monotone-in-pi.c}

\begin{theorem}[$\mu_{0c}(X)$ monotone in $\pi_c(X)$ in sensitivity analysis]\label{thm:mu.0c-monotone-in-pi.c}
    Consider $\mu_{0c}(X)$ as given by the identification result under A0-A2 combined with any one of the sensitivity assumptions A4-OR, A4-GOR, A4-MR (see Proposition~\ref{thm:ratio.params-id}) or A4-SMDe (see Proposition~\ref{thm:smd-id}). It is a function of $\mu_0(X)$, $\pi_c(X)$ and the sensitivity parameter ($\rho$ or $\eta$). Define the symmetric sensitivity parameter, $\rho_c=c\rho+(1-c)/\rho$ or $\eta_c=c\eta-(1-c)\eta$. Then we have: if the symmetric sensitivity parameter is larger than its null value, i.e., $\rho_c>1$ or $\eta_c>0$, then $\mu_{0c}(X)$ is monotone decreasing in $\pi_c(X)$; and if the symmetric sensitivity parameter is smaller than its null value, i.e., $\rho_c<1$ or $\eta_c<0$, then $\mu_{0c}(X)$ is monotone increasing in $\pi_c(X)$.
\end{theorem}

\noindent
\begin{proof}[Proof of Proposition~\ref{thm:mu.0c-monotone-in-pi.c}]
We consider any arbitrary value of $X$ and drop the ``$(X)$'' to ease notation. 
Also to ease notation (in this proof only), we will just write $\mu$ for $\mu_0(X)$, $\pi$ for $\pi_c(X)$, $\rho$ for $\rho_c$ and $\eta$ for $\eta_c$.

As $\mu_{0c}^\text{GOR}(X)$ is a scaled and shifted version of $\mu_{0c}^\text{OR}(X)$, we will just consider the latter, which is simpler. Thus we have three cases, with $\pi\in(0,1)$, $\rho>0,\rho\neq1$ and $\eta\neq0$:

\medskip

\noindent
\begin{tabular}{ll}
    MR case: & $\displaystyle\mu_{0c}^\text{MR}(\pi,\mu,\rho)=\frac{\rho\mu}{(\rho-1)\pi+1}$,~~~$\mu>0$.
    \\[.5em]
    OR case: & $\displaystyle\mu_{0c}^\text{OR}(\pi,\mu,\rho)=\frac{\alpha-\beta}{2(\rho-1)\pi}$,~~~$\alpha:=(\pi+\mu)(\rho-1)+1,\beta:=\sqrt{\alpha^2-4\pi\mu\rho(\rho-1)}$,~~~$\mu\in[0,1]$.
    \\[.5em]
    SMDe case: & $\displaystyle\mu_{0c}^\text{SMDe}(\pi,\mu,\eta)=\mu+\eta\frac{(1-\pi)\sigma_0}{\sqrt{1+\eta^2\pi(1-\pi)}}$.
\end{tabular}

For all three cases, all we need to show is that $\partial\mu_{0c}/\partial\pi$ is of the opposite sign of $(\rho-1)$ or of $\eta$.

\smallskip
    
This is easy to show for the MR and SMDe cases:
\begin{align*}
    \frac{\partial\mu_{0c}^\text{MR}}{\partial\pi}
    &=-\frac{\rho\mu(\rho-1)}{[(\rho-1)\pi+1]^2}
    \\
    &=-(\rho-1)\underbrace{\frac{\rho\mu}{[(\rho-1)\pi+1]^2}}_{\textstyle >0}.
    \\
    \frac{\partial\mu_{0c}^\text{SMDe}}{\partial\pi}
    &=\eta\sigma_0\left[\frac{-1}{[1+\eta^2\pi(1-\pi)]^{1/2}}-\frac{(1-\pi)\eta^2(1-2\pi)}{2[1+\eta^2\pi(1-\pi)]^{3/2}}\right]
    \\
    &=-\eta\frac{\sigma_0}{2[1+\eta^2\pi(1-\pi)]^{3/2}}\left\{2[1+\eta^2\pi(1-\pi)]+(1-\pi)\eta^2(1-2\pi)\right\}
    \\
    &=-\eta\underbrace{\frac{\sigma_0[1+\eta^2(1-\pi)]}{[1+\eta^2\pi(1-\pi)]^{3/2}}}_{\textstyle >0}.
\end{align*}

It is a lot more involved for the OR case.
\begin{align*}
    \partial\alpha/\partial\pi&=\rho-1,
    \\
    \partial\beta/\partial\pi&=\frac{2\alpha\frac{\partial\alpha}{\partial\pi}-4\mu\rho(\rho-1)}{2\beta}
    =\frac{\rho-1}{\beta}(\alpha-2\mu\rho),
    \\
    \frac{\partial}{\partial\pi}(\alpha-\beta)&=-\frac{\rho-1}{\beta}(\alpha-\beta-2\mu\rho)
\end{align*}
so
\begin{align*}
    \frac{\partial\mu_{0c}^\text{OR}}{\partial\pi}
    &=\frac{\frac{\partial}{\partial\pi}(\alpha-\beta)}{2(\rho-1)\pi}-\frac{\alpha-\beta}{2(\rho-1)\pi^2}
    \\
    &=\frac{-\frac{\rho-1}{\beta}(\alpha-\beta-2\mu\rho)}{2(\rho-1)\pi}-\frac{\alpha-\beta}{2(\rho-1)\pi^2}
    \\
    &=-\frac{\pi(\rho-1)(\alpha-\beta-2\mu\rho)+(\alpha-\beta)\beta}{2\pi^2\beta(\rho-1)}.
\end{align*}
The expression above has a negative sign in front, and the denominator is of the same sign as $\rho-1$, so we want to show that the numerator is always positive. We work with the numerator,
\begin{align*}
    \mathcal{N}
    &=\pi(\rho-1)(\alpha-\beta-2\mu\rho)+(\alpha-\beta)\beta
    \\
    &=\pi(\rho-1)(\alpha-\beta)-2\mu\rho\pi(\rho-1)+(\alpha-\beta)\beta
    \\
    &=[\alpha-\mu(\rho-1)-1](\alpha-\beta)-2\mu\rho\pi(\rho-1)+(\alpha-\beta)\beta
    \\
    &=\alpha(\alpha-\beta)-[\mu(\rho-1)+1](\alpha-\beta)-2\mu\rho\pi(\rho-1)+(\alpha-\beta)\beta
    \\
    &=(\alpha+\beta)(\alpha-\beta)-[\mu(\rho-1)+1](\alpha-\beta)-2\mu\rho\pi(\rho-1)
    \\
    &=4\mu\rho\pi(\rho-1)-[\mu(\rho-1)+1](\alpha-\beta)-2\mu\rho\pi(\rho-1)
    \\
    &=2\mu\rho\pi(\rho-1)-(\alpha-\beta)[\mu(\rho-1)+1].
\end{align*}
To show $\mathcal{N}>0$, we examine the behavior of $\mathcal{N}$ as a function of $\pi$. To see the function better, we put aside the constraint of the range of $\pi$, and just treat it as a generic variable on the real line. To ease notation, let $a=\mu\rho$, $b=1-\mu$; these do not involve $\pi$. Then
\begin{align*}
    \mathcal{N}
    &=2a\pi(\rho-1)-(\alpha-\beta)(a+b),
    \\
    \alpha
    &=\pi(\rho-1)+a+b,
    \\
    \beta
    &=\sqrt{\alpha^2-4a\pi(\rho-1)}
\end{align*}
\begin{align*}
    \mathcal{N}':=\frac{\partial}{\partial\pi}\mathcal{N}
    &=2a(\rho-1)-(a+b)\frac{\partial}{\partial\pi}(\alpha-\beta)
    \\
    &=2a(\rho-1)+(a+b)\frac{\rho-1}{\beta}(\alpha-\beta-2a)
    \\
    &=\frac{\rho-1}{\beta}\left[2\beta a+(a+b)(\alpha-\beta-2a)\right]
    \\
    &=\frac{\rho-1}{\beta}[\beta(a-b)+(a+b)(\alpha-2a)]
    \\
    &=\frac{\rho-1}{\beta}\{\beta(a-b)+(a+b)[\pi(\rho-1)-(a-b)]\}
    \\
    &=\frac{\rho-1}{\beta}\{\beta(a-b)+(a+b)\pi(\rho-1)-(a+b)(a-b)\}.
\end{align*}
We set $\mathcal{N}'$ to 0 and solve for $\pi$ to find the critical points of this function. At this point, it helps to use shorthand notation $u=\pi(\rho-1)$, so we need to solve for $u$. Dropping $(\rho-1)/\beta$ (which is non-zero), isolating the term with $\beta$, and plugging in
\begin{align*}
    \beta
    &=\sqrt{(u+a+b)^2-4au}
    \\
    &=\sqrt{u^2+2(a+b)u+(a+b)^2-4au}
    \\
    &=\sqrt{u^2-2(a-b)u+(a+b)^2}
\end{align*}
obtains
\begin{align}
    \sqrt{u^2-2(a-b)u+(a+b)^2}(a-b)=-u(a+b)+(a+b)(a-b).\label{eq:tocheck}
\end{align}
To make progress, we square both sides of the equation to undo the square root. Because squaring is a two-to-one operation, after solving for $u$ we will need to check the solutions against (\ref{eq:tocheck}).
\begin{align*}
    [u^2-2(a-b)u+(a+b)^2](a-b)^2=u^2(a+b)^2-2u(a+b)^2(a-b)+(a+b)^2(a-b)^2.
\end{align*}
Moving everything to the RHS,
we have
\begin{align*}
    u^2[(a+b)^2-(a-b)^2]-2u(a-b)[(a+b)^2-(a-b)^2]+(a+b)^2(a-b)^2-(a+b)^2(a-b)^2=0,
\end{align*}
which simplifies to
\begin{align*}
    u^2-2(a-b)u=0,
\end{align*}
which has two solutions $u_1=0$, $u_2=2(a-b)$. $u_1$ satisfies (\ref{eq:tocheck}) but $u_2$ does not. Hence the equation $\mathcal{N}=0$ has a unique solution $\pi=0$ (implied by $u_1$).

So we know that the function $\mathcal{N}(\pi)$ has one critical point at $\pi=0$. At this point, $\mathcal{N}=0$. Now we take the second derivative.
\begin{align*}
    \mathcal{N}'':=\frac{\partial}{\partial\pi}\mathcal{N}'
    &=(\rho-1)\left\{\frac{(a-b)\frac{\partial\beta}{\partial\pi}+(a+b)(\rho-1)}{\beta}-\frac{[\beta(a-b)+(a+b)\pi(\rho-1)-(a+b)(a-b)]\frac{\partial\beta}{\partial\pi}}{\beta^2}\right\}
    \\
    &=\frac{\rho-1}{\beta^2}\left\{\beta(a-b)\frac{\partial\beta}{\partial\pi}+\beta(a+b)(\rho-1)-[\beta(a-b)+(a+b)\pi(\rho-1)-(a+b)(a-b)]\frac{\partial\beta}{\partial\pi}\right\}
    \\
    &=\frac{\rho-1}{\beta^2}\left\{\beta(a+b)(\rho-1)-[(a+b)\pi(\rho-1)-(a+b)(a-b)]\frac{\partial\beta}{\partial\pi}\right\}
    \\
    &=\frac{(\rho-1)(a+b)}{\beta^2}\left\{\beta(\rho-1)-[\pi(\rho-1)-(a-b)]\frac{\partial\beta}{\partial\pi}\right\}
    \\
    &=\frac{(\rho-1)(a+b)}{\beta^2}\left\{\beta(\rho-1)-[\pi(\rho-1)-(a-b)]\frac{(\rho-1)[\pi(\rho-1)-(a-b)]}{\beta}\right\}
    \\
    &=\frac{(\rho-1)^2}{\beta^3}(a+b)\left\{\beta^2-[\pi(\rho-1)-(a-b)]^2\right\}
    \\
    &=\frac{(\rho-1)^2}{\beta^3}(a+b)\left\{[\pi(\rho-1)+(a+b)]^2-4a\pi(\rho-1)-[\pi(\rho-1)-(a-b)]^2\right\}
    \\
    &=\frac{(\rho-1)^2}{\beta^3}(a+b)\left\{-4a\pi(\rho-1)+[\pi(\rho-1)+(a+b)]^2-[\pi(\rho-1)-(a-b)]^2\right\}
    \\
    &=\frac{(\rho-1)^2}{\beta^3}(a+b)\left\{-4a\pi(\rho-1)+[\pi(\rho-1)+(a+b)]^2-[\pi(\rho-1)-(a-b)]^2\right\}
    \\
    &=\frac{(\rho-1)^2}{\beta^3}(a+b)\left\{-4a\pi(\rho-1)+[2\pi(\rho-1)+2b]2a\right\}
    \\
    &=\frac{(\rho-1)^2}{\beta^3}(a+b)(4ab)
    \\
    &=\frac{(\rho-1)^2}{\beta^3}[\mu\rho+(1-\mu)]4\mu\rho(1-\mu)
    \\
    &>0.
\end{align*}
This means $\mathcal{N}(\pi)$ is a convex function. This implies $\pi=0$ is the minimal point. It follows that on the interval of interest, $(0,1)$, this function is greater than $\mathcal{N}(0)=0$. 

We thus have shown that $\displaystyle\frac{\partial}{\partial\pi}\mu_{0c}^\text{OR}$ is of the opposite sign of $(\rho_c-1)$, completing the proof. 
\end{proof}

\subsection{Bias correction -- a limited simulation study}\label{appendix-sec:bias-correction}

\subsubsection{Bootstrap-based bias corrections}
\label{assec:bias-corrections}

A practical question then is whether the finite-sample bias seen in the sensitivity analyses can and should be corrected. Two bootstrap-based bias correction techniques\cite{efron1994IntroductionBootstrap,Chang2015} are $\hat\theta_\textsc{bc1}:=2(\hat\theta-\bar{\hat\theta}^*)+\bar{\hat\theta}^*$ based on the single bootstrap, and $\hat\theta_\textsc{bc2}:=3(\hat\theta-\bar{\hat\theta}^*)+\bar{\hat\theta}^{**}$ based on the iterated bootstrap. Figure \ref{fig:finitesamplebias} shows BC1 and BC2 results in two shades of blue. (For BC2, we apply the warp-speed version in \cite{Chang2015} that uses single double bootstrap draws.) But is it beneficial to use either correction? To answer this question, we conduct a limited simulation study (see below) based on JOBS II data, generating five hundred samples of the same size as our analytic sample, and implementing both corrections. Results show that both corrections reduce bias and only slightly increase variance; and that the two corrections result in similar estimates. 
From a practical perspective, BC1 is simpler, as the standard bootstrap is used for bias correction in addition to confidence interval estimation at no additional cost. 
For the JOBS II analysis in Section \ref{sec:illustration} we show BC1 results.

\smallskip

A note of caution: As this kind of bias correction is mean-based, we do not recommend it for unbounded estimators, because $\bar{\hat\theta}^*$ and $\bar{\hat\theta}^{**}$ may be influenced by extreme bootstrap estimates.

\smallskip

\begin{remark}
    With BC point estimates and BCa confidence intervals, both our point and interval estimates are bias-corrected.
    Interestingly, while the bias correction methods are different (mean-based for point estimates and quantile-based for intervals), BC point estimates seem to be more centered (than uncorrected estimates) in the BCa intervals.
\end{remark}

\subsubsection{The simulation study}

The purpose of this limited simulation study is to examine the bias-reduction impact of BC1 and BC2 as well as how these corrections affect the variance of the estimator. Our focus is on finite-sample bias, which is the departure of the expectation of the estimator (in a finite sample) from the estimator's probability limit. We are not concerned here about the estimator's asymptotic bias, which is the departure of the estimator's probability limit from the true parameter.

\paragraph{Data generating model.}

The data generating model mimics the joint distribution of variables observed in the analytic JOBS II sample. With the factorization
\begin{align*}
    f(O)&=f_1(X,Z)\times
    \\
    &~~~~~[f_2(C\mid X,Z=1)]^Z\times
    \\
    &~~~~~[f_3(Y\mid X,Z=1,C=1)]^{ZC}[f_4(Y\mid X,Z=1,C=0)]^{Z(1-C)}[f_5(Y\mid X,Z=0)]^{1-Z},
\end{align*}
we simulate variables in the following order.

First, we use the R-packge synthpop \cite{nowok2016SynthpopBespokeCreation} to simulate $(X,Z)$ based on JOBS II $(X,Z)$ data. The package uses nonparametric (classification and regression tree) methods to generate synthetic samples that reflect the joint distribution in the provided data. 

Then we generate $(C,Y)$ conditional on $(X,Z)$ based on models fit to JOBS II data. $C$ is generated as a Bernoulli random variable based on the fitted principal score model. The outcome \textit{work for pay} is also generated Bernoulli based on the fitted outcome models. The outcome \textit{depressive symptoms} is generated as a (scaled-and-shifted) beta random variable using the fitted mean and dispersion from the fitted outcome models. The outcome \textit{earnings} is generated conditional on working, but not using the Gamma outcome models from the analysis, because these models treat the outcome as unbounded, so using them would generate earnings that look quite different from JOBS II earnings. To respect the observed data range, we generate this variable as a (scaled-and-shifted) beta random variable, similar to \textit{depressive symptoms}; this is based on models newly fit for this purpose.%
\footnote{These models use the generalized logit link (like the models for \textit{depressive symptoms}), losing the log link feature of the analysis models. It would be preferable to generate \textit{earnings} from a truncated Gamma model fit with log link. Unfortunately, fitting a truncated Gamma model, especially with covariates, is a hard problem.}

The latter means that for \textit{earnings} the analysis of simulated data (which uses the same models as the real data analysis) is based on miss-specified outcome models. This is appropriate, as the real data analysis for this outcome is also based on miss-specified models. While this results in variation where the analyses for two outcomes are correctly specified and for one outcome is not, this detail is not of central interest here, because we are now only concerned with finite-sample, not asymptotic, bias.

\begin{figure}[t!]
    \caption{Simulation results 1: Pattern of point estimate, bootstrap mean estimates and bias-corrected estimates, shown in expectation (i.e., averaged over simulated datasets), with true value benchmark.}
    \label{fig:sim-pattern}
    \includegraphics[width=1.05\textwidth]{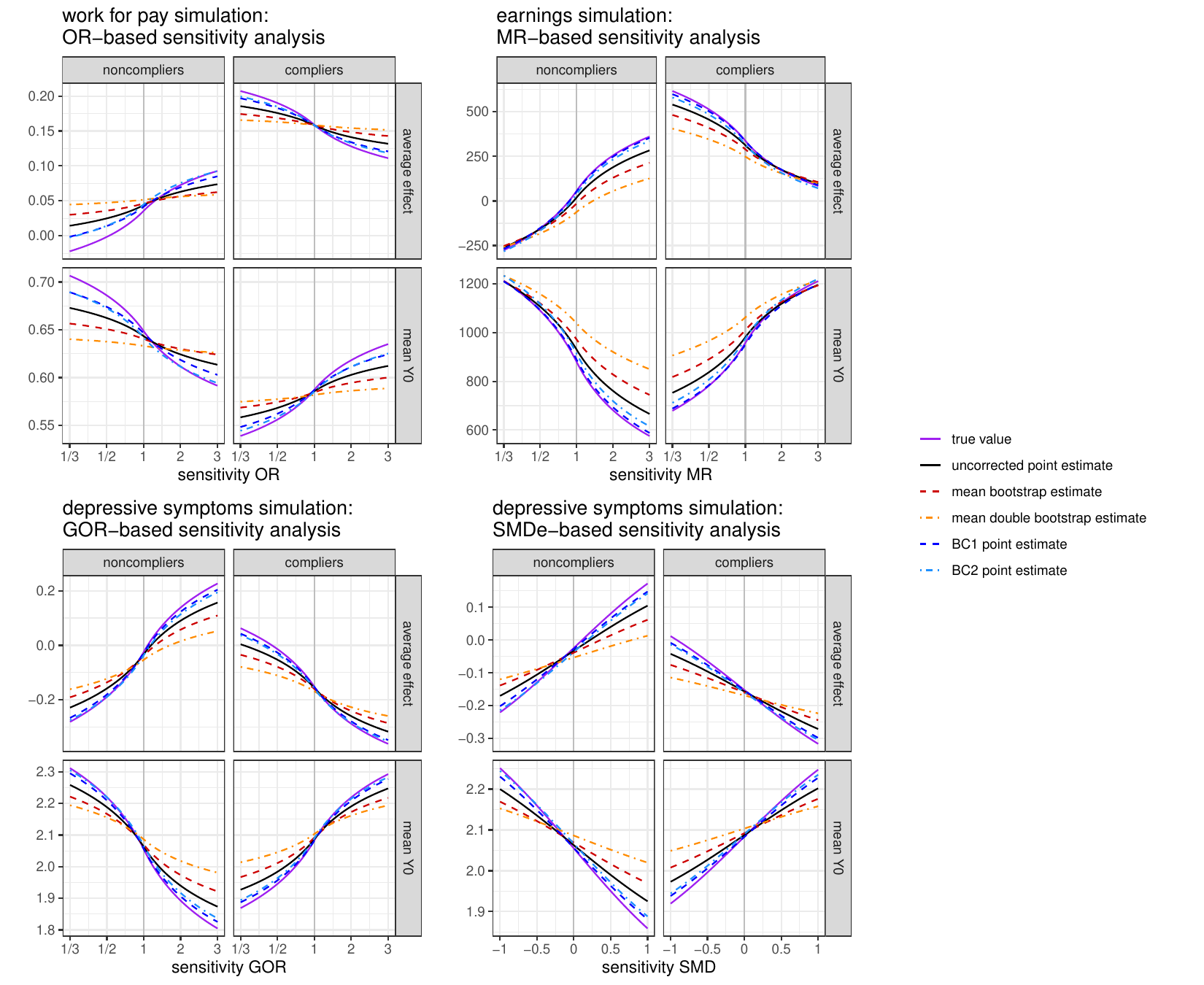}
\end{figure}

\paragraph{True value approximation.}

Since the probability limit of the estimator is only available at $n=\infty$, we approximate it by an estimate of the expectation of the estimator at $n=1,000,000$ (by averaging estimates from 20 samples). This sample size cap ensures that our computer memory can handle the model fitting. We loosely refer to this approximate value as the probability limit itself (ignoring the remaining finite-sample bias, which is likely minimal), and denote it by $\theta$.

\paragraph{Data simulation and analysis of simulated data.}

We draw $K=500$ samples of size $n=465$, the size of our JOBS II analytic sample. On each sample $k\in\{1,\dots,K\}$, we obtain the point estimate ($\hat\theta_k$) from the estimator, the bootstrap and double bootstrap mean estimates ($\bar{\hat\theta}_k^*$ and $\bar{\hat\theta}_k^{**}$), and the bootstrap-based bias-corrected estimates ($\hat\theta_{\textsc{bc1},k}$ and $\hat\theta_{\textsc{bc2},k}$). 

\paragraph{Pattern of estimates.}

Figure \ref{fig:sim-pattern} uses the same structure as Figure \ref{fig:finitesamplebias-all-noCI}, except here instead of showing results from a single dataset, the figure shows the averages over the simulated datasets of the point estimate, the bootstrap mean estimates and the bias-corrected estimates. Specifically,
$$\bar{\hat\theta}_0=\frac{1}{K}\sum_{k=1}^{K}\hat\theta_k,~~~
\overline{\bar{\hat\theta}^*}=\frac{1}{K}\sum_{k=1}^{K}\bar{\hat\theta}_k^*,~~~
\overline{\bar{\hat\theta}^{**}}=\frac{1}{K}\sum_{k=1}^{K}\bar{\hat\theta}_k^{**},~~~
\bar{\hat\theta}_\textsc{bc1}=\frac{1}{K}\sum_{k=1}^{K}\hat\theta_{\textsc{bc1},k},~~~
\bar{\hat\theta}_\textsc{bc2}=\frac{1}{K}\sum_{k=1}^{K}\hat\theta_{\textsc{bc2},k}.$$
Also shown in Figure \ref{fig:sim-pattern} is the true value $\theta$.

All the plots in Figure \ref{fig:sim-pattern} show that going from the point estimate to the mean bootstrap estimate to the mean double bootstrap estimate, in expectation, we increasingly depart from the true value. They also show that the bias-corrections pull back closer to the true value.

The plots in Figure \ref{fig:sim-pattern} look similar to those in Figure \ref{fig:finitesamplebias-all-noCI}. This is not surprising, as the simulation is based on JOBS II data.

\paragraph{Bias correction performance.}

We estimate bias, standard error (SE), standardized bias, and standardized SE increase as
\begin{alignat*}{4}
    &\widehat{\text{bias}}_0=\bar{\hat\theta}-\theta,
    ~~~&&\widehat{\text{SE}}_0=\left[\frac{\sum_k(\hat\theta_k-\bar{\hat\theta})^2}{K-1}\right]^{1/2},
    ~~~&&\widehat{\text{std.bias}}_0=\frac{\widehat{\text{bias}}_0}{\widehat{\text{SE}}_0},
    \\
    &\widehat{\text{bias}}_1=\bar{\hat\theta}_\textsc{bc1}-\theta,
    ~~~&&\widehat{\text{SE}}_1=\left[\frac{\sum_k(\hat\theta_{\textsc{bc1},k}-\bar{\hat\theta}_\textsc{bc1})^2}{K-1}\right]^{1/2},
    ~~~&&\widehat{\text{std.bias}}_1=\frac{\widehat{\text{bias}}_1}{\widehat{\text{SE}}_0},
    ~~~&&\widehat{\text{std.SE.diff}}_1=\frac{\widehat{\text{SE}}_1-\widehat{\text{SE}}_0}{\widehat{\text{SE}}_0},
    \\
    &\widehat{\text{bias}}_2=\bar{\hat\theta}_\textsc{bc2}-\theta,
    ~~~&&\widehat{\text{SE}}_2=\left[\frac{\sum_k(\hat\theta_{\textsc{bc2},k}-\bar{\hat\theta}_\textsc{bc2})^2}{K-1}\right]^{1/2},
    ~~~&&\widehat{\text{std.bias}}_2=\frac{\widehat{\text{bias}}_2}{\widehat{\text{SE}}_0},
    ~~~&&\widehat{\text{std.SE.diff}}_2=\frac{\widehat{\text{SE}}_2-\widehat{\text{SE}}_0}{\widehat{\text{SE}}_0}.
\end{alignat*}
All standardization uses the same denominator (SE of the uncorrected $\hat\theta_k$) to put the metrics on the same scale to facilitate comparison.
Figures \ref{fig:sim-work}, \ref{fig:sim-depress-GOR}, \ref{fig:sim-depress-SMDe} and \ref{fig:sim-earn} show these results for the four sensitivity analyses (OR-based for \textit{work for pay}, GOR- and SMDe-based for \textit{depressive symptoms}, and MR-based for \textit{earnings}), including bias, standardized bias and standardized SE change. The $y$-axis scale of the bias plot is ten times that of the analysis result plot in Figure \ref{fig:meanbased-sens}.

In these simulations both bias corrections reduce bias while only slightly increase variance. Also, BC1 and BC2 have similar performance.

\begin{figure}[h!]
    \caption{Simulation results 2a (\textit{work for pay}, OR-based): Bias and standardized bias before and after bias correction; standardized standard error change due to bias correction.}
    \label{fig:sim-work}
    \includegraphics[width=\textwidth]{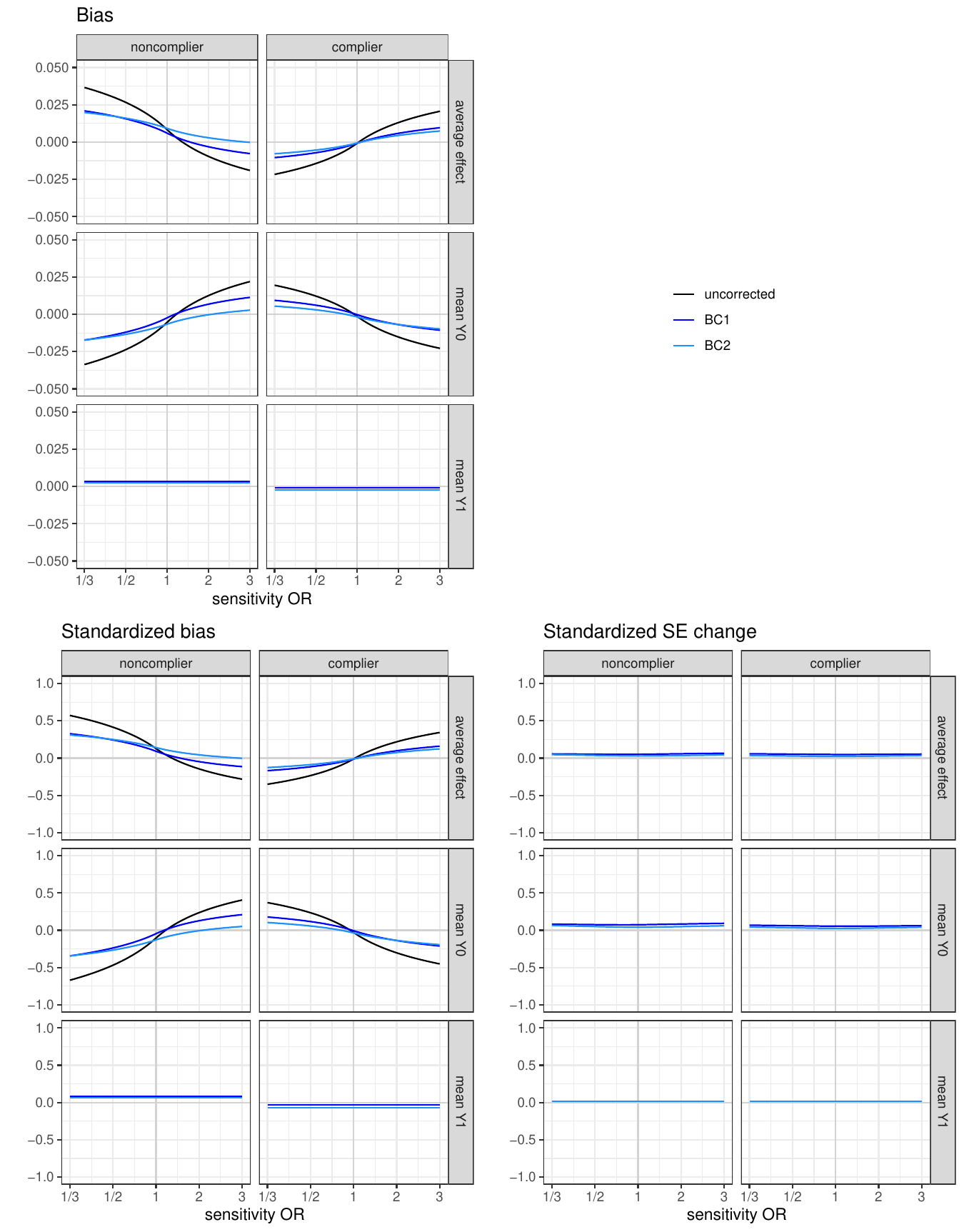}
\end{figure}

\begin{figure}[h!]
    \caption{Simulation results 2b (\textit{depressive symptoms}, GOR-based): Bias and standardized bias before and after bias correction; standardized standard error change due to bias correction.}
    \label{fig:sim-depress-GOR}
    \includegraphics[width=\textwidth]{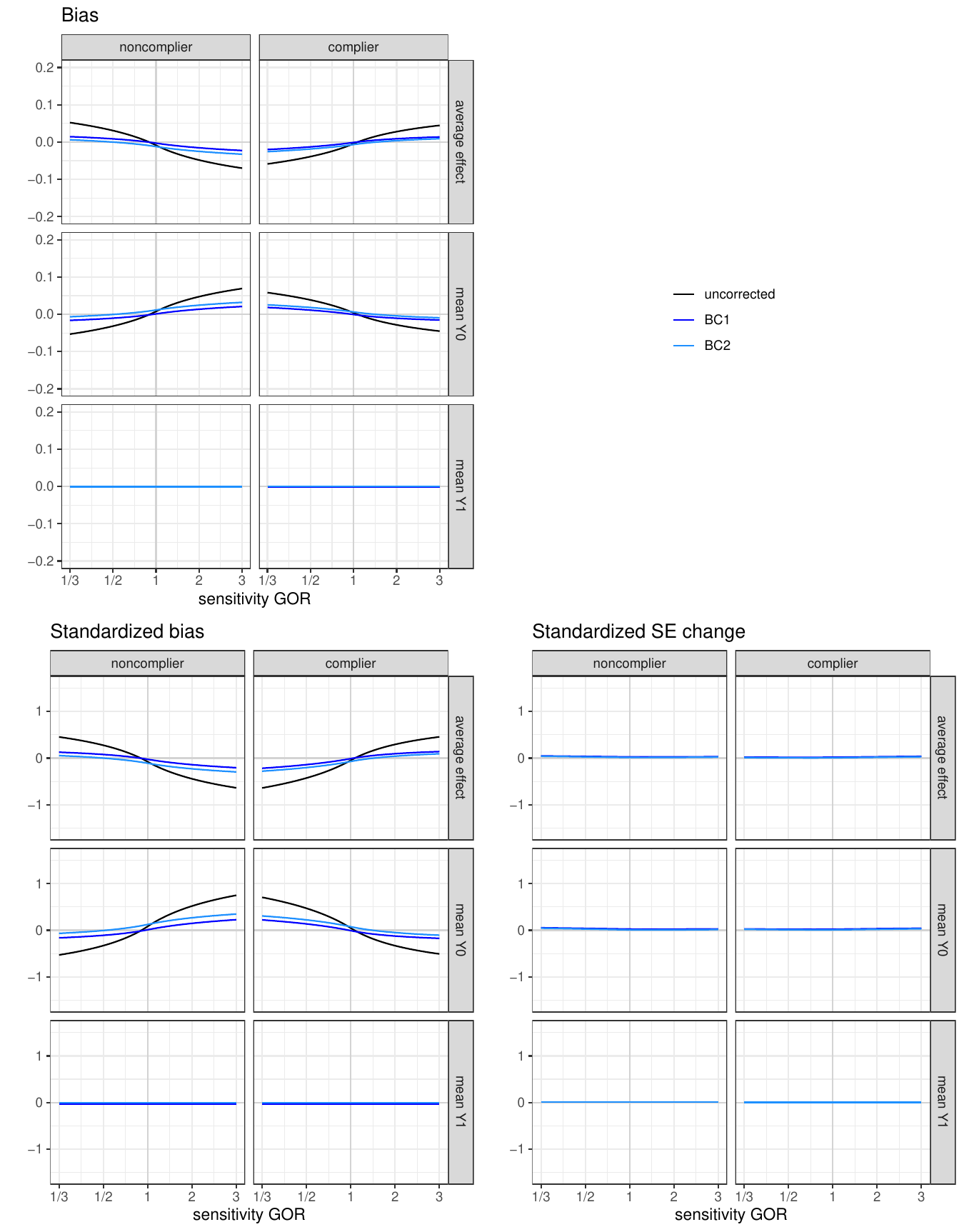}
\end{figure}

\begin{figure}[h!]
    \caption{Simulation results 2c (\textit{depressive symptoms}, SMDe-based): Bias and standardized bias before and after bias correction; standardized standard error change due to bias correction.}
    \label{fig:sim-depress-SMDe}
    \includegraphics[width=\textwidth]{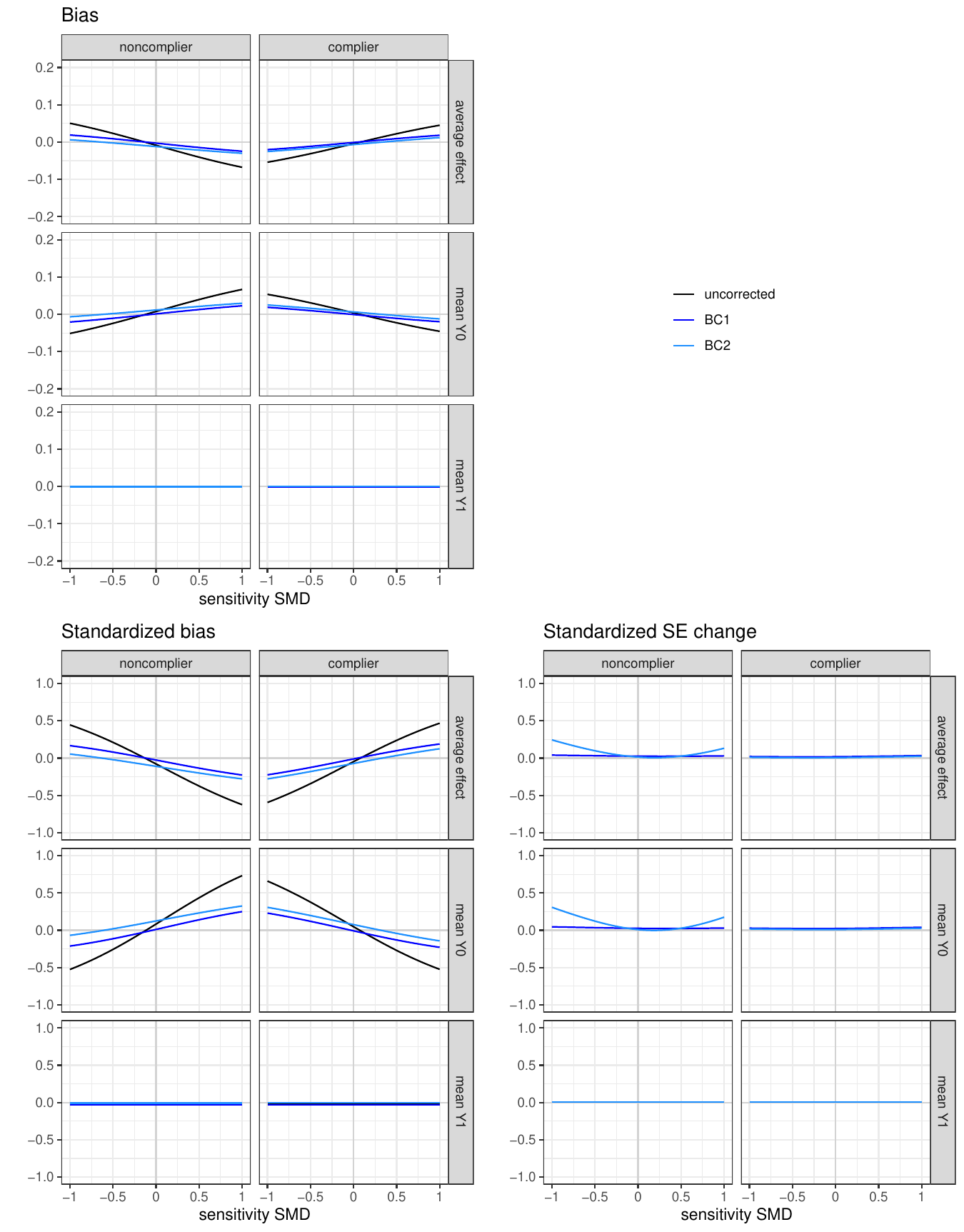}
\end{figure}

\begin{figure}[h!]
    \caption{Simulation results 2d (\textit{earnings}, MR-based): Bias and standardized bias before and after bias correction; standardized standard error change due to bias correction.}
    \label{fig:sim-earn}
    \includegraphics[width=\textwidth]{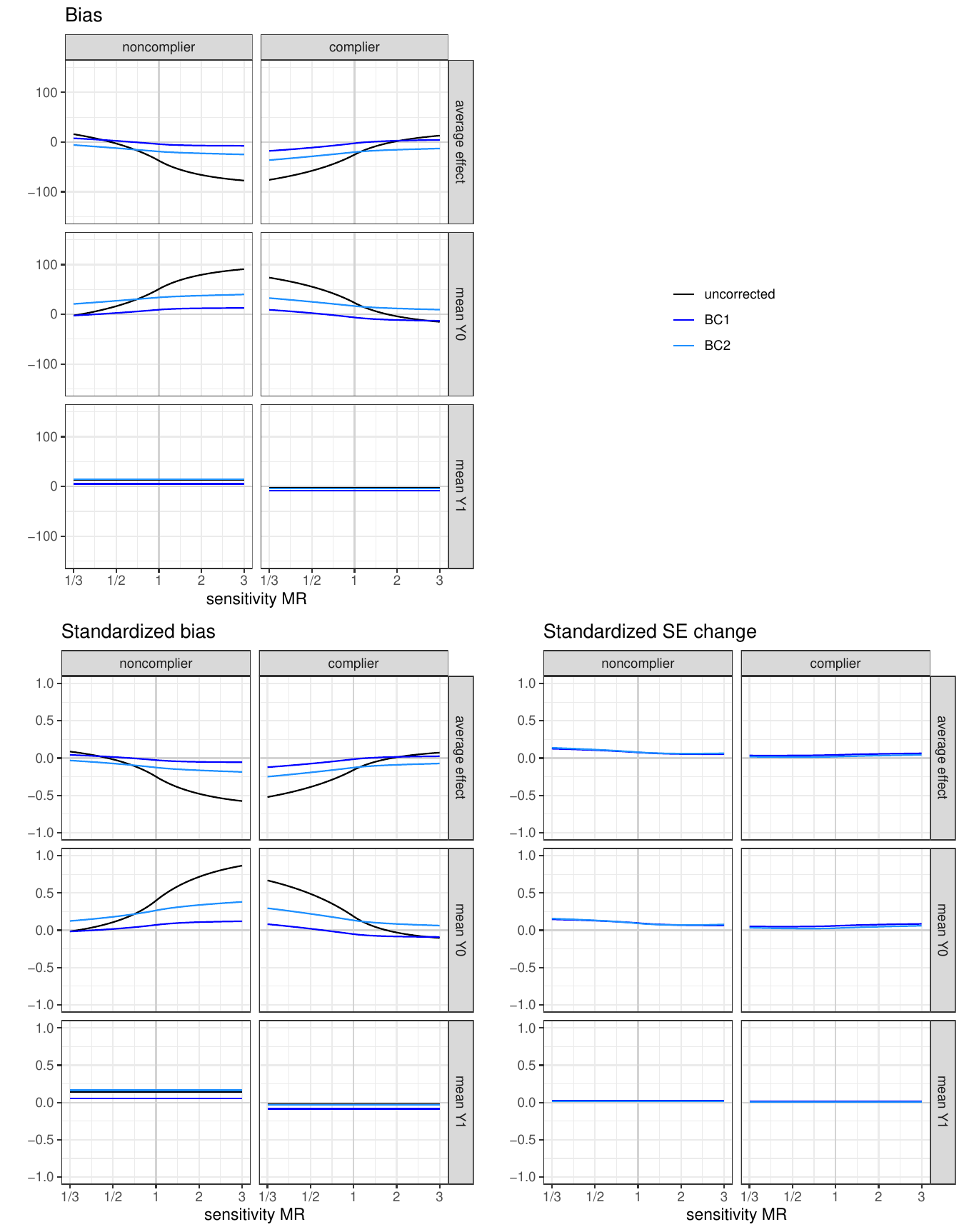}
\end{figure}

\newpage
~
\newpage
~
\newpage
~
\newpage
~
\newpage
\section{Additional figures for Section~\ref{sec:illustration} -- JOBS II illustrative analysis}\label{appendix:example}

\begin{figure}[h]
    \caption{Covariate balance before/after propensity score weighting (left), before/after combined propensity-and-principal score weighting for CACE (middle) and for NACE (right)}\label{fig:love.plots}
    \includegraphics[width=.325\textwidth]{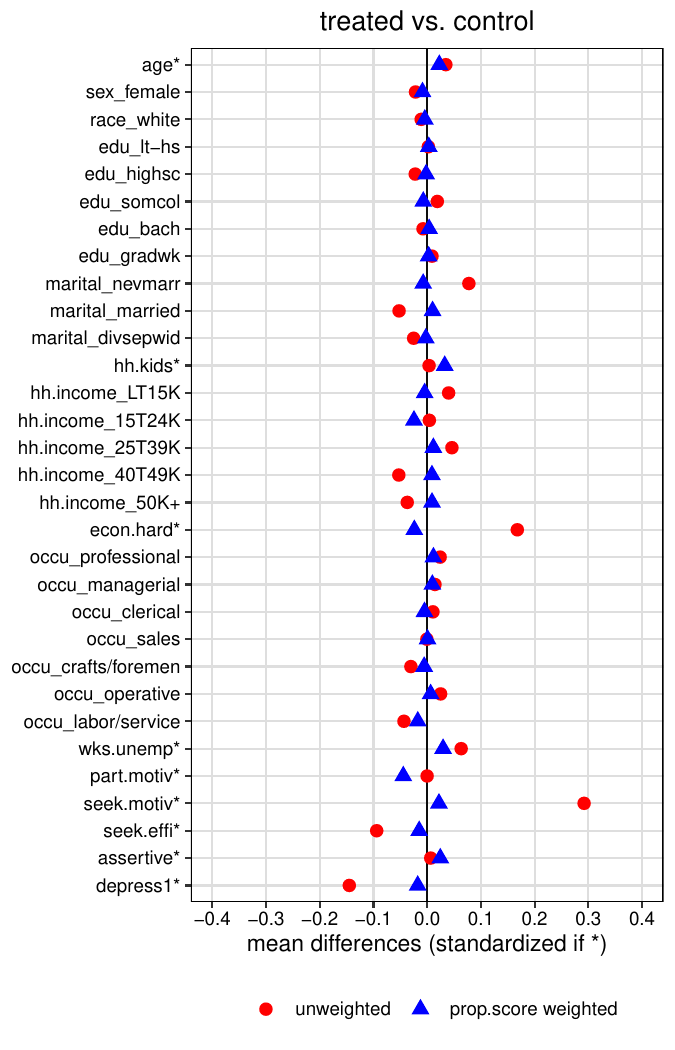}
    \includegraphics[width=.325\textwidth]{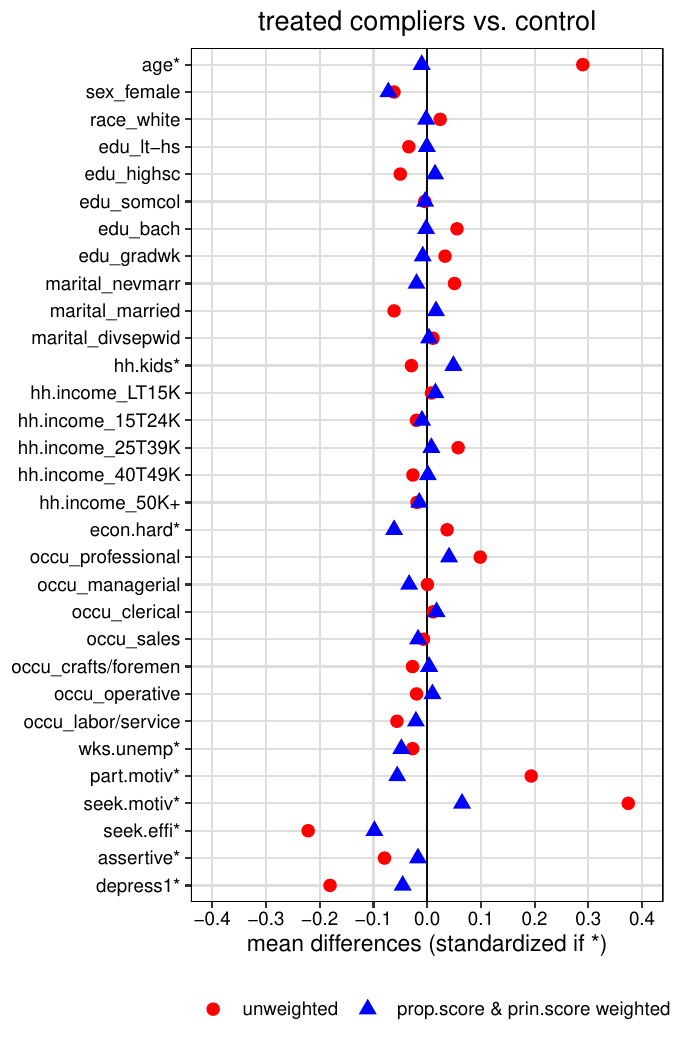}
    \includegraphics[width=.325\textwidth]{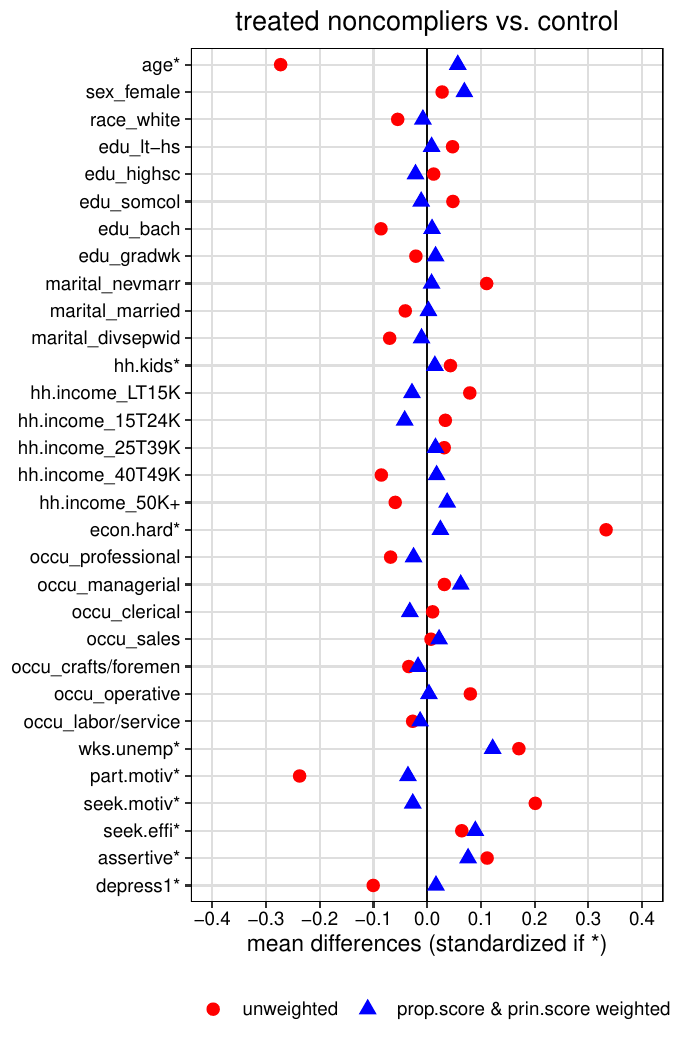}
\end{figure}


\begin{figure}[h]
    \caption{Implications of different sensitivity MR values about stratum-specific distributions of conditional mean $Y_0$ values, i.e., $\mu_{01}(X)$ values among compliers and $\mu_{00}(X)$ values among noncompliers}\label{fig:MRimplication}
    \includegraphics[width=\textwidth]{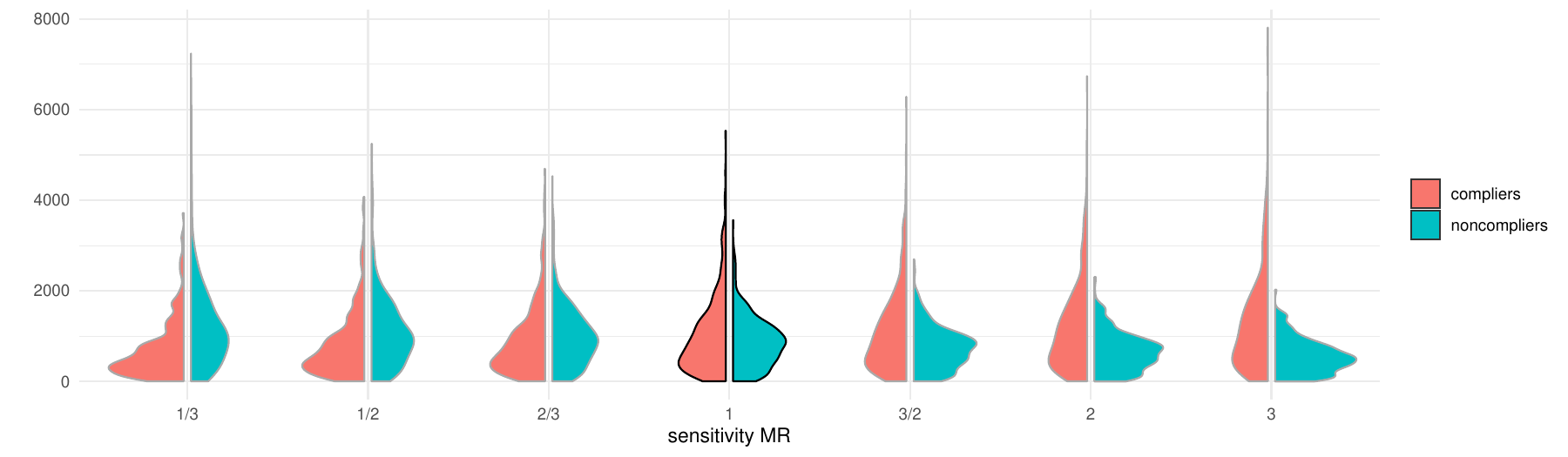}
\end{figure}

\newpage

\end{document}